\documentclass[12pt]{article}
\usepackage{amsmath,amsthm,amssymb}
\usepackage[margin=1in]{geometry}
\usepackage[numbers,square,sort]{natbib}
\usepackage{mathrsfs}
\usepackage{booktabs}
\usepackage{nicefrac}
\usepackage{array}
\usepackage{varwidth}

\usepackage{fontawesome5}
\usepackage{tabularx}
\usepackage{thm-restate}
\usepackage{bbm}
\usepackage{graphicx}
\usepackage{subcaption}
\usepackage{caption}
\usepackage{enumitem}
\usepackage{wrapfig}
\usepackage{tikz}
\usepackage{rotating}
\usetikzlibrary{positioning, calc, arrows.meta, shadows.blur, shapes.geometric,
decorations.pathreplacing, calligraphy, fit, backgrounds, external}

\usepackage{tcolorbox}
\definecolor{systemcolor}{HTML}{A9A9A9}
\definecolor{usercolor}{HTML}{2980B9}
\definecolor{assistantcolor}{HTML}{27AE60}
\newtcolorbox{systemmessagebox}{
  colback=systemcolor!10, colframe=systemcolor,
  sharp corners, boxrule=1pt, left=5pt, right=5pt, top=5pt, bottom=5pt,
  fonttitle=\bfseries, title=System
}

\newtcolorbox{usermessagebox}{
  colback=usercolor!10, colframe=usercolor,
  sharp corners, boxrule=1pt, left=5pt, right=5pt, top=5pt, bottom=5pt,
  fonttitle=\bfseries, title=User
}

\newtcolorbox{assistantmessagebox}{
  colback=assistantcolor!10, colframe=assistantcolor,
  sharp corners, boxrule=1pt, left=5pt, right=5pt, top=5pt, bottom=5pt,
  fonttitle=\bfseries, title=Assistant
}

\usepackage{hyperref}
\usepackage{cleveref}
\usepackage{xcolor}
\def\R{\mathbb{R}}
\def\N{\mathbb{N}}
\def\STOP{\texttt{STOP}}

\newtheorem{theorem}{Theorem}[section]
\newtheorem{lemma}[theorem]{Lemma}
\newtheorem{proposition}[theorem]{Proposition}

\theoremstyle{definition}
\newtheorem{definition}[theorem]{Definition}

\def\gemma{\textsf{gemma2}}
\def\llamaone{\textsf{llama3.1}}
\def\llamatwo{\textsf{llama3.2}}
\def\mistral{\textsf{mistral}}
\def\nemotron{\textsf{nemotron}}
\def\phithree{\textsf{phi3.5}}
\def\qwen{\textsf{qwen2.5}}
\def\cogview{\textsf{cogview4}}
\def\flux{\textsf{FLUX.1-dev}}
\def\pixart{\textsf{pixart-sigma}}
\def\sd{\textsf{stable-diffusion3.5}}

\def\lpips{\textsf{LPIPS}}

\def\dreamsim{\textsf{DreamSim}}

\def\game{\mathscr{G}}
\def\eqset{\mathscr{E}}
\newcommand\eqsetn[1]{\mathscr{E}(#1, \bd, s_\gamma)}

\newcommand{\WI}{\mathsf{WI}}
\newcommand{\WIargs}[2]{\WI(#1 \, \| \, #2)}
\newcommand{\WIavgargs}[2]{\WI^{\text{avg}}(#1, #2)}

\def\dTV{d_{\text{TV}}}

\newcommand{\Var}{\mathrm{Var}}

\def\taueq{\tau_{\textsc{eq}}}
\def\tauopt{\tau_{\textsc{opt}}}
\def\popt{\bp_{\textsc{opt}}}
\def\tpopt{\tilde{\bp}_{\textsc{opt}}}
\def\peq{\bp_{\textsc{eq}}}
\def\dgam{\,d \gamma}
\newcommand{\dpat}[1]{\Big|_{p=#1}}
\def\maxm{\overline{\mathcal{M}}}
\def\minm{\underline{\mathcal{M}}}
\def\sup{\mathcal{S}_{\scalebox{0.4}{$\nearrow$}}}
\def\sdown{\mathcal{S}_{\scalebox{0.4}{$\searrow$}}}
\def\sall{\mathcal{S}}

\newcommand{\pin}[1]{\pi^{(#1)}}
\def\mc{\mathcal}

\def\bp{\mathbf{p}}
\def\bplim{\bp^{\dagger}(\gamma)}
\def\bq{\mathbf{q}}
\def\bP{\mathbf{P}}

\def\bPp{\mathbf{P}^+}
\def\tbp{\mathbf{\hat{p}}}

\newcommand{\bpn}[1]{\bp^{(#1)}}
\newcommand{\bPn}[1]{\bP^{(#1)}}
\newcommand{\bqn}[1]{\bq^{(#1)}}

\def\gz{\gamma_0}
\def\go{\gamma_1}
\newcommand{\bps}{\bp^*}
\newcommand{\bpg}{\bp^{(\gz)}}
\newcommand{\bpgp}{\bp^{(\go)}}
\newcommand{\pkn}[1]{\bp_k^{(#1)}}
\newcommand{\pn}[1]{\bp^{(#1)}}

\newcommand{\qkn}[1]{\bq_k^{(#1)}}
\newcommand{\qkpn}[1]{\bq_{k'}^{(#1)}}
\newcommand{\qn}[1]{\bq^{(#1)}}

\def\Cn{C^{(n)}}
\def\ckn{c_k^{(n)}}
\def\ckpn{c_{k'}^{(n)}}
\def\dkn{\delta_k^{(n)}}
\def\bd{\mathbf{d}}
\def\ba{\mathbf{a}}
\def\tbd{\mathbf{\tilde{d}}}

\newcommand\barpn[1]{\bar{\bp}^{(#1)}}

\def\du{\bar{u}}
\def\dw{\bar{w}}
\def\df{\bar{f}}
\def\ddf{\bar{\bar{f}}}
\def\dg{\bar{g}}
\def\dz{\bar{z}}
\def\ddg{\bar{\bar{g}}}
\def\ddu{\bar{\bar{u}}}
\def\given{\; | \;}

\newtheorem{assumption}{Assumption}
\Crefname{assumption}{Assumption}{Assumptions}
\newtheorem{hypothesis}{Hypothesis}
\Crefname{hypothesis}{Hypothesis}{Hypotheses}

\newcommand{\p}[1]{\left( #1 \right)}
\renewcommand{\b}[1]{\left[ #1 \right]}
\newcommand{\E}[1]{\mathbb{E} \b{#1}}
\newcommand{\EE}[2]{\mathbb{E}_{#1} \b{#2}}
\newcommand{\ind}[1]{\mathbbm{1}\b{#1}}

\newcommand\numberthis{\addtocounter{equation}{1}\tag{\theequation}}
\DeclareMathOperator{\avg}{avg}
\DeclareMathOperator*{\argmax}{arg\,max}
\DeclareMathOperator*{\argmin}{arg\,min}
\newboolean{smallEqs}
\setboolean{smallEqs}{false}

\title{Competition and Diversity in Generative AI}
\author{Manish Raghavan\thanks{MIT Sloan School of Management \& Department of
Electrical Engineering and Computer Science. \texttt{mragh@mit.edu}}}
\date{}
\newenvironment{acks}
{
  \paragraph*{Acknowledgements.}  
}
{}

\begin{document}
\maketitle
\begin{abstract}

  Recent evidence, both in the lab and in the wild, suggests that the use of generative artificial intelligence reduces the diversity of content produced. The use of the same or similar AI models appears to lead to more homogeneous behavior. Our work begins with the observation that there is a force pushing in the opposite direction: competition. When producers compete with one another (e.g., for customers or attention), they are incentivized to create novel or unique content. We explore the impact competition has on both content diversity and overall social welfare. Through a formal game-theoretic model, we show that competitive markets select for diverse AI models, mitigating monoculture. We further show that a generative AI model that performs well in isolation (i.e., according to a benchmark) may fail to provide value in a competitive market. Our results highlight the importance of evaluating generative AI models across the breadth of their output distributions, particularly when they will be deployed in competitive environments. We validate our results empirically by using language models to play Scattergories, a word game in which players are rewarded for answers that are both correct and unique. Overall, our results suggest that homogenization due to generative AI is unlikely to persist in competitive markets, and instead, competition in downstream markets may drive diversification in AI model development.

\end{abstract}

\section{Introduction}
\label{sec:intro}

A growing body of literature on generative artificial intelligence reveals a
surprisingly consistent stylized fact: when people use generative AI tools, the
set of content they produce tends to be more homogeneous than content produced
by more traditional means~\citep{padmakumar2023does, liang2024monitoring,
anderson2024homogenization, zhang2024generative, zhou2023generative,
doshi2024generative, kirk2023understanding, park2024diminished,
shur2024growing}. Across a wide range of domains including peer
review~\citep{liang2024monitoring}, writing~\citep{padmakumar2023does}, digital
art~\citep{zhou2023generative}, and survey
responses~\citep{zhang2024generative}, access to generative AI tools (GAITs)
leads to less diverse outcomes. Researchers refer to this phenomenon---where the
use of similar or identical underlying AI tools leads to convergence in
outcomes---as \textit{algorithmic monoculture}~\citep{kleinberg2021algorithmic}
or \textit{homogenization}~\citep{bommasani2022picking}.

Much of the empirical evidence on homogenization via generative AI comes from
lab settings, or from relatively recent AI adoption. Should we expect
homogenization to persist over time in the real world? While monoculture indeed
pushes towards homogenization, there is a force pushing in the other direction:
\textit{competition}. In applications like marketing, creative writing, and art,
users of generative AI tools compete with one another for attention and market
share. The primary intuition for our work can be summarized as follows:
\textbf{In markets that reward novelty, competition will mitigate
homogenization and promote diversity.}

In this paper, we formalize this statement and explore its implications. We make
three broad contributions. First, we propose a formal model of algorithmic
monoculture via generative AI in \Cref{sec:model}. Our model captures the idea
that two users of the same generative AI tool are likely to get similar
responses or information.
Notably, the existing theoretical literature on
algorithmic monoculture tends to instantiate it in terms of correlated
errors~\citep{kleinberg2021algorithmic, peng2024wisdom, peng2023monoculture,
jo2025homogeneous}, making it clear how correlated or identical
algorithms induce similar behavior. The same is not obviously true for generative
AI; users have far more degrees of freedom in the inputs they provide, and it is
not clear why two users of the same GAIT should get ``similar'' outcomes.

Second, under this framework, we build and analyze a formal game-theoretic model
of competition via generative AI in \Cref{sec:model,sec:theory,sec:inter-tool}.
At a high level, we formally prove that our intuition holds: stronger
competition leads to more diversity
(\Cref{thm:diversity-n,thm:diversity-gamma}), but equilibria are less diverse
than socially optimal behavior (\Cref{thm:diversity-opt}). Interestingly, we
find that the quality of an AI tool is not one-dimensional; some tools may
perform better in low-competition environments, while others might perform
better in the presence of strong competition (\Cref{lem:dominant-strength}).

We empirically validate our model in \Cref{sec:empirical}.\footnote{Our code is
available at \url{https://github.com/mraghavan/llm-scattergories}.} We find that
our theoretical results hold when using large language models (LLMs) to play a
competitive word game, described below. We find strong evidence that the
``best'' LLM in this setting depends on the level of competition, and moreover,
our findings suggest that ``worse'' models will still capture non-zero market
share. Finally, we conduct a set of experiments in \Cref{sec:continuous} under a
variant of our model, where players compete via text-to-image models to match
target images (inspired by \citet{vafa2025s,jahani2024prompt}). Again, we find
that the qualitative predictions of our theoretical model continue to hold
empirically. We discuss the implications of our findings for the development,
evaluation, and deployment of generative AI in \Cref{sec:discussion}.

\paragraph*{Scattergories as a metaphor.}

To build intuition for the interaction between monoculture and competition, we
introduce a prototypical example of generative competition: the word game
Scattergories.\footnote{A variant of Scattergories has also been proposed as an
  alternative to the Turing test in prior
work~\citep{kennedy2023participation}.}
In Scattergories, players must come up with a word or phrase that begins with a
given letter and matches a given category.\footnote{The validity of an answer is
typically determined by some agreement within the group of players.} Critically,
a player only scores a point if \textit{their answer is unique}. The ``market''
therefore rewards novelty.
For example, given the letter `S' and the category ``Fish,'' a player could
respond with ``Salmon'' or ``Swordfish.'' But a player might be better off
choosing a less common answer like ``Snapper'' since others might be less
likely to think of it.

Suppose Alice and Bob are playing Scattergories with a few other friends. Alice
notices that Bob has been using ChatGPT under the table to come up with his
responses. Alice wonders: Should she use ChatGPT as well? While it appears to be
providing Bob with valid answers, she worries that if she uses it, she'll get
\textit{the same} valid answers that Bob does, and neither of them will score
points. Alice decides to use a different chatbot instead, under the
intuition that its answers will be slightly different.

Is Alice's intuition valid? As long as Alice and Bob use slightly different
prompts, they could in principle receive completely different (distributions
over) answers from ChatGPT, despite the fact that they have the same ``intent.''
And yet, experiments that show AI homogenization in practice
\citep[e.g.,~][]{anderson2024homogenization} suggest that input variation alone
does not suffice to yield substantial output variation. While it is hard to
quantify a GAIT's sensitivity to input variations on an \textit{absolute} scale,
we can at least formulate a testable \textit{relative} statement informed by
Alice's intuition above: as long as the user's intent is fixed, varying the
input produces less output variability than simply switching to a different
GAIT.

\begin{hypothesis}[Monoculture via low within-GAIT variability]
  \label{hyp:variability}
  For a given user intent, within-GAIT variability is significantly lower than
  across-GAIT variability in responses.
\end{hypothesis}

In other words, two users of the same GAIT are more likely to get similar
outcomes than two users of different GAITs, regardless of the precise expression
of their intent. \Cref{hyp:variability} provides a starting point for us to
reason formally about algorithmic monoculture among generative AI users.

Beyond anthropomorphizing about different GAITs having different ``opinions,''
what reason do we have to believe \Cref{hyp:variability}? In
\Cref{sec:as-validation}, we provide experimental evidence that different GAITs
produce substantially different outputs, while different prompts do little to
induce variability. There are a number of reasons why this might be the case in
modern GAITs.
First, GAITs are typically trained on slightly different data, meaning we should
expect them to exhibit at least some disagreement, particularly over tasks for
which there is no single correct answer.
Second, the literature on stability and consistency in generative AI takes
minimizing within-GAIT variability as an objective~\citep{petroni2019language,
elazar2021measuring, singh2024robustness, yang2025enhancing, zhou2022prompt,
fu2025same, zhou2024larger}, seeking to ensure that GAIT responses are
insensitive to variations in inputs that do not change their meaning. To the
extent that these techniques are adopted in practice, they will reduce a GAIT's
sensitivity to slightly different inputs. Finally, prompt rewriting techniques
have the potential to reduce individual variability in prompts by ``smoothing''
out differences~\citep{abhishek2025enhancemyprompt, deng2023rephrase,
hsieh2024automatic}.

\paragraph*{Technical overview.}

Our formal model, which we present in \Cref{sec:model}, draws inspiration from
and generalizes our AI-assisted Scattergories vignette.
We seek to capture competition between producers who use GAITs to create content
and compete over demand for that content.
While we formalize competition via exact collisions in a discrete space for
tractability, the underlying economic principle is more general: competition
incentivizes differentiation.
Producers benefit by creating unique, distinctive content.
In \Cref{sec:continuous}, we study an extension to more general forms of spatial
competition.

In our model,
producers generate content from one of $K$ distinct types, where $K$ might be
very large. In Scattergories, $K$ would represent the number of possible
responses, valid or otherwise, to a (letter, category) pair.\footnote{In our
experiments, $K$ can be on the order of $10^{30}$.} We model competition as
follows: A producer who generates
content of type $k \in [K]$ derives utility that depends on (1) the value of
that type of content $\bd_k$ (e.g., 1 if it is a valid Scattergories answer and
0 otherwise) and (2) the total number of players $C_k$ who produce content of
the same type. If more players choose the same output, they each derive lower
utility, capturing negative externalities for similarity.\footnote{Assigning
  penalties for similarity has some parallels with the creativity literature,
  which seeks to capture originality by awarding higher scores to novel
responses~\citep{guilford1967nature}.}

Players use GAITs to produce content. We model each player's strategy as a
distribution $\bp$ over $[K]$, which results from sampling an output from a
GAIT. Critically, in order to model monoculture, we introduce an assumption
(\Cref{as:ranked}) based on \Cref{hyp:variability} that constrains players using
the same GAIT to have similar distributions. Thus, players can partially control
a GAIT's output distribution by varying the prompt or decoding parameters, but
monoculture prevents GAITs from producing arbitrary probability distributions.
Note that an assumption of this form is critical; without it, players could get
arbitrarily different outputs out of the same GAIT, and we wouldn't see any
homogenization at all.

\paragraph*{Limitations of our setting.}

We focus on the case in which all competitors express the same intent, meaning
they have identical preferences from the GAIT's perspective. This may be because
they truly have the same preferences (e.g., profit-maximization), because
expressing their full preferences is costly \citep[e.g.,~][]{castro2024human}, or
because they do not know their own preferences. Regardless of the explanation,
we assume that there are sufficiently many competitors with (near-)identical
expressed intents who form a competitive market.

Thus far, we have implicitly assumed that externalities are negative: Producers
are penalized for generating similar content. This is not true in all markets.
When goods are complementary, producers benefit when their behaviors align.
We explicitly study markets in which this is not the case. In many
of the domains discussed earlier---creative ideation, writing, art,
etc.---content producers typically face negative externalities when creating
similar content. There are several plausible mechanisms here. The demand for a
particular micro-genre is finite, and as more producers crowd the space, their
audiences may fragment. Prices could also fall as lookalikes enter the market,
reducing producer surplus. Abstracting away from the exact mechanisms
at play, we restrict our attention to \textit{settings in which externalities
are negative}, and conformity reduces producers' utilities.

\section{Related Work}
\label{sec:related}

\paragraph*{Algorithmic monoculture and homogenization.}

Theoretical work on algorithmic monoculture has typically focused on prediction
or allocation tasks, where multiple previously independent decision-makers begin
to correlate their behavior via identical or similar algorithmic
components~\citep{kleinberg2021algorithmic,bommasani2022picking,creel2022algorithmic,peng2023monoculture,peng2024wisdom,jain2024algorithmic,toups2024ecosystem}.
In these settings, results often suggest that monoculture harms the welfare of
decision-makers by reducing the overall information in a
system~\citep{kleinberg2021algorithmic,kaashoek2024impact,peng2023monoculture}.
Related work also considers the trade-off (or lack thereof) between accuracy and
diversity~\citep{peng2024reconciling}.

Recent work has studied homogenization through generative AI empirically in a
number of different settings~\citep{padmakumar2023does, liang2024monitoring,
  anderson2024homogenization, zhang2024generative, zhou2023generative,
  doshi2024generative, kirk2023understanding, park2024diminished,
shur2024growing}. Many of these studies argue implicitly or explicitly that
similarity imposes negative externalities. Researchers have articulated
objections to homogenization via AI beyond the purely economic considerations
discussed in this work, including its effects on
language~\citep{levent2023model,hancock2020ai},
creativity~\citep{buschek2021nine},
culture~\citep{epstein2023art,manovich2018ai}, and
science~\citep{fishman2022should,messeri2024artificial}. Others argue that we
can and should use generative AI to explicitly introduce diverse and stochastic
ideas~\citep{cai2023designaid,ugander2024art}.
\citet{zhang2025noveltybench} introduce a benchmark that seeks to measure
novelty over a set of responses. \citet{kreminski2025endless} argues that AI
homogeneity stems from the limited information a user provides to a tool, which
formalizes the broad strokes put forth by \citet{chiang2024why}.

\paragraph*{Alignment.}

A slightly different perspective on diversifying generative AI distributions
comes from \textit{alignment}, in which researchers seek to align output
distributions from AI tools with some notion of the ``right''
distribution~\citep{sorensen2024roadmap}. This research often indicates that the
use of reinforcement learning with human feedback (RLHF) is responsible for
misalignment between a model's output distribution and its training
distribution, leading to a form of mode collapse~\citep{wu2024generative,
kirk2023understanding, park2024diminished, padmakumar2023does}. All of the tools
we experiment with in \Cref{sec:empirical} have been instruction-tuned, so we
should expect some sort of ``narrowing'' of their distributions relative to what
their pre-trained baselines might produce. Other work on distributional
alignment seeks to modify training or fine-tuning processes to promote output
diversity~\citep{siththaranjan2023distributional, xiao2024algorithmic,
wang2023beyond}. Simple strategies like temperature
scaling~\citep{guo2017calibration}, which we use in our experiments in
\Cref{sec:empirical}, sometimes improve the calibration of output
distributions~\citep{kadavath2022language, tian2023just, shur2024growing}.

Empirical studies on distributional alignment tend to differ from our work in
their conception of what the ``right'' distribution is. Examples from the
literature include reference comparisons to training
distributions~\citep{wu2024generative}, events with known distributions (e.g.,
coin flips)~\citep{lovering2024language}, and public
opinion~\citep{santurkar2023whose}. In contrast, we make no assumptions about
what the distribution ``should'' be; our analysis takes downstream economic
consequences as the primary measure of interest. As we will see, it may be
optimal for a tool to be quite \textit{misaligned} as long as it fills a niche
left open by existing tools. Despite this, the Scattergories setting we
introduce could be an interesting testbed for future research on alignment.

\paragraph*{Related game-theoretic models.}

Much of the prior work on competition in content production studies incentives
in platforms mediated by recommender systems~\citep[e.g.,~][]{basat2017game,
  ben2018game, ben2019recommendation, mladenov2020optimizing,
  ghosh2011incentivizing, yao2023bad, yao2024user}.\footnote{Other related
  models study algorithmic competition in curation~\citep{immorlica2011dueling}
and decision-making \citep{jagadeesan2024improved}.} The goal of this work is
often to design mechanisms (i.e.,~recommender systems) that induce desirable
incentives for content producers.

More similar to ours is a recent line of work, inspired by
\citet{hotelling1929stability}, modeling competition between producers who
generate content represented in a vector space to meet consumer demand, also
represented as vectors~\citep{jagadeesan2024supply,hron2022modeling,
hu2023incentivizing, acharya2024producers}. In these models, consumers choose
the content that most closely aligns with their preference vectors. Our model
has some similar features: we could think of a GAIT's distribution over $K$
discrete types to be a $K$-dimensional vector. However, our ``choice model'' is
quite different. Demand is not fixed, and there is no cross-type competition. We
study an extension in \Cref{sec:continuous} that adopts many of the features of
these models.

Recent research also models competition between humans and generative AI.
\citet{yao2024human} and \citet{esmaeili2024strategize} study a setting where AI
continually learns from past content. \citet{taitler2024braess} analyze the
impacts of generative AI on incentives for human content producers in a
question-answering forum like StackOverflow. In contrast, all players in our
model use generative AI tools.

Analogous models to ours appear in a number of settings. For example, models of
spatial competition in the ride-sharing literature often feature discrete
locations where drivers choose to operate, and in those locations drivers impose
externalities on one another~\citep{brancaccio2023search, ghili2021spatial,
chin2023unified, buchholz2022spatial}. Parimutuel or racetrack betting markets
also share some features, where winning bettors divide the pot in proportion to
the sizes of their bets~\citep{thaler1988anomalies,hausch1981efficiency}.

Our work also has connections to more traditional game theory. We can interpret
our model as a class of atomic congestion games
~\citep[e.g.,~][]{rosenthal1973class,koutsoupias1999worst} with parallel edges,
particular cost functions, and restrictions on the strategy space: Each player
must choose a mixed strategy that obeys a ranking constraint. The game we study
turns out to be a \textit{valid utility game}~\citep{vetta2002nash}, which allows us
to leverage existing price of anarchy results. It is also a \textit{potential
game}~\citep{rosenthal1973class,monderer1996potential} (see
\Cref{lem:potential-game}), so best-response dynamics converge.

\section{Model}
\label{sec:model}

We begin by introducing our formal model. At a high level, each player in an
$n$-player game will sample an output from a GAIT. Their utility will depend on
the quality of that output as well as the number of other players who produce
the same output. We next describe how we model our two opposing forces:
monoculture and competition.

\paragraph*{Modeling algorithmic monoculture.}

First, we must formalize our intuition in \Cref{hyp:variability}, that users of
the same GAIT get ``similar'' outputs. In what follows, we will formalize limits
on the extent to which players can strategically tailor their GAIT use to
produce different output distributions. We assume each player provides an input
$\theta$ to a GAIT $T$ and receives a sample from $T$'s output distribution
$\bpn{T}(\theta)$. We assume that there are $K$ possible outputs (which we'll
refer to as ``types'') that $T$ can sample from, where $K$ can be very large.
For a given intent $I$, there is a set $\Theta_I$ of ways to express that intent
through different prompts or decoding parameters. We will introduce formal
assumptions to encode the idea that varying $\theta \in \Theta_I$ does not
fundamentally change the information that $T$ provides via $\bpn{T}(\theta)$.

We summarize the information that $\bpn{T}(\theta)$ provides via the
\textit{ranking} it induces over the types $[K]$.
For a distribution $\bp$, let $\bp_k$ be the probability $\bp$ places on type
$k$.
Define $\pi(\bp)$ to be the permutation over $[K]$ induced by the distribution
$\bp$, meaning $\bp_{\pi_1} \ge \bp_{\pi_2} \ge \dots \ge \bp_{\pi_K}$.
Intuitively, $\pi(\bp)$ summarizes $\bp$'s notion of quality: ``better'' types
should be sampled with higher likelihood than ``worse'' types. If $T$ is stable
to input variations, we would expect that its notion of response quality should
be invariant to $\theta \in \Theta_I$. In other words, we can formalize
\Cref{hyp:variability} as follows:

\begin{assumption}[GAITs preserve output rank]
  \label{as:ranked}
  Given a GAIT $T$ and intent $I$,
  for all $\theta, \theta' \in \Theta_I$, $\pi(\bpn{T}(\theta)) =
  \pi(\bpn{T}(\theta'))$.
\end{assumption}

In other words, while players can strategically choose their inputs $\theta$,
they cannot change $T$'s ranking of ``better'' responses over ``worse''
responses given intent $I$. \Cref{as:ranked} is fairly stylized, and will not be
true in practice; models are not perfectly stable. Very detailed, strategic
prompts or the inclusion of personal information could lead to meaningful
differences in the information that GAITs provide. On the other hand,
\Cref{as:ranked} captures the fact that empirically, idiosyncratic prompting
does not suffice to break GAITs out of their default behaviors, as evidenced by
much of the experimental literature on AI homogeneity
\citep[e.g.,~][]{anderson2024homogenization, doshi2024generative}. We provide
evidence in \Cref{sec:as-validation} that \Cref{as:ranked} is
\textit{directionally} correct by studying how rankings vary across different
prompts $\theta$. We find that $\pi(\bpn{T}(\theta))$ does not change much as
$\theta$ varies, but it can change dramatically across different GAITs $T, T'$.

For a given GAIT $T$ and intent $I$, \Cref{as:ranked} implies that there is some
canonical ordering over responses $\pi^{(T, I)}$ such that $\pi(\bpn{T}(\theta))
= \pi^{(T, I)}$ for all $\theta \in \Theta_I$.
Let $\Delta(\pi)$ denote the set of probability distributions on $[K]$ that have
the same ranking as $\pi$. We make one further assumption:
\begin{assumption}[GAITs are expressive]
  \label{as:expressive}
  For any $\bp \in \Delta(\pi^{(T, I)})$, there exists $\theta \in \Theta_I$
  such that $\bpn{T}(\theta) = \bp$.
\end{assumption}

Informally, we're assuming that $\Theta_I$ is sufficiently expressive: there is
some input $\theta$ that elicits any particular distribution $\bp$, subject to
the constraint that $\bp$ provides the same ranking $\pi^{(T, I)}$. Again, this
is a stylized assumption that will not be strictly true in practice, although
$\Theta_I$ contains a high degree of expressivity in both the exact wording of a
prompt and decoding parameters. It is possible that a weaker or approximate
version of \Cref{as:expressive} would suffice for our results, though we do not
investigate this here. While our theoretical results rely on
\Cref{as:ranked,as:expressive}, our experiments in \Cref{sec:empirical} do not.
Nevertheless, we find that our results hold empirically.

\paragraph*{Modeling competition.}

We study a symmetric $n$-player game and fix the intent $I$. For now, we'll
focus on the case in which all players use the same GAIT $T$, represented by a
permutation $\pi$, and we'll drop $(T, I)$ from our notation.
We'll generalize this to the multi-GAIT
case in \Cref{sec:inter-tool}. Players strategically choose a distribution $\bp
\in \Delta(\pi)$ (which is equivalent to strategically choosing an input $\theta
\in \Theta_I$ under \Cref{as:ranked,as:expressive}) over the $K$ types. Then,
each player samples type $k \in [K]$ independently with probability $\bp_k$.

Each type $k$ has a known nonnegative expected value $\bd_k$, which is the
utility a producer receives if she is the only producer who samples it. For
example, if consumers have heterogeneous preferences, $\bd_k$ might represent
the fraction of consumers with demand for type $k$ and only type $k$. This
captures the idea that some types are more valuable than others, perhaps because
they have more consumer demand.
But
because producers impose negative externalities on one another, producers
receive utility that decreases in the number of competitors who sample the same
type. In particular, let $C_k$ be the number of players who produce type $k$.
Then, each receives utility $\bd_k / s(C_k)$, where $s(\cdot)$ is an increasing
function.\footnote{In \Cref{sec:continuous}, we extend our model to handle
  competition in a continuous output space. Restricting ourselves to discrete
  types allows us to impose structural assumptions on GAIT behavior (in
particular, \Cref{as:ranked}), which enables our theoretical analysis.} For
example, if $s(x) = x$, then players receive an equal share of the overall value
$\bd_k$. Choosing $s(x) = x^\infty$ yields the scoring rule for Scattergories:
$\bd_k$ if the response is unique, and 0 otherwise. We refer to $s$ as a
\textit{score function}, and it plays a role similar to a classical congestion
function. The mechanics of our game $\game(n, \bd, s)$ are summarized in
\Cref{fig:game-mechanics}.

\begin{figure}[ht]
  \centering
  \includegraphics[width=0.95\textwidth]{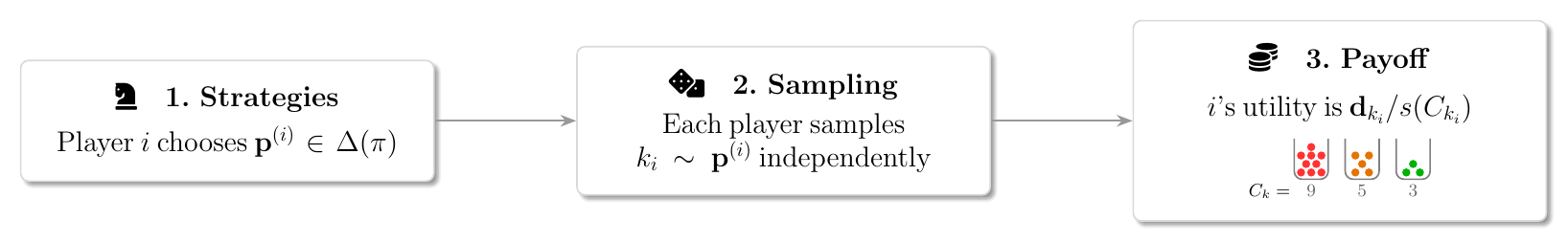}
  \caption{Description of game mechanics.}
  \label{fig:game-mechanics}
\end{figure}

\paragraph*{Score functions and social welfare.}

To fully specify the model, we must provide a class of score functions.
Throughout this paper, we will frequently use score functions of the form
$s_\gamma(x) \triangleq x^\gamma$ for $\gamma > 0$. (Observe that when $\gamma =
0$, there is no competition, and players behave independently from one another.)
The functions $s_\gamma$ have several desirable properties. $s_\gamma(1) = 1$,
meaning if a single player samples type $k$, they face no competition.
$s_\gamma$ is increasing, so more competition implies higher congestion. And
$1/s_\gamma$ is convex, so marginal congestion is decreasing in the amount of
competition. Intuitively, the jump from no competition to a single competitor is
significant, but if competition is already strong, one additional competitor
makes less of a marginal difference.

Our definition of $s_\gamma$ reveals a slight subtlety that will impact how we
think about social welfare, reflecting modeling differences in Hotelling games.
In one formulation (for example, Hotelling games on disconnected
graphs~\citep{knoblauch1991generalizing,fournier2019location}), prices are fixed
(or 0), but demand served is not. Social welfare is the total demand served,
which is simply the sum of player utilities:\footnote{By convention, we let
$s(0) = 1$ for any $s$ to ensure that denominators are always nonzero.}
\begin{equation}
  \label{eq:sw-up}
  \sum_{i \in [n]} \bd_{k_i} \frac{1}{s_\gamma(C_{k_i})} 
  = \sum_{k \in [K]} \bd_k \frac{C_k}{s_\gamma(C_k)}
  = \sum_{k \in [K]} \bd_k C_k^{1-\gamma}.
\end{equation}
This maps nicely to the case where $\gamma \le 1$. While society experiences
diminishing marginal gains from items of type $k$, overall demand met for type
$k$ increases in $C_k$. We could microfound this, for example, with a
model in which consumers can derive benefit from more than one item. Thus, more
production of type $k$ increases overall demand for type $k$ met.

In another formulation (e.g.,~Hotelling games with endogenous
prices~\citep{hotelling1929stability}), demand is fixed but prices are not.
Players who sample the same type experience (potentially imperfect) price
competition, meaning the total surplus players derive from type $k$ (i.e., total
demand met $\times$ price per unit demand) \textit{decreases} in $C_k$. As is
common in the literature, we ignore transfers between producers and consumers in
our definition of social welfare, meaning social welfare is simply the overall
demand met. Taking $\bd_k$ to be the fixed demand for each $k$, social welfare
is
\begin{equation}
  \label{eq:sw-down}
  \sum_{k \in [K]} \bd_k \ind{C_k > 0}.
\end{equation}
This corresponds to the case where $\gamma \ge 1$. Here, we could microfound a
model in which demand for type $k$ is fixed at $\bd_k$, and as $C_k$ increases,
price competition drives down producer surplus per unit of demand met. Perfect
(Bertrand) price competition would imply that producer $i$'s utility is
$\bd_{k_i}$ when $C_{k_i} = 1$ and $0$ if $C_{k_i} > 1$, corresponding to the
score function $s_\infty(x) = x^\infty$.

We have thus described two slightly different classes of score functions
($s_\gamma$ for $\gamma \le 1$ and $\gamma \ge 1$) with different definitions of
social welfare~\eqref{eq:sw-up} and~\eqref{eq:sw-down} respectively.
These definitions coincide at $\gamma = 1$. See \Cref{fig:gamma-funcs} for
further intuition: For a particular type, individual utilities (given by
$1/s_\gamma(x)$) are always decreasing in congestion (\Cref{fig:congestion}),
but the sum of player utilities (given by $x/s_\gamma(x)$) can either increase
or decrease (\Cref{fig:x-over-s}). In cases where the sum of utilities
decreases,~\eqref{eq:sw-down} defines social welfare to be equivalent to the
$\gamma = 1$ case.

In \Cref{app:score-funcs}, we define a more general class of score functions
$\sall$ for which our results hold. As with our score functions $s_\gamma$, if
the sum of player utilities is increasing in congestion (i.e., $x/s(x)$ is
increasing), we apply the social welfare definition~\eqref{eq:sw-up}. We denote
this subset of score functions $\sup \subset \sall$. Analogously, we use $\sdown
\subset \sall$ to denote score functions for which $x/s(x)$ is decreasing and
apply social welfare definition~\eqref{eq:sw-down}. By construction, $\sall =
\sup \cup \sdown$, and the identity score function $s_1(x) = x$ is the only
function in their intersection $\sup \cap \sdown$.

\begin{figure}[ht]
  \centering
  \begin{subfigure}[t]{0.49\textwidth}
    \centering
    \includegraphics[width=\textwidth]{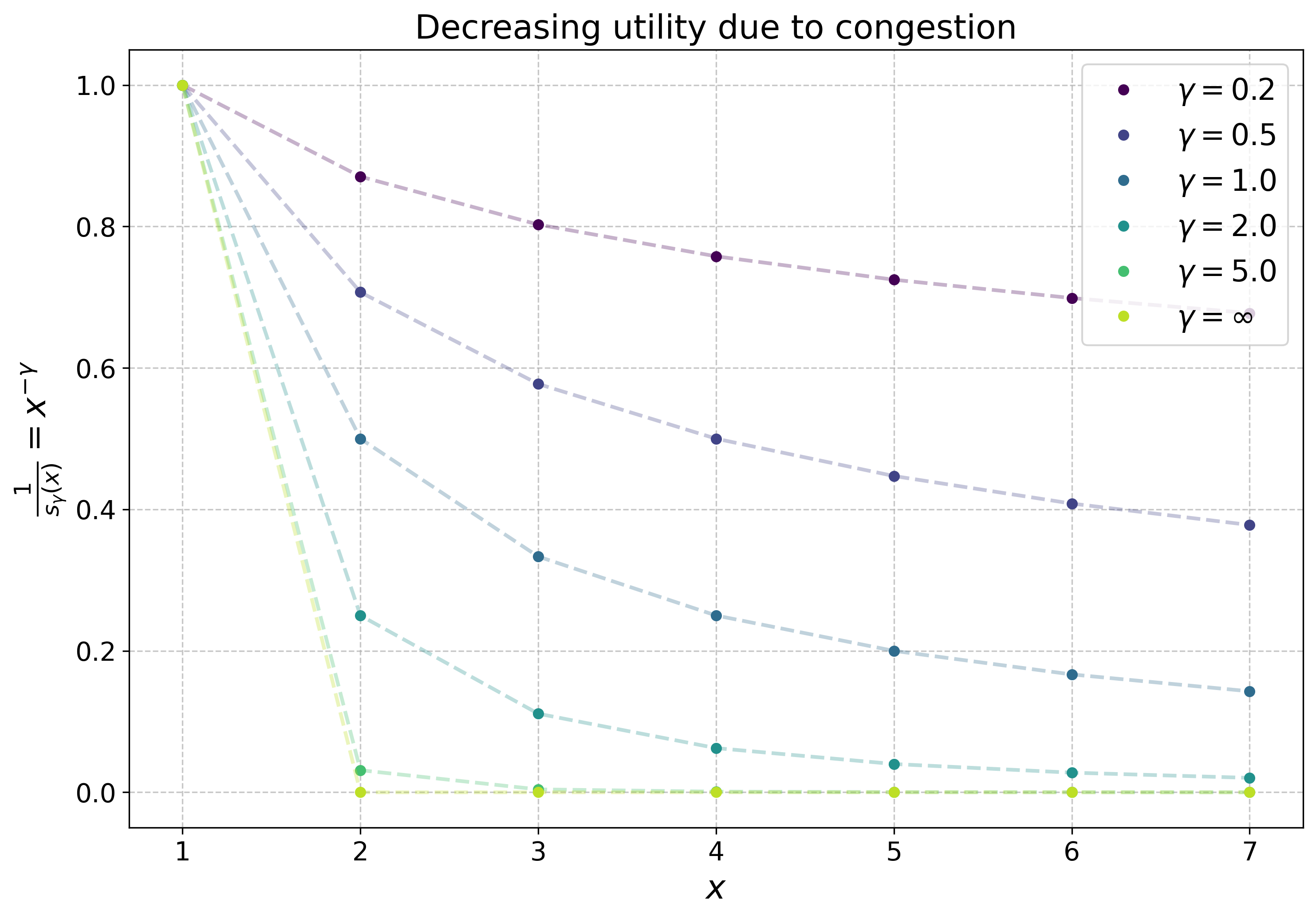}
    \caption{Players' utilities decrease with congestion.}
  \label{fig:congestion}
  \end{subfigure}
  \hfill
  \begin{subfigure}[t]{0.49\textwidth}
    \centering
    \includegraphics[width=\textwidth]{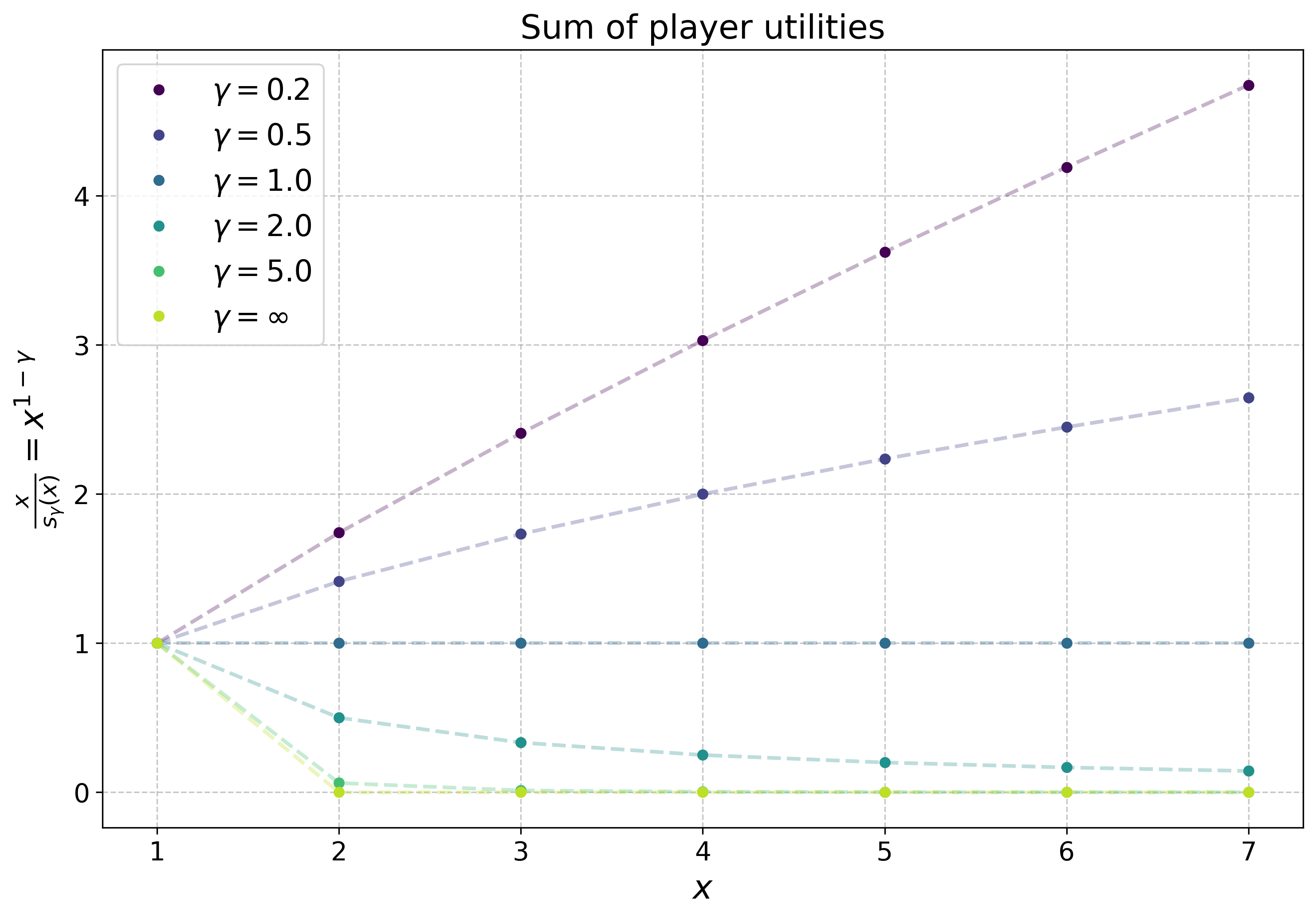}
    \caption{The \textit{sum} of players' utilities may increase or decrease
      with congestion. If $x/s(x)$ is increasing, then $s \in \sup$. If $x/s(x)$
    is decreasing, then $s \in \sdown$.}
    \label{fig:x-over-s}
  \end{subfigure}
  \caption{Sample score functions $s_\gamma(x) = x^\gamma$.}
  \label{fig:gamma-funcs}
\end{figure}

\section{Competition under a Single Tool}
\label{sec:theory}

We are now ready to analyze behavior under our model. We begin in this section
with the case where all players have access to a single GAIT with ordering $\pi
= [1, 2, \dots, K]$ and can choose any $\bp \in \Delta(\pi)$. In
\Cref{sec:inter-tool}, we will generalize to the case where players can choose
from a set $\Pi$ of available GAITs. First, we discuss our solution concept,
symmetric Nash equilibria, in \Cref{sec:equilibria}. Then, in
\Cref{sec:diversity}, we study the \textit{diversity} of equilibria, finding
that stronger competition induces more diversity, but equilibria are less
diverse than socially optimal behavior. We characterize equilibria in the limit
of infinite competition in \Cref{sec:compeitition-limit}. We show in
\Cref{sec:poa} that the \textit{price of anarchy}~\citep{koutsoupias1999worst}
is 2, limiting the extent to which equilibria can be suboptimal. Throughout this
section, we will assume $n \ge 2$, $\bd \in \R_{\ge 0}^K$, and $s \in \sall$.

\subsection{Equilibria}
\label{sec:equilibria}

Our focus in this section will be on symmetric behavior, though we will show
later that asymmetric equilibria converge to symmetric equilibria
(\Cref{lem:sym-asym-converge}).
Because $\game$ is symmetric and utilities are continuous, it is guaranteed to
have a symmetric mixed Nash
equilibrium~\citep{glicksberg1952further,becker2006existence}. Every mixed
strategy is equivalent to a pure strategy since $\Delta(\pi)$ is convex (see
\Cref{lem:mixed-pure-eq}). Thus, $\game$ is guaranteed to have a symmetric pure
Nash equilibrium.

\paragraph*{Utility and social welfare.}
Let $X(n, p)$ denote a binomial random variable with
those parameters. Given that the $n-1$ other players choose the strategy $\bp$,
player $i$'s expected utility for strategy $\bp'$ is
\ifthenelse{\boolean{smallEqs}}{
\begin{align*}
  U(\bp', \bp; n, \bd, s)
   &\triangleq
  \sum_{k \in [K]} \bp_k' \bd_k \E{\frac{1}{s(C_k)} \given k_i = k} \\
  &= \sum_{k \in [K]} \bp_k' \bd_k \E{\frac{1}{s(1 + X(n-1, \bp_k))}}.
\end{align*}
}{
\begin{align*}
  U(\bp', \bp; n, \bd, s)
  &\triangleq
  \sum_{k \in [K]} \bp_k' \bd_k \E{\frac{1}{s(C_k)} \given k_i = k}
  =
  \sum_{k \in [K]} \bp_k' \bd_k \E{\frac{1}{s(1 + X(n-1, \bp_k))}}.
\end{align*}
}
When it is clear from context, we will omit the $(n, \bd, s)$. In a slight abuse
of notation, we will sometimes write $U(\bp, \bp)$ as simply $U(\bp)$. $\bp$
is a symmetric equilibrium for $\game(n, \bd, s)$ if and only if
$U(\bp, \bp; n, \bd, s) \ge U(\bp', \bp; n, \bd, s)$ for all $\bp' \in
\Delta(\pi)$.
We next define expected social welfare for symmetric strategies, taking care to
account for the different interpretations discussed in \Cref{sec:model}. In the
case where $s \in \sup$, following~\eqref{eq:sw-up},
\ifthenelse{\boolean{smallEqs}}{
\begin{align*}
  W(\bp; n, \bd, s)
  &\triangleq \sum_{k \in [K]} \bd_k \E{\frac{C_k}{s(C_k)}} \\
  &= \sum_{k \in [K]} \bd_k \E{\frac{X(n, \bp_k)}{s(X(n, \bp_k))}} \\
  &= \sum_{i \in [n]} U(\bp; n, \bd, s).
\end{align*}
}{\begin{align*}
  W(\bp; n, \bd, s)
  &\triangleq \sum_{k \in [K]} \bd_k \E{\frac{C_k}{s(C_k)}}
  = \sum_{k \in [K]} \bd_k \E{\frac{X(n, \bp_k)}{s(X(n, \bp_k))}}
  = \sum_{i \in [n]} U(\bp; n, \bd, s).
\end{align*}
}
When $s \in \sdown$, following~\eqref{eq:sw-down},
\ifthenelse{\boolean{smallEqs}}{
\begin{align*}
  W(\bp; n, \bd, s)
  &\triangleq \sum_{k \in [K]} \bd_k \E{\ind{C_k > 0}} \\
  &= \sum_{k \in [K]} \bd_k \Pr[X(n, \bp_k) > 0] \\
  &= \sum_{k \in [K]} \bd_k \p{1 - (1 - \bp_k)^n}.
\end{align*}
}{\begin{align*}
  W(\bp; n, \bd, s)
  &\triangleq \sum_{k \in [K]} \bd_k \E{\ind{C_k > 0}}
  = \sum_{k \in [K]} \bd_k \Pr[X(n, \bp_k) > 0]
  = \sum_{k \in [K]} \bd_k \p{1 - (1 - \bp_k)^n}.
\end{align*}
}
Again, we will omit $(n, \bd, s)$ when it is clear from context.
Note that $W$ no longer depends on $s$ when $s \in \sdown$, so
welfare-maximizing strategies are identical for all $s \in \sdown$.

\paragraph*{Without loss of generality, a simplifying assumption.}

Consider the special case where $\bd$ is decreasing in $k$, meaning $\pi$
perfectly ranks outputs by their expected value:
\begin{assumption}
  \label{as:d-ranked}
  $\bd_k$ is weakly decreasing in $k$.
\end{assumption}
\Cref{as:d-ranked} turns out to dramatically simplify the characterization of
symmetric equilibrium and optimal strategies. As we show in
\Cref{app:equilibria}, the constraint $\bp \in \Delta(\pi)$ is no longer
binding, and we can express symmetric equilibrium and optimal strategies
(denoted $\peq$ and $\popt$ respectively) as the unique solutions to relatively
simple convex optimization problems (see
\Cref{lem:unique-eq-decreasing,lem:unique-opt-decreasing}).

Perhaps surprisingly, we can make \Cref{as:d-ranked} \textit{without loss of
generality}: For every $\bd$, there exists some decreasing $\tbd$ with identical
strategies and utilities for symmetric equilibrium and optimal play.

\begin{restatable}{theorem}{uniqueeq}
  \label{thm:unique-eq}
  For every $\bd \in \R_{\ge 0}^K$, there is some $\tbd \in \R_{\ge 0}^K$ such
  that $\tbd$ is decreasing in $k$, and for all $n \ge 2$, $s \in \sall$,
  \begin{enumerate}

    \item $\peq(n, \bd, s) = \peq(n, \tbd, s)$, and per-player utilities are
      equal.
      
    \item $\popt(n, \bd, s) = \popt(n, \tbd, s)$, and per-player utilities and
      social welfare are equal.

  \end{enumerate}
  By \Cref{lem:unique-eq-decreasing,lem:unique-opt-decreasing}, $\peq(n,
  \bd, s)$ and $\popt(n, \bd, s)$ are both unique.
\end{restatable}

We defer all proofs in this section to \Cref{app:proofs}. By
\Cref{thm:unique-eq}, every instance $\bd$ has a corresponding instance $\tbd$
satisfying \Cref{as:d-ranked} with identical equilibria and optimal strategies.
This means that we can restrict our attention to the analytically simpler case
where \Cref{as:d-ranked} holds, and \Cref{thm:unique-eq} allows us to
immediately lift all of our results to the more general setting. A version of
\Cref{thm:unique-eq} holds for asymmetric equilibria as well
(\Cref{lem:dec-d-asym}).

We briefly provide some intuition for what the equivalent game $\tbd$ looks
like. Suppose $\bd = [1; ~ 0; ~ 1]$. Ideally, a player could choose a
distribution like $\bp' = [\nicefrac{1}{2}; ~ 0; ~ \nicefrac{1}{2}]$, but of
course, $\bp' \notin \Delta(\pi)$. Under \Cref{as:ranked}, $\bp$ must respect
the constraint $\bp_2 \ge \bp_3$. In this example, it is never optimal or
rational to set $\bp_2 > \bp_3$, since a player's utility strictly increases if
they transfer mass from $\bp_2$ to $\bp_3$ or $\bp_1$.
Thus, for any equilibrium or socially optimal strategy, $\bp_2 = \bp_3$.

The final step of the proof is to observe that if $\bp_2 = \bp_3$, then $\bd =
[1; ~ 0; ~ 1]$ is effectively the same as $\tbd = [1; ~ \nicefrac{1}{2}; ~
\nicefrac{1}{2}]$. More generally, we can replace sequences of (weakly)
increasing values of $\bd$ by their means, since $\peq$ and $\popt$ will assign
equal probabilities to all of those elements anyways. Iteratively applying the
argument yields the desired result. This procedure is often known as the ``Pool
Adjacent Violators'' algorithm~\citep{barlow1972statistical} and is commonly
used for isotonic regression~\citep{van1956maximum, ayer1955empirical,
brunk1955maximum}.

\paragraph*{A simple example.}

To build intuition, consider an example with $n = 2$ players and $K = 2$ types
with values $\bd = [3; ~ 2]$. In the absence of competition, each player would
choose the distribution $\bp = [1; ~ 0]$, since type 1 is higher-quality than
type 2. But with the score function $s_1(x) = x$, these strategies would yield
utility $3/2$ per player. They would be better off diversifying: there is
unmet demand $\bd_2 = 2$, and by shifting probability mass from $\bp_1$ to
$\bp_2$, players could increase their utilities.

In fact, the unique symmetric pure Nash equilibrium for this example is $\bp =
[4/5; ~ 1/5]$: by \Cref{lem:unique-eq-decreasing}, the symmetric equilibrium
satisfies
\begin{align*}
  \E{\frac{\bd_1}{s(1+X(n-1, \bp_1))}}
  &= \E{\frac{\bd_2}{s(1+X(n-1, \bp_2))}} \\
  3 \p{1 - \bp_1 + \frac{\bp_1}{2}}
  &= 2 \p{1 - \bp_2 + \frac{\bp_2}{2}} \\
  \bp_1 &= \frac{4}{5}.
\end{align*}
Competition between 2 players has diversified equilibrium
behavior. And with more competition (as $n$ increases), we might
expect even more diverse behavior. On the other hand, the
symmetric socially optimal strategy is $\bp^* = [3/5; ~ 2/5]$, which by
\Cref{lem:unique-opt-decreasing} satisfies
\begin{align*}
  \bd_1 (1 - \bp_1^*)^{n-1}
  &= \bd_2 (1 - \bp_2^*)^{n-1} \\
  3 (1 - \bp_1^*)
  &= 2 (1 - \bp_2^*) \\
  \bp_1^*
  &= \frac{3}{5}.
\end{align*}
In this example, while competition induces diversity, socially optimal behavior
is even more diverse. We will show that these statements hold more generally in
what follows.

\subsection{Diversity of equilibria}
\label{sec:diversity}

Our main results in this section formalize our intuition from the above example.
We show in \Cref{thm:diversity-n,thm:diversity-gamma} that
more competitive environments do indeed induce more diverse behavior. On the
other hand, we show in \Cref{thm:diversity-opt} that the equilibrium
distribution $\peq(n, \bd, s)$ is less diverse than the optimal distribution
$\popt(n, \bd, s)$. We then discuss the limiting behavior of competition in
\Cref{sec:compeitition-limit}.

First, we must define ``diversity'' more precisely. All of our results will use
a particularly strong notion of the relative diversity of distributions known as
\textit{majorization}:

\begin{definition}[Majorization]
  \label{def:majorization}
  For distributions $\bp, \bq \in \Delta(\pi)$, $\bp$ majorizes $\bq$ if
  \begin{align*}
    \sum_{\ell=1}^k \bp_\ell \ge \sum_{\ell=1}^k \bq_\ell.
    \tag{$\forall k \in [K]$}
  \end{align*}
  We write this as $\bp \succ \bq$.
\end{definition}

Majorization provides a useful notion of the relative diversity of discrete
distributions. By definition, if $\bp \succ \bq$, then $\phi(\bp) \le \phi(\bq)$
for any Schur-concave function $\phi$ \citep[e.g.,~][]{olkin2014inequalities}.
(And similarly, $\bp \succ \bq$ implies $\phi(\bp) \ge \phi(\bq)$ for any
Schur-convex function $\phi$.) Standard measures of diversity, such as
Shannon~\citep{shannon1948mathematical}, R{\'e}nyi~\citep{renyi1961measures},
and Tsallis~\citep{tsallis1988possible} entropies, are Schur-concave.
Thus, the following holds:
\begin{proposition}[Majorization implies entropy ordering
  \citep{olkin2014inequalities}]
  \label{fact:majorization}
  If $\bp \succ \bq$, then $H(\bp) \le H(\bq)$ for Shannon, R{\'e}nyi, or
  Tsallis entropy $H$.
\end{proposition}
Similarly, measures of inequality like the Gini
coefficient~\citep{gini1912variabilita} are Schur-convex, so we can also
interpret $\bp \succ \bq$ to say that $\bp$ is more ``unequal'' than $\bq$.
Majorization is equivalent to first-order stochastic dominance over ranks in the
case where the permutation $\pi$ is fixed.

Given that our goal will be to show that competition increases entropy, it
suffices to state our results in terms of majorization.
Intuitively, we can understand ``$\bp$ majorizes $\bq$'' ($\bp \succ \bq$)
as ``$\bp$ is less diverse than $\bq$.'' We begin by showing that increasing the
number of players $n$ increases equilibrium diversity.
All proofs in this
section are deferred to \Cref{app:diversity}.

\begin{restatable}[More players $\Longrightarrow$ more diverse
  equilibria]{theorem}{divn}
  \label{thm:diversity-n}
  For all $n \ge 2$, $\bd \in \R_{\ge 0}^K$, $s \in \sall$,
  \begin{align*}
    \peq(n, \bd, s) \succ \peq(n+1, \bd, s).
  \end{align*}
\end{restatable}

We provide a sketch of the proof here. Consider the following thought
experiment: what if, in the $n+1$ player game $\game(n+1, \bd, s)$, all players
played $\bp = \peq(n, \bd, s)$? We show that if they did, expected utility
conditioned on sampling type $k$ would be
\textit{increasing} in $k$ (for $k$ such that $\bp_k > 0$): playing the
$n$-player equilibrium strategy in the $n+1$-player game would over-exploit the
types with high $\bd_k$ and under-exploit those with lower $\bd_k$.

By \Cref{lem:unique-eq-decreasing}, under the $n+1$-player equilibrium
$\peq(n+1, \bd, s)$, types should have equal conditional utility. So in order to
transform $\peq(n, \bd, s)$ into $\peq(n+1, \bd, s)$, we have to decrease
probabilities on the high-$\bd_k$ types and raise them on the low-$\bd_k$ types.
This seems like it should lead to a more diverse distribution, since we are
lowering the larger probabilities and raising the smaller probabilities, and
indeed, we show that this implies $\peq(n, \bd, s) \succ \peq(n+1, \bd,
s)$.\footnote{Intuitively, we might think the same holds for $\popt$;
\Cref{fig:ps-to-inf} shows that this is not the case.}

\Cref{thm:diversity-n} tells us that increased competition incentivizes more
diverse behavior. We show that a similar result holds when increased competition
takes the form of a stronger score function $s$. In particular, we consider
score functions of the form $s_\gamma(x) = x^\gamma$.

\begin{restatable}[Greater externalities $\Longrightarrow$ more diverse
  equilibria]{theorem}{divgamma}
  \label{thm:diversity-gamma}
  For all $n \ge 2$, $\bd \in \R_{\ge 0}^K$, $\gz \le \go$,
  \begin{align*}
    \peq(n, \bd, s_{\gz}) \succ \peq(n, \bd, s_{\go}).
  \end{align*}
\end{restatable}

The proof idea behind \Cref{thm:diversity-gamma} is somewhat similar to that of
\Cref{thm:diversity-n}: we ask what would happen if players in $\game(n, \bd,
s_{\go})$ were to play the strategies $\peq(n, \bd, s_{\gz})$, finding that they
would again place too much probability mass on types with high values of
$\bd_k$.

Taken together, \Cref{thm:diversity-n,thm:diversity-gamma} provide some
intuition on when we should expect to see more diverse vs. more homogeneous
behavior from producers using GAITs. In settings where competitors do not impose
strong negative externalities on one another, behavior should be relatively
homogeneous. Suppose, for example, teachers use a GAIT to generate problem sets.
We might expect to see homogeneity across teachers if the value of a homework
problem is unaffected by other teachers' behavior. But if some teachers begin
publishing solution manuals to their problem sets online, a student may be able
to use another teacher's solutions for their homework. Similar assignments now
impose negative externalities (from the teachers' perspective). As a result, at
equilibrium, teachers should diversify their use of GAITs to avoid generating
similar problem sets.

Although competition induces diversity, equilibria are still more homogeneous
than would be socially optimal:

\begin{restatable}[Optimal strategies are more diverse than
  equilibria]{theorem}{divopt}
  \label{thm:diversity-opt}
  For all $n \ge 2$, $\bd \in \R_{\ge 0}^K$, $s \in \sall$,
  \begin{align*}
    \peq(n, \bd, s) \succ \popt(n, \bd, s).
  \end{align*}
\end{restatable}

Again, the proof considers the hypothetical in which players play $\popt$,
showing that they could increase their own utility (at the expense of social
welfare) by placing more probability mass on high-$\bd_k$ types. To summarize,
we have shown so far that competition induces diversity
(\Cref{thm:diversity-n,thm:diversity-gamma}), but equilibria are less diverse
than would be socially optimal (\Cref{thm:diversity-opt}).

\subsection{Competition in the limit.}
\label{sec:compeitition-limit}

Next, we explore the effects of competition in the limit. For the
score functions $s_\gamma$, we consider how $\peq$ and $\popt$ relate as $n$
grows, finding a surprising discontinuity. For $\gamma < 1$, $\peq$ and $\popt$
approach the same distribution $\bplim$, which is proportional to
$\bd_k^{1/\gamma}$. But for $\gamma \ge 1$, they diverge: $\peq$ still
approaches $\bp^\dagger$, but $\popt$ approaches the uniform distribution.
Intuitively, we can explain this qualitative shift as follows. When $\gamma <
1$, there is no limit to social welfare. More production of type $k$ results in
greater welfare, though with diminishing returns. But for $\gamma \ge 1$, social
welfare is bounded by $\sum_k \bd_k$, and maximizing social welfare amounts to
minimizing the chance that any type fails to get produced. Note that this
discontinuity is not simply an artifact of our two definitions of social welfare
for $\sup$ and $\sdown$, since it arises even if we restrict ourselves
to $\sup$ (i.e., for $\gamma < 1$ vs. $\gamma = 1$). We state the result
formally below. See \Cref{fig:ps-to-inf} for further intuition.

\begin{restatable}[Strategies as $n \to \infty$]{lemma}{gammaconv}
  \label{lem:n-to-infty}

  Assume without loss of generality that $\bd \in \R_{> 0}^K$ is weakly
  decreasing. Let $\bpn{n} = \peq(n, \bd, s_\gamma)$, and let $\bqn{n} =
  \popt(n, \bd, s_\gamma)$. Let $\bplim_k \triangleq \bd_k^{1/\gamma} / \sum_{k' \in
  [K]} \bd_{k'}^{1/\gamma}$. Then,
  \begin{itemize}
    \item $\left\{\bpn{n}\right\}_{n \to \infty}$ approaches $\bplim$.
    \item If $\gamma < 1$, $\left\{\bqn{n}\right\}_{n \to \infty}$ approaches
      $\bplim$.
    \item If $\gamma \ge 1$ (or more generally, for any $s \in \sdown$),
      $\left\{\bqn{n}\right\}_{n \to \infty}$ approaches the uniform
      distribution.
  \end{itemize}
  Each converges at a rate of $O\p{\sqrt{\frac{\log n}{n}}}$.

\end{restatable}

\ifdefined\smallfigs
\else
\begin{figure}[ht]
  \centering
  \begin{subfigure}[ht]{0.49\textwidth}
    \centering
    \includegraphics[width=\textwidth]{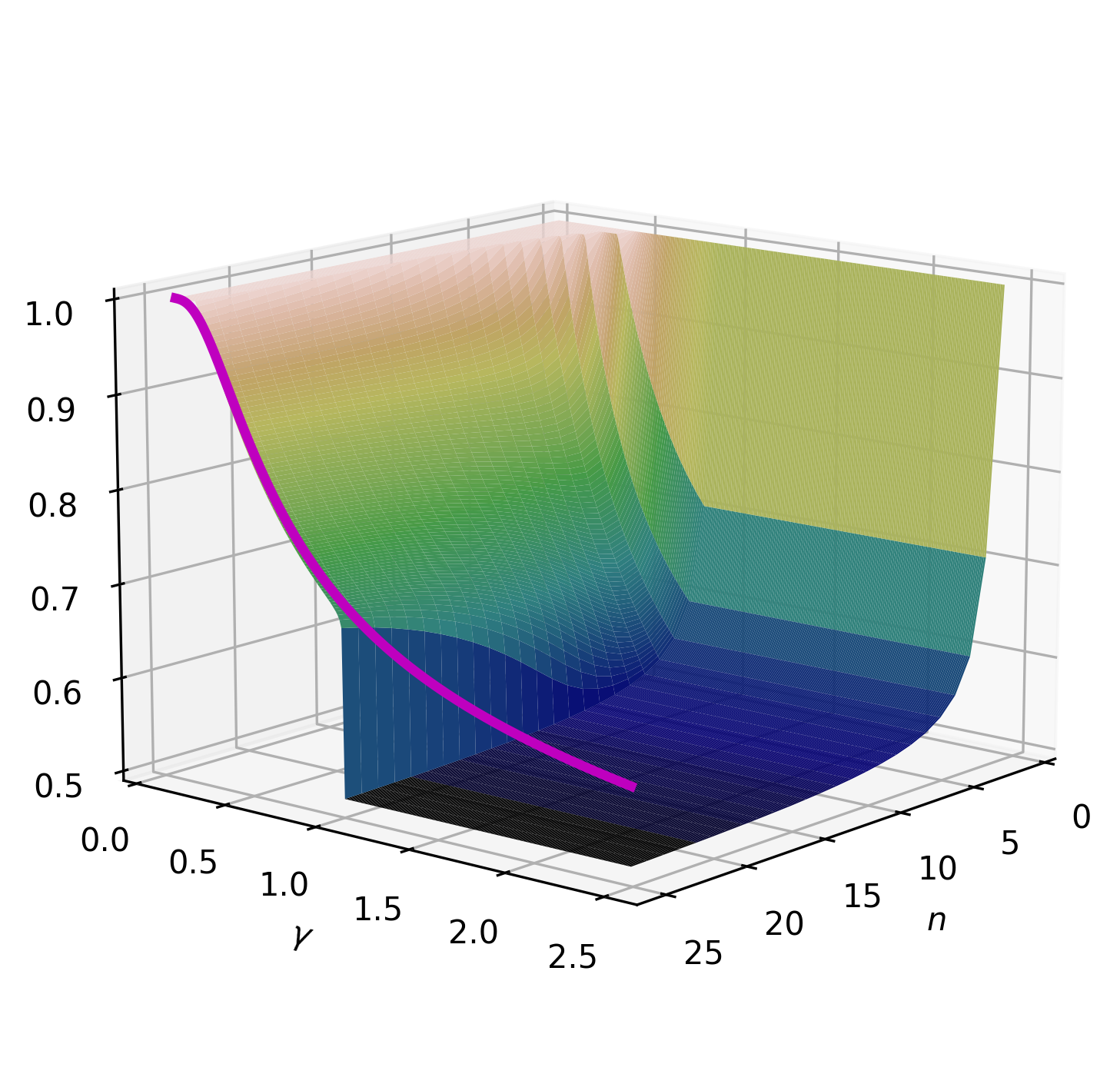}
    \caption{$\popt(n, \bd, s_\gamma)_1$ for $\bd = [5; ~ 2]$.}
    \label{fig:ps-to-inf-opt}
  \end{subfigure}
  \hfill
  \begin{subfigure}[ht]{0.49\textwidth}
    \centering
    \includegraphics[width=\textwidth]{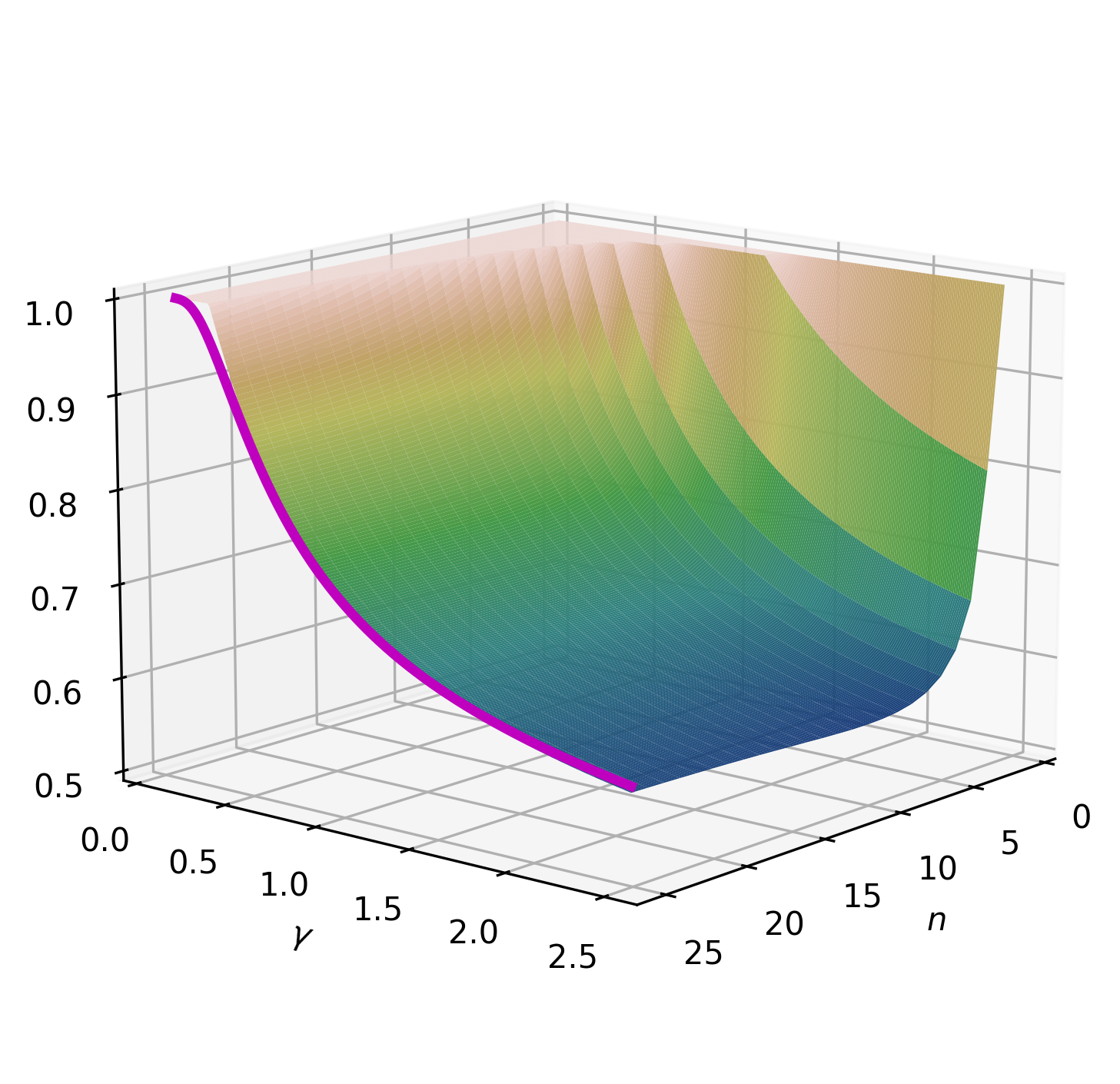}
    \caption{$\peq(n, \bd, s_\gamma)_1$ for $\bd = [5; ~ 2]$.}
    \label{fig:ps-to-inf-eq}
  \end{subfigure}
  \caption{For further intuition on \Cref{lem:n-to-infty},
    we show how $\popt(n, \bd, s_\gamma)_1$ and $\peq(n, \bd, s_\gamma)_1$
    respectively vary with both $n$ and $\gamma$ for the instance given by $\bd
    = [5; ~ 2]$. For $\gamma < 1$, both $\popt$ and $\peq$ converge to $\bplim$,
    shown as a magenta line in each panel. For $\gamma \ge 1$, $\popt$ converges
    to the uniform
    distribution, with a discontinuity at $\gamma = 1$. Further, while $\peq$ is
    monotone in both $n$ and $\gamma$
    (\Cref{thm:diversity-n,thm:diversity-gamma}), $\popt$ is not. Increasing $n$
  can lead to a \textit{less} diverse symmetric socially optimal strategy.}
  \label{fig:ps-to-inf}
\end{figure}
\fi

While the proof is quite involved, the idea is simple: for sufficiently large
$n$, $X(n, p)$ concentrates around $np$, so $s_\gamma(X(n, p)) \approx
(np)^{\gamma}$. Another interesting limiting case is where $\gamma = \infty$,
i.e., players only derive utility when their sampled types are unique (which is
how Scattergories is intended to be played). In this case, equilibria coincide
with optimal strategies under any $s \in \sdown$.

\begin{restatable}[Equivalence when $\gamma = \infty$]{lemma}{infequiv}
  \label{lem:opt-inf-equivalence}
  For all $n \ge 2$, $\bd \in \R_{\ge 0}^K$, $s \in \sdown$,
  \begin{align*}
    \popt(n, \bd, s) = \peq(n, \bd, s_\infty).
  \end{align*}
\end{restatable}

It turns out that $\gamma = \infty$ also yields a bound on the diversity of
equilibria. No $s \in \sall$ yields an equilibrium that is more diverse than the
one induced by $s_\infty$.

\begin{restatable}[$\gamma = \infty$ yields the most diverse
  equilibria]{lemma}{infmostdiverse}
  \label{lem:inf-most-diverse}
  For all $n \ge 2$, $\bd \in \R_{\ge 0}^K$, $s \in \sall$,
  \begin{align*}
    \peq(n, \bd, s) \succ \peq(n, \bd, s_\infty).
  \end{align*}
\end{restatable}

We conclude by extending our asymptotic findings to asymmetric equilibria.
Let $\bP \in
\Delta(\pi)^n$ be a $K \times n$ matrix in which each column $i$ is player $i$'s
strategy.
While
each player may have different strategies in an asymmetric equilibrium $\bP$, we
show that the average strategy in asymmetric equilibria converges to $\peq$ as
the number of players grows.
Formally, let
$\eqsetn{n}$ be the set of (potentially asymmetric) equilibria to $\game(n, \bd,
s_\gamma)$.
\begin{restatable}[Asymmetric equilibria converge to symmetric
  equilibra]{lemma}{asymconverge}
  \label{lem:sym-asym-converge}
  Choose $\{\bPn{n}\}_{n=2}^\infty$ such that $\bPn{n} \in \eqsetn{n}$.
  Let $\barpn{n} = \frac{1}{n} \sum_{i=1}^n \bPn{n}(i)$ be the average strategy
  in $\bPn{n}$.
  Then, for all $k$,
  \begin{align*}
    |\peq(n, \bd, s_\gamma)_k - \barpn{n}_k| \le O\p{\sqrt{\frac{\log n}{n}}}.
  \end{align*}
\end{restatable}
As a result, a generalization of \Cref{lem:n-to-infty} holds for asymmetric
equilibria, which we formalize in \Cref{lem:p-convergence}.

\subsection{Price of anarchy}
\label{sec:poa}

Optimal behavior is more diverse than equilibrium behavior, and while they may
coincide in the limit, equilibria will in general be suboptimal. We conclude
this section by showing that they cannot be \textit{too} suboptimal: the price
of anarchy~\citep{koutsoupias1999worst,roughgarden2002bad} is at most 2.
This holds even if we broaden our scope to consider \textbf{asymmetric}
strategies $\bP \in \Delta(\pi)^n$ as well.
Social welfare is then defined as before, where each player
independently samples from potentially different distributions. We show in
\Cref{app:poa} that $\game$ is a \textit{valid utility
game}~\citep{vetta2002nash}, which directly yields the following result. We also
give a family of instances for which it is tight.

\begin{restatable}[Price of anarchy]{theorem}{poa}
  \label{thm:poa}
  The price of anarchy for $\game$ is 2. Formally, let $\bP$ be any equilibrium
  of $\game(n, \bd, s)$, and let $\bP^*$ be any socially optimal strategy
  matrix. Then,
  \begin{equation}
    \label{eq:poa}
    W(\bP^*) \le 2 W(\bP).
  \end{equation}
  As a special case, the same holds true for symmetric equilibria $\bp$ and optimal
  strategies $\bp^*$:
  \begin{align*}
    W(\bp^*) \le 2 W(\bp)
  \end{align*}
  There exists a family of instances $\{\game(n, \bd^{(n)}, s)\}_{n \to \infty}$
  for which~\eqref{eq:poa} is asymptotically tight.
\end{restatable}

Thus, while selfish behavior leads to lower diversity, its negative consequences
on social welfare are bounded. Having analyzed the setting in which all players
use a single GAIT, we next generalize to the case where players can
strategically choose a GAIT.

\section{Competition across Tools}
\label{sec:inter-tool}

In this section, we allow players to strategically choose a GAIT, represented by
a permutation $\pin{j}$ over $[K]$, from a set of choices $\Pi = \{\pin{1},
\pin{2}, \dots, \pin{J}\}$. Subject to their choice of GAIT, each player chooses
a distribution $\bp \in \Delta(\pin{j})$ that obeys the ordering given by
$\pin{j}$. This strictly generalizes the setting described in \Cref{sec:theory},
in which $\Pi$ was a singleton set. Our results here will not be as general as
before; indeed, symmetric pure strategy equilibria need not exist in this
setting. Our goal will instead be to provide illustrative examples of how
competition across tools might unfold, which will ground our subsequent
empirical results.

\paragraph*{Analyzing market share.}

Our primary focus in this section will be on \textit{market share at
equilibrium}. How many players choose to use each GAIT? A theme in both our
theoretical results here as well as our experiments in
\Cref{sec:pairwise,sec:continuous-exp} is that suboptimal tools can
succeed in competitive markets by capturing particular niches within the market.
Faced with heterogeneous demand, producers strategically exploit heterogeneity
across GAITs, increasing diversity in overall production.

In our theoretical analysis, we will restrict our
attention to pure partially symmetric equilibria, in which all players who
choose $\pin{j}$ also choose the same distribution $\bpn{j}$. Given a set of
GAITs $\Pi$ and values $\bd$, a pure partially symmetric equilibrium can be
specified by $\{m_j\}_{j=1}^J$, the number of players who use GAIT $j$, and
$\{\bpn{j}\}_{j=1}^J$, their respective strategies. Let $\mc E(n, \Pi, \bd)$ be
the set of all such pure partially symmetric equilibria. With this, we can
define market share.
\begin{definition}
  \label{def:market-share}
  Given an instance $(n, \Pi, \bd)$, the market share of GAIT $j$ is
  \ifthenelse{\boolean{smallEqs}}{
  \begin{align*}
    \mc M_j(n, \Pi, \bd)
    \triangleq
    &\{m : \exists \text{ an equilibrium in } \\
    &\mc E(n, \Pi, \bd) \text{ such that } m_j = m\}.
  \end{align*}
}{\begin{align*}
    \mc M_j(n, \Pi, \bd)
    \triangleq \{m : \exists \text{ an equilibrium in } \mc E(n, \Pi, \bd) \text{
    such that } m_j = m\}.
  \end{align*}
}
\end{definition}

In other words, $\mc M_j$ is the set of possible market shares for tool $j$.
Define $\maxm_j(n, \Pi, \bd) \triangleq \max \mc M_j(n, \Pi, \bd)$ and
$\minm_j(n, \Pi, \bd) \triangleq \min \mc M_j(n, \Pi, \bd)$ to be the maximal
and minimal (respectively) market shares that tool $j$ can have at any
equilibrium in $\mc E$.

Intuitively, we might expect market share to be related to the quality of a
GAIT: more players should choose a ``better'' tool over a worse one. Because
tools may take a range of market shares across different equilibria, we need to
be careful in how we formalize this. Let $\Pi = \{\pin{1}, \pin{2}\}$. Tool $1$
is ``better'' than tool $2$ in isolation if $\minm_1(1, \Pi, \bd) = 1$ (and
therefore $\maxm_2(1, \Pi, \bd) = 0$). This means that a single player facing no
competition strictly prefers tool $1$ to tool $2$. Under this condition, we
might expect that with $n$ players, GAIT $1$'s market share would be at least as
high as $2$'s. Conservatively, we can express this as $\maxm_1(n, \Pi, \bd) \ge
\minm_2(n, \Pi, \bd)$, meaning that the ``best'' equilibrium for tool $1$ yields
market share at least as large as the ``worst'' equilibrium for tool $2$.

Our first result in this setting shows that this is not the case: it is possible
for a GAIT that appears better in isolation to get strictly less market share
than a worse one as competition (i.e., the number of players $n$) increases.

\begin{restatable}[The dominant GAIT depends on the strength of
  competition]{lemma}{dominantstrength}
  \label{lem:dominant-strength}
  There exists $(n, \Pi = \{\pin{1}, \pin{2}\}, \bd)$ such that $\minm_1(1, \Pi,
  \bd) = 1$ and $\maxm_1(n, \Pi, \bd) < \minm_2(n, \Pi, \bd)$.
\end{restatable}

All proofs for this section can be found in \Cref{app:pairwise-proofs}.
\Cref{lem:dominant-strength} shows that GAIT quality is not one-dimensional.
A benchmark might suggest that $\pin{1}$ is better than $\pin{2}$, but in a
competitive marketplace, $\pin{2}$ might turn out to have an edge. We next show
conditions under which a tool gets completely locked out of the market.

\begin{restatable}[Strict domination is possible]{lemma}{strict}
  \label{lem:strict-domination}

  Suppose $s(x) < \infty$ for $x \in \N_{> 0}$, and suppose $\exists~\hat k$
  such that $\bd_{\hat k} = 0$. For any $\pin{1}$ such that $\hat k$ is not ranked
  first by $\pin{1}$, there exists $\pin{2}$ such that for all $\Pi \supseteq
  \{\pin{1}, \pin{2}\}$, $\maxm_2(n, \Pi, \bd) = 0$ for all $n$.

\end{restatable}

\Cref{lem:strict-domination} is a bit dense, but the intuition is simple. As
long as there is a 0-value type that $\pin{1}$ doesn't rank first, we can find a
strictly worse GAIT $\pin{2}$ that never gets any market share. The construction
is fairly simple: we obtain $\pin{2}$ from $\pin{1}$ by simply taking the
0-value element $\hat k$ and putting it at the first position of $\pin{2}$. Our
final result states that the perfect GAIT weakly captures the entire market
against \textit{all possible} competition. On the other hand, for any imperfect
tool $\pin{1}$, there exists some other tool $\pin{2}$ that exploits $\pin{1}$'s
imperfection.

\begin{restatable}[Only the perfect tool weakly dominates]{lemma}{perfect}
  \label{lem:perfect-dominates}

  If $\pin{1}$ is a perfect ordering according to $\bd$, i.e., $\bd_{\pin{1}_1}
  \ge \bd_{\pin{1}_2} \ge \dots \ge \bd_{\pin{1}_K}$, then $\maxm_1(n, \Pi, \bd)
  = n$ for all $\Pi \ni \pin{1}$. Otherwise, there exist $\pin{2}$ and $n_0$
  such that for all $n \ge n_0$, $\maxm_1(n, \{\pin{1}, \pin{2}\}, \bd) < n$.
  
\end{restatable}

As a result, a strategic tool developer need not attempt to create a perfectly
aligned tool; instead, they might seek ``competitive alignment'' by filling a
niche left open by existing tools. A tool might, for example, prioritize a
particular visual or linguistic style even if it isn't what the majority of
users want simply because existing tools fail to provide it. This concludes our
theoretical exploration. Next, we validate our results empirically.

\section{Case Study: Scattergories}
\label{sec:empirical}

\begin{figure}[ht]
  \centering
  \includegraphics[width=0.8\textwidth]{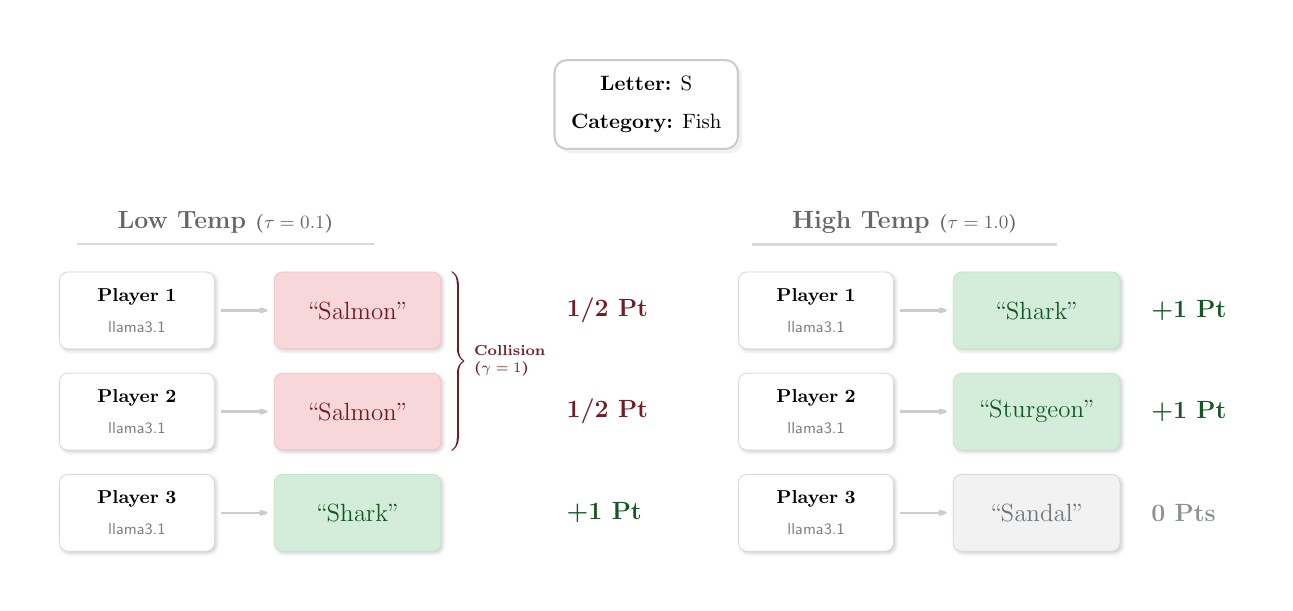}
  \caption{Empirical setup for \Cref{sec:same-model}. At low temperatures,
    collisions are likely. Higher temperatures reduce collisions, but they
  increase the chance of errors (e.g.,~``Sandal'').}
  \label{fig:scat-fig}
\end{figure}

We begin with three sets of experiments. First, we consider a setting where
players use the same LLM to play Scattergories, mirroring the theoretical setup
in \Cref{sec:theory}. We then study the case where players can strategically
choose which LLM to use. In both of these experiments, we allow players to
strategically choose the LLM's temperature but not the prompt, since finding
equilibria in the space of prompts dramatically increases the computational
costs. We investigate the role of prompting in \Cref{sec:as-validation}, finding
that, consistent with \Cref{hyp:variability}, prompt variations do not
substantially change the output ranking of each LLM. Instead, across-model
output differences are much greater than across-prompt differences.

Our empirical setup (shown in \Cref{fig:scat-fig}) is simple: $n$ players each
use an LLM to generate a single response to a Scattergories instance given by a
(letter, category) pair. We take each player's strategy space to be the
temperature $\tau$ at which they operate their LLM. Formally, the LLM outputs a
logit $\ell_t$ for each possible token $t$ in its vocabulary $\mc V$. As is
standard practice \citep[e.g.~][]{hinton2015distilling}, at temperature $\tau$,
the LLM samples token $t$ with probability proportional to $\exp(\ell_t/\tau)$.
Higher temperatures induce more diverse distributions at the individual token
level. In the special case where $\tau = 0$, the LLM deterministically produces
the token with the highest logit, $\argmax_{t \in \mc V} \ell_t$.

Because Scattergories requires complete words or phrases, answers typically
consist of multiple tokens. We generate tokens sequentially until we encounter a
\STOP\ token. (Most LLMs have a built in end-of-message token used to signal
that generation is complete.) We additionally treat all tokens containing
non-alphanumeric characters as \STOP\ tokens.
In all of our experiments, we use score functions of the form $s_\gamma(x) =
x^\gamma$. Although the original rules of Scattergories call for $\gamma =
\infty$, we use other values because equilibrium and optimal strategies coincide
under $s_\infty$ (see \Cref{lem:opt-inf-equivalence}).

\paragraph*{Data and language models.}

\begin{table}[ht]
  \centering
  \begin{tabularx}{\textwidth}{lX}
    \toprule
    Role
    &
    Content \\
    \midrule
    System
    & 
    \texttt{You are a helpful assistant. Answer in as few words as possible,
    with no explanations.} \\
    User
    &
    \texttt{We are playing a game of Scattegories. I will give you a letter and
      a category, and you will give me a word or short phrase that starts with
      that letter and matches the category. For example, if I say "Fruit" and "A,"
    you could say "Apple" or "Apricot."} \\
    Assistant
    &
    \texttt{I understand.} \\
    User
    &
    \texttt{Letter: C\textbackslash nCategory: Countries} \\
    Assistant
    &
    \texttt{Canada} \\
    User
    &
    \texttt{Letter: X\textbackslash nCategory: Instruments} \\
    Assistant
    &
    \texttt{Xylophone} \\
    User
    &
    \texttt{Letter: \{letter\}\textbackslash nCategory: \{category\}} \\
    \bottomrule
  \end{tabularx}
  \caption{Default prompt}
  \label{tab:default-prompt}
\end{table}

The categories we use for our games come from a free online Scattergories
game.\footnote{\url{https://swellgarfo.com/scattergories/}} We randomly sample
25 instances, each of which is a (letter, category) pair, and we average across
them in all of our results. A full list of our sample can be found in
\Cref{tab:instances}. Our experiments compare six LLMs, each of which is a
relatively small open-source language model: gemma-2-2b-it (\gemma),
Llama-3.1-8B-Instruct (\llamaone), Llama-3.2-1B-Instruct (\llamatwo),
Mistral-7B-Instruct-v0.3 (\mistral), Nemotron-Mini-4B-Instruct (\nemotron), and
Phi-3.5-mini-instruct (\phithree). In addition to these, we use a seventh
language model, Qwen2.5-7B-Instruct (\qwen), to validate answers. See
\Cref{tab:models} for more details.

While relying on language models to evaluate one another comes with several
downsides~\citep{chiang2023can,zheng2023judging,wang2023large}, our setting has
a number of mitigating factors that make it a more reasonable choice. Even with
human players, Scattergories works via self-evaluation. Evaluation is fairly
simple: each evaluation is a binary response to a short answer.
We do not seek to make strong claims about the correctness of the evaluations
made by \qwen, although manual spot-checking reveals that its judgements are
reasonable. Instead, we view our results as confirmation that there exists a
reasonable setting in which our theoretical results hold. We validate answers
using a separate language model to minimize the risk of
self-preferencing~\citep{panickssery2024llm,koo2023benchmarking,liu2023llms}.
One downside we do not address in this work is the fact that we treat
restatements of the same answer (e.g.,~``swordfish'' and ``sword fish'') as
distinct. For both answer generation and validation, we include two in-context
examples to improve performance (see
\Cref{tab:default-prompt,tab:verification-prompt} for the full prompts).

\ifdefined\smallfigs
\begin{wrapfigure}{L}{0.6\textwidth}
  \centering
  \includegraphics[width=0.6\textwidth]{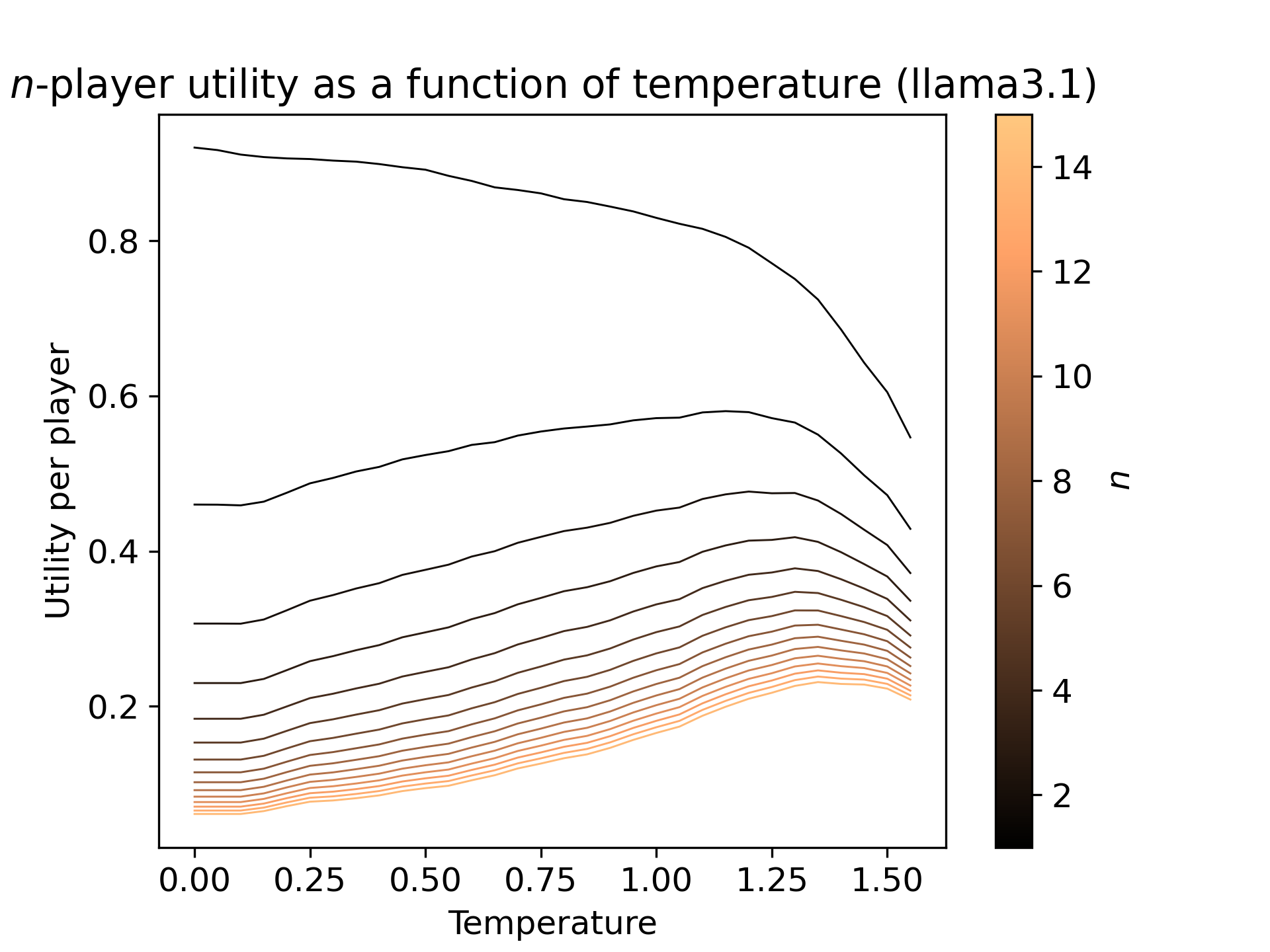}
  \caption{Per-player utility as a function of $\tau$ for different values of
    $n$. The highest line corresponds to $n=1$, and the lowest
  corresponds to $n=15$. $\gamma = 1.0$.}
  \label{fig:llama3.1-opt-over-temp}
\end{wrapfigure}
\else
\begin{figure}[ht]
  \centering
  \includegraphics[width=0.8\textwidth]{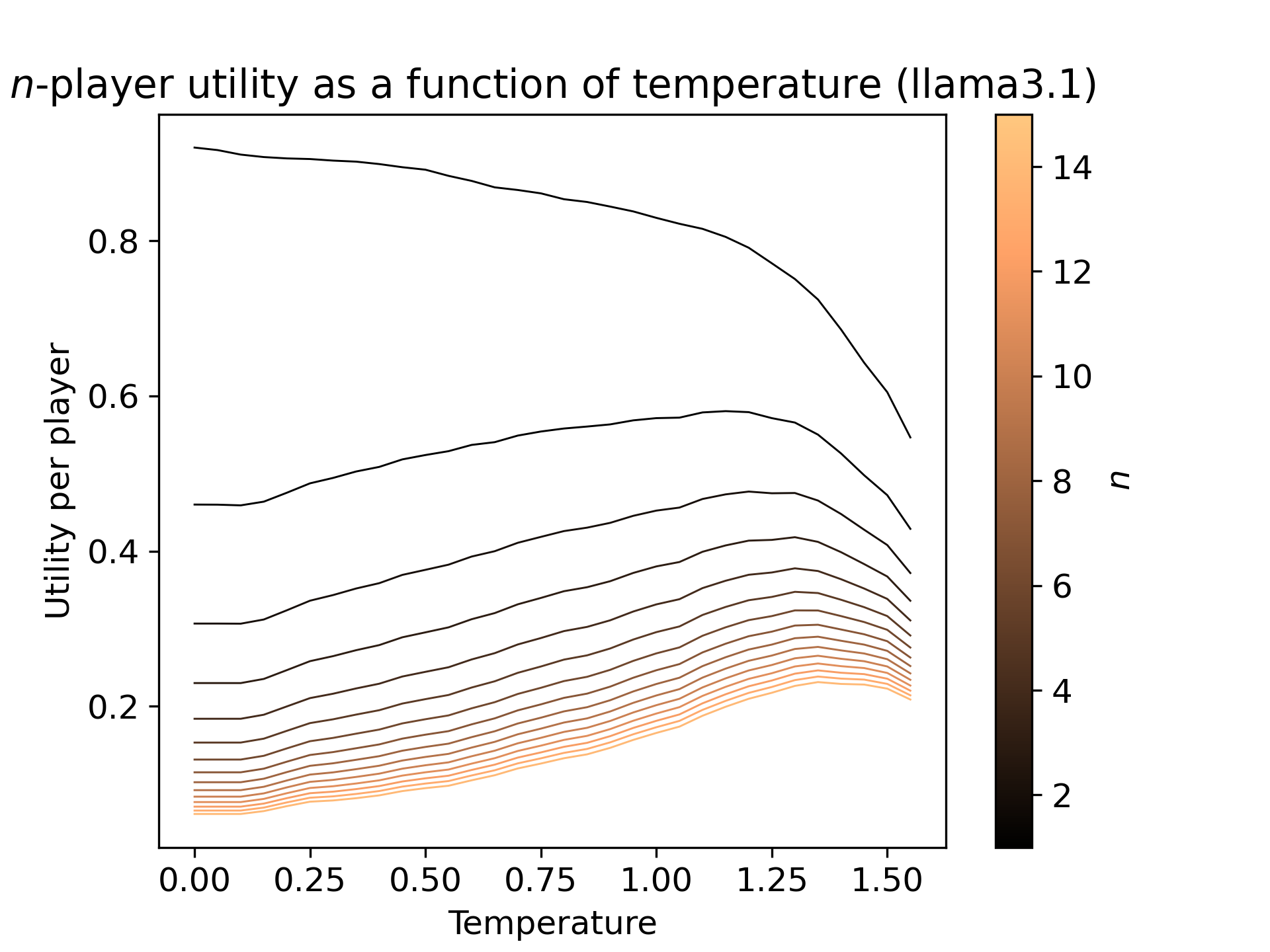}
  \caption{Per-player utility as a function of $\tau$ for different values of
    $n$. The highest line corresponds to $n=1$, and the lowest
  corresponds to $n=15$. $\gamma = 1.0$.}
  \label{fig:llama3.1-opt-over-temp}
\end{figure}
\fi

\begin{figure}[ht]
  \centering
\end{figure}

\begin{figure}[ht]
  \centering
  \includegraphics[width=0.9\textwidth]{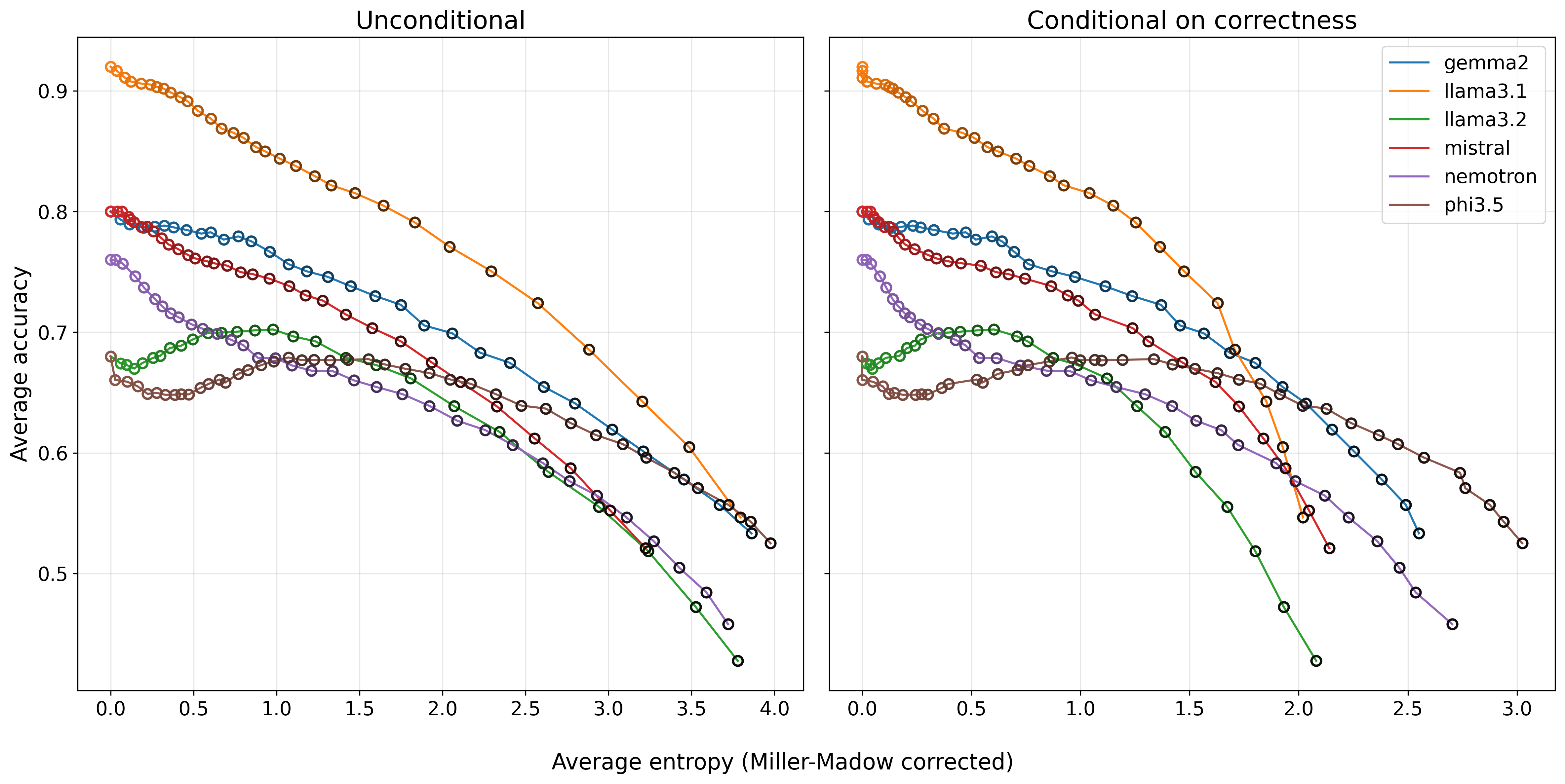}
  \caption{Average accuracy (as verified by \qwen) vs entropy for each LLM.
    Each point represents a different temperature, using darker circles for
    higher temperatures. As one might expect, entropy increases monotonically
    with temperature.  The left panel shows
    average entropy for the unconditional distribution, while the right panel
    shows entropy conditional on correctness. Observe that \llamaone\ falls from
    the accuracy-entropy frontier at high temperatures when we only consider
  correct answers.}
  \label{fig:entropies}
\end{figure}

Even in this limited setting where answers are short sequences of tokens,
computing the distribution of answers for a given (language model, temperature,
instance) combination is still non-trivial. Ideally, we would simply write down
an LLM's output distribution as a function of temperature $\tau$. Indeed, if we
only considered single-token answers, this would be straightforward: each token
$t$ is sampled with probability $\propto \exp(\ell_t/\tau)$. If an LLM's
vocabulary contains $|\mc V|$ tokens, we need $\Theta(|\mc V|)$ space just to
write down the distribution. For multi-token answers, this becomes prohibitive,
requiring $\Theta(|\mc V|^d)$ time and space, where $d$ is the maximum answer
length.\footnote{For context, \nemotron\ has a valid vocabulary size $|\mc V| =
$ 169,159, and we allow responses up to 6 tokens.}

Instead, we take a sampling-based approach: we can simulate games of
Scattergories by sampling completions from an LLM (at each temperature
$\tau$) for each player. While this approach is more tractable, it comes with
its own disadvantages. We generate samples independently for each $\tau$ under
consideration, meaning we must take temperatures at discrete points (we choose
intervals of $0.05$). With enough samples (we choose 2,000 for each (language
model, temperature, instance) combination), we can get fairly tight estimates of
players' utility. See \Cref{app:exp-details} for details on how we use samples
optimally and bound error.

\subsection{Competition under a single tool}
\label{sec:same-model}

We begin by analyzing the setting in which all players share the same LLM and
can only vary the temperature $\tau$. \Cref{fig:llama3.1-opt-over-temp} shows
utility as a function of $\tau$ in the $n$-player game with the score function
$s_1(x) = x$ when all players use \llamaone. Observe that utility is
approximately quasiconcave as a function of $\tau$: it increases with $\tau$ up
to a certain point, and then decreases monotonically. Moreover, the peak occurs
at larger values of $\tau$ as competition gets stronger with larger values of
$n$. Intuitively, lower temperatures are most likely to yield valid answers but
produce low-entropy distributions, since entropy increases monotonically with
temperature (see \Cref{fig:entropies}, left panel).\footnote{Because entropies
are estimated from incomplete samples, we apply the Miller-Madow correction
\citep{miller1955note,madow1948limiting}.} While this is good for mild
competition, higher temperatures trade off quality for diversity, which
increases welfare under stronger competition. Eventually, at higher
temperatures, LLM performance degrades dramatically, and utilities drop off. We
provide analogous figures for the other five LLMs in
\Cref{fig:utility_over_temp}.

\paragraph*{Diversity and utility of equilibria.}

Because \Cref{as:ranked} is not strictly true (we analyze the extent to which it is
violated later in this section), pure symmetric equilibria do not necessarily
exist here.
Nevertheless, let $\taueq(n, \gamma)$ be the symmetric equilibrium temperature
if it exists for a single LLM given $n$ players and score function
$s_\gamma$.\footnote{In principle, there could be multiple pure strategy
equilibria. Empirically, this is not the case here.} From our results in
\Cref{sec:theory}, we expect $\taueq$ to increase in both $n$ and $\gamma$
(\Cref{thm:diversity-n,thm:diversity-gamma}). We'd also expect the socially
optimal temperature $\tauopt(n, \gamma)$ to be larger than the corresponding
equilibrium temperature $\taueq(n, \gamma)$ (\Cref{thm:diversity-opt}).

\Cref{fig:opt_and_eq_temp_over_n} supports our theoretical results in
\Cref{thm:diversity-n}, that stronger competition induces more diversity (larger
values of $\taueq$). Analogous results hold as we vary $\gamma$
(\Cref{thm:diversity-gamma}) in \Cref{fig:opt_and_eq_temp_over_gamma}.
Equilibria reflect a balance between the benefits of diversity and the costs of
invalid answers that occur at high temperatures. Symmetric pure equilibria do
not always exist (e.g.,~\phithree\ with $n=4$). When they do exist,
\Cref{fig:opt_and_eq_temp_over_n,fig:opt_and_eq_temp_over_gamma} show that
$\taueq(n, \gamma) \le \tauopt(n, \gamma)$, as \Cref{thm:diversity-opt}
suggests.

\ifdefined\smallfigs
\begin{figure}[ht]
  \centering
  \begin{subfigure}[t]{0.48\textwidth}
    \centering
    \includegraphics[width=\textwidth]{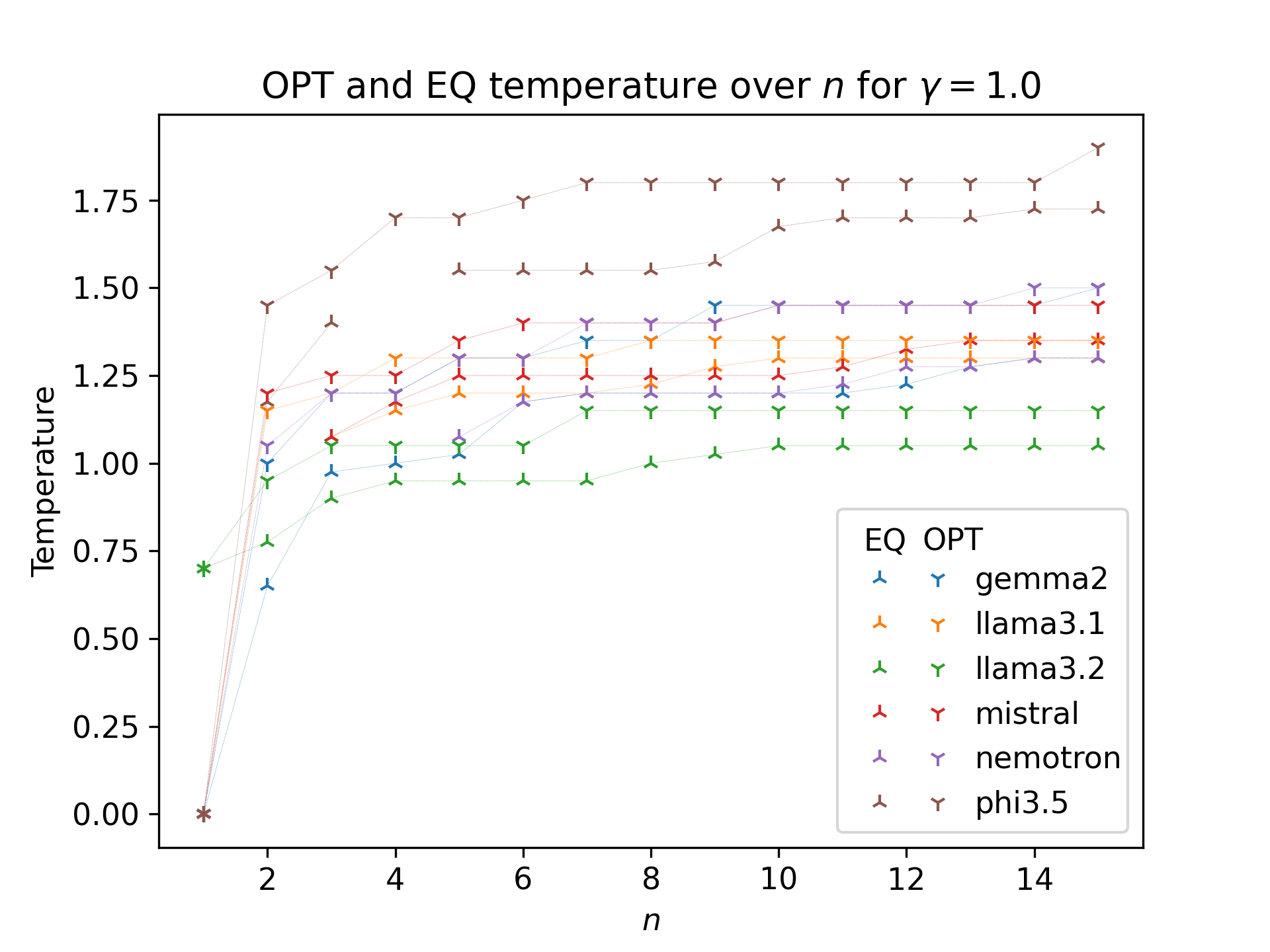}
    \caption{Socially optimal and equilibrium temperatures for each language
      model as a function of $n$. Observe that (1) all curves are increasing in
      $n$, and (2) equilibrium temperatures are lower than their socially
    optimal counterparts.}
    \label{fig:opt_and_eq_temp_over_n}
  \end{subfigure}
  \hfill
  \begin{subfigure}[t]{0.48\textwidth}
    \centering
    \includegraphics[width=\textwidth]{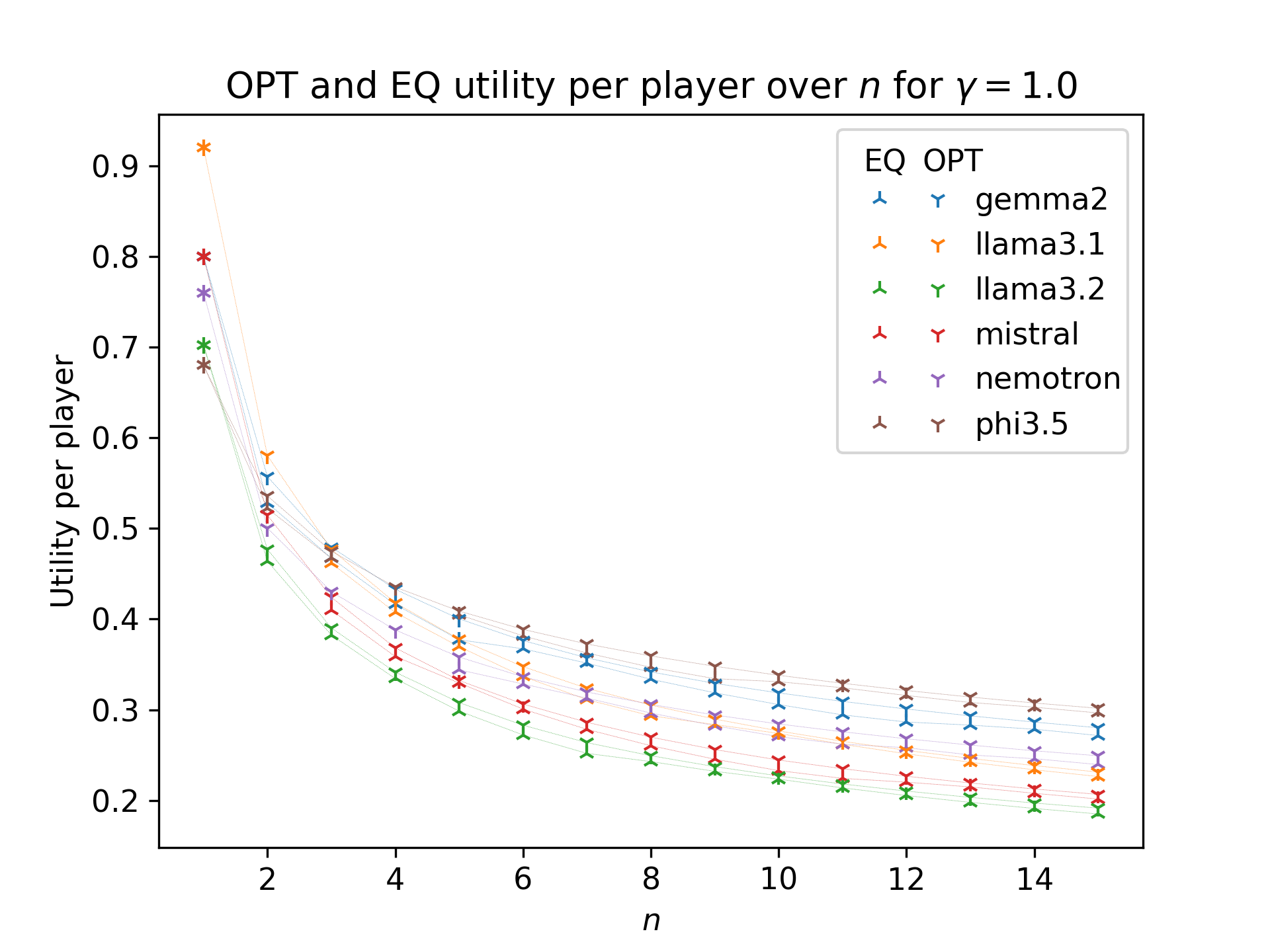}
    \caption{Socially optimal and equilibrium utility (per player) for each
      language model as a function of $n$. Observe that (1) welfare is
      decreasing in $n$, and (2) equilibrium welfare is not too far from
      optimal. \phithree\ and \gemma\ perform better in the presence of stronger
    competition.}
    \label{fig:opt_and_eq_util_over_n}
  \end{subfigure}
  \caption{Equilibrium vs. optimal temperatures and utilities. In both plots,
  $\gamma = 1.0$.}
\end{figure}
\else
\begin{figure}[ht]
  \centering
  \includegraphics[width=0.8\textwidth]{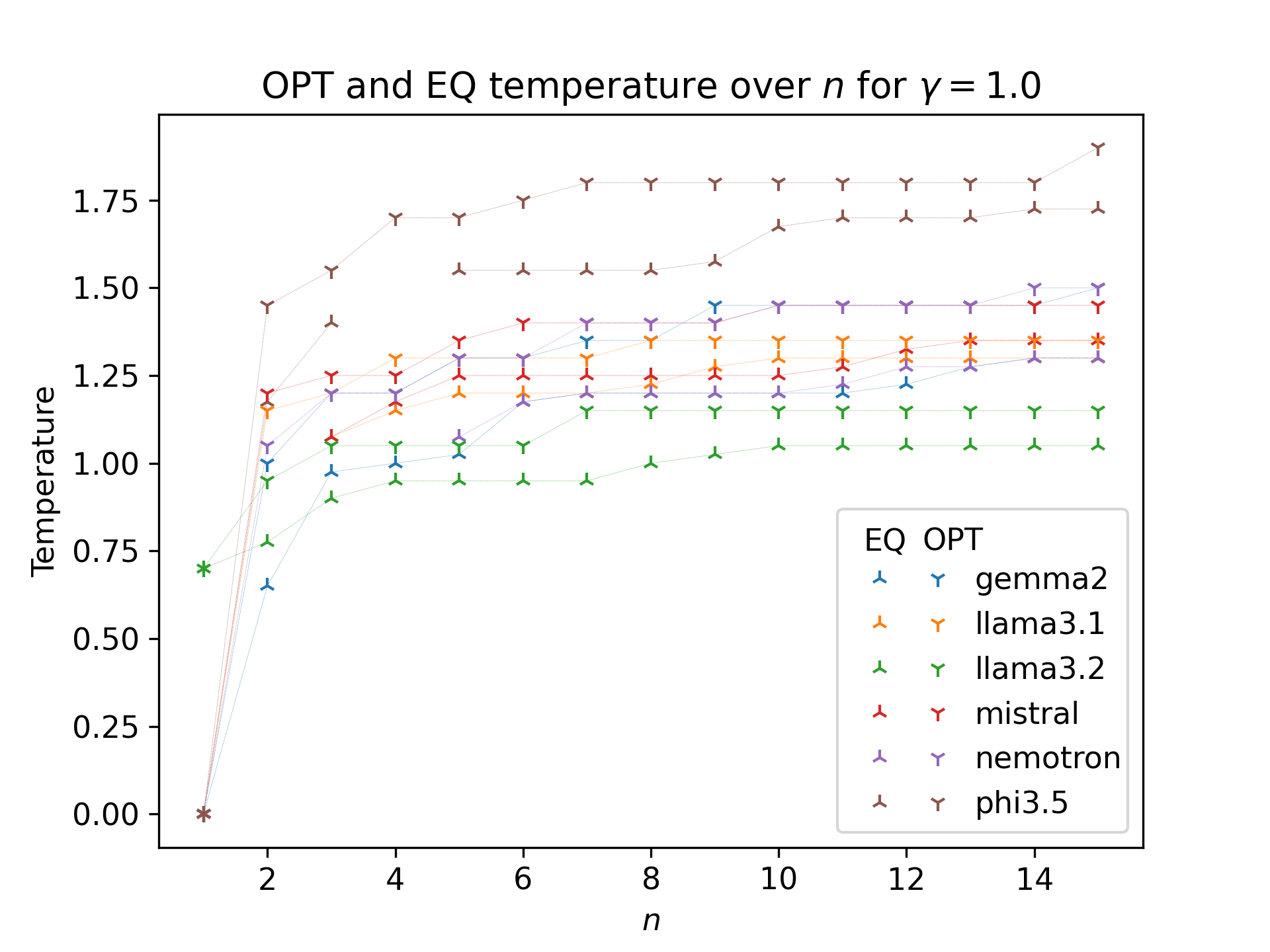}
  \caption{Socially optimal and equilibrium temperatures for each language
    model as a function of $n$. Observe that (1) all curves are increasing in
    $n$, and (2) equilibrium temperatures are lower than their socially
  optimal counterparts. $\gamma = 1.0$}
  \label{fig:opt_and_eq_temp_over_n}
\end{figure}
\begin{figure}[ht]
  \centering
  \includegraphics[width=0.8\textwidth]{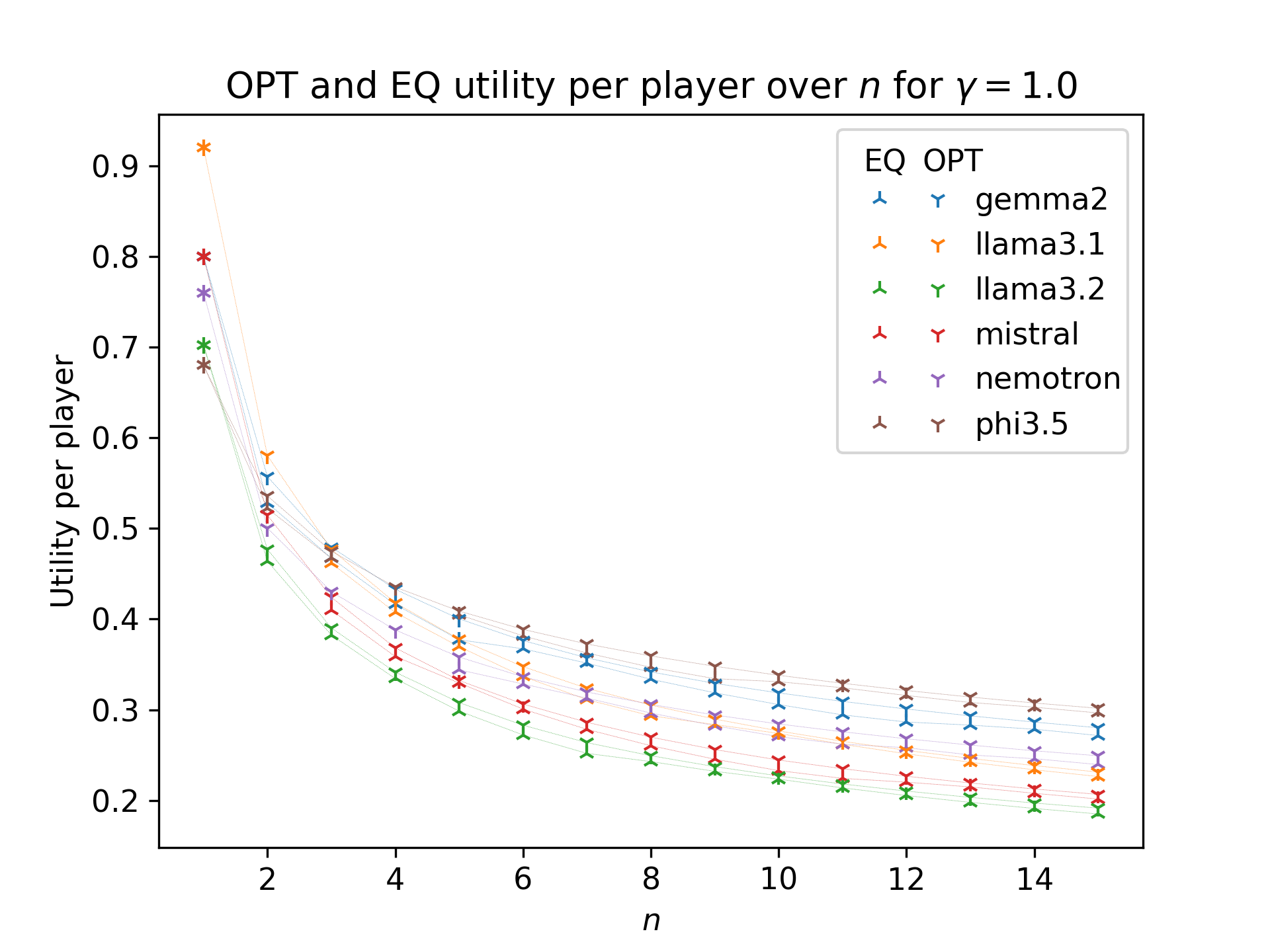}
  \caption{Socially optimal and equilibrium welfare (per player) for each
    language model as a function of $n$. Observe that (1) welfare is
    decreasing in $n$, and (2) equilibrium welfare is not too far from
    optimal. \phithree\ and \gemma\ perform better in the presence of stronger
  competition. $\gamma = 1.0$.}
  \label{fig:opt_and_eq_util_over_n}
\end{figure}
\fi

Next, we consider players' utilities in \Cref{fig:opt_and_eq_util_over_n}.
Equilibrium utilities are, of course, lower than would be socially optimal, but
the gap between them is not substantial; it is well within the factor of two
predicted by our price of anarchy bound in \Cref{thm:poa}.

Perhaps more interestingly, we find strong evidence that the relative quality of
LLMs depends on the strength of competition. With $n=1$, \llamaone\ yields the
best performance out of all of our LLMs, providing a valid answer 92\% of the
time at a temperature of $\tauopt(1, \cdot) = \taueq(1, \cdot) =
0.0$.\footnote{When $n=1$, the score function $s$ does not matter.} In contrast
\phithree\ only yields a valid answer 68\% of the time, also at an optimal
temperature of $0.0$. According to this naive benchmark, it might seem that
\llamaone\ is the ``better'' of the two LLMs in the Scattergories context.
However, this comparison masks a surprising fact: under competition, the
relative performance between these LLMs flips. \phithree\ outperforms
\llamaone\ when there are either more players
(\Cref{fig:opt_and_eq_util_over_n}) or stronger penalties for congestion
(\Cref{fig:opt_and_eq_util_over_gamma}). See \Cref{app:additional-figures} for
analogous plots varying $\gamma$.

Intuitively, \llamaone\ is a ``one-trick pony'' in this setting. Its best answer
is usually valid, but diversifying its distribution fails to place much
probability mass on distinct, valid answers. \phithree, \gemma, and \nemotron,
on the other hand, make more mistakes when producing a single answer, but are
able to draw from a large set of valid answers at higher temperatures.
\Cref{fig:entropies} makes this more explicit. The left panel shows how
\phithree\ and \gemma\ approach the ``frontier'' traced out by \llamaone\ at
high temperatures.\footnote{We choose maximal temperatures for each LLM such
that no equilibrium or optimal strategy relies on a temperature near our
cutoffs; see \Cref{app:exp-details} for more details.} While these LLMs appear
to be comparable in their accuracy-entropy trade-offs, the right panel of
\Cref{fig:entropies} reveals a stark difference between them: If we restrict our
attention to their output distributions \textit{conditioned on their responses
being correct}, we find that \llamaone's entropy drops dramatically relative to
other LLMs. Its overall increased entropy at high temperatures appears to come
from its ability to sample from a diverse distribution over \textit{incorrect}
answers, whereas \phithree\ and \gemma\ sample from diverse distributions over
\textit{correct} answers. Our competitive framework thus reveals subtle
distinctions in how LLMs productively diversify their output distributions,
which high-level accuracy-diversity trade-offs fail to capture.

One possible explanation for these findings lies in the differences in
post-training strategies across LLMs. \citet{dubey2024llama} describe a
complex post-training pipeline involving ``millions of human instructions and
preference judgments'' for \llamaone, whereas \citet{abdin2024phi} describe a
targeted attempt to produce a high-quality LLM with minimal data for
\phithree. Several studies have drawn the link between post-training and mode
collapse~\citep{wu2024generative, kirk2023understanding, park2024diminished,
padmakumar2023does}, where loss functions steer LLMs towards a single
``correct'' answer at the expense of high-quality responses across the entire
distribution~\citep{siththaranjan2023distributional}. Thus, it is plausible that
\llamaone's more extensive post-training process shifts its performance profile
towards higher accuracy of its first-choice answer at the expense of quality
across its entire output distribution. We caution that our findings on these
particular LLMs are specific to our setting and should not be used to draw more
general conclusions about them without more comprehensive testing. Future work
could compare these LLMs to their pre-trained precursors
(i.e.,~\textsf{Llama-3.1-8B-Instruct} vs. \textsf{Llama-3.1-8B}) to provide a
more complete picture of the role of post-training in the accuracy-diversity
trade-off.

\subsection{Pairwise competition across tools}
\label{sec:pairwise}

\begin{figure}[ht]
  \centering
  \begin{subfigure}[t]{0.48\textwidth}
    \centering
    \includegraphics[width=\textwidth]{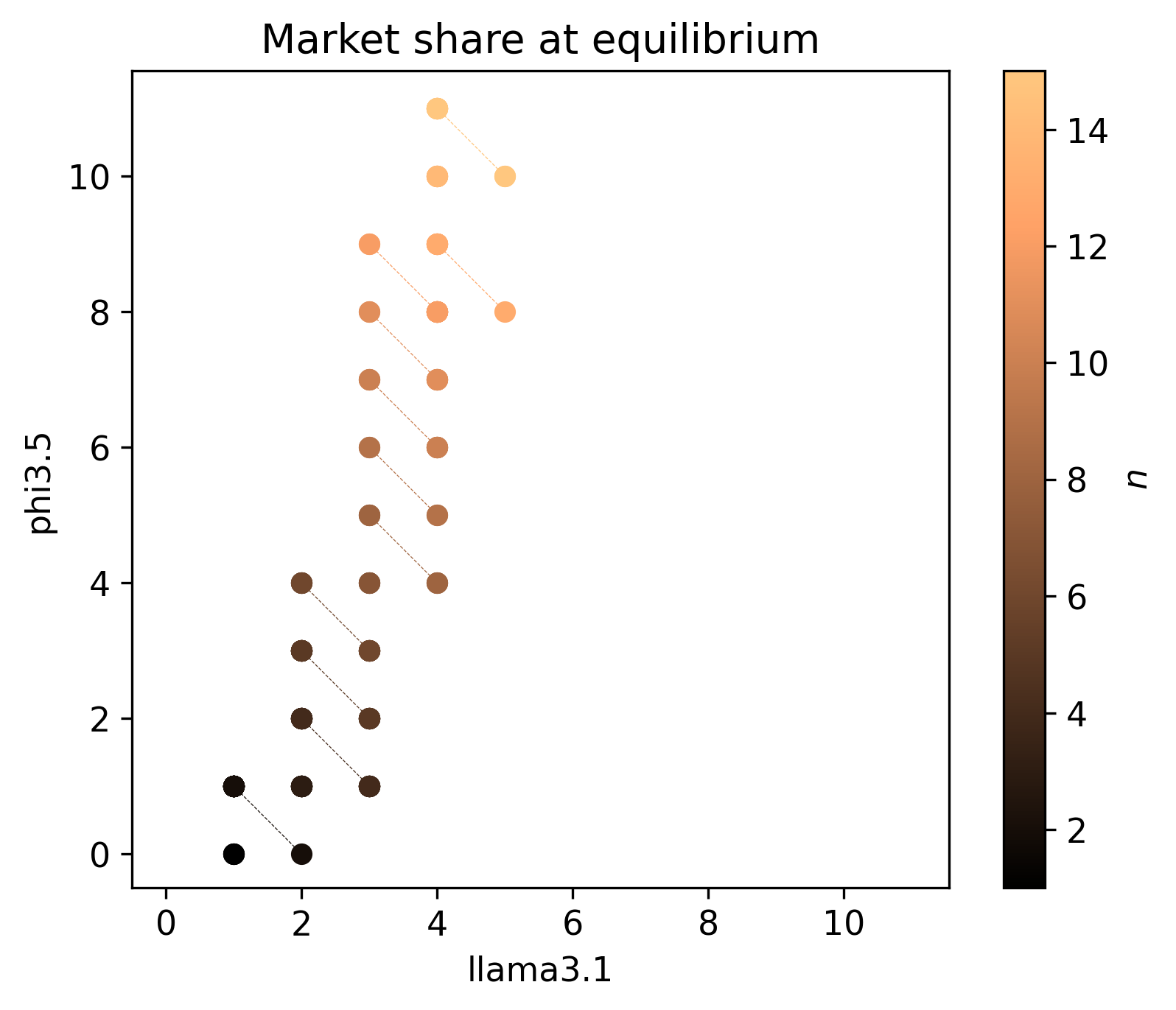}
    \caption{Market shares in pairwise equilibria between \llamaone\ and
    \phithree.}
    \label{fig:llama3.1-phi3.5-market-shares}
  \end{subfigure}
  \hfill
  \begin{subfigure}[t]{0.48\textwidth}
    \centering
    \includegraphics[width=\textwidth]{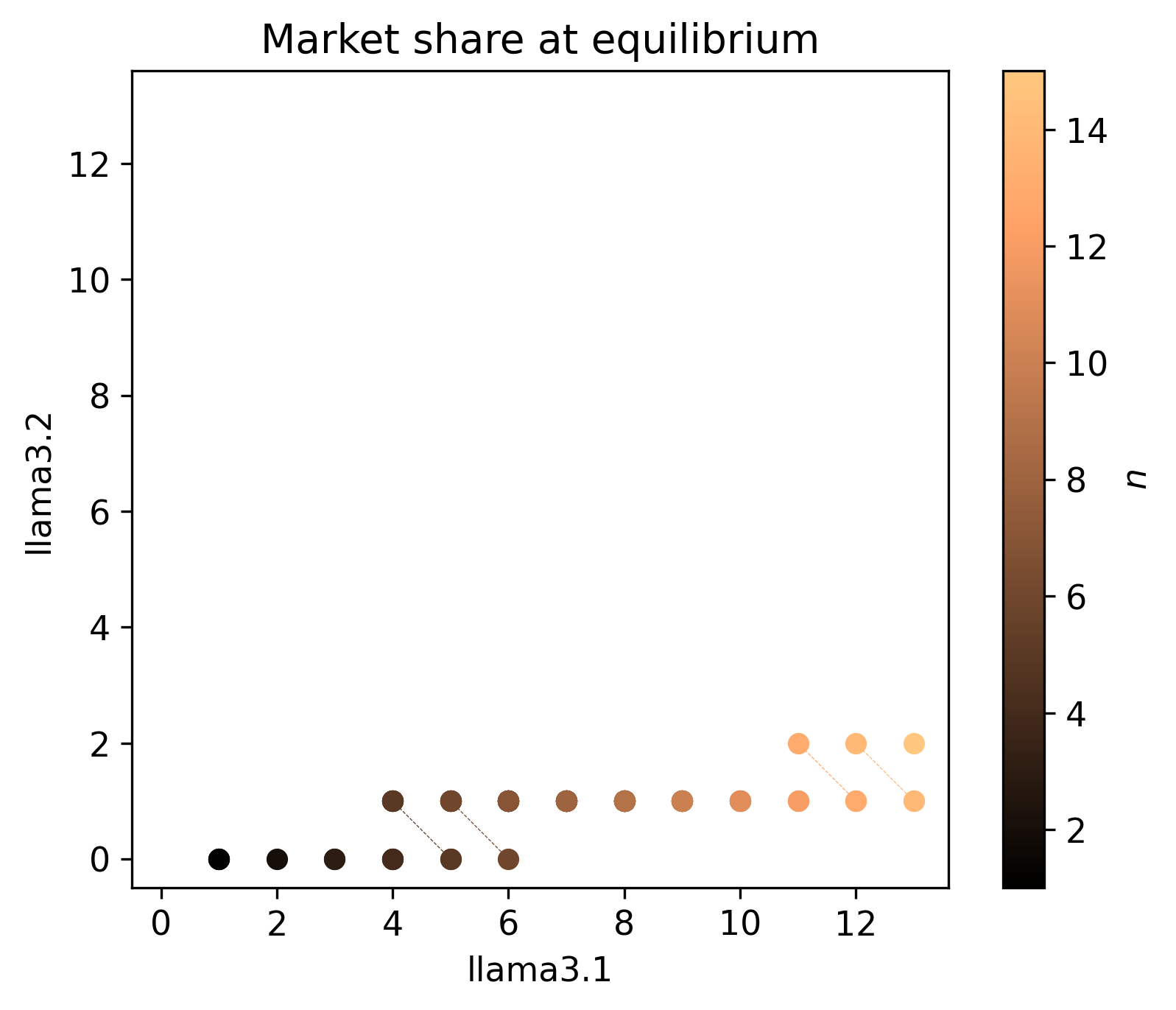}
    \caption{Market shares in pairwise equilibria between \llamaone\ and
    \llamatwo.}
    \label{fig:llama3.1-llama3.2-market-shares}
  \end{subfigure}
  \caption{Each dot represents a possible pair $(m_1, m_2)$ of equilibrium
    market shares for each language model. Dashed lines connect points with the
  same total number of players $n$, with a maximum of $n=15$. $\gamma = 1.0$.}
  \label{fig:market-shares} 
\end{figure}

Next, we consider the setting in \Cref{sec:inter-tool}, where players
strategically choose both a GAIT and a temperature $\tau$. For simplicity, we
limit ourselves to pairwise comparisons, and as before, our focus will be on
market share at partially symmetric pure strategy equilibria.

Based on \Cref{fig:opt_and_eq_util_over_n}, we might expect that the market
share of each GAIT will shift dramatically as competition grows. For example,
given the choice between \llamaone\ and \phithree, a single player ($n=1$) would
clearly prefer \llamaone. But with more players, we could see a mix where some
choose \llamaone\ and some choose \phithree, and we might anticipate that for
sufficiently large $n$, a majority of players would choose \phithree. This is
precisely what \Cref{fig:llama3.1-phi3.5-market-shares} shows. Each point $(m_1,
m_2)$ represents market shares for \llamaone\ and \phithree\ respectively at an
equilibrium with $n = m_1 + m_2$ players. With $n=1$ player, \llamaone\ is
optimal. But for larger $n$, the clear majority of players prefer \phithree. As
in \Cref{lem:dominant-strength}, the ``better'' language model in isolation
cedes market power under competition.

Next, we ask whether one language model ``dominates'' another in the sense of
\Cref{lem:strict-domination}. In our sample, for any pair of language models,
there exists sufficiently large $n$ such that both get a nonzero share of the
market. For example, as we can see from
\Cref{fig:opt_and_eq_util_over_n,fig:opt_and_eq_util_over_gamma}, \llamaone\
yields higher equilibrium utility than \llamatwo\ at any value of competition,
either via $n$ or $\gamma$. Despite this, we find that \llamatwo\ can still get
nonzero market share (\Cref{fig:llama3.1-llama3.2-market-shares}). This is
perhaps even more surprising given that \llamaone\ and \llamatwo\ are from the
same family of language models, meaning their output distributions are
especially likely to be homogeneous (see \Cref{fig:spectral-embedding} for
evidence to this effect). Our results show that despite this
relationship, \llamatwo\ can still prove useful to a (relatively small) fraction
of the market. The same is true across all pairs of language models in our
sample (\Cref{fig:all-pairwise}). Even without any sort of price
differentiation, the market naturally supports, and indeed encourages,
competition. Of course, had we included a sufficiently low-quality language
model in our sample, it would likely fail to capture any market share. We also
do not find evidence that any of our language models is close enough to the
perfect GAIT for \Cref{lem:perfect-dominates} to hold.

\subsection{Validating \Cref{as:ranked}}
\label{sec:as-validation}

\definecolor{successbg}{HTML}{D4EDDA} 
\definecolor{successborder}{HTML}{C3E6CB}
\definecolor{successtext}{HTML}{155724}
\definecolor{failbg}{HTML}{F8D7DA} 
\definecolor{failborder}{HTML}{F5C6CB}
\definecolor{failtext}{HTML}{721C24}
\definecolor{neutralbg}{HTML}{F2F2F2} 
\definecolor{neutralborder}{HTML}{D6D8DB}
\definecolor{neutraltext}{HTML}{6c757d}
\definecolor{primarytext}{HTML}{212529}

\definecolor{cT1}{HTML}{0d6efd}
\definecolor{cT1bg}{HTML}{cfe2ff}
\definecolor{cT2}{HTML}{fd7e14}
\definecolor{cT2bg}{HTML}{ffe5d0}
\definecolor{cTh1}{HTML}{6f42c1}
\definecolor{cTh1bg}{HTML}{e0cffc}
\definecolor{cTh2}{HTML}{20c997}
\definecolor{cTh2bg}{HTML}{cef4e6}

\newcommand{\llama}{\textsf{llama3.1}}
\newcommand{\phiM}{\textsf{phi3.5}}

\begin{figure}[ht]
  \centering

  \includegraphics[width=\textwidth]{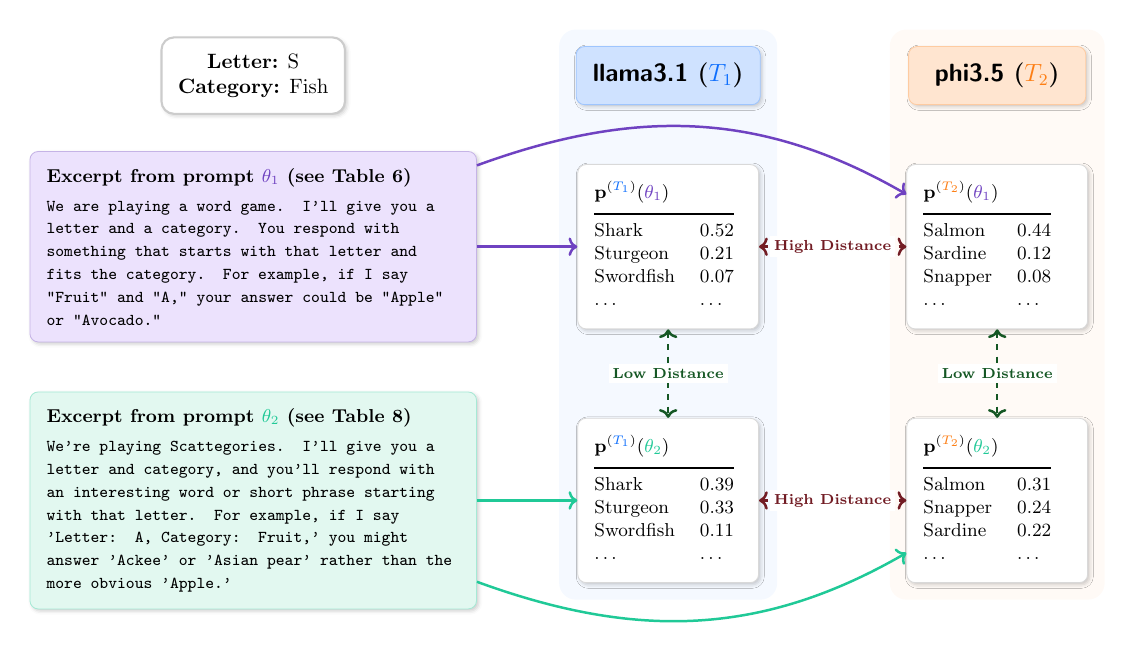}

\caption{Illustration of our setup for \Cref{sec:as-validation}. We compute
  distributions over answers for multiple LLMs and prompts (at a variety of
  temperatures). We compute distances between distributions based on how
  ``mis-ranked'' they are relative to one another. In this example, while the two
  rankings from \phithree\ are not identical, they are ``close'' in the sense that
  they can be made identical by shifting a small amount of probability mass. At
  a high level, we find that differences across prompts are much less salient
  than differences across tools. We use prompts that are quite different to
  capture potential strategic behavior---for example, one explicitly mentions
  Scattergories and provides examples of less ``obvious'' answers, whereas the
  other does not. Our experiments use different in-context examples for each
prompt. See \Cref{app:prompts} for further details.}

  \label{fig:ranking-diagram}
\end{figure}

\begin{figure}[ht]
  \centering
  \includegraphics[width=0.8\textwidth]{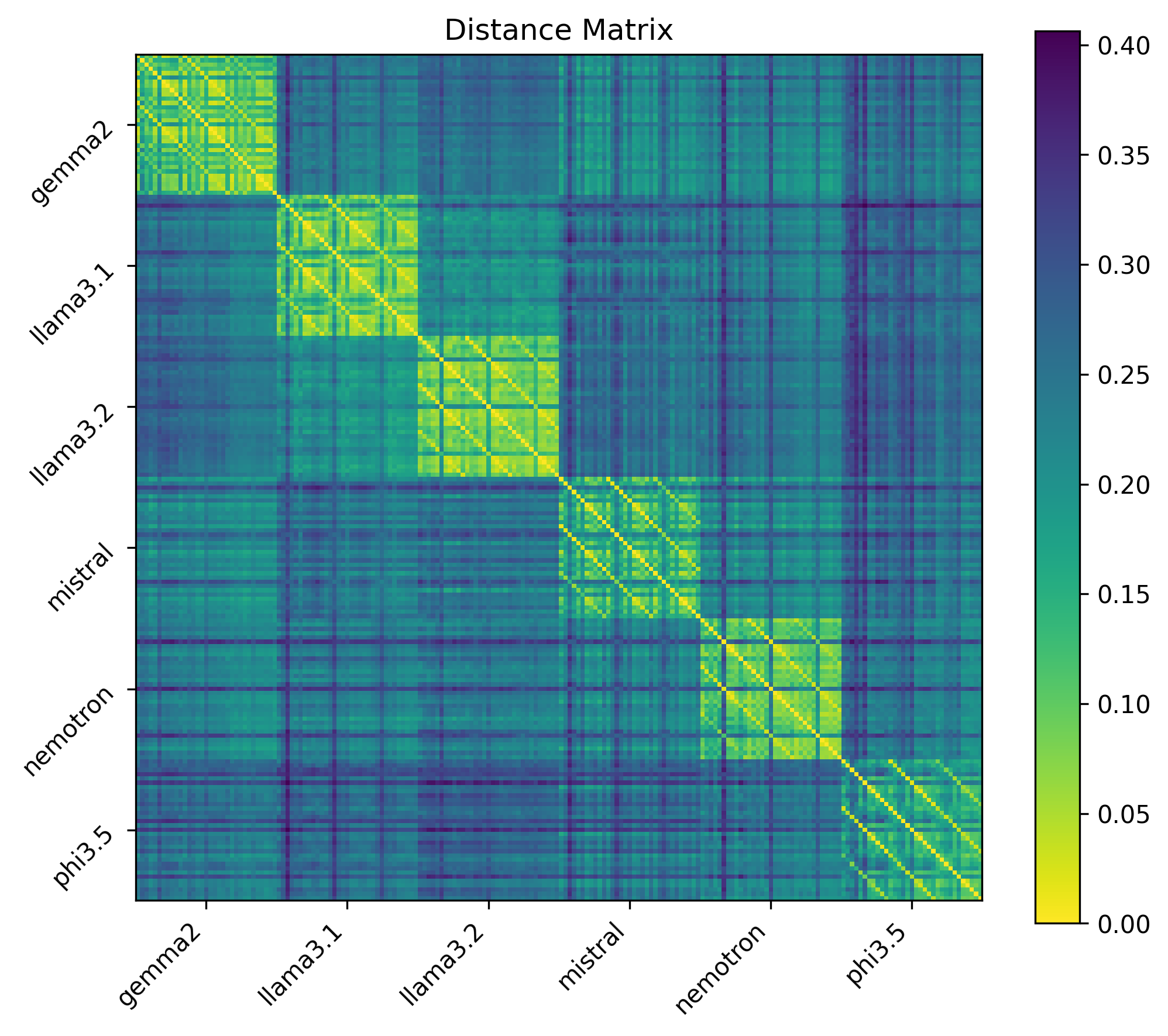}
  \caption{Distance matrix $D$ defined in~\eqref{eq:D-def}.}
  \label{fig:distance-matrix}
\end{figure}

\begin{figure}[ht]
  \centering
  \includegraphics[width=0.8\textwidth]{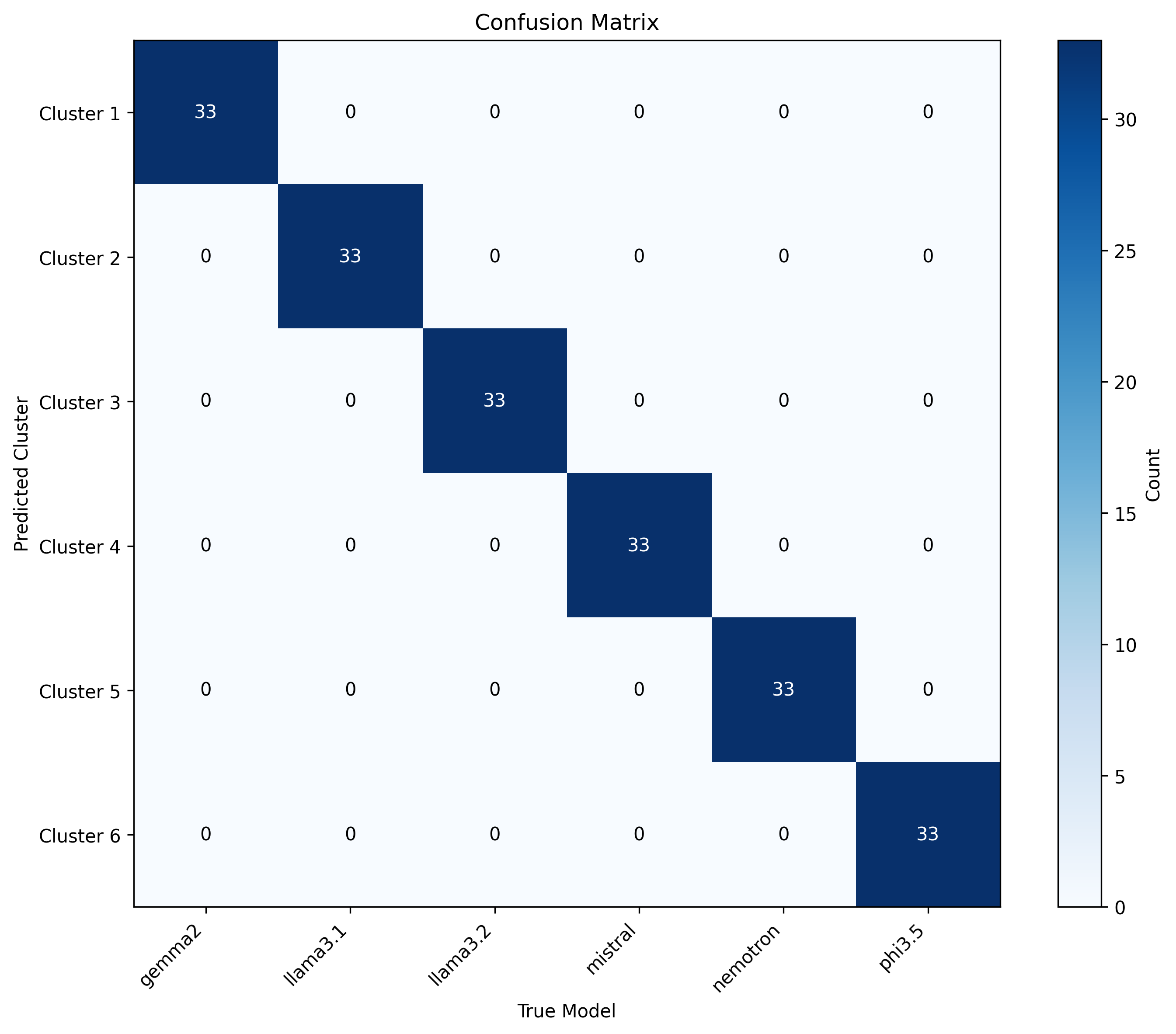}
  \caption{Confusion matrix from spectral clustering.}
  \label{fig:confusion}
\end{figure}

Our theoretical results in \Cref{sec:theory,sec:inter-tool} depended on
\Cref{as:ranked}. Empirically, we found that our conclusions held despite the
fact that our assumptions were not strictly met in practice. Here, we provide
additional evidence in support of \Cref{as:ranked} by measuring the variability
in rankings induced by different inputs. \Cref{fig:ranking-diagram} illustrates
our approach. We introduce 10 additional prompt variations (see
\Cref{app:exp-monoculture} for details), and we take 3 evenly spaced
temperatures per model for a total of $|\Theta_I| = 33$ input variants. If
\Cref{as:ranked} (approximately) holds, we'd expect the rankings over types to
be relatively insensitive to variations of $\theta \in \Theta_I$, but to
potentially be different across different GAITs.

To validate this empirically, we need a notion of distance between distributions
that penalizes inversions. While Kendall tau distance is a natural candidate,
because many outputs appear with probability approximately zero, it would be
dominated by inversions between outputs that almost never occur. Instead, we
define the following weighted notion of inversions, which penalizes inversions
that are in a sense ``larger'':
\begin{equation}
  \label{eq:WI-def}
  \WIargs{\bp}{\bp'}
  \triangleq \min_{\bp^* : \pi(\bp^*) = \pi(\bp)} \dTV(\bp^*, \bp'),
\end{equation}
where $\dTV$ is the total variation distance. Intuitively, $\WI$ reflects the
minimal change that would need to be made to $\bp'$ in order for it to share
$\bp$'s ranking over responses.
To make this symmetric, define
\begin{equation}
  \label{eq:WI-avg-def}
  \WIavgargs{\bp}{\bp'}
  \triangleq \frac{\WIargs{\bp}{\bp'} + \WIargs{\bp'}{\bp}}{2}
\end{equation}
With this definition, we can compare the variability due to changes in $\theta$
to the variability across GAITs.

Consider a graph where node $v$ represents the pair $(T_v, \theta_v)$.
Define the weight of the edge between nodes $v$ and $w$ to be the distance
\begin{equation}
  \label{eq:D-def}
  D_{vw}
  =  D_{wv}
  \triangleq
  \WIavgargs{\bpn{T_v}(\theta_v)}{\bpn{T_w}(\theta_w)}
\end{equation}
averaged across all instances $I$.\footnote{The full distributions $\bp$ are too
  large to compute, so we use heuristics to approximate them. See
\Cref{app:exp-monoculture} for details.}
Intuitively, $D_{vw}$ should be large if $v$ and $w$ have very different
rankings over the outputs, and small otherwise.
Under \Cref{as:ranked}, we'd expect nodes with the same LLM $T_v = T_w$ to
be ``close'' together, and nodes with different LLMs $T_v \ne T_w$ to be farther
apart.

\Cref{fig:distance-matrix} visualizes the adjacency matrix $D$. We group the
nodes by LLM, so each contiguous block of rows and columns consists of $|\Theta|
= 33$ input variants for a given LLM $T$. By construction, the diagonal entries
are all 0. Under \Cref{as:ranked}, we'd expect each of these 6 blocks to be made
up of 0's as well. While this does not strictly hold, it appears to be
directionally true: within-GAIT distances are noticeably smaller than
across-GAIT distances. An off-the-shelf spectral clustering\footnote{Using
default settings, and a standard Gaussian kernel to convert distances into
similarities: $S_{vw} = \exp(-D_{vw}^2/2)$.} confirms our visual intuition: the 6
clusters we learn exactly correspond to our 6 LLMs (see
\Cref{fig:confusion} for the confusion matrix of our clustering). This is true
despite the fact that the prompts are quite different from one another; for
example, some explicitly ask for creative responses that no one else will think
of, while others do not. See \Cref{app:prompts} for the prompts we use.

Consistent with \Cref{hyp:variability}, variation across $\theta \in \Theta_I$
has relatively little effect on the rankings a GAIT induces over the outputs,
but we find substantial across-GAIT heterogeneity. Interestingly, we note that
the two LLMs made by Meta (\llamaone\ and \llamatwo) appear quite close together
when we visualize the two-dimensional spectral embedding in
\Cref{fig:spectral-embedding}, providing suggestive evidence that they are more
alike than other pairs of models. These results support \Cref{as:ranked}: the
rankings over types that GAITs produce are relatively insensitive to the precise
input $\theta$, but they can vary significantly between GAITs.

Our prompt variations are intended to capture idiosyncratic, low-effort
differences that could emerge from heterogeneous users. More sophisticated
prompting strategies like extensive in-context examples or the inclusion of
personal preferences could lead to different response distributions across
users. Empirical studies of AI homogeneity suggest that at least for now,
differences in behavior do not induce enough prompt variation to produce
meaningful heterogeneity. In practice, some commercial systems use prompt
rewriting to improve quality, which ``smooths out'' idiosyncratic prompt
differences and reduces their impact~\citep{jahani2024prompt}.

\section{Competition in Continuous Spaces}
\label{sec:continuous}

\begin{figure}[htpb]
  \centering
  \includegraphics[width=\textwidth]{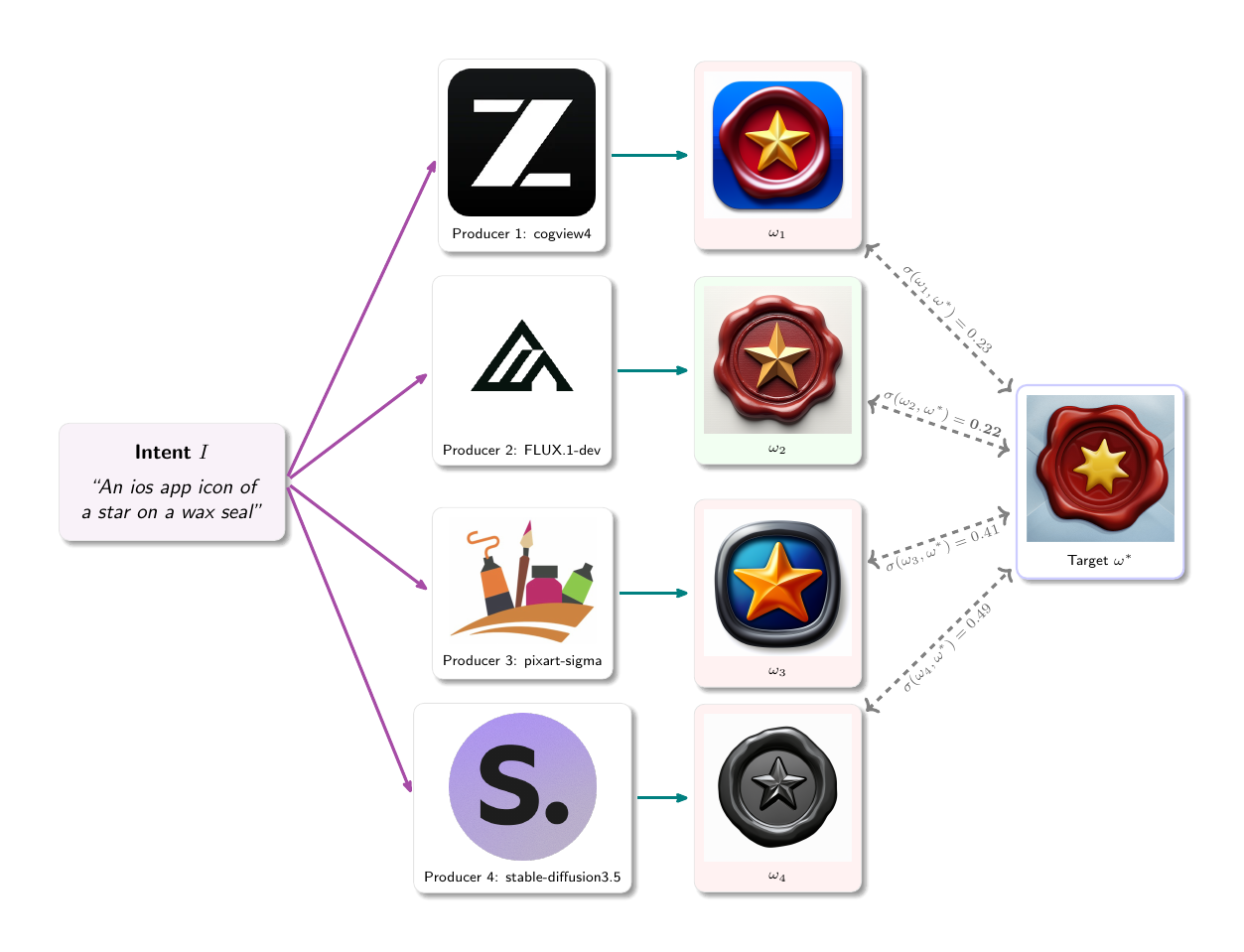}
  
  \caption{Illustration of our setting for \Cref{sec:continuous}. A user arrives
    with an intent $I$, instantiated as a text description of their desired
    image $\omega^*$. Each producer samples an output, using their chosen GAIT.
    In this example, each producer chooses a different GAIT. The user selects the
    closest image as measured by $\sigma$ (in this case, $\omega_2$). Player 2
gets utility 1, and the others get utility 0.} \label{fig:continuous-diagram}
\end{figure}

Our investigation so far has studied and empirically validated a stylized model
of competition. We conclude with an extension to determine whether our
conclusions apply more broadly.
We generalize our model to capture competition in continuous output spaces, in
contrast to the discrete setting discussed thus far. At a high level, we find
that our main results bear out empirically, and we leave a deeper theoretical
investigation of this setting to future work. We begin by describing our
extension before presenting experiments.

\subsection{Modeling competition in continuous spaces}
\label{sec:continuous-model}

Our theoretical model uses collisions to capture the negative externalities of
similarity. In high-dimensional domains like image generation, exact collisions
are rare, but the underlying market forces remain the same: negative
externalities arise through \textit{proximity} instead of collisions. Producers
suffer when their outputs are ``close'' to their competitors' outputs. Here, we
extend our analysis to a continuous setting resembling spatial Hotelling models.

We draw inspiration from \citet{jagadeesan2024supply,hron2022modeling}, who
model competition in Euclidean space. In their setting, a consumer's demand can
be represented as a vector, and given producer outputs, the consumer selects the
one closest to their demand. With this starting point, we assume both demand and
outputs lie in some space $\Omega$ equipped with a distance measure $\sigma$. A
consumer has demand for some $\omega^* \in \Omega$. Producer $i$ generates
$\omega_i \in \Omega$ with the help of a GAIT, and the consumer selects the
producer whose output is closest to their target: $i^* \triangleq \argmin_i
\sigma(\omega_i, \omega^*)$. Producer $i^*$ gets utility 1, and all other
producers get utility 0.\footnote{\citet{hron2022modeling} consider consumer
choice models that are smoother than this one. For simplicity, we only consider
the $\argmin$ choice, as do \citet{jagadeesan2024supply}.} The sum of producer
utilities is by definition 1, but as a measure of welfare, we take the distance
between the target output and the chosen output, i.e.,~$\sigma(\omega_{i^*},
\omega^*)$. Thus, the socially optimal strategy profile is the one that
minimizes this distance.

Our setting differs from that of \citet{jagadeesan2024supply,hron2022modeling}
in two ways. First, producers can create a different output for each consumer,
as opposed to a fixed output for all consumers. However, they do not observe the
consumer's true demand $\omega^*$. Instead, the consumer provides their
``intent'' $I$ describing their target output. In our experiments, we will take
$\omega^*$ to be
an image and $I$ to be a text prompt describing that target image. Second, we
assume producers cannot generate arbitrary outputs. Instead, each must commit to
a GAIT $T$, which takes as input $I$ and samples $\omega \sim \mc D_T(I)$. As in
\Cref{sec:same-model,sec:pairwise}, for the sake of tractability, we do not
consider the role of prompt variation. This is consistent with both
information-theoretic~\citep{rabii2023oatmeal,kreminski2025endless} and
empirical~\citep{zhou2023generative} evidence that prompting alone is not enough
to generate output diversity in image generation.

This extension of our model can be summarized as follows.
\begin{enumerate}
  \item Each producer $i$ strategically chooses a GAIT $T$ from some set $\mc T$.
  \item A random consumer arrives with intent $I$ describing their (unobserved)
    target output $\omega^*$.
  \item Each producer $i$ samples $\omega_i \sim \mc D_{T_i}(I)$.
  \item Producer $i^* = \argmin_i \sigma(\omega_i, \omega^*)$ gets utility 1,
    and all others get utility 0.
\end{enumerate}
Our experiments in \Cref{sec:empirical} can be interpreted as a variant of this
framework: For a given Scattergories instance (i.e., ``intent'') $I$, a consumer
arrives with ``demand'' for a random (valid) answer $k^*$. Producers who
generate $k^*$ split the demand equally, since they are equally close to the
consumer's desired output. In the continuous setting introduced in this section,
collisions almost never occur, so the consumer always selects a single producer.

Our extension departs from our theoretical model from \Cref{sec:model} in its
notion of competition. Here, competition is (1) in a continuous space, and (2)
constant-sum. Despite these differences, we can still ask many of the same
questions: Does competition lead to more diverse production? How does this
compare to socially optimal behavior? And what is the relationship between GAIT
quality and market share? In what follows, we answer these questions
empirically.

\subsection{Competing via image production}
\label{sec:continuous-exp}

\paragraph*{Experimental setup.}

We instantiate this model via experiments using text-to-image (T2I) models to
generate iOS app icons. We begin with a dataset of app icons and their text
descriptions.\footnote{\url{https://huggingface.co/datasets/ppierzc/ios-app-icons}}
Each player $i$ can strategically choose a T2I model from $\mc T = \{$\cogview,
\flux, \pixart, \sd$\}$. Then, we randomly sample (description, icon) pairs $(I,
\omega^*)$ from the dataset.\footnote{As before, we use $I$ to denote
``intent,'' in this case given by the text description of the target image.}
Each player $i$ generates an image $\omega_i \sim \mc D_{T_i}(I)$, and the
winner is the one whose image most closely matches $\omega^*$. We use
DreamSim~\citep{dreamsim} as our measure $\sigma$ of perceptual distance between
images, and present qualitatively similar results for
LPIPS~\citep{zhang2018unreasonable} in \Cref{app:icon}.
\Cref{fig:continuous-diagram} depicts our setup in the case where $n=4$, and
each player chooses a different T2I model. In this instance, the winner is
player 2 using \flux.

As in \Cref{sec:empirical}, we take a sample of instances and generate several
images per instance to run simulations. Because image generation is more
computationally expensive than text generation, we take 25 samples per model for
each of 30 $(I, \omega^*)$ instances. \Cref{tab:model_performance_dreamsim}
shows the average performance of each model, i.e., $\sigma(\omega, \omega^*)$
averaged across all instances and image samples. \cogview\ performs best, and
\pixart\ is the worst (lower distances mean samples are closer to the target).

\begin{table}
\centering
\begin{tabular}{lc}
\toprule
Model & Average Distance \\
\midrule
\cogview & \textbf{0.413} \\
\flux & 0.428 \\
\pixart & 0.462 \\
\sd & 0.442 \\
\bottomrule
\end{tabular}
\caption{Average distance to target image (\dreamsim)}
\label{tab:model_performance_dreamsim}
\end{table}

We compute pure-strategy equilibria for this game for $n = 2, \dots, 20$. This
setup most closely resembles the one in \Cref{sec:pairwise}, where producers
strategically choose their GAIT. Here, because T2I models do not have a natural
equivalent of the temperature parameter, we do not allow for further strategic
behavior beyond choice of GAIT.

Qualitatively, our existing theoretical results suggest that competition
increases diversity (\Cref{thm:diversity-n}), socially optimal behavior is more
diverse than equilibrium behavior (\Cref{thm:diversity-opt}), and suboptimal
models should still get some market share (\Cref{lem:dominant-strength}).
Despite the fact that the setting we study here extends beyond the scope of our
theoretical model, we indeed find that these predictions hold.

\paragraph*{Results.}

\Cref{fig:icon-ms-dreamsim} shows market shares when instantiating our model
across four T2I GAITs. Our findings reflect the intuition that
\Cref{lem:dominant-strength} provides: even though \cogview\ is the ``best''
GAIT in isolation, it is not sufficiently good to lock other GAITs out of the
market. Competing tools can capture their own niches. If a GAIT produces images
with a distinctive visual style, it can meet demand for consumers who want that
style, even if they represent a small fraction of the overall market. Our
results are analogous to those in \Cref{fig:market-shares}, where even when
there are performance disparities between tools, each still achieves nonzero
market share when competition is sufficiently strong.

As equilibrium strategies spread market share more equally across models when
more players participate (\Cref{fig:icon-ms-dreamsim}, left panel), we might
naturally expect this to yield a more diverse distribution over content. Unlike
in \Cref{sec:empirical}, where we could measure diversity as the entropy of a
discrete distribution, we must use a different notion of content diversity. We
instead measure the average distance between a random pair of producers'
outputs, again using \dreamsim. \Cref{fig:icon-diversity-dreamsim} shows that
heterogeneity goes up as the number of players increases, which aligns
qualitatively with \Cref{thm:diversity-n}. For larger values of $n$, players
spread across a broader set of models.

Moreover, \Cref{fig:icon-diversity-dreamsim} also shows that socially optimal
strategies would lead to even more heterogeneity, which is consistent with
\Cref{thm:diversity-opt}. Interestingly, the heterogeneity of socially optimal
strategies is relatively stable as a function of the number of players $n$. This
is in part because the strategy space we study is small and discrete, meaning
more players ``crowd'' the space fairly quickly and produce more similar content
to one another, since they are constrained to choose one of four GAITs. Finally,
while our price of anarchy results (\Cref{thm:poa}) do not immediately
generalize to this setting, \Cref{fig:avg-distance-dreamsim} shows that
equilibrium strategies nearly match socially optimal strategies in terms of the
user welfare they provide. Taken together, these empirical results validate the
qualitative predictions of our theory.

\begin{figure}[ht]
  \centering
  \includegraphics[width=\textwidth]{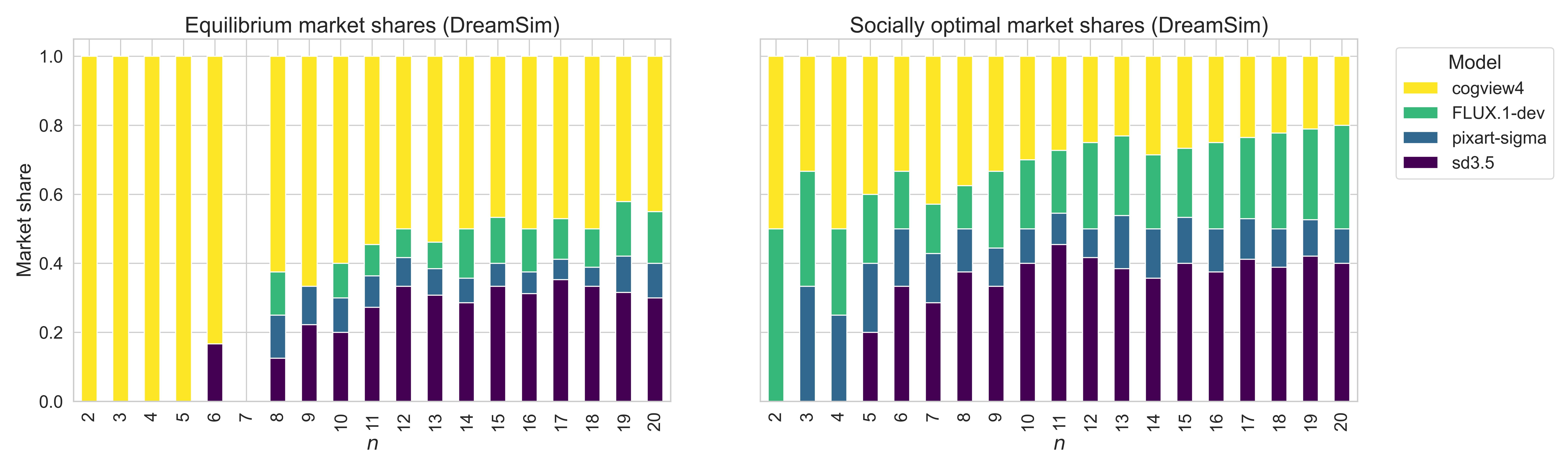}
  \caption{Market shares for equilibrium and optimal strategies. There is no
  pure strategy equilibrium for $n=7$.}
  \label{fig:icon-ms-dreamsim}
\end{figure}

\begin{figure}[ht]
  \centering
  \includegraphics[width=\textwidth]{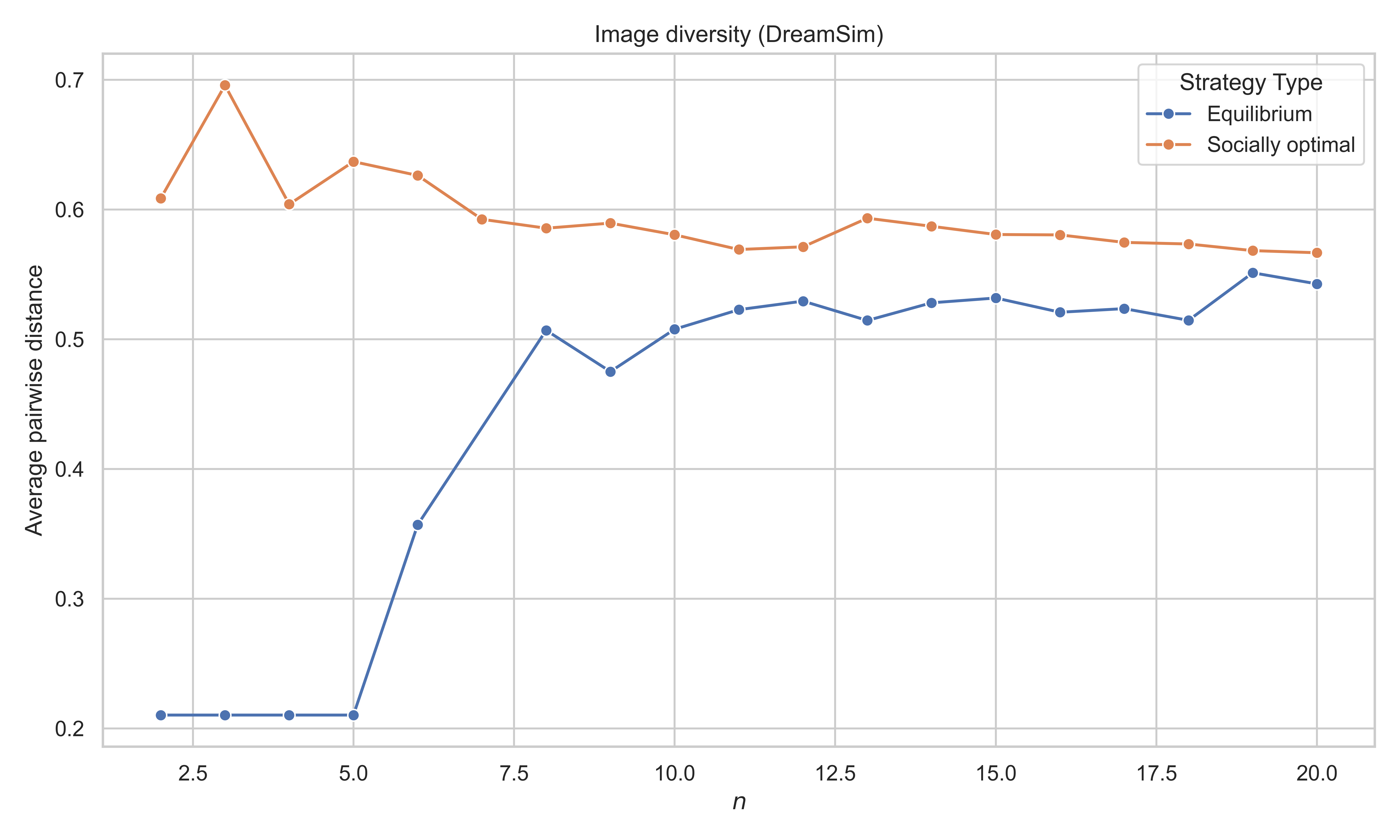}
  \caption{Image diversity with optimal and equilibrium strategies, as measured
    by average pairwise distance. Equilibria generally yield lower pairwise
    distances (more visual similarity) than socially optimal strategies, though
  competition (larger $n$) increases distances.}
  \label{fig:icon-diversity-dreamsim}
\end{figure}

\section{Discussion}
\label{sec:discussion}

Our work studies the interplay between competition and creation through a
game-theoretic lens. Using a simple model and empirical setting, we have
described a rich set of phenomena that shed light on the diversity and quality
of production under generative AI. Despite our model's simplicity, its key
predictions---on diversity, quality, and market shares under competition---are
borne out in our experiments, even as we remove assumptions and vary the
structure of the output space.

Our findings have important implications for the development, evaluation, and
use of generative AI tools. In highly competitive settings, GAITs are
particularly useful if they can create high-quality outputs that are distinct
from those of other tools on the market. Our model suggests that the value of
data to a developer depends on both its quality in an absolute sense (i.e., the
improvement it will enable on a benchmark) as well as its value in a competitive
setting: data may be especially valuable if it captures signal others lack,
increasing the value of exclusive rights to proprietary data. Further, our model
provides additional motivation for objectives like
pluralistic alignment~\citep{sorensen2024roadmap}, which seek to ensure that
tools sample from a diverse or representative distribution of outputs.

Existing GAIT evaluation methods typically consist of benchmarking their
performance in isolation. Our results indicate that we learn something different
when we evaluate tools jointly. Efforts like LMArena~\citep{chiang2024chatbot},
which allows users to choose between a pair of responses generated by anonymous
language models, begin to get at this. It would be interesting to see whether a
similar platform with $k$-way choices instead of pairwise choices would yield a
different ranking over language models. We might expect high-quality but similar
chatbots to split their demand, allowing a lower-quality but distinctive chatbot
to be highly ranked.

Finally, users of GAITs must be increasingly sensitive to their competitors'
behavior, particularly as AI adoption grows. Recruiters have already begun
to warn job applicants that AI-written resumes ``look exactly the
same''~\citep{chatgpt-resume}; applicants who want to use AI successfully need
to explicitly diversify their use of tools~\citep{kaashoek2024impact}. Possible
strategies include altering sampling parameters, varying a prompt, or including
personal information to increase the likelihood of getting a unique response.

We see a wide range of potential future directions to build upon our work. From
a modeling perspective, future work could incorporate more general notions of
similarity and competition than our type-based model, building on our extension
in \Cref{sec:continuous}. Our model of producer-GAIT collaboration is somewhat
limited; suppose, instead, producers sample multiple outputs from a GAIT and
strategically choose one. Could a small number of
choices~\citep{azar1994balanced,mitzenmacher2001power} dramatically increase
welfare and mitigate monoculture? We might also question whether our conception
of social welfare is broad enough to reveal the full costs of homogenization. In
our model, consumer welfare only appears insofar as it directly impacts producer
welfare. But perhaps some of the value of content diversity indirectly accrues
to consumers who do not contribute to producer welfare at all. Somewhat
paradoxically, AI developers who scrape publicly available data benefit
immensely from its diversity, but because the original creators of that data do
not share in those benefits (and indeed, they may not even be aware that their
data has been scraped), their behavior will be unaffected by developers'
preferences for diversity. Rational economic behavior on the part of data
creators~\citep[e.g.,~][]{taitler2024braess} may result in far more
homogenization than would be socially optimal. Accounting for social welfare
beyond its direct impact on producer utilities is a promising direction for
future work.

On the empirical side, Scattergories is a natural competitive benchmark for
language models. It could prove useful to researchers in working towards
alignment---better-aligned models should have a competitive advantage here. More
broadly, future work could build competitive benchmarks that go beyond the
language modality, providing a comprehensive testbed for \textit{competitive
alignment}. Our hope is that both our theoretical model and empirical setting
provide a starting point for future work to study the rich interaction between
competition and diversity in generative AI.

\begin{acks}
  
  We thank David Autor, Kate Donahue, Maryam Farboodi, Nikhil Garg, Amritha
  Jayanti, Jon Kleinberg, Simon Johnson, Sendhil Mullainathan, Ashesh Rambachan,
  Dave Rand, and participants of the Fall 2024 Schwarzman College of Computing
  Special Topics in Computing Workshop and Tuck Operations and Management
  Science seminar for helpful discussions and feedback.

\end{acks}

\bibliography{refs}
\appendix
\clearpage
\paragraph*{Organization of the appendix.}

All deferred proofs for \Cref{sec:theory}, along with auxiliary supporting
results, can be found in \Cref{app:proofs}. \Cref{app:pairwise-proofs} contains
deferred proofs for \Cref{sec:inter-tool}. We restate each result before proving
it. In \Cref{app:empirical,app:icon}, we provide details about our experiments and
additional figures for \Cref{sec:empirical,sec:continuous} respectively.

\section{Score Functions}
\label{app:score-funcs}

Here, we
provide a more general description of the family of score functions for which
our results will hold. Let $f$ be discrete-convex if for any $x \in \N$,
\begin{align*}
  f(2+x) - f(1+x) \ge f(1+x) - f(x).
\end{align*}
$f$ is discrete-concave if $-f$ is discrete-convex. Let $\N_{>0}$ be the set of
nonnegative integers. With this, we can define our class of score functions $\mc
S$.
\begin{definition}[Score function]
  \label{def:S-class}
  A score function $s \in \sall$ satisfies
  \begin{enumerate}[label=\textbf{D.\arabic*}]
    \item $s(x)$ is increasing on $\N_{> 0}$.
      \label{def:increasing}
    \item $1/s(x)$ is discrete-convex on $\N_{> 0}$.
      \label{def:convex}
    \item
      \label{def:bounds}
      $s(1) = 1$ and $s(x) > 1$ for some $x \in \N_{>0}$.
    \item
      \label{def:convex-or-concave}
      At least one of the following two conditions:
      \begin{enumerate}[label=\textbf{\ref{def:convex-or-concave}(\alph*)}]
        \item
          \label{def:S-concave}
          $x/s(x)$ is increasing and discrete-concave on $\N_{> 0}$.
        \item
          \label{def:S-convex}
          $x/s(x)$ is decreasing on $\N_{> 0}$.
      \end{enumerate}
  \end{enumerate}
  Let $\sup$ be the set of functions $s$ that
  satisfy~\labelcref{def:increasing,def:convex,def:bounds,def:S-concave}. Let
  $\sdown$ be the set of functions $s$
  satisfying~\labelcref{def:increasing,def:convex,def:bounds,def:S-convex}. By
  definition, $\sall = \sup \cup \sdown$.
\end{definition}

Observe that together, \labelcref{def:convex,def:bounds} imply $s(2) > 1$.
We include a discrete-concavity
requirement in \labelcref{def:S-concave} so that social welfare exhibits
diminishing returns in $C_k$ for any type $k$. By this definition, $s_\gamma
\in \sup$ for $\gamma \in (0, 1]$ and $s_\gamma \in \sdown$ for $\gamma \ge
1$ (including $\gamma = \infty$). The
only function in their intersection $\sup \cap \sdown$ is the identity function
$s_1(x) = x$.

\section{Omitted Proofs and Auxiliary Results for \Cref{sec:theory}}
\label{app:proofs}
Here, we provide omitted proofs and auxiliary claims needed to support them. We
restate each result before proving it. We begin with a few definitions.
\begin{align*}
  b(k, n, p)
  &\triangleq \binom{n}{k} p^k (1-p)^{n-k} \\
  B(k, n, p)
  &\triangleq \sum_{\ell=0}^k b(k, n, p)
\end{align*}
For a sequence $\{\ba_k\}$, and $k_2 \ge k_1$, define
\begin{align*}
  \avg(\ba; k_1, k_2)_k \triangleq \begin{cases}
    \frac{\sum_{\ell=k_1}^{k_2} \ba_\ell}{k_2 - k_1 + 1} & k_1 \le k \le k_2 \\
    \ba_k & \text{otherwise}
  \end{cases}.
\end{align*}
Intuitively, this operation averages all entries between $k_1$ and $k_2$.
For a function $f(x)$, define $\df(x)$ to be the discrete derivative $\df(x)
\triangleq f(x+1) - f(x)$. For a function that takes two arguments (e.g., $g(n,
p)$), we will take the discrete derivative of the first dimension by convention:
$\dg(n, p) \triangleq g(n+1, p) - g(n, p)$.

\subsection{Omitted proofs for \Cref{sec:equilibria}}
\label{app:equilibria}

Define
\begin{align*}
  u_s(n, p) \triangleq \E{\frac{1}{s(1 + X(n, p))}}.
\end{align*}
When it is clear from context, we will drop the subscript $s$ and just write
$u(n, p)$.

\begin{lemma}
  \label{lem:unique-eq-decreasing}
  Under \Cref{as:d-ranked}, $\game(n, \bd, s)$ has a unique symmetric
  equilibrium $\bp = \peq(n, \bd, s)$, which is a pure strategy equilibrium. Let $k^* = \max \{k : \bp_k > 0\}$. $\bp$ satisfies the following properties:
  \begin{enumerate}
    \item $\bp \in \Delta(\pi)$.
    \item For all $k \le k^*$,
      \begin{align*}
        \bd_k u(n-1, \bp_k) = U(\bp).
      \end{align*}
    \item For all $k > k^*$, $U(\bp) \ge \bd_k$.
  \end{enumerate}
  This holds even if we remove \Cref{as:ranked} (i.e., if players can choose
  $\bp \notin \Delta(\pi)$).
\end{lemma}
\begin{proof}
  By \Cref{lem:sym-mixed-eq,lem:mixed-pure-eq}, $\game(n, \bd, s)$ has a
  symmetric pure strategy equilibrium that yields utility $U(\bp)$. For now, we
  remove the restriction that $\bp \in \Delta(\pi)$ and instead allow for any
  distribution $\bp$.

  Player $i$'s utility for strategy $\bp$ is
  \ifthenelse{\boolean{smallEqs}}{
  \begin{align*}
    U(\bp)
    &= \sum_{k \in [K]} \bd_k \bp_k \E{\frac{1}{s(1 + X(n-1, \bp_k))}} \\
    &= \sum_{k \in [K]} \bd_k \bp_k u(n-1, \bp_k).
  \end{align*}
  }{\begin{align*}
      U(\bp)
    &= \sum_{k \in [K]} \bd_k \bp_k \E{\frac{1}{s(1 + X(n-1, \bp_k))}} 
    = \sum_{k \in [K]} \bd_k \bp_k u(n-1, \bp_k).
\end{align*}
}
  If we remove the restriction that $\bp \in \Delta(\pi)$, then player $i$ can
  deviate from $\bp$ to any arbitrary distribution $\bp'$. This means that for
  all $k$,
  \begin{align*}
    \bd_k u(n-1, \bp_k) \le U(\bp),
  \end{align*}
  because otherwise, a player could strictly increase their utility by placing
  all of their probability mass on type $k$. Similarly, for all $k$ such that
  $\bp_k > 0$,
  \begin{align*}
    \bd_k u(n-1, \bp_k) \ge U(\bp),
  \end{align*}
  because otherwise, a player could strictly increase their utility by removing
  mass from $k$ and distributing it to any $k'$ such that $\bd_{k'} u(n-1,
  \bp_{k'}) \ge U(\bp)$. (Such a $k'$ must exist because these quantities
  average to $U(\bp)$.)

  Finally, note that $u(n-1, \cdot)$ is strictly decreasing by
  \Cref{lem:u-dec-in-p} and $u(n-1, 0) = 1$. Therefore, if $\bd_k > U(\bp)$,
  \begin{align*}
    \bd_k u(n-1, \bp_k) \le U(\bp) < \bd_k,
  \end{align*}
  so $\bp_k > 0$. If $\bd_k \le U(\bp)$, then
  \begin{align*}
    \bd_k u(n-1, \bp_k) \le \bd_k \le U(\bp).
  \end{align*}
  In this case, $\bp_k = 0$ because otherwise the first inequality is strict.
  Let $k^*$ be the largest $k$ such that $\bd_{k^*} > U(\bp)$. Putting these
  together, under \Cref{as:d-ranked}, 
  \begin{align*}
    \bp_k > 0
    &\Longleftrightarrow k \le k^* \\
    \bd_k u(n-1, \bp_k)
    &= U(\bp)
    \tag{$k \le k^*$}
  \end{align*}

  For $k \le k^*$, we have
  \begin{align*}
    u(n-1, \bp_k) = \frac{U(\bp)}{\bd_k}.
  \end{align*}
  Again using the fact that $u(n-1, \cdot)$ is strictly decreasing, this means
  that $\bp_k$ must also be decreasing, for $k \le k^*$. And since $\bp_k = 0$
  for $k > k^*$, we have shown $\bp \in \Delta(\pi)$.

  Finally, we show that $\bp$ is the unique symmetric equilibrium. Assume
  towards contradiction that there is some other symmetric equilibrium $\bp' \ne
  \bp$. There must be some $k$ such that $\bp_k' > 0$, $\bp_k > 0$, and $\bp_k'
  \ne \bp_k$. Because $u(n-1, \cdot)$ is strictly decreasing, $U(\bp') = \bd_k
  u(n-1, \bp_k') \ne \bd_k u(n-1, \bp_k) = U(\bp)$.

  Let $u^{-1}(n, p)$ be the inverse of $u$, such that $u^{-1}(n, u(n, p)) = p$.
  We only require that $u^{-1}$ is defined on the range $[u(n, 1), u(n, 0)]$; by
  convention, define it to be 0 otherwise. It is well-defined because $u(n,
  \cdot)$ is strictly decreasing in that range. Observe that $u^{-1}(n-1,
  \cdot)$  is strictly decreasing on the range where it is nonzero. Then, let
  \begin{align*}
    S(c)
    &\triangleq \sum_{k \in [K]} u^{-1}\p{n-1, \frac{c}{\bd_k}}.
  \end{align*}
  Because at least one $u^{-1}$ term is strictly decreasing in $c$ (i.e., all
  terms corresponding to nonzero $\bp_k$'s), $S$ is strictly decreasing in $c$.
  Observe that
  \begin{align*}
    S(U(\bp))
    &= \sum_{k \in [K]} u^{-1}\p{n-1, \frac{U(\bp)}{\bd_k}} \\
    &= \sum_{k \in [K]} \bp_k \\
    &= 1,
  \end{align*}
  and similarly, $S(U(\bp')) = 1$. But we know that $S$ is strictly decreasing,
  so $S(U(\bp)) = S(U(\bp'))$ implies $U(\bp) = U(\bp')$, which is a
  contradiction. Therefore, $\bp$ is the unique symmetric equilibrium.
\end{proof}

\begin{lemma}
  \label{lem:unique-opt-decreasing}
  Under \Cref{as:d-ranked}, $\game(n, \bd, s)$ has a unique symmetric
  optimal strategy $\bp^* = \popt(n, \bd, s)$. Let $c = \frac{1}{n} \|\nabla
  W(\bp^*)\|_{\infty}$, and let $k^* = \max \{k : \bp_k^* > 0\}$. $\bp^*$
  satisfies the following properties:
  \begin{enumerate}
    \item $\bp^* \in \Delta(\pi)$.
    \item If $s \in \sup$, for all $k \le k^*$,
      \begin{align*}
        \frac{1}{n} \p{\nabla W(\bp^*)}_k
        = \frac{\partial}{\partial p} \bd_k p u_s(n-1, p) \dpat{\bp_k^*} = c.
      \end{align*}
    \item If $s \in \sdown$, for all $k \le k^*$, $\bd_k (1-\bp_k^*)^{n-1} = c$.
    \item For all $k > k^*$, $c \ge \bd_k$.
  \end{enumerate}
  This holds even if we remove \Cref{as:ranked} (i.e., if players can choose
  $\bp \notin \Delta(\pi)$).
\end{lemma}
\begin{proof}
  We begin with the case where $s \in \sup$.
  By \Cref{lem:mixed-pure-eq}, it suffices to consider pure strategies. The
  socially optimal symmetric pure strategy is the solution to the optimization
  problem
  \begin{align*}
    \max_{\bp} ~
    &W(\bp) \\
    \text{s.t.}~
    &\sum_{k \in [K]} \bp_k \le 1 \\
    &\bp_{k+1} \le \bp_k \tag{$\forall k \in [K-1]$} \\
    &\bp_K \ge 0.
  \end{align*}
  Because $\game$ is symmetric and social welfare is the sum of utilities for $s
  \in \sup$, this is equivalent to maximizing a single player's utility:
  \begin{align*}
    \max_{\bp} ~
    &U(\bp) \\
    \text{s.t.}~
    &\sum_{k \in [K]} \bp_k \le 1 \\
    &\bp_{k+1} \le \bp_k \tag{$\forall k \in [K-1]$} \\
    &\bp_K \ge 0.
  \end{align*}
  As in the proof of \Cref{lem:unique-eq-decreasing}, we will remove the
  assumption that $\bp \in \Delta(\pi)$ and show that even without this
  assumption, $\popt \in \Delta(\pi)$. Without the assumption that $\bp \in
  \Delta(\pi)$, we can rewrite utility-maximization problem as
  \begin{align*}
    \max_{\bp} ~
    &U(\bp) \\
    \text{s.t.}~
    &\sum_{k \in [K]} \bp_k \le 1 \\
    &\bp_k \ge 0.
    \tag{$\forall k \in [K]$}
  \end{align*}
  Intuitively, we would expect $\frac{\partial}{\partial \bp_k} U(\bp)$ to be
  equal for all $k$ such that $\bp_k > 0$, since the optimal solution should
  equalize the marginal gain from increasing any nonzero $\bp_k$. We formalize
  this in what follows. The Lagrangian is
  \ifthenelse{\boolean{smallEqs}}{
  \begin{align*}
    \mc L(\bp; \mu, \lambda)
    &= -U(\bp) + \mu\p{\sum_{k \in [K]} \bp_k - 1}
    - \sum_{k \in [K]} \lambda_k \bp_k \\
    &= - \sum_{k \in [K]} \bd_k \bp_k u(n-1, \bp_k) \\
    &+ \mu\p{\sum_{k \in [K]}
    \bp_k - 1}
    - \sum_{k \in [K]} \lambda_k \bp_k.
  \end{align*}
}{\begin{align*}
    \mc L(\bp; \mu, \lambda)
    &= -U(\bp) + \mu\p{\sum_{k \in [K]} \bp_k - 1} - \sum_{k \in [K]} \lambda_k
    \bp_k \\
    &= - \sum_{k \in [K]} \bd_k \bp_k u(n-1, \bp_k) + \mu\p{\sum_{k \in [K]}
    \bp_k - 1} - \sum_{k \in [K]} \lambda_k \bp_k.
  \end{align*}
}
  By stationarity, for the socially optimal strategy $\bp^*$, for all $k \in
  [K]$,
  \begin{align*}
    \frac{\partial}{\partial \bp_k} \mc L(\bp; \mu, \lambda)
    \Big|_{\bp_k=\bp_k^*}
    &= 0 \\
    - \bd_k \frac{\partial}{\partial p} p u(n-1, p) \dpat{\bp_k^*} + \mu -
    \lambda_k
    &= 0.
  \end{align*}
  Consider an optimal solution $\bp^*$. By complementary slackness, 
  \begin{align*}
    \bp_k^* > 0 \Longrightarrow \lambda_k = 0.
  \end{align*}
  Therefore, for all $k$ such that $\bp_k^* > 0$,
  \begin{equation}
    \label{eq:partial-mu}
    \bd_k \frac{\partial}{\partial p} p u(n-1, p) \dpat{\bp_k^*} = \mu.
  \end{equation}
  For all $k$ such that $\bp_k^* = 0$,
  \begin{align*}
     \bd_k \frac{\partial}{\partial p} p u(n-1, p) \dpat{\bp_k^*}
     &= \mu - \lambda_k \\
     \bd_k \frac{\partial}{\partial p} p u(n-1, p) \dpat{0}
     &\le \mu.
  \end{align*}
  Note that
  \ifthenelse{\boolean{smallEqs}}{
  \begin{align*}
    \frac{\partial}{\partial p} p u(n-1, p) \dpat{0}
    &= \p{p \cdot \frac{\partial}{\partial p} u(n-1, p)} \dpat{0} \\
    &+ u(n-1, 0) \\
    &= 1.
  \end{align*}
}{\begin{align*}
    \frac{\partial}{\partial p} p u(n-1, p) \dpat{0}
    &= \p{p \cdot \frac{\partial}{\partial p} u(n-1, p)} \dpat{0} 
    + u(n-1, 0)
    = 1.
  \end{align*}
}
  Therefore, for $k$ such that $\bp_k^* = 0$,
  \begin{align*}
    \bd_k \le \mu.
  \end{align*}
  Next, we will show that $\frac{\partial}{\partial p} p u(n-1, p)$ is strictly
  decreasing in $p$. Define
  \begin{align*}
    w_s(n, p)
    &\triangleq \E{\frac{X(n, p)}{s(X(n, p))}}
    = n p u(n-1, p).
  \end{align*}
  As before, we will drop the subscript $s$ when it is clear from context. Then,
  \begin{align*}
    \frac{\partial^2}{\partial p^2} p u(n-1, p) < 0
    \Longleftrightarrow
    \frac{\partial^2}{\partial p^2} w(n, p) < 0.
  \end{align*}
  Because $x/s(x)$ is strictly concave by assumption
  (\labelcref{def:bounds,def:S-concave}), $w$ is
  the expectation of a strictly concave function of a binomial random variable,
  which is strictly concave by \Cref{lem:binom-ex-convex}. This implies
  \begin{align*}
    \frac{\partial^2}{\partial p^2} p u(n-1, p) < 0.
  \end{align*}

  Given $\mu$, we therefore have a tight characterization of $\bp^*$:
  \begin{enumerate}
    \item If $\bd_k > \mu$, then we must have $\bp_k^* > 0$, and in particular,
      \begin{align*}
        \bd_k \frac{\partial}{\partial p} p u(n-1, p) \dpat{\bp_k^*} = \mu.
      \end{align*}
    \item If $\bd_k \le \mu$, then we must have
      \begin{align*}
        \bd_k \frac{\partial}{\partial p} p u(n-1, p) \dpat{\bp_k^*} \le \bd_k
        \le \mu.
      \end{align*}
      If $\bp_k^* > 0$, the first of these inequalities is strict, which
      contradicts~\eqref{eq:partial-mu}. Therefore, $\bd_k \le \mu$ implies
      $\bp_k^* = 0$.
  \end{enumerate}

  To show uniqueness, it suffices to note that the objective of our optimization
  problem is strictly concave, meaning it has a unique optimum. Because $\bp_1^*
  > 0$, $\mu = \frac{\partial}{\partial \bp_1} U(\bp) \big|_{\bp=\bp^*}$. Thus,
  $\mu = \|\nabla U(\bp^*)\|_{\infty} = \frac{1}{n} \|\nabla
  W(\bp^*)\|_{\infty}$.

  Observe that $\bp^* \in \Delta(\pi)$: for any $k$ such that $\bp_k^* > 0$ and
  $\bp_{k+1}^* > 0$,
  \begin{align*}
    \bd_k \frac{\partial}{\partial p} p u(n-1, p) \dpat{\bp_k^*}
    &= \mu
    = \bd_{k+1} \frac{\partial}{\partial p} p u(n-1, p) \dpat{\bp_{k+1}^*} \\
    \frac{\partial}{\partial p} p u(n-1, p) \dpat{\bp_k^*}
    &\le \frac{\partial}{\partial p} p u(n-1, p) \dpat{\bp_{k+1}^*}
    \tag{$\bd_k \ge \bd_{k+1}$} \\
    \bp_k^*
    &\ge \bp_{k+1}^*.
    \tag{$\frac{\partial}{\partial p} p u(n-1, p)$ strictly decreasing}
  \end{align*}

  Next, we handle the case where $s \in \sdown$. By definition,
  \begin{align*}
    W(\bp)
    &= \sum_{k \in [K]} \bd_k(1 - (1 - \bp_k^*)^n) \\
    &= \sum_{k \in [K]} \E{\frac{X(n, p)}{X(n, p)} \cdot \ind{X(n, p > 0)}} \\
    &= \sum_{k \in [K]} \E{\frac{X(n, p)}{s_1(X(n, p))}}.
  \end{align*}
  Thus, welfare under any $s \in \sdown$ is equivalent to welfare under the
  identity score function $s_1(x) = x$.
  Because $s_1 \in \sup$, our results above hold. Let $c = \frac{1}{n} \|\nabla
  W(\bp^*)\|_\infty$.
  For all $k$ such that $\bp_k^* > 0$,
  \begin{align*}
    c
    = \frac{\bd_k}{n}
    \frac{\partial}{\partial p} w_{s_1}(n, p) \dpat{\bp_k^*}.
  \end{align*}
  For $k > k^*$, $c \ge \bd_k$. To conclude, we must show that for $k$ such that
  $\bp_k^* > 0$,
  \begin{align*}
    c
    &= \frac{\bd_k}{n} \frac{\partial}{\partial p} w_{s_1}(n, p) \dpat{\bp_k^*}
    \\
    &= \frac{\bd_k}{n} \frac{\partial}{\partial p} \E{\ind{X(n, p) > 0}}
    \dpat{\bp_k^*} \\
    &= \frac{\bd_k}{n} \frac{\partial}{\partial p} (1 - (1-p)^n) \dpat{\bp_k^*}
    \\
    &= \frac{\bd_k}{n} n (1-p)^{n-1} \dpat{\bp_k^*} \\
    &= \bd_k (1-\bp_k^*)^{n-1}.
  \end{align*}
\end{proof}

\uniqueeq*
\begin{proof}

  First, \Cref{lem:sym-mixed-eq,lem:mixed-pure-eq} imply that $\game(n, \bd, s)$ has
  a symmetric pure Nash equilibrium. By \Cref{lem:dec-d-asym}, there exists
  decreasing $\tbd$ such that every equilibrium or symmetric socially optimal
  symmetric strategy in $\game(n, \bd, s)$ is an equilibrium or symmetric
  socially optimal strategy respectively in $\game(n, \tbd, s)$ with the same
  utilities and social welfare. By
  \Cref{lem:unique-eq-decreasing,lem:unique-opt-decreasing}, $\peq(n, \tbd, s)$
  and $\popt(n, \tbd, s)$ are both unique for $\game(n, \tbd, s)$. This implies
  that $\peq(n, \bd, s)$ and $\popt(n, \bd, s)$ are both unique for $\game(n,
  \bd, s)$ as well.

\end{proof}

\subsubsection{Additional results on equilibria}

\begin{lemma}
  \label{lem:sym-mixed-eq}
  $\game(n, \bd, s)$ has a symmetric mixed Nash equilibrium.
\end{lemma}
\begin{proof}
  $\game$ is a symmetric game with a compact, convex strategy space $\Delta(\pi)$.
  Utilities are continuous in strategies. Classical results by
  \citet{glicksberg1952further}, \citet[][Theorem 1]{becker2006existence} show
  that this implies $\game$ has a symmetric mixed Nash equilibrium.
\end{proof}

\begin{lemma}
  \label{lem:mixed-pure-eq}
  Every mixed strategy in the game $\game(n, \bd, s)$ is equivalent to a pure
  strategy.
\end{lemma}
\begin{proof}
  Consider a mixed strategy $q$, which is a distribution over distributions
  $\bpn{1}, \bpn{2}, \dots$. Let $\tbp_k(q) = \EE{\bp \sim q}{\bp_k}$. A
  player samples type $k$ with probability $\tbp_k(q)$ under both $q$ and
  $\tbp(q)$, meaning no downstream game can distinguish between them. $\tbp \in
  \Delta(\pi)$ because $\Delta(\pi)$ is convex: for any $k$ and $i$, $\bpn{i}_k
  \ge \bpn{i}_{k+1}$. Therefore,
  \begin{align*}
    \tbp_k(q)
    = \sum_i q_i \bpn{i}_k
    \ge \sum_i q_i \bpn{i}_{k+1}
    = \tbp_{k+1}(q).
  \end{align*}
\end{proof}

\begin{lemma}
  \label{lem:dec-d-asym}
  For all $n \ge 2$, $s \in \sall$ and $\bd$, let $\eqset(n, \bd, s) \subset
  \Delta(\pi)^n$ be the set of equilibria for $\game(n, \bd, s)$. Let $\popt(n,
  \bd, s)$ be the unique symmetric socially optimal strategy for $\game(n, \bd,
  s)$ There exists $\tbd$ such that:
  \begin{itemize}
    \item $\tbd$ is decreasing in $k$
    \item $\eqset(n, \bd, s) \subseteq \eqset(n, \tbd, s)$
    \item For all $\bP \in \eqset(n, \bd, s)$ and for all $i$, $U_i(\bP; n, \bd,
      s) = U_i(\bP; n, \tbd, s)$.
    \item $\popt(n, \bd, s) = \popt(n, \tbd, s) = \bp^*$
    \item For all $i$, $U_i(\bp^*; n, \bd, s) = U_i(\bp^*; n, \tbd, s)$.
    \item $W(\bp^*; n, \bd, s) = W(\bp^*; n, \tbd, s)$.
  \end{itemize}
\end{lemma}
\begin{proof}
  We present an iterative argument that every (weakly) increasing sequence of
  $\bd_k$'s can be replaced by their means. Consider $k_1, k_2$ such that
  $\bd_{k_1} \le \bd_{k_1+1} \le \dots \le \bd_{k_2}$ with at least one of these
  inequalities being strict. By \Cref{lem:inc-d-best-response}, every $\bP \in
  \eqset(n, \bd, s)$ satisfies $\bP(i) = \avg(\bP(i);
  k_1, k_2)$ for all $i$. Similarly, by \Cref{lem:inc-d-social-welfare},
  $\popt(n, \bd, s) = \avg(\popt(n, \bd, s); k_1, k_2)$.

  Let $\Delta(\pi; k_1, k_2)$ be the subset of $\Delta(\pi)$ that additionally
  satisfies $\bp = \avg(\bp; k_1, k_2)$. We can rewrite the above conclusions as
  \begin{align*}
    \forall \bP \in \eqset(n, \bd, s) ~ ~ ~
    &\bP \in \Delta(\pi; k_1, k_2)^n
    \numberthis \label{eq:bP-in-Delta-k1k2}
    \\
    &\popt(n, \bd, s) \in \Delta(\pi; k_1, k_2)
    \numberthis \label{eq:popt-in-Delta-k1k2}
  \end{align*}

  Consider the game $\tilde{\game}$ in which we restrict the strategy space to
  $\Delta(\pi; k_1, k_2)$, with corresponding equilibria $\tilde \eqset$ and
  symmetric optimal strategy $\tpopt$.
  Because all best responses for any $\bP \in \Delta(\pi)$ lie in $\Delta(\pi;
  k_1, k_2)$ by~\eqref{eq:bP-in-Delta-k1k2},
  $\eqset(n, \bd, s) = \tilde \eqset(n, \bd, s)$.
  By~\eqref{eq:popt-in-Delta-k1k2},
  $\popt(n, \bd, s) = \tpopt(n, \bd, s)$.

  Let $\tbd = \avg(\bd; k_1, k_2)$.
  For all $\bP \in \Delta(\pi; k_1, k_2)^n$,
  \ifthenelse{\boolean{smallEqs}}{
  \begin{align*}
    U_i&(\bP; n, \bd, s) - U_i(\bP; n, \tbd, s) \\
    &= \sum_{k=1}^K \bP(i)_k \bd_k \E{\frac{1}{s(1 + X_k(\bP(-i)))}} \\
    &-\sum_{k=1}^K \bP(i)_k \tbd_k \E{\frac{1}{s(1 + X_k(\bP(-i)))}} \\
    &= \sum_{k=k_1}^{k_2} \bP(i)_k (\bd_k-\tbd_k) \E{\frac{1}{s(1 + X_k(\bP(-i)))}} \\
    &= \bP(i)_{k_1} \E{\frac{1}{s(1 + X_{k_1}(\bP(-i)))}} \sum_{k=k_1}^{k_2}
    (\bd_k-\tbd_k) \tag{$\bP \in \Delta(\pi; k_1, k_2)^n$} \\
    &= 0.
  \end{align*}
}{\begin{align*}
    U_i(\bP; n, \bd, s) - U_i(\bP; n, \tbd, s)
    &= \sum_{k=1}^K \bP(i)_k \bd_k \E{\frac{1}{s(1 + X_k(\bP(-i)))}} \\
    &-\sum_{k=1}^K \bP(i)_k \tbd_k \E{\frac{1}{s(1 + X_k(\bP(-i)))}} \\
    &= \sum_{k=k_1}^{k_2} \bP(i)_k (\bd_k-\tbd_k) \E{\frac{1}{s(1 + X_k(\bP(-i)))}} \\
    &= \bP(i)_{k_1} \E{\frac{1}{s(1 + X_{k_1}(\bP(-i)))}} \sum_{k=k_1}^{k_2}
    (\bd_k-\tbd_k) \tag{$\bP \in \Delta(\pi; k_1, k_2)^n$} \\
    &= 0.
  \end{align*}
}
  Therefore, for all $\bP \in \Delta(\pi; k_1, k_2)$,
  \begin{equation}
    \label{eq:bd-tbd-equal-U}
    U_i(\bP; n, \bd, s) = U_i(\bP; n, \tbd, s).
  \end{equation}
  In other words, if we restrict players to strategies in $\Delta(\pi; k_1,
  k_2)$, games with $\bd$ and $\tbd$ have identical utilities, and thus,
  identical equilibria. As a result, $\tilde \eqset(n, \bd, s) = \tilde
  \eqset(n, \tbd, s)$.

  Next, we show the same is true for social welfare.
  For all $s \in \sup$, summing~\eqref{eq:bd-tbd-equal-U} across all $i$ yields
  that for any $\bP \in \Delta(\pi; k_1, k_2)$,
  \begin{align*}
    W(\bP; n, \bd, s) = W(\bP; n, \tbd, s).
  \end{align*}
  For $s \in \sdown$,
  \ifthenelse{\boolean{smallEqs}}{
  \begin{align*}
    W&(\bP; n, \bd, s) - W(\bP; n, \tbd, s) \\
    &= \sum_{k=1}^K \bd_k \Pr[X_k(\bP) > 0]
    - \sum_{k=1}^K \tbd_k \Pr[X_k(\bP) > 0] \\
    &= \sum_{k=k_1}^{k_2} (\bd_k - \tbd_k) \Pr[X_k(\bP) > 0] \\
    &= \Pr[X_{k_1}(\bP) > 0] \sum_{k=k_1}^{k_2} (\bd_k - \tbd_k) \\
    &=0.
  \end{align*}
}{\begin{align*}
    W(\bP; n, \bd, s) - W(\bP; n, \tbd, s)
    &= \sum_{k=1}^K \bd_k \Pr[X_k(\bP) > 0]
    - \sum_{k=1}^K \tbd_k \Pr[X_k(\bP) > 0] \\
    &= \sum_{k=k_1}^{k_2} (\bd_k - \tbd_k) \Pr[X_k(\bP) > 0] \\
    &= \Pr[X_{k_1}(\bP) > 0] \sum_{k=k_1}^{k_2} (\bd_k - \tbd_k) \\
    &=0.
  \end{align*}
}
  This implies $\tpopt(n, \bd, s) = \tpopt(n, \tbd, s)$.

  Next, we will show that $\tilde \eqset(n, \tbd, s) \subseteq \eqset(n, \tbd,
  s)$. We must show that if $\bP \in \tilde \eqset(n, \tbd, s)$, then $\bP \in
  \eqset(n, \tbd, s)$. Consider some $\bP \in \tilde \eqset(n, \tbd, s)$, and
  assume towards contradiction that $\bP \notin \eqset(n, \tbd, s)$. Then, some
  player $i$ must have some strict best response in $\game(n, \tbd, s)$, and by
  \Cref{lem:inc-d-best-response}, at least one of those best responses must be
  in $\Delta(\pi; k_1, k_2)$. This contradicts the fact that $\bP \in \tilde
  \eqset(n, \tbd, s)$.

  A similar argument holds for $\popt$, where we can use the fact that $\popt(n,
  \bd, s)$ is unique (\Cref{lem:unique-opt-decreasing}). Assume towards
  contradiction that $\tpopt(n, \tbd, s) \ne \popt(n, \tbd, s)$. Then, by
  \Cref{lem:inc-d-social-welfare}, there is some $\bp' \in \Delta(\pi; k_1,
  k_2)$ such that
  \ifthenelse{\boolean{smallEqs}}{
  \begin{align*}
    W(\bp'; n, \tbd, s)
    &> W(\popt(n, \tbd, s); n, \tbd, s) \\
    &= W(\tpopt(n, \tbd, s); n, \tbd, s).
  \end{align*}
}{\begin{align*}
    W(\bp'; n, \tbd, s)
    > W(\popt(n, \tbd, s); n, \tbd, s)
    = W(\tpopt(n, \tbd, s); n, \tbd, s).
  \end{align*}
}
  This contradicts the fact that $\tpopt(n, \tbd, s)$ is optimal for $\tilde
  \game(n, \tbd, s)$.

  Thus, we have shown that:
  \begin{itemize}
    \item $\eqset(n, \bd, s) \subseteq \eqset(n, \tbd, s)$ and
    $\popt(n, \bd, s) = \popt(n, \tbd, s) = \bp^*$
    \item For any $\bP \in \eqset(n, \bd, s)$, for all $i$,
      $U_i(\bP; n, \bd, s) = U_i(\bP; n, \tbd, s)$.
    \item For all $i$, $U_i(\bp^*; n, \bd, s) = U_i(\bp^*; n, \tbd, s)$.
    \item $W(\bp^*; n, \bd, s) = W(\bp^*; n, \tbd, s)$.
  \end{itemize}
  To complete argument, we must show that this holds for decreasing $\tbd$. We
  simply iteratively apply this argument until there are no strict increases in
  $\bd$. This amounts to running the Pool Adjacent Violators
  Algorithm~\citep{barlow1972statistical}, which is known to terminate.
\end{proof}

\begin{lemma}
  \label{lem:inc-d-best-response}
  For all $n \ge 2$, $s \in \sall$, suppose $\bd$ contains a sequence $\bd_{k_1}
  \le \bd_{k_1 + 1} \le \dots \le \bd_{k_2}$. Then, for any $\bP \in
  \Delta(\pi)^n$, every player has a best response to $\bP(-i)$ such that
  $\bP(i)_{k_1} = \bP(i)_{k_1 + 1} = \dots = \bP(i)_{k_2}$. If at least one
  inequality is strict (i.e., $\bd_{k} < \bd_{k'}$ for $k_1 \le k < k' \le
  k_2$), then every best response to $\bP(-i)$ has $\bP(i)_{k_1} = \bP(i)_{k_1 +
  1} = \dots = \bP(i)_{k_2}$.
\end{lemma}
\begin{proof}
  Player $i$'s best response is the solution to
  \begin{align*}
    \argmax_{\bp \in \Delta(\pi)} ~ \sum_{k=1}^K \bp_k V_k(\bP(-i))
  \end{align*}
  where
  \begin{align*}
    V_k(\bP(-i)) = \bd_k \E{\frac{1}{s(1 + X_k(\bP(-i)))}}.
  \end{align*}
  Let $\tbp$ be such a best response. Let $c = \sum_{k_1 \le k \le k_2}
  \tbp_k$, and let $p = \tbp_{k_2+1}$. Then, $\tbp$ satisfies
  \begin{align*}
    \tbp_{k_1}, \dots \tbp_{k_2} \in \argmax_{\bp_{k_1}, \dots, \bp_{k_2}}
    ~&\sum_{k=k_1}^{k_2} \bp_k V_k(\bP(-i)) \\
    \text{s.t.} ~
     &\bp_{k_1} \ge \bp_{k_1 + 1} \ge \dots \ge \bp_{k_2} \\
     &\sum_{k=k_1}^{k_2} \bp_k = c \\
     &\bp_{k_2} \ge p
  \end{align*}

  Because $\bP(-i) \in \Delta(\pi)^{n-1}$, each $X_k(\bP(-i))$ stochastically
  dominates $X_{k+1}(\bP(-i))$. Therefore, for any $k' > k$ between $k_1$ and
  $k_2$, $V_{k'}(\bP(-i)) \ge V_k(\bP(-i))$. If one of the inequalities is
  strict, meaning $\bd_{k} < \bd_{k'}$ for some $k
  < k'$, then
  $V_{k'}(\bP(-i)) > V_k(\bP(-i))$.

  If we ignore the constraint that $\bp_{k_2} \ge p$, we can apply
  \Cref{lem:inc-product} to show that all of the $\bp_k$'s are equal, with
  $b_\ell = \tbp_\ell$, $a_\ell = V_\ell(\bP(-i))$, $f(x) = x$, and $\ell$
  ranging from $k_1$ to $k_2$. This solution maximizes $\bp_{k_2}$, so it cannot
  violate the $\bp_{k_2} \ge p$ constraint. As a result, player $i$ has a best
  response $\tbp$ where $\tbp_{k_1} = \tbp_{k_1+1} = \dots = \tbp_{k_2}$, and if
  any of inequalities on $\bd$ is strict, then every best response has this
  property.
\end{proof}

\begin{lemma}
  \label{lem:inc-d-social-welfare}
  For all $n \ge 2$, $s \in \sall$, suppose $\bd$ contains a sequence $\bd_{k_1}
  \le \bd_{k_1 + 1} \le \dots \le \bd_{k_2}$. Then, for any $\bp \in
  \Delta(\pi)$, there exists a strategy matrix $\bp^* \in \Delta(\pi)$ such
  that
  \begin{equation}
    \label{eq:welfare-inc-d}
    W(\bp^*; n, \bd, s) \ge W(\bp; n, \bd, s)
  \end{equation}
  and
  $\bp^*(i) = \avg(\bp^*(i); k_1, k_2)$ for all $i$.
  If at least one
  inequality is strict (i.e., $\bd_{k} < \bd_{k'}$ for $k_1 \le k < k' \le
  k_2$), then~\eqref{eq:welfare-inc-d} is strict.
\end{lemma}
\begin{proof}
  We use a similar approach as the proof to \Cref{lem:inc-d-best-response}. Let
  $c = \sum_{k_1 \le k \le k_2} \bp_k^*$ and $p = \bp_{k_2+1}^*$. For
  $s \in \sup$, a symmetric socially optimal strategy $\bp^*$ satisfies
  \begin{align*}
    \bp^*_{k_1}, \dots, \bp^*_{k_2} \in \argmax_{\bp_{k_1}, \dots, \bp_{k_2}} ~
    &\sum_{k=k_1}^{k_2} \bp_k^* \bd_k u(n-1, \bp_k^*) \\
    \text{s.t.} ~
    &\bp_{k_1}^* \ge \bp_{k_1 + 1}^* \ge \dots \ge \bp_{k_2}^* \\
    &\sum_{k=k_1}^{k_2} \bp_k^* = c \\
    &\bp_{k_2}^* \ge p
  \end{align*}
  Because $\bp_k^*$ is decreasing in $k$, $u(n-1, \bp_k^*)$ is increasing in
  $k$. Therefore, $\bd_k u(n-1, \bp_k^*)$ is increasing in $k$, and if $\bd_k <
  \bd_{k'}$ for some $k_1 \le k < k' \le k_2$, then it must strictly increase
  somewhere. We apply \Cref{lem:inc-product} to show that there exists a
  solution in which all $\bp_k^*$'s are equal, with $b_\ell = \bp_\ell^*$,
  $a_\ell = \bd_k u(n-1, \bp_\ell^*)$, $f(x) = x$, and $\ell$ ranging from $k_1$
  to $k_2$. As in the proof of \Cref{lem:inc-d-best-response}, we can safely
  ignore the constraint $\bp_{k_2}^* \ge p$ because the all-equal solution
  maximizes $\bp_{k_2}^*$. And if there is a strict increase somewhere, then
  every optimal solution has equal $\bp_k^*$'s between $k_1$ and $k_2$.

  To extend this argument to $s \in \sdown$, note that because the identity
  score $s_1 \in \sup$, the lemma holds for $s_1$. Welfare under any other $s
  \in \sdown$ is equivalent to welfare under $s_1$, so the lemma extends to any
  $s \in \sdown$.
\end{proof}

\subsection{Omitted proofs for \Cref{sec:diversity}}
\label{app:diversity}

\divn*
\begin{proof}

  Assume without loss of generality that $\bd$ is decreasing.
  Let $\bpn{n} = \peq(n, \bd, s)$ and $\bpn{n+1} = \peq(n+1, \bd, s)$. Our plan
  will be to show that there is some $k^*$ such that $\bpn{n}_k > \bpn{n+1}_k$
  if and only if $k \le k^*$, and then apply \Cref{lem:inc-majorization} to
  conclude that $\bpn{n} \succ \bpn{n+1}$.

  Our goal will be to show that for $k < k'$, if $\bpn{n}_{k'} > 0$, then
  \begin{equation}
    \label{eq:small-k-big-change}
    \bd_k u(n, \bpn{n}_k) \le \bd_{k'} u(n, \bpn{n}_{k'}).
  \end{equation}
  (Note that $\bpn{n}_{k'} > 0$ implies $\bpn{n}_k > 0$ because $\bpn{n} \in
  \Delta(\pi)$.) Let $c^{(n)} = U(\bpn{n}; n, \bd, s)$. By
  \Cref{lem:unique-eq-decreasing}, for all $k$ such that $\bp_k^{(n)} > 0$,
  $\bd_k u(n, \bp_k^{(n)}) = c^{(n)}$. Therefore,
  \begin{equation}
    \label{eq:small-k-next-step}
    \bd_k u(n-1, \bp_k^{(n)}) = c^{(n)} = \bd_{k'} u(n-1, \bpn{n}_{k'}).
  \end{equation}
  Dividing each side of~\eqref{eq:small-k-big-change} by each respective side
  of~\eqref{eq:small-k-next-step}, our goal is to show that
  \begin{align*}
    \frac{\bd_k u(n, \bp_k^{(n)})}{\bd_k u(n-1, \bp_k^{(n)})}
    &\le \frac{\bd_{k'}
    u(n, \bp_{k'}^{(n)})}{\bd_{k'} u(n-1, \bp_{k'}^{(n)})} \\
    \frac{u(n, \bp_k^{(n)})}{u(n-1, \bp_k^{(n)})}
    &\le \frac{u(n, \bp_{k'}^{(n)})}{u(n-1, \bp_{k'}^{(n)})}.
  \end{align*}
  Because $\bp_{k'}^{(n)} \le \bp_k^{(n)}$, it suffices to show that
  \begin{align*}
    \frac{\partial}{\partial p} \frac{u(n, p)}{u(n-1, p)} \le 0.
  \end{align*}
  This holds by~\Cref{lem:dp-dn-ratio}.
  We can now use
  \Cref{lem:p-majorize} with $\bp = \bpn{n+1}$, $\bp' = \bpn{n}$, and $Q_k =
  \bd_k u(n, \bpn{n})$. By~\eqref{eq:small-k-big-change}, $\{Q_k\}$ is
  increasing as long as $\bpn{n}_k > 0$, so $\bpn{n} \succ \bpn{n+1}$.
\end{proof}

\divgamma*
\begin{proof}

  Assume without loss of generality that $\bd$ is decreasing.
  Let $\bpg = \peq(n, \bd, s_{\gz})$ and $\bpgp = \peq(n, \bd, s_{\go})$.
  We will use a similar approach to the proof of \Cref{thm:diversity-n} to show
  that $Q_k = \{\bd_k u_{s_{\go}}(n-1, \bpg_k)\}$ is increasing as long as
  $\bpg_k > 0$. We will then apply \Cref{lem:p-majorize} to conclude that $\bpg
  \succ \bpgp$.
  Abusing notation slightly, define
  \begin{align*}
    u_\gamma(n, p)
    \triangleq \E{\frac{1}{s_\gamma(1 + X(n, p))}}
    = \E{\frac{1}{(1+X(n, p))^\gamma}}.
  \end{align*}
  As we did with \Cref{thm:diversity-n}, we begin by
  showing that for $k < k'$, if $\bpg_{k'} > 0$, then
  \begin{equation}
    \label{eq:gamma-main-eq}
    \bd_k u_{\go}(n-1, \bpg_k) \le \bd_{k'} u_{\go}(n-1, \bpg_{k'}).
  \end{equation}
  Observe that
  \begin{equation}
    \label{eq:gamma-int}
    u_{\go}(n-1, p)
    = u_{\gz}(n-1, p) + \int_{\gz}^{\go} \frac{\partial}{\partial
    \gamma} u_\gamma(n-1, p) \dgam.
  \end{equation}
  Further, by \Cref{lem:unique-eq-decreasing}, if $k < k'$ and $\bpg_k \ge
  \bpg_{k'} > 0$,
  then
  \ifthenelse{\boolean{smallEqs}}{
  \begin{align*}
    \bd_k u_{\gz}(n-1, \bpg_k)
    &= U(\bpg; n, \bd, s_{\gz}) \\
    &= \bd_{k'} u_{\gz}(n-1, \bpg_k).
    \numberthis \label{eq:gamma-denom}
  \end{align*}
}{\begin{equation}
    \label{eq:gamma-denom}
    \bd_k u_{\gz}(n-1, \bpg_k) = U(\bpg; n, \bd, s_{\gz}) = \bd_{k'}
    u_{\gz}(n-1, \bpg_k).
  \end{equation}
}
  Substituting~\eqref{eq:gamma-int} and dividing both sides by the left and
  right of~\eqref{eq:gamma-denom} respectively, \eqref{eq:gamma-main-eq} is true
  if and only if
  \begin{align*}
    &\frac{\bd_k \p{u_{\gz}(n-1, \bpg_k) + \int_{\gz}^{\go} \frac{\partial}{\partial
    \gamma} u_\gamma(n-1, \bpg_k) \dgam}}{\bd_k u_{\gz}(n-1, \bpg_k)} \\
    \le
    &\frac{\bd_{k'} \p{u_{\gz}(n-1, \bpg_{k'}) + \int_{\gz}^{\go}
      \frac{\partial}{\partial \gamma} u_\gamma(n-1, \bpg_{k'}) \dgam}}{\bd_{k'}
    u_{\gz}(n-1, \bpg_{k'})}.
  \end{align*}
  Simplifying, this is
  \ifthenelse{\boolean{smallEqs}}{
  \begin{align*}
    \frac{\int_{\gz}^{\go} \frac{\partial}{\partial
    \gamma} u_\gamma(n-1, \bpg_k) \dgam}{u_{\gz}(n-1, \bpg_k)} \\
    -
    \frac{\int_{\gz}^{\go} \frac{\partial}{\partial \gamma} u_\gamma(n-1,
    \bpg_{k'}) \dgam}{u_{\gz}(n-1, \bpg_{k'})}
    &\le 0 \\
    \int_{\gz}^{\go} \frac{\partial}{\partial \gamma}
    \p{\frac{u_\gamma(n-1, \bpg_k)}{u_{\gz}(n-1, \bpg_k)} - \frac{u_\gamma(n-1,
    \bpg_{k'})}{u_{\gz}(n-1, \bpg_{k'})}}
    \dgam
    &\le 0.
  \end{align*}
}{\begin{align*}
    \frac{\int_{\gz}^{\go} \frac{\partial}{\partial
    \gamma} u_\gamma(n-1, \bpg_k) \dgam}{u_{\gz}(n-1, \bpg_k)}
    -
    \frac{\int_{\gz}^{\go} \frac{\partial}{\partial \gamma} u_\gamma(n-1,
    \bpg_{k'}) \dgam}{u_{\gz}(n-1, \bpg_{k'})}
    &\le 0 \\
    \int_{\gz}^{\go} \frac{\partial}{\partial \gamma}
    \p{\frac{u_\gamma(n-1, \bpg_k)}{u_{\gz}(n-1, \bpg_k)} - \frac{u_\gamma(n-1,
    \bpg_{k'})}{u_{\gz}(n-1, \bpg_{k'})}}
    \dgam
    &\le 0.
  \end{align*}
}
  It suffices to show that the term inside the integral is nonpositive for all
  $\gamma \in [\gz, \go]$. For now, assume that $\go \le \gz + \frac{1}{\log
  n}$. We will show that under this assumption,~\eqref{eq:gamma-main-eq} holds,
  and we will use \Cref{lem:p-majorize} to conclude that $\bpg \succ \bpgp$ for
  all such $\go$. Iteratively applying this result will allow us to remove the
  assumption. Our goal is to prove that
  \ifthenelse{\boolean{smallEqs}}{
  \begin{align*}
    \frac{\partial}{\partial \gamma}
    \p{\frac{u_\gamma(n-1, \bpg_k)}{u_{\gz}(n-1, \bpg_k)} - \frac{u_\gamma(n-1,
    \bpg_{k'})}{u_{\gz}(n-1, \bpg_{k'})}}
    &\le 0 \\
    \frac{\frac{\partial}{\partial
    \gamma} u_\gamma(n-1, \bpg_k)}{u_{\gz}(n-1, \bpg_k)}
    -
    \frac{\frac{\partial}{\partial \gamma} u_\gamma(n-1,
    \bpg_{k'})}{u_{\gz}(n-1, \bpg_{k'})}
    &\le 0.
  \end{align*}
}{\begin{align*}
    \frac{\partial}{\partial \gamma}
    \p{\frac{u_\gamma(n-1, \bpg_k)}{u_{\gz}(n-1, \bpg_k)} - \frac{u_\gamma(n-1,
    \bpg_{k'})}{u_{\gz}(n-1, \bpg_{k'})}}
    &\le 0 \\
    \frac{\frac{\partial}{\partial
    \gamma} u_\gamma(n-1, \bpg_k)}{u_{\gz}(n-1, \bpg_k)}
    &\le
    \frac{\frac{\partial}{\partial \gamma} u_\gamma(n-1,
    \bpg_{k'})}{u_{\gz}(n-1, \bpg_{k'})}.
  \end{align*}
}
  Because $\bpg_k \ge \bpg_{k'}$, it suffices to show that this ratio is
  decreasing in $p$:
  \begin{equation}
    \label{eq:p-gamma-partials}
    \frac{\partial}{\partial p} 
    \frac{\frac{\partial}{\partial
    \gamma} u_\gamma(n-1, p)}{u_{\gz}(n-1, p)} \le 0.
  \end{equation}
  First, we take the partial with respect to $\gamma$:
  \begin{align*}
    \frac{\partial}{\partial \gamma} u_\gamma(n-1, p)
    &= \frac{\partial}{\partial \gamma} \E{(1 + X(n-1, p))^{-\gamma}} \\
    &= \E{\frac{\partial}{\partial \gamma} (1 + X(n-1, p))^{-\gamma}} \\
    &= \E{-\frac{\log(1 + X(n-1, p))}{(1 + X(n-1, p))^{-\gamma}}}.
  \end{align*}
  Define
  \begin{align*}
    z_\gamma(n, p)
    \triangleq \E{\frac{\log(1 + X(n, p))}{(1 + X(n, p))^{\gamma}}}
    = - \frac{\partial}{\partial \gamma} u_\gamma(n, p).
  \end{align*}
  Plugging this into~\eqref{eq:p-gamma-partials}, we want to show that
  \begin{equation}
    \label{eq:p-partial-only}
    \frac{\partial}{\partial p} \frac{z_\gamma(n-1, p)}{u_{\gz}(n-1, p)} \ge 0.
  \end{equation}
  This holds if and only if 
  \begin{align*}
    \frac{\partial}{\partial p} \log\p{\frac{z_\gamma(n-1, p)}{u_{\gz}(n-1, p)}}
    &\ge 0 \\
    \frac{\partial}{\partial p} \log\p{z_\gamma(n-1, p)}
    &\ge \frac{\partial}{\partial p} \log\p{u_{\gz}(n-1, p)} \\
    \frac{\frac{\partial}{\partial p} z_\gamma(n-1, p)}{z_\gamma(n-1, p)}
    &\ge \frac{\frac{\partial}{\partial p} u_{\gz}(n-1, p)}{u_{\gz}(n-1, p)} \\
    \frac{\frac{n-1}{p} \dz_\gamma(n-2, p)}{z_\gamma(n-1, p)}
    &\ge \frac{\frac{n-1}{p} \du_{\gz}(n-2, p)}{u_{\gz}(n-1, p)}
    \tag{by \Cref{lem:sah-diff}} \\
    \frac{z_\gamma(n-1, p) - z_\gamma(n-2, p)}{z_\gamma(n-1, p)}
    &\ge \frac{u_{\gz}(n-1, p) - u_{\gz}(n-2, p)}{u_{\gz}(n-1, p)} \\
    \frac{z_\gamma(n-2, p)}{z_\gamma(n-1, p)}
    &\le \frac{u_{\gz}(n-2, p)}{u_{\gz}(n-1, p)}
  \end{align*}
  Here, we need one more pair of definitions:
  \def\uplus{u^{+}}
  \def\zplus{z^{+}}
  \begin{align*}
    \uplus_{\gamma}(n, p)
    &\triangleq \E{(2 + X(n, p))^{-\gamma}} \\
    \zplus_{\gamma}(n, p)
    &\triangleq \E{\frac{\log(2 + X(n, p))}{(2 + X(n, p))^{\gamma}}}.
  \end{align*}
  With these, we must show that
  \ifthenelse{\boolean{smallEqs}}{
  \begin{align*}
    \frac{z_\gamma(n-2, p)}{z_\gamma(n-1, p)}
    &\le \frac{u_{\gz}(n-2, p)}{u_{\gz}(n-1, p)} \\
    \frac{z_\gamma(n-1, p)}{z_\gamma(n-2, p)}
    &\ge \frac{u_{\gz}(n-1, p)}{u_{\gz}(n-2, p)}
  \end{align*}
  \begin{align*}
    &\frac{(1-p)z_\gamma(n-2, p) + p \zplus_\gamma(n-2, p)}{z_\gamma(n-2, p)} \\
    &\ge \frac{(1-p) u_{\gz}(n-2, p) + p \uplus_{\gz}(n-2, p)}{u_{\gz}(n-2, p)}
  \end{align*}
  \begin{align*}
    1 - p + \frac{p \zplus_\gamma(n-2, p)}{z_\gamma(n-2, p)}
    &\ge 1 - p + \frac{p \uplus_{\gz}(n-2, p)}{u_{\gz}(n-2, p)} \\
    \frac{\zplus_\gamma(n-2, p)}{z_\gamma(n-2, p)}
    &\ge \frac{\uplus_{\gz}(n-2, p)}{u_{\gz}(n-2, p)} \\
    \zplus_\gamma(n-2, p) u_{\gz}(n-2, p)
    &\ge \uplus_{\gz}(n-2, p) z_\gamma(n-2, p).
  \end{align*}
}{\begin{align*}
    \frac{z_\gamma(n-2, p)}{z_\gamma(n-1, p)}
    &\le \frac{u_{\gz}(n-2, p)}{u_{\gz}(n-1, p)} \\
    \frac{z_\gamma(n-1, p)}{z_\gamma(n-2, p)}
    &\ge \frac{u_{\gz}(n-1, p)}{u_{\gz}(n-2, p)} \\
    \frac{(1-p)z_\gamma(n-2, p) + p \zplus_\gamma(n-2, p)}{z_\gamma(n-2, p)}
    &\ge \frac{(1-p) u_{\gz}(n-2, p) + p \uplus_{\gz}(n-2, p)}{u_{\gz}(n-2, p)}
    \\
    1 - p + \frac{p \zplus_\gamma(n-2, p)}{z_\gamma(n-2, p)}
    &\ge 1 - p + \frac{p \uplus_{\gz}(n-2, p)}{u_{\gz}(n-2, p)} \\
    \frac{\zplus_\gamma(n-2, p)}{z_\gamma(n-2, p)}
    &\ge \frac{\uplus_{\gz}(n-2, p)}{u_{\gz}(n-2, p)} \\
    \zplus_\gamma(n-2, p) u_{\gz}(n-2, p)
    &\ge \uplus_{\gz}(n-2, p) z_\gamma(n-2, p).
  \end{align*}
}
  This holds if and only if
  \ifthenelse{\boolean{smallEqs}}{
  \begin{align*}
    &\p{\sum_{\ell=0}^{n-2} b(\ell, n-2, p) \frac{\log(2+\ell)}{(2+\ell)^\gamma}}
    \p{\sum_{\ell=0}^{n-2}  \frac{b(\ell, n-2, p)}{(1+\ell)^{\gz}}} \\
    &-
    \p{\sum_{\ell=0}^{n-2} b(\ell, n-2, p) \frac{\log(1+\ell)}{(1+\ell)^\gamma}}
    \p{\sum_{\ell=0}^{n-2}  \frac{b(\ell, n-2, p)}{(2+\ell)^{\gz}}} \\
    &\ge 0.
  \end{align*}
}{\begin{align*}
    \p{\sum_{\ell=0}^{n-2} b(\ell, n-2, p) \frac{\log(2+\ell)}{(2+\ell)^\gamma}}
    \p{\sum_{\ell=0}^{n-2} b(\ell, n-2, p) \frac{1}{(1+\ell)^{\gz}}} \\
    -
    \p{\sum_{\ell=0}^{n-2} b(\ell, n-2, p) \frac{\log(1+\ell)}{(1+\ell)^\gamma}}
    \p{\sum_{\ell=0}^{n-2} b(\ell, n-2, p) \frac{1}{(2+\ell)^{\gz}}}
    \ge 0.
  \end{align*}
}
  Expanding terms, this is
  \ifthenelse{\boolean{smallEqs}}{
  \begin{align*}
    &\sum_{\ell_1 = 0}^{n-3} \sum_{\ell_2 = \ell_1 + 1}^{n-2} b(\ell_1, n-2, p)
    b(\ell_2, n-2, p) \cdot \\
    &\left[\frac{\log(2+\ell_1)}{(2+\ell_1)^\gamma} \cdot \frac{1}{(1 +
      \ell_2)^{\gz}} + \frac{\log(2+\ell_2)}{(2+\ell_2)^\gamma} \cdot
    \frac{1}{(1 + \ell_1)^{\gz}}\right. \\
    &\left. - \frac{\log(1+\ell_1)}{(1+\ell_1)^\gamma} \cdot \frac{1}{(2 +
      \ell_2)^{\gz}} - \frac{\log(1+\ell_2)}{(1+\ell_2)^\gamma} \cdot
    \frac{1}{(2 + \ell_1)^{\gz}}\right] \numberthis \label{eq:bracket-1} \\
    &+ \sum_{\ell=0}^{n-2} b(\ell, n-2, p)^2 \cdot \\
    &\left[\frac{\log(2+\ell)}{(2+\ell)^\gamma} \cdot \frac{1}{(1 + \ell)^{\gz}}
      + \frac{\log(2+\ell)}{(2+\ell)^\gamma} \cdot \frac{1}{(1 +
    \ell)^{\gz}}\right. \\
    &\left. - \frac{\log(1+\ell)}{(1+\ell)^\gamma} \cdot \frac{1}{(2 +
      \ell)^{\gz}} - \frac{\log(1+\ell)}{(1+\ell)^\gamma} \cdot \frac{1}{(2 +
    \ell)^{\gz}}\right]
    \ge 0. \numberthis \label{eq:bracket-2}
  \end{align*}
}{\begin{align*}
    \sum_{\ell_1 = 0}^{n-3} \sum_{\ell_2 = \ell_1 + 1}^{n-2} b(\ell_1, n-2, p)
    b(\ell_2, n-2, p)
    &\left[\frac{\log(2+\ell_1)}{(2+\ell_1)^\gamma} \cdot \frac{1}{(1 +
      \ell_2)^{\gz}} + \frac{\log(2+\ell_2)}{(2+\ell_2)^\gamma} \cdot
    \frac{1}{(1 + \ell_1)^{\gz}}\right. \\
    &\left. - \frac{\log(1+\ell_1)}{(1+\ell_1)^\gamma} \cdot \frac{1}{(2 +
      \ell_2)^{\gz}} - \frac{\log(1+\ell_2)}{(1+\ell_2)^\gamma} \cdot
    \frac{1}{(2 + \ell_1)^{\gz}}\right] \numberthis \label{eq:bracket-1} \\
    + \sum_{\ell=0}^{n-2} b(\ell, n-2, p)^2
    &\left[\frac{\log(2+\ell)}{(2+\ell)^\gamma} \cdot \frac{1}{(1 + \ell)^{\gz}}
      + \frac{\log(2+\ell)}{(2+\ell)^\gamma} \cdot \frac{1}{(1 +
    \ell)^{\gz}}\right. \\
    &\left. - \frac{\log(1+\ell)}{(1+\ell)^\gamma} \cdot \frac{1}{(2 +
      \ell)^{\gz}} - \frac{\log(1+\ell)}{(1+\ell)^\gamma} \cdot \frac{1}{(2 +
    \ell)^{\gz}}\right]
    \ge 0. \numberthis \label{eq:bracket-2}
  \end{align*}
}
  We will show that the terms in brackets in~\eqref{eq:bracket-1}
  and~\eqref{eq:bracket-2} are both nonnegative. We begin
  with~\eqref{eq:bracket-1}.
  \begin{align*}
    &\frac{\log(2+\ell_1)}{(2+\ell_1)^\gamma} \cdot \frac{1}{(1 + \ell_2)^{\gz}}
    + \frac{\log(2+\ell_2)}{(2+\ell_2)^\gamma} \cdot \frac{1}{(1 +
    \ell_1)^{\gz}} \\
    &- \frac{\log(1+\ell_1)}{(1+\ell_1)^\gamma} \cdot \frac{1}{(2
    + \ell_2)^{\gz}} - \frac{\log(1+\ell_2)}{(1+\ell_2)^\gamma} \cdot
    \frac{1}{(2 + \ell_1)^{\gz}} \\
    =
    & \frac{1}{(2+\ell_1)^{\gz} (1 + \ell_2)^{\gz}} \p{\frac{\log(2 +
      \ell_1)}{(2+\ell_1)^{\gamma - \gz}} -
    \frac{\log(1+\ell_2)}{(1+\ell_2)^{\gamma-\gz}}} \\
    &+ \frac{1}{(1+\ell_1)^{\gz} (2 + \ell_2)^{\gz}}
    \p{\frac{\log(2+\ell_2)}{(2 + \ell_2)^{\gamma - \gz}} -
    \frac{\log(1+\ell_1)}{(1+\ell_1)^{\gamma - \gz}}}
  \end{align*}
  Let $t(\ell) \triangleq \frac{\log(\ell)}{\ell^{\gamma-\gz}}$. Our goal is to
  show that the above is nonnegative, or
  \ifthenelse{\boolean{smallEqs}}{
  \begin{align*}
    &\frac{1}{(1+\ell_1)^{\gz} (2 + \ell_2)^{\gz}} \p{t(2+\ell_2) - t(1 +
    \ell_1)} \\
    &\ge \frac{1}{(2+\ell_1)^{\gz} (1 + \ell_2)^{\gz}} \p{t(1 + \ell_2) - t(2 +
    \ell_1)} \\
    &t(2 + \ell_2) - t(1 + \ell_1) \\
    &\ge \frac{((1+\ell_1)(2 + \ell_2))^{\gz}}{((2+\ell_1) (1 + \ell_2))^{\gz}}
    \p{t(1 + \ell_2) - t(2 + \ell_1)}
  \end{align*}
}{\begin{align*}
    \frac{1}{(1+\ell_1)^{\gz} (2 + \ell_2)^{\gz}} \p{t(2+\ell_2) - t(1 +
    \ell_1)}
    &\ge \frac{1}{(2+\ell_1)^{\gz} (1 + \ell_2)^{\gz}} \p{t(1 + \ell_2) - t(2 +
    \ell_1)} \\
    t(2 + \ell_2) - t(1 + \ell_1)
    &\ge \frac{((1+\ell_1)(2 + \ell_2))^{\gz}}{((2+\ell_1) (1 + \ell_2))^{\gz}}
    \p{t(1 + \ell_2) - t(2 + \ell_1)}
  \end{align*}
}
  Because $\ell_1 < \ell_2$,
  \begin{align*}
    \frac{((1+\ell_1)(2 + \ell_2))^{\gz}}{((2+\ell_1) (1 + \ell_2))^{\gz}}
    &= \frac{(2 + 2 \ell_1 + \ell_2 + \ell_1 \ell_2)^{\gz}}{(2 + \ell_1 + 2
    \ell_2 + \ell_1 \ell_2)^{\gz}} 
    \le 1.
  \end{align*}
  It therefore suffices to show that for $\ell_2 > \ell_1$,
  \begin{align*}
    t(2 + \ell_2) - t(1 + \ell_1)
    &\ge t(1 + \ell_2) - t(2 + \ell_1) \\
    t(2 + \ell_1) + t(2 + \ell_2)
    &\ge t(1 + \ell_1) + t(1 + \ell_2)
    \numberthis \label{eq:need-t-increasing}
  \end{align*}
  If $t$ is weakly increasing, this is trivially true: each term on the left
  hand side is weakly larger than the corresponding term on the right hand side.
  While $t$ is not increasing in general, it is as long as its arguments are
  sufficiently small relative to $\gamma - \gz$:
  \begin{align*}
    \frac{d}{d\ell} t(\ell)
    &\ge 0 \\
    \frac{\ell^{\gamma - \gz} \ell^{-1} - \log(\ell) (\gamma - \gz)
    \ell^{\gamma - \gz - 1}}{\ell^{2(\gamma - \gz)}}
    &\ge 0 \\
    1 - (\gamma - \gz) \log(\ell)
    &\ge 0 \\
    (\gamma - \gz)
    &\le \frac{1}{\log(\ell)}
  \end{align*}
  Here, we make use of our assumption $\go - \gz \le \frac{1}{\log n}$.
  Because $\gamma \in [\gz, \go]$, $t$ is increasing as long as we plug in
  values that are at most $n$. The largest value $\ell_2$ can take is $n-2$, so
  the largest value we plug in is $2 + n - 2 = n$.
  Thus,~\eqref{eq:need-t-increasing} holds, and the first term in
  brackets~\eqref{eq:bracket-1} is nonnegative. We use a similar, though
  simpler, argument to handle the second term~\eqref{eq:bracket-2}.
  \begin{align*}
    &\frac{\log(2+\ell)}{(2+\ell)^\gamma} \cdot \frac{1}{(1 + \ell)^{\gz}}
      + \frac{\log(2+\ell)}{(2+\ell)^\gamma} \cdot \frac{1}{(1 +
    \ell)^{\gz}} \\
    &- \frac{\log(1+\ell)}{(1+\ell)^\gamma} \cdot \frac{1}{(2 +
      \ell)^{\gz}} - \frac{\log(1+\ell)}{(1+\ell)^\gamma} \cdot \frac{1}{(2 +
    \ell)^{\gz}} \\
    =
    & \frac{2}{(2+\ell)^{\gz} (1+\ell)^{\gz}}
    \p{\frac{\log(2+\ell)}{(2+\ell)^{\gamma-\gz}} -
    \frac{\log(1+\ell)}{(1+\ell)^{\gamma - \gz}}} \\
    & \frac{2}{(2+\ell)^{\gz} (1+\ell)^{\gz}}
    \p{t(2+\ell) - t(1+\ell)}
  \end{align*}
  Again, this is nonnegative as long as $t$ is weakly increasing, which it is
  for $\ell \le n-2$. With this, we have shown~\eqref{eq:gamma-main-eq} under
  the assumption that $\go \le \gz + \frac{1}{\log n}$. Still under this
  assumption, we use \Cref{lem:p-majorize} with $\bp' = \bpg$, $\bp = \bpgp$,
  and $Q_k = \bd_k u_{\go}(n-1, \bpg_k)$. By~\eqref{eq:gamma-main-eq}, $\{Q_k\}$
  is increasing as long as $\bpg_k > 0$, so $\bpg \succ \bpgp$.

  To finish the proof, we must remove the assumption that $\go \le \gz +
  \frac{1}{\log n}$. For any $\go \ge \gz$, we simply need to iteratively apply
  this argument $\lceil (\go-\gz)\log n\rceil$ times to get
  \begin{align*}
    \bpg \succ \bp^{(\gz + 1/\log(n))} \succ \bp^{(\gz + 2/\log(n))} \succ
    \dots \succ \bpgp.
  \end{align*}
  This suffices for any finite $\gamma_1$. To handle the case where $\gamma_1 =
  \infty$, we invoke \Cref{lem:inf-most-diverse} which states that for any $s
  \in \sall$ (including $s_{\gz}$), $\peq(n, \bd, s) \succ \peq(n, \bd,
  s_\infty)$.
\end{proof}

\divopt*
\begin{proof}

  Assume without loss of generality that $\bd$ is decreasing.
  Let $\bps = \popt(n, \bd, s)$. We begin with the case where $s \in \sup$. By
  \Cref{lem:unique-opt-decreasing}, for all $k$ such that $\bp_k^* > 0$,
  \begin{align*}
    \frac{\partial}{\partial p} \bd_k p u(n-1, p) \dpat{\bp_k^*} = c
  \end{align*}
  for $c = \frac{1}{n} \|\nabla W(\bp)\|_\infty$. Expanding,
  \ifthenelse{\boolean{smallEqs}}{
  \begin{align*}
    \frac{\partial}{\partial p} \bd_k p u(n-1, p) \dpat{\bp_k^*}
    &= c \\
    \bd_k u(n-1, \bp_k^*) + \bd_k \bp_k^* \frac{\partial}{\partial p}  u(n-1, p)
    \dpat{\bp_k^*}
    &= c \\
    \bd_k u(n-1, \bp_k^*) + \bd_k \bp_k^* \frac{n-1}{\bp_k^*} \du(n-2, \bp_k^*)
    &= c
    \tag{by \Cref{lem:sah-diff}} \\
    \bd_k u(n-1, \bp_k^*)
    = c - \bd_k (n-1) \du(n-2, \bp_k^*)
    \numberthis \label{eq:opt-as-diff}
  \end{align*}
}{\begin{align*}
    \frac{\partial}{\partial p} \bd_k p u(n-1, p) \dpat{\bp_k^*}
    &= c \\
    \bd_k u(n-1, \bp_k^*) + \bd_k \bp_k^* \frac{\partial}{\partial p}  u(n-1, p)
    \dpat{\bp_k^*}
    &= c \\
    \bd_k u(n-1, \bp_k^*) + \bd_k \bp_k^* \frac{n-1}{\bp_k^*} \du(n-2, \bp_k^*)
    &= c
    \tag{by \Cref{lem:sah-diff}} \\
    \bd_k u(n-1, \bp_k^*)
    &= c - \bd_k (n-1) \du(n-2, \bp_k^*)
    \numberthis \label{eq:opt-as-diff}
  \end{align*}
}
  We will show that the left hand side is decreasing in $k$: For $k < k'$, our
  goal is to show that
  \begin{equation}
    \label{eq:opt-div-setup}
    \bd_k u(n-1, \bp_k^*)
    \ge \bd_{k'} u(n-1, \bp_{k'}^*).
  \end{equation}
  If $\bp_k^* = 0$, then $\bp_{k'}^* = 0$ as well. In this case, because $u(n-1,
  0) = 1$, this is simply $\bd_k \ge \bd_{k'}$, which is true by assumption. We
  can therefore focus on the case where $\bp_k^* > 0$.
  \ifthenelse{\boolean{smallEqs}}{
  \begin{align*}
    &\bd_k u(n-1, \bp_k^*)
    \ge \bd_{k'} u(n-1, \bp_{k'}^*) \\
    \Longleftrightarrow~ &
    c - \bd_k (n-1) \du(n-2, \bp_k^*) \\
    &\ge c - \bd_{k'} (n-1) \du(n-2, \bp_{k'}^*)
    \tag{by~\eqref{eq:opt-as-diff}} \\
    \Longleftrightarrow~ &
    \bd_k \du(n-2, \bp_k^*)
    \le \bd_{k'} \du(n-2, \bp_{k'}^*)
    \numberthis \label{eq:opt-ineq-before-dividing}
  \end{align*}
}{\begin{align*}
    &\bd_k u(n-1, \bp_k^*)
    \ge \bd_{k'} u(n-1, \bp_{k'}^*) \\
    \Longleftrightarrow~ &
    c - \bd_k (n-1) \du(n-2, \bp_k^*)
    \ge c - \bd_{k'} (n-1) \du(n-2, \bp_{k'}^*)
    \tag{by~\eqref{eq:opt-as-diff}} \\
    \Longleftrightarrow~ &
    \bd_k \du(n-2, \bp_k^*)
    \le \bd_{k'} \du(n-2, \bp_{k'}^*)
    \numberthis \label{eq:opt-ineq-before-dividing}
  \end{align*}
}
  Next, observe that by \Cref{lem:unique-opt-decreasing},
  \ifthenelse{\boolean{smallEqs}}{
  \begin{align*}
    \bd_k &u(n-1, \bp_k^*) + \bd_k (n-1) \du(n-2, \bp_k^*) \\
    &\ge \bd_{k'} u(n-1, \bp_{k'}^*) + \bd_{k'} (n-1) \du(n-2, \bp_{k'}^*).
    \numberthis \label{eq:opt-denom}
  \end{align*}
}{\begin{equation}
    \label{eq:opt-denom}
    \bd_k u(n-1, \bp_k^*) + \bd_k (n-1) \du(n-2, \bp_k^*)
    \ge \bd_{k'} u(n-1, \bp_{k'}^*) + \bd_{k'} (n-1) \du(n-2, \bp_{k'}^*).
  \end{equation}
}
  This holds because by assumption, $\bp_k^* > 0$, so the left hand side is $c$,
  and the right hand side is at most $c$. Note that both sides
  of~\eqref{eq:opt-ineq-before-dividing} are weakly negative, and both sides
  of~\eqref{eq:opt-denom} are positive. Therefore, it suffices to show that the
  inequality holds if we divide both sides
  of~\eqref{eq:opt-ineq-before-dividing} by the respective sides
  of~\eqref{eq:opt-denom} to get
  \ifthenelse{\boolean{smallEqs}}{
  \begin{align*}
    &
    \frac{\bd_k \du(n-2, \bp_k^*)}{\bd_k u(n-1, \bp_k^*) + \bd_k
    (n-1) \du(n-2, \bp_k^*)} \\
    &\le \frac{\bd_{k'} \du(n-2, \bp_{k'}^*)}{\bd_{k'} u(n-1, \bp_{k'}^*) +
    \bd_{k'} (n-1) \du(n-2, \bp_{k'}^*)} \\
    \Longleftrightarrow~
    &
    \frac{u(n-1, \bp_k^*) + (n-1) \du(n-2, \bp_k^*)}{\du(n-2, \bp_k^*)} \\
    &\ge \frac{u(n-1, \bp_{k'}^*) + (n-1) \du(n-2, \bp_{k'}^*)}{\du(n-2,
    \bp_{k'}^*)} \\
    \Longleftrightarrow~
    &
    \frac{u(n-1, \bp_k^*)}{\du(n-2, \bp_k^*)} \ge \frac{u(n-1,
    \bp_{k'}^*)}{\du(n-2, \bp_{k'}^*)} \\
    \Longleftrightarrow~
    &
    \frac{\du(n-2, \bp_k^*)}{u(n-1, \bp_k^*)} \le \frac{\du(n-2,
    \bp_{k'}^*)}{u(n-1, \bp_{k'}^*)} \\
    \Longleftrightarrow~
    &
    \frac{u(n-1, \bp_k^*) - u(n-2, \bp_k^*)}{u(n-1, \bp_k^*)} \\
    &\le \frac{u(n-1,
    \bp_{k'}^*) - u(n-2, \bp_{k'}^*)}{u(n-1, \bp_{k'}^*)} \\
    \Longleftrightarrow~ &
    -\frac{u(n-2, \bp_k^*)}{u(n-1, \bp_k^*)} \le - \frac{u(n-2,
    \bp_{k'}^*)}{u(n-1, \bp_{k'}^*)} \\
    \Longleftrightarrow~ &
    \frac{u(n-2, \bp_k^*)}{u(n-1, \bp_k^*)} \ge \frac{u(n-2,
    \bp_{k'}^*)}{u(n-1, \bp_{k'}^*)} \\
    \Longleftrightarrow~ &
    \frac{u(n-1, \bp_k^*)}{u(n-2, \bp_k^*)} \le \frac{u(n-1, \bp_{k'}^*)}{u(n-2,
    \bp_{k'}^*)} \\
    \Longleftarrow~ &
    \frac{\partial}{\partial p} \frac{u(n-1, p)}{u(n-2, p)} \le 0.
  \end{align*}
}{\begin{align*}
    &
    \frac{\bd_k \du(n-2, \bp_k^*)}{\bd_k u(n-1, \bp_k^*) + \bd_k
    (n-1) \du(n-2, \bp_k^*)}
    \le \frac{\bd_{k'} \du(n-2, \bp_{k'}^*)}{\bd_{k'} u(n-1, \bp_{k'}^*) +
    \bd_{k'} (n-1) \du(n-2, \bp_{k'}^*)} \\
    \Longleftrightarrow~ &
    \frac{u(n-1, \bp_k^*) + (n-1) \du(n-2, \bp_k^*)}{\du(n-2, \bp_k^*)}
    \ge \frac{u(n-1, \bp_{k'}^*) + (n-1) \du(n-2, \bp_{k'}^*)}{\du(n-2,
    \bp_{k'}^*)} \\
    \Longleftrightarrow~ &
    \frac{u(n-1, \bp_k^*)}{\du(n-2, \bp_k^*)} \ge \frac{u(n-1,
    \bp_{k'}^*)}{\du(n-2, \bp_{k'}^*)} \\
    \Longleftrightarrow~ &
    \frac{\du(n-2, \bp_k^*)}{u(n-1, \bp_k^*)} \le \frac{\du(n-2,
    \bp_{k'}^*)}{u(n-1, \bp_{k'}^*)} \\
    \Longleftrightarrow~ &
    \frac{u(n-1, \bp_k^*) - u(n-2, \bp_k^*)}{u(n-1, \bp_k^*)} \le \frac{u(n-1,
    \bp_{k'}^*) - u(n-2, \bp_{k'}^*)}{u(n-1, \bp_{k'}^*)} \\
    \Longleftrightarrow~ &
    -\frac{u(n-2, \bp_k^*)}{u(n-1, \bp_k^*)} \le - \frac{u(n-2,
    \bp_{k'}^*)}{u(n-1, \bp_{k'}^*)} \\
    \Longleftrightarrow~ &
    \frac{u(n-2, \bp_k^*)}{u(n-1, \bp_k^*)} \ge \frac{u(n-2,
    \bp_{k'}^*)}{u(n-1, \bp_{k'}^*)} \\
    \Longleftrightarrow~ &
    \frac{u(n-1, \bp_k^*)}{u(n-2, \bp_k^*)} \le \frac{u(n-1, \bp_{k'}^*)}{u(n-2,
    \bp_{k'}^*)} \\
    \Longleftarrow~ &
    \frac{\partial}{\partial p} \frac{u(n-1, p)}{u(n-2, p)} \le 0.
  \end{align*}
}
  This last inequality holds by \Cref{lem:dp-dn-ratio}.
  With this, we can apply \Cref{lem:p-majorize} with $\bp' = \bp^*$, $\bp =
  \peq(n, \bd, s)$, and $Q_k = \bd_k u(n-1, \bp_k^*)$. As we have shown,
  $\{Q_k\}$ is decreasing, so $\peq(n, \bd, s) \succ \bp^*$.

  Next, we handle the case where $s \in \sdown$.
  \begin{align*}
    \peq(n, \bd, s)
    &\succ \peq(n, \bd, s_\infty) \tag{\Cref{lem:inf-most-diverse}} \\
    &= \popt(n, \bd, s) \tag{\Cref{lem:opt-inf-equivalence}},
  \end{align*}
  which proves the claim.
\end{proof}

\subsection{Proofs for \Cref{sec:compeitition-limit}}

\gammaconv*
\begin{proof}
  For finite $\gamma$,
  \begin{align*}
    |\bpn{n} - \bplim| \le O\p{\sqrt{\frac{\log n}{n}}}
  \end{align*}
  by \Cref{lem:p-convergence}. We will return to the $\gamma = \infty$ case
  later.

  Next, we consider $\qkn{n} = \popt(n, \bd, s_\gamma)$. We
  begin with the case where $\gamma < 1$, so $s_\gamma \in \sup$.
  By \Cref{lem:unique-opt-decreasing}, for all $k > k'$ such that $\qkn{n} > 0$,
  \begin{align*}
    \frac{\bd_k}{n} \frac{\partial}{\partial p} w(n, p) \dpat{\qkn{n}}
    &= \frac{\bd_1}{n} \frac{\partial}{\partial p} w(n, p) \dpat{\qkpn{n}} \\
    \frac{\bd_k}{\qkn{n}} \dw(n-1, \qkn{n})
    &= \frac{\bd_1}{\qkpn{n}} \dw(n-1, \qkpn{n}). \numberthis \label{eq:dw-equal}
  \end{align*}
  where the last step follows from \Cref{lem:sah-diff}.
  Observe that
  \ifthenelse{\boolean{smallEqs}}{
  \begin{align*}
    &\frac{1}{p} \dw(n-1, p) \\
    &= \frac{1}{p} (w(n, p) - w(n-1, p)) \\
    &= \frac{1}{p} \p{\E{\frac{X(n, p)}{s_\gamma(X(n, p))}} - \E{\frac{X(n-1,
    p)}{s_\gamma(X(n-1, p))}}} \\
    &= \frac{1}{p} \p{\E{X(n, p)^{1-\gamma}} - \E{X(n-1, p)^{1-\gamma}}} \\
    &= \frac{1}{p} \left(p \E{(1 + X(n-1, p))^{1-\gamma}}\right. \\
    &+ \left.(1-p) \E{X(n-1, p)^{1-\gamma}} - \E{X(n-1, p)^{1-\gamma}}\right) \\
    &= \E{(1 + X(n-1, p))^{1-\gamma} - X(n-1, p)^{1-\gamma}}.
  \end{align*}
}{\begin{align*}
    \frac{1}{p} \dw(n-1, p)
    &= \frac{1}{p} (w(n, p) - w(n-1, p)) \\
    &= \frac{1}{p} \p{\E{\frac{X(n, p)}{s_\gamma(X(n, p))}} - \E{\frac{X(n-1,
    p)}{s_\gamma(X(n-1, p))}}} \\
    &= \frac{1}{p} \p{\E{X(n, p)^{1-\gamma}} - \E{X(n-1, p)^{1-\gamma}}} \\
    &= \frac{1}{p} \p{p \E{(1 + X(n-1, p))^{1-\gamma}} + (1-p) \E{X(n-1,
    p)^{1-\gamma}} - \E{X(n-1, p)^{1-\gamma}}} \\
    &= \E{(1 + X(n-1, p))^{1-\gamma} - X(n-1, p)^{1-\gamma}}.
  \end{align*}
}

  We use the fact that $X(n, p)$ concentrates. Defining $\mu = np$, let $\mc E$
  be the event that $|X(n, p) - \mu| \le (1+\delta) \mu$. Let $\overline{\mc E}$
  be its complement. Then,
  \begin{align*}
    &\E{(1 + X(n, p))^{1-\gamma} - X(n, p)^{1-\gamma}} \\
    &= \E{(1 + X(n, p))^{1-\gamma} - X(n, p)^{1-\gamma} \given \mc E}
    \cdot \Pr[\mc E] \\
    &+ \E{(1 + X(n, p))^{1-\gamma} - X(n, p)^{1-\gamma} \given \overline{\mc E}}
    \cdot \Pr[\overline{\mc E}]
  \end{align*}
  The first term is
  \begin{align*}
    &\E{(1 + X(n, p))^{1-\gamma} - X(n, p)^{1-\gamma} \given \mc E} \\
    &=\E{(1-\gamma) (\alpha_1 + X(n, p))^{-\gamma} \given \mc E}
    \tag{Taylor's theorem, $\alpha_1 \in [0, 1]$} \\
    &= (1-\gamma) (\alpha_1 + \alpha_2 + \mu)^{-\gamma}
    \tag{Mean value theorem, $\alpha_2 \in [-\delta \mu, \delta \mu]$}
  \end{align*}
  Thus, we get the two-sided bound
  \begin{align*}
    &(1-\gamma)(1 + \delta \mu + \mu)^{-\gamma} \\
    &\le \E{(1 + X(n, p))^{1-\gamma} - X(n, p)^{1-\gamma} \given \mc E} \\
    &\le (1-\gamma)(-\delta \mu + \mu)^{-\gamma}.
  \end{align*}
  We apply the Chernoff bound
  \begin{align*}
    \Pr[\overline{\mc E}]
    = \Pr[|X(n, p) - \mu| \ge (1+\delta) \mu] \le 2 e^{-\frac{\delta^2 \mu}{3}}.
  \end{align*}
  Using the fact that $(1 + x)^{1-\gamma} -
  x^{1-\gamma} \le 1$ for $x \ge 0$, this becomes
  \ifthenelse{\boolean{smallEqs}}{
  \begin{align*}
    &(1-\gamma)(1 + \delta \mu + \mu)^{-\gamma} \p{1 -
    2e^{-\frac{np\delta^2}{3}}} \\
    &\le \frac{1}{p} \dw(n, p) \\
    &\le (1-\gamma)(-\delta \mu + \mu)^{-\gamma} + 2e^{-\frac{np\delta^2}{3}}.
  \end{align*}
}{\begin{align*}
    (1-\gamma)(1 + \delta \mu + \mu)^{-\gamma} \p{1 - 2e^{-\frac{np\delta^2}{3}}}
    &\le \frac{1}{p} \dw(n, p) \\
    &\le (1-\gamma)(-\delta \mu + \mu)^{-\gamma} + 2e^{-\frac{np\delta^2}{3}}.
  \end{align*}
}
  Choose $\delta = \sqrt{\frac{3(1+\gamma) \log n}{n}}$. As long as $p$ is
  bounded from below by a constant, $\mu = \Theta(n)$, and
  \begin{equation}
    \label{eq:dw-asymptotic}
    \frac{1}{p} \dw(n, p)
    = (1-\gamma) (np)^{-\gamma} \p{1 \pm O\p{\sqrt{\frac{\log n}{n}}}}.
  \end{equation}
  Here, we have used the fact that for $x$ going to 0,
  \begin{align*}
    (1 \pm x)^\gamma = 1 \pm \gamma x \pm O(x^2) = 1 \pm O(x).
  \end{align*}
  By \Cref{thm:diversity-n}, $\pn{n} \succ \qn{n}$ for all $n$. By
  \Cref{lem:lower-bound-on-p}, $\pn{n} = \Omega(1)$. Therefore, there exists
  some constant $a$ such that for sufficiently large $n$, $\qn{n} \ge a$.
  Applying~\eqref{eq:dw-asymptotic}, for all $k$,
  \ifthenelse{\boolean{smallEqs}}{
    \begin{align*}
    \frac{1}{\qkn{n}} &\dw(n-1, \qkn{n}) \\
    &= (1-\gamma) ((n-1)\qkn{n})^{-\gamma} \p{1 \pm O\p{\sqrt{\frac{\log n}{n}}}}.
  \end{align*}
}{\begin{align*}
    \frac{1}{\qkn{n}} &\dw(n-1, \qkn{n}) \\
    &= (1-\gamma) ((n-1)\qkn{n})^{-\gamma} \p{1 \pm O\p{\sqrt{\frac{\log n}{n}}}}.
  \end{align*}
}

  With this, we can return to~\eqref{eq:dw-equal}.
  \ifthenelse{\boolean{smallEqs}}{
  \begin{align*}
    \frac{\bd_k}{\qkn{n}} \dw(n-1, \qkn{n})
    - \frac{\bd_{k'}}{\qkpn{n}} \dw(n-1, \qkpn{n})
    &= 0 \\
    \bd_k ((n-1)\qkn{n})^{-\gamma}
    \p{1 \pm O\p{\sqrt{\frac{\log n}{n}}}} \\
    - \bd_{k'} ((n-1)\qkpn{n})^{-\gamma}
    \p{1 \pm O\p{\sqrt{\frac{\log n}{n}}}}
    &= 0 \\
    \bd_k (\qkn{n})^{-\gamma}
    \p{1 \pm O\p{\sqrt{\frac{\log n}{n}}}} \\
    - \bd_{k'} (\qkpn{n})^{-\gamma}
    \p{1 \pm O\p{\sqrt{\frac{\log n}{n}}}}
    &= 0.
  \end{align*}
}{\begin{align*}
    \frac{\bd_k}{\qkn{n}} \dw(n-1, \qkn{n})
    - \frac{\bd_{k'}}{\qkpn{n}} \dw(n-1, \qkpn{n})
    &= 0 \\
    \bd_k ((n-1)\qkn{n})^{-\gamma}
    \p{1 \pm O\p{\sqrt{\frac{\log n}{n}}}}
    - \bd_{k'} ((n-1)\qkpn{n})^{-\gamma}
    \p{1 \pm O\p{\sqrt{\frac{\log n}{n}}}}
    &= 0 \\
    \bd_k (\qkn{n})^{-\gamma}
    \p{1 \pm O\p{\sqrt{\frac{\log n}{n}}}}
    - \bd_{k'} (\qkpn{n})^{-\gamma}
    \p{1 \pm O\p{\sqrt{\frac{\log n}{n}}}}
    &= 0.
  \end{align*}
}
  Because $\qkn{n} = \Omega(1)$, this implies
  \begin{align*}
    | \bd_k (\qkn{n})^{-\gamma}
    - \bd_{k'} (\qkpn{n})^{-\gamma}|
    \le O\p{\sqrt{\frac{\log n}{n}}}.
  \end{align*}
  This is identical to~\eqref{eq:dk-pk-pairwise-bound}, and by an identical
  argument,
  \begin{align*}
    |\qkn{n} - \bplim_k| \le O\p{\sqrt{\frac{\log n}{n}}}.
  \end{align*}

  Next, we handle the case where $\gamma \ge 1$ (and more generally, any $s
  \in \sdown$). By \Cref{lem:unique-opt-decreasing}, for any $s \in \sdown$,
  \begin{align*}
    \frac{1}{n} \frac{\partial}{\partial p} w(n, p)
    &= \frac{1}{n} \frac{\partial}{\partial p} \Pr[X(n, p) > 0] \\
    &= \frac{1}{n} \frac{\partial}{\partial p} 1 - (1-p)^n \\
    &= (1-p)^{n-1}.
  \end{align*}
  Then, for any $k, k'$,
  \begin{align*}
    \bd_k (1-\qkn{n})^{n-1}
    &= \bd_{k'} (1 - \qkpn{n})^{n-1} \\
    \frac{1 - \qkn{n}}{1 - \qkpn{n}}
    &= \p{\frac{\bd_{k'}}{\bd_k}}^{1/(n-1)}
  \end{align*}
  Therefore,
  \begin{align*}
    |\qkn{n} - \qkpn{n}|
    &= |1-\qkn{n} - (1-\qkpn{n})| \\
    &= (1-\qkpn{n})\left|\p{\frac{\bd_{k'}}{\bd_k}}^{1/(n-1)} - 1\right| \\
    &\le O\p{e^{\frac{\log(\bd_{k'} / \bd_k)}{n-1}} - 1} \\
    &= O\p{\frac{\log(\bd_{k'} / \bd_k)}{n-1} + \sum_{\ell=2}^\infty
  \frac{1}{\ell!} \p{\frac{\log(\bd_{k'} / \bd_k)}{n-1}}^\ell} \\
    &= O\p{\frac{1}{n}}.
  \end{align*}
  Because the average value of $\qkn{n}$ is $\frac{1}{K}$, this implies
  \begin{align*}
    \left|\qkn{n} - \frac{1}{K} \right| \le O\p{\frac{1}{n}},
  \end{align*}
  and $\qkn{\infty} = \frac{1}{K}$.
  
  Finally, we can use this to return to the case where $\gamma = \infty$ for
  $\pkn{\infty}$. By \Cref{lem:opt-inf-equivalence}, when $\gamma = \infty$,
  $\peq(n, \bd, s_\infty) = \popt(n, \bd, s_\infty)$. Thus, in this case,
  $\pkn{\infty} = 1/K \propto \bd^{1/\infty}$ as well.
\end{proof}

\infequiv*
\begin{proof}

  Assume without loss of generality that $\bd$ is decreasing.
  Let $\bp^* = \popt(n, \bd, s)$ and $\bp = \peq(n, \bd, s_\infty)$. By
  \Cref{lem:unique-opt-decreasing}, for $s \in \sdown$,
  \begin{align*}
    \bp_k^* > 0 \Longleftrightarrow \bd_k (1 - \bp_k^*)^{n-1} = c
  \end{align*}
  for $c = \frac{1}{n} \|\nabla W(\bp^*)\|_\infty$. In addition,
  \begin{align*}
    \bp_k^* = 0 \Longleftrightarrow \bd_k \le c.
  \end{align*}
  By \Cref{lem:unique-eq-decreasing},
  \begin{equation}
    \label{eq:inf-kkt}
    \bp_k > 0 \Longleftrightarrow \bd_k u_{s_\infty}(n-1, \bp_k) = c'
  \end{equation}
  for $c' = U(\bp; n, \bd, s_\infty)$. In addition,
  \begin{align*}
    \bp_k = 0 \Longleftrightarrow \bd_k \le c'.
  \end{align*}
  Observe that
  \begin{align*}
    u_{s_\infty}(n-1, \bp_k)
    &= \E{\frac{1}{(1 + X(n-1, \bp_k))^\infty}} \\
    &= \Pr[X(n-1, \bp_k) = 0] \\
    &= (1 - \bp_k)^{n-1}.
  \end{align*}
  Substituting into~\eqref{eq:inf-kkt},
  \begin{align*}
    \bp_k > 0 \Longleftrightarrow \bd_k (1 - \bp_k)^{n-1} = c'.
  \end{align*}
  Because $(1 - p)^{n-1}$ is strictly decreasing in $p$ and $\sum_{k \in [K]}
  \bp_k^* = \sum_{k \in [K]} \bp_k = 1$, by a water-filling argument, it must be
  the case that $c' = c$. Therefore, $\bp^* = \bp$.
\end{proof}

\infmostdiverse*
\begin{proof}

  Assume without loss of generality that $\bd$ is decreasing.
  Let $\bp = \peq(n, \bd, s)$ and let $\bq = \peq(n, \bd, s_\infty)$. Our goal
  will be to show that for $k' > k$, if $\bp_{k'} > 0$,
  \begin{equation}
    \label{eq:inf-most-diverse-setup}
    \bd_k u_{s_\infty}(n-1, \bp_k) \le \bd_{k'} u_{s_\infty}(n-1, \bp_{k'}).
  \end{equation}
  We will then apply \Cref{lem:p-majorize}. To do so, observe that under these
  conditions,
  \begin{equation}
    \label{eq:inf-denom}
    \bd_k u_{s}(n-1, \bp_k) = \bd_{k'} u_{s}(n-1, \bp_{k'})
  \end{equation}
  by \Cref{lem:unique-eq-decreasing}. Dividing both sides
  of~\eqref{eq:inf-most-diverse-setup} by the respective sides
  of~\eqref{eq:inf-denom}, it suffices to show that
  \ifthenelse{\boolean{smallEqs}}{
  \begin{align*}
    &\frac{u_{s_\infty}(n-1, \bp_k)}{u_{s}(n-1, \bp_k)}
    \le \frac{u_{s_\infty}(n-1, \bp_{k'})}{u_{s}(n-1, \bp_{k'})} \\
    \Longleftrightarrow
    &
    \frac{\Pr[X(n-1, \bp_k) = 0]}{u_{s}(n-1, \bp_k)}
    \le \frac{\Pr[X(n-1, \bp_{k'}) = 0]}{u_{s}(n-1, \bp_{k'})} \\
    \Longleftarrow
    &
    \frac{\partial}{\partial p} 
    \frac{\Pr[X(n-1, p) = 0]}{u_{s}(n-1, p)} \le 0 \\
    \Longleftrightarrow
    &
    \frac{\partial}{\partial p} 
    \log\p{\frac{\Pr[X(n-1, p) = 0]}{u_{s}(n-1, p)}} \le 0 \\
    \Longleftrightarrow
    &
    \frac{\frac{\partial}{\partial p} \Pr[X(n-1, p) = 0]}{\Pr[X(n-1, p) = 0]}
    \le \frac{\frac{\partial}{\partial p} u_{s}(n-1, p)}{u_{s}(n-1, p)} \\
    \Longleftrightarrow
    &
    \frac{-(n-1)(1-p)^{n-2}}{(1-p)^{n-1}} \le \frac{\frac{n-1}{p} \du_s(n-2,
    p)}{u_s(n-1, p)} \\
    \Longleftrightarrow
    &
    -\frac{1}{1-p} \le \frac{1}{p} \frac{u_s(n-1, p) - u_s(n-2, p)}{u_s(n-1, p)}
    \\
    \Longleftrightarrow
    &
    -\frac{p}{1-p} \le 1 - \frac{u_s(n-2, p)}{u_s(n-1, p)} \\
    \Longleftrightarrow
    &
    \frac{u_s(n-2, p)}{u_s(n-1, p)} \le \frac{1}{1-p} \\
    \Longleftrightarrow
    &
    (1-p) u_s(n-2, p) \le u_s(n-1, p).
  \end{align*}
}{\begin{align*}
    \frac{u_{s_\infty}(n-1, \bp_k)}{u_{s}(n-1, \bp_k)}
    \le \frac{u_{s_\infty}(n-1, \bp_{k'})}{u_{s}(n-1, \bp_{k'})}
    &\Longleftrightarrow
    \frac{\Pr[X(n-1, \bp_k) = 0]}{u_{s}(n-1, \bp_k)}
    \le \frac{\Pr[X(n-1, \bp_{k'}) = 0]}{u_{s}(n-1, \bp_{k'})} \\
    &\Longleftarrow
    \frac{\partial}{\partial p} 
    \frac{\Pr[X(n-1, p) = 0]}{u_{s}(n-1, p)} \le 0 \\
    &\Longleftrightarrow
    \frac{\partial}{\partial p} 
    \log\p{\frac{\Pr[X(n-1, p) = 0]}{u_{s}(n-1, p)}} \le 0 \\
    &\Longleftrightarrow
    \frac{\frac{\partial}{\partial p} \Pr[X(n-1, p) = 0]}{\Pr[X(n-1, p) = 0]}
    \le \frac{\frac{\partial}{\partial p} u_{s}(n-1, p)}{u_{s}(n-1, p)} \\
    &\Longleftrightarrow
    \frac{-(n-1)(1-p)^{n-2}}{(1-p)^{n-1}} \le \frac{\frac{n-1}{p} \du_s(n-2,
    p)}{u_s(n-1, p)} \\
    &\Longleftrightarrow
    -\frac{1}{1-p} \le \frac{1}{p} \frac{u_s(n-1, p) - u_s(n-2, p)}{u_s(n-1, p)}
    \\
    &\Longleftrightarrow
    -\frac{p}{1-p} \le 1 - \frac{u_s(n-2, p)}{u_s(n-1, p)} \\
    &\Longleftrightarrow
    \frac{u_s(n-2, p)}{u_s(n-1, p)} \le \frac{1}{1-p} \\
    &\Longleftrightarrow
    (1-p) u_s(n-2, p) \le u_s(n-1, p).
  \end{align*}
}
  Observe that
  \ifthenelse{\boolean{smallEqs}}{
  \begin{align*}
    u_s(n-1, p)
    &= \E{\frac{1}{s(1 + X(n-1, p))}} \\
    &= (1-p) \E{\frac{1}{s(1 + x(n-2, p))}} \\
    &+ p \E{\frac{1}{s(2 + X(n-2, p))}} \\
    &\ge (1-p) \E{\frac{1}{s(1 + x(n-2, p))}} \\
    &= (1-p) u_s(n-2, p)
  \end{align*}
}{\begin{align*}
    u_s(n-1, p)
    &= \E{\frac{1}{s(1 + X(n-1, p))}} \\
    &= (1-p) \E{\frac{1}{s(1 + x(n-2, p))}} + p \E{\frac{1}{s(2 + X(n-2, p))}}
    \\
    &\ge (1-p) \E{\frac{1}{s(1 + x(n-2, p))}} \\
    &= (1-p) u_s(n-2, p)
  \end{align*}
}
  as desired. Therefore, if $k' > k$ and $p_{k'} >
  0$,~\eqref{eq:inf-most-diverse-setup} holds. Applying \Cref{lem:p-majorize}
  with $Q_k = \bd_k u_{s_\infty}(n-1, \bp_k)$, because $\{Q_k\}$ is increasing
  when $\bp_k > 0$, $\bp \succ \bq$.
\end{proof}

\asymconverge*
\begin{proof}
  This follows directly from \Cref{lem:p-convergence}. Let $\bpn{n} = \peq(n,
  \bd, s)$.
  Because $[\bpn{n}, \bpn{n}, \dots, \bpn{n}] \in \eqsetn{n}$,
  \begin{align*}
    |\bpn{n}_k - \barpn{n}_k|
    \le 
    |\bpn{n}_k - \bplim_k| + |\bplim_k - \barpn{n}_k|
    \le O\p{\sqrt{\frac{\log n}{n}}}.
  \end{align*}
\end{proof}

\subsection{Supporting results}

\subsubsection{Supporting results on diversity}

\begin{lemma}
  \label{lem:p-majorize}
  Let $\bp = \peq(n, \bd, s)$. Let $\bp'$ be some other distribution. Consider
  the sequence $\{Q_k\}_{k \in [K]}$ where
  \begin{align*}
    Q_k \triangleq \bd_k u_s(n-1, \bp_k').
  \end{align*}
  If $\{Q_k\}$ is decreasing, then $\bp \succ \bp'$. If $\{Q_k\}$ is increasing
  for $k$ such that $\bp_k' > 0$, then $\bp' \succ \bp$.
\end{lemma}
\begin{proof}
  Assume without loss of generality that $\bd$ is decreasing.
  Because $u(n-1, \cdot)$ is strictly
  decreasing,
  \begin{equation}
    \label{eq:comparison-eq}
    \bp_k' > \bp_k \Longleftrightarrow \bd_k u(n-1, \bp_k') < \bd_k u(n-1,
    \bp_k).
  \end{equation}
  Let $c = U(\bp; n, \bd, s)$.
  \paragraph*{Case 1: $\{Q_k\}$ is increasing.}

  Consider any $k, \hat k$ such that $k < \hat k$. If $\bp_k > 0$, then by
  \Cref{lem:unique-eq-decreasing},
  \begin{equation}
    \label{eq:k-to-khat-decrease}
    \bd_k u(n-1, \bp_k) = c \ge \bd_{\hat k} u(n-1, \bp_{\hat k}).
  \end{equation}
  If $\bp_k = 0$, then $\bp_{\hat k} = 0$ as well, so
  \begin{align*}
    \bd_k u(n-1, \bp_k) = \bd_k \ge \bd_{\hat k} \ge \bd_{\hat k} u(n-1,
    \bp_{\hat k}).
  \end{align*}
  In either case,
  \begin{align*}
    \bd_k u(n-1, \bp_k) \ge \bd_{\hat k} u(n-1, \bp_{\hat k}).
  \end{align*}
  Therefore, if $\bp_{\hat k}' > \bp_{\hat k}$, then
  \begin{align*}
    Q_k
    &\le Q_{\hat k} \tag{$\{Q_k\}$ increasing for $\bp_{\hat k}' > 0$} \\
    &= \bd_{\hat k} u(n-1, \bp_{\hat k}') \\
    &< \bd_{\hat k} u(n-1, \bp_{\hat k}) \tag{by~\eqref{eq:comparison-eq}} \\
    &\le \bd_k u(n-1, \bp_k). \tag{by~\eqref{eq:k-to-khat-decrease}}
  \end{align*}
  Again using~\eqref{eq:comparison-eq}, $Q_k < \bd_k u(n-1, \bp_k)$ implies
  $\bp_k' > \bp_k$. As a result, we have
  \begin{align*}
    \bp_{\hat k}' > \bp_{\hat k} \Longrightarrow \bp_k' > \bp_k
  \end{align*}
  for any $k < \hat k$. Let $k^*$ be the largest $k$ such that $\bp_k' > \bp_k$.
  We have shown that for $k \le k^*$, $\bp_k' > \bp_k$. On the other hand, for
  $k > k^*$, $\bp_k' = 0$, so $\bp_k' \le \bp_k$. Therefore, $\bp_k' > \bp_k
  \Longleftrightarrow k \le k^*$. This implies $\bp' \succ \bp$ by
  \Cref{lem:inc-majorization}.

  \paragraph*{Case 2: $\{Q_k\}$ is decreasing.}

  Similarly, we will show that for any $k < \hat k$,
  \begin{align*}
    \bp_{\hat k} > \bp_{\hat k}' \Longrightarrow \bp_k > \bp_k'.
  \end{align*}
  The premise implies that $\bp_{\hat k} > 0$ and therefore $\bp_k > 0$.
  By \Cref{lem:unique-eq-decreasing}, this means
  \begin{align*}
    \bd_k u(n-1, \bp_k) = \bd_{\hat k} u(n-1, \bp_{\hat k}) = c
  \end{align*}
  for some $c$. This means that
  \begin{align*}
    \bp_{\hat k} > \bp_{\hat k}'
    &\Longleftrightarrow \bd_{\hat k} u(n-1, \bp_{\hat k}) < \bd_{\hat k}
    u(n-1, \bp_{\hat k}')
    \tag{by~\eqref{eq:comparison-eq}} \\
    &\Longleftrightarrow c < Q_{\hat k} \\
    &\Longrightarrow c < Q_k \tag{$\{Q_k\}$ is decreasing} \\
    &\Longleftrightarrow \bd_k u(n-1, \bp_k) < \bd_k
    u(n-1, \bp_k') \\
    &\Longleftrightarrow \bp_k > \bp_k'.
    \tag{by~\eqref{eq:comparison-eq}}
  \end{align*}
  Again, this means that there is some $k^*$ such that for all $k \le k^*$, $\bp_k
  > \bp_k'$. \Cref{lem:inc-majorization} yields $\bp \succ \bp'$.
\end{proof}

\begin{lemma}
  \label{lem:inc-majorization}
  Suppose distributions $\bp$ and $\bq$ satisfy
  \begin{align*}
    \bp_k < \bq_k \Longleftrightarrow k \le k^*
  \end{align*}
  for some $k^*$. Then, $\bq \succ \bp$.
\end{lemma}
\begin{proof}
  By definition~\citep[e.g.,~][]{olkin2014inequalities},
  $\bq \succ \bp$ if and only if for all $k$,
  \begin{equation}
    \label{eq:majorize-def}
    \sum_{\ell=0}^k \bq_\ell \ge 
    \sum_{\ell=0}^k \bp_\ell.
  \end{equation}
  By assumption, there is some $k^*$ such that $\bq_k > \bp_k$ if and only if
  $k \le k^*$. For $k \le k^*$,
  \begin{align*}
    \sum_{\ell=0}^k \bq_\ell \ge
    \sum_{\ell=0}^k \bp_\ell
  \end{align*}
  because $\bq_\ell > \bp_\ell$ for each $\ell$. For $k > k^*$,
  \begin{align*}
    \sum_{\ell=0}^k \bq_\ell
    &= 1 - \sum_{\ell=k+1}^K \bq_\ell \\
    &\ge 1 - \sum_{\ell=k+1}^K \bp_\ell \tag{$\bq_\ell \le \bp_\ell$ for $k >
    k^*$} \\
    &= \sum_{\ell=0}^k \bp.
  \end{align*}
  Therefore, $\bq \succ \bp$.
\end{proof}

\subsubsection{Supporting results on asymptotics}
\label{sec:support-asymp}

Here, we provide additional results to support our asymptotic convergence
results for all equilibria, including asymmetric equilibria. As before, let
$\eqsetn{n}$ be the set of equilibria for $\game(n, \bd, s)$. For any $\bP \in
\Delta(\pi)^n$, define
\begin{align*}
  X_k(\bP)
  &\triangleq \sum_{i \in [n]} X(1, \bP(i)_k) \\
  U_i(\bP)
  &\triangleq \sum_{k \in [K]} \bd_k \bP(i)_k \E{\frac{1}{s(1 + X_k(\bP(-i)))}}
  \\
  W(\bP)
  &\triangleq
  \begin{cases}
    \sum_{k \in [K]} \bd_k \E{\frac{X_k(\bP)}{s(X_k(\bP))}} & {s \in \sup} \\
    \sum_{k \in [K]} \bd_k \Pr[X_k(\bP) > 0] & {s \in \sdown}
  \end{cases} \\
  V_k(\bP)
  &\triangleq \bd_k \E{(1 + X_k(\bP))^{-\gamma}}
\end{align*}
We index the $i$th player's strategy as $\bP(i)
\in \Delta(\pi)$, and we use $\bP(i)_k$ to denote the probability that player
$i$ chooses type $k$. We use $\bP(-i)$ to denote the strategy matrix $\bP$ with
player $i$'s strategy removed.
Our results will hold for sequences
$\{\bPn{n}\}_{n=2}^\infty$ where each $\bPn{n} \in \eqsetn{n}$. We will denote
the average strategy by $\barpn{n} \triangleq \frac{1}{n} \sum_{i=1}^n
\bPn{n}(i)$.

\begin{lemma}
  \label{lem:p-convergence}
  Assume without loss of generality that $\bd \in \R_{\ge 0}^K$ is weakly
  decreasing.
  Choose $\{\bPn{n}\}_{n=2}^\infty$ such that $\bPn{n} \in
  \eqsetn{n}$.
  Let $\bplim_k \triangleq \bd_k^{1/\gamma} / \sum_{k' \in
  [K]} \bd_{k'}^{1/\gamma}$. Then,
  \begin{align*}
    |\barpn{n}_k - \bplim_k| \le O\p{\sqrt{\frac{\log n}{n}}}.
  \end{align*}
\end{lemma}
\begin{proof}
  By \Cref{lem:V-overall-bound}, for all pairs $k, k'$,
  \begin{align*}
    |V_k(\bPn{n}) - V_{k'}(\bPn{n})| \le O(n^{-1-\gamma}).
  \end{align*}
  By \Cref{lem:pb-expectation,lem:lower-bound-on-p}, for all $k$,
  \begin{align*}
    V_k(\bPn{n})
    = \bd_k (1 + n\barpn{n}_k)^{-\gamma}\p{1 + O\p{\sqrt{\frac{\log n}{n}}}}.
  \end{align*}
  Combining,
  \ifthenelse{\boolean{smallEqs}}{
  \begin{align*}
    \left|
    \bd_k (1 + n\barpn{n}_k)^{-\gamma}\p{1 + O\p{\sqrt{\frac{\log n}{n}}}}
    \right. \\
    \left.- \bd_{k'} (1 + n\barpn{n}_{k'})^{-\gamma}\p{1 + O\p{\sqrt{\frac{\log
    n}{n}}}}
    \right|
    &\le O(n^{-1-\gamma}) \\
    \left|
    \bd_k (1 + n\barpn{n}_k)^{-\gamma}
    - \bd_{k'} (1 + n\barpn{n}_{k'})^{-\gamma}
    \right| \\
    +
    (\bd_k (1 + n \barpn{n}_k)^{-\gamma} \\
    + \bd_{k'} (1 + n
    \barpn{n}_{k'})^{-\gamma}) O\p{\sqrt{\frac{\log n}{n}}}
    &\le O(n^{-1-\gamma})
  \end{align*}
}{\begin{align*}
    \left|
    \bd_k (1 + n\barpn{n}_k)^{-\gamma}\p{1 + O\p{\sqrt{\frac{\log n}{n}}}}
    - \bd_{k'} (1 + n\barpn{n}_{k'})^{-\gamma}\p{1 + O\p{\sqrt{\frac{\log
    n}{n}}}}
    \right|
    &\le O(n^{-1-\gamma}) \\
    \left|
    \bd_k (1 + n\barpn{n}_k)^{-\gamma}
    - \bd_{k'} (1 + n\barpn{n}_{k'})^{-\gamma}
    \right| \\
    -
    (\bd_k (1 + n \barpn{n}_k)^{-\gamma} + \bd_{k'} (1 + n
    \barpn{n}_{k'})^{-\gamma}) O\p{\sqrt{\frac{\log n}{n}}}
    &\le O(n^{-1-\gamma})
  \end{align*}
}
  Because $\barpn{n}_k$ and $\barpn{n}_{k'}$ are $\Omega(1)$, $(1 +
  n\barpn{n}_k)^{-\gamma} = O(n^{-\gamma})$. Therefore,
  \begin{align*}
    |
    \bd_k (1 + n\barpn{n}_k)^{-\gamma}
    - \bd_{k'} (1 + n\barpn{n}_{k'})^{-\gamma}
    |
    &\le O\p{\frac{\sqrt{\log n}}{n^{1/2+\gamma}}}.
  \end{align*}
  Next, we use the fact that $(1 + n \barpn{n}_k)^{-\gamma} \approx (n
  \barpn{n}_k)^{-\gamma}$ via Taylor's theorem:
  \ifthenelse{\boolean{smallEqs}}{
  \begin{align*}
    &\bd_k (1 + n\barpn{n}_k)^{-\gamma}
    - \bd_{k'} (1 + n\barpn{n}_{k'})^{-\gamma} \\
    =& \bd_k (n\barpn{n}_k)^{-\gamma} - \gamma (\xi_k +
    n\barpn{n}_k)^{-1-\gamma} \\
     &- \bd_{k'} (n\barpn{n}_{k'})^{-\gamma} + \gamma (\xi_{k'} +
    n\barpn{n}_{k'})^{-1-\gamma}
  \end{align*}
}{\begin{align*}
    &\bd_k (1 + n\barpn{n}_k)^{-\gamma}
    - \bd_{k'} (1 + n\barpn{n}_{k'})^{-\gamma} \\
    =&
    \bd_k (n\barpn{n}_k)^{-\gamma} - \gamma (\xi_k + n\barpn{n}_k)^{-1-\gamma}
    - \bd_{k'} (n\barpn{n}_{k'})^{-\gamma} + \gamma (\xi_{k'} +
    n\barpn{n}_{k'})^{-1-\gamma}
  \end{align*}
}
  where $\xi_k, \xi_{k'} \in [0, 1]$. Using the fact that $x^{-1-\gamma}$ is
  decreasing,
  \begin{align*}
    |
    \bd_k (1 + n\barpn{n}_k)^{-\gamma}
    - \bd_{k'} (1 + n\barpn{n}_{k'})^{-\gamma}
    |
    &\le O\p{\frac{\sqrt{\log n}}{n^{1/2+\gamma}}} \\
    |
    \bd_k (n\barpn{n}_k)^{-\gamma}
    - \bd_{k'} (n\barpn{n}_{k'})^{-\gamma}
    |
    - O(n^{-1-\gamma})
    &\le O\p{\frac{\sqrt{\log n}}{n^{1/2+\gamma}}} \\
    |
    \bd_k (n\barpn{n}_k)^{-\gamma}
    - \bd_{k'} (n\barpn{n}_{k'})^{-\gamma}
    |
    &\le O\p{\frac{\sqrt{\log n}}{n^{1/2+\gamma}}} \\
    |
    \bd_k (\barpn{n}_k)^{-\gamma}
    - \bd_{k'} (\barpn{n}_{k'})^{-\gamma}
    |
    &\le O\p{\sqrt{\frac{\log n}{n}}}.
    \numberthis \label{eq:dk-pk-pairwise-bound}
  \end{align*}
  Define
  \begin{align*}
    \ckn
    &\triangleq \bd_k (\barpn{n}_k)^{-\gamma} \\
    \Cn
    &\triangleq \frac{1}{K} \sum_{i=1}^K \ckn \\
    D_\gamma
    &= \sum_{k=1}^K \bd_k^{1/\gamma}
  \end{align*}
  By~\eqref{eq:dk-pk-pairwise-bound},
  \begin{align*}
    |\ckn - \ckpn|
    &\le O\p{\sqrt{\frac{\log n}{n}}}.
  \end{align*}
  Therefore,
  \begin{align*}
    \ckn
    &= \Cn + \dkn,
  \end{align*}
  where $\dkn \le O\p{\sqrt{\frac{\log n}{n}}}$.
  Next,
  \begin{align*}
    1
    &= \sum_{k=1}^K \barpn{n}_k \\
    &= \sum_{k=1}^K \p{\frac{\bd_k}{\ckn}}^{1/\gamma} \\
    &= \sum_{k=1}^K \p{\frac{\bd_k}{\Cn + \dkn}}^{1/\gamma} \\
    &= \sum_{k=1}^K \p{\frac{\bd_k}{\Cn}}^{1/\gamma} \p{1 +
    \frac{\dkn}{\Cn}}^{-1/\gamma} \\
    &= {\Cn}^{-1/\gamma} \sum_{k=1}^K \bd_k^{1/\gamma}
    \p{1 + \frac{\dkn}{\Cn}}^{-1/\gamma}
  \end{align*}
  A Taylor expansion yields
  \begin{align*}
    \p{1 + \frac{\dkn}{\Cn}}^{-1/\gamma}
    &= 1 - \frac{1}{\gamma} \frac{\dkn}{\Cn} \pm O\p{\p{\frac{\dkn}{\Cn} }^2} \\
    &= 1 - \frac{1}{\gamma} \frac{\dkn}{\Cn} \pm O\p{\frac{\log n}{n}} \\
    &= 1 \pm O\p{\sqrt{\frac{\log n}{n}}}
  \end{align*}
  Therefore,
  \begin{align*}
    1
    &= {\Cn}^{-1/\gamma} \sum_{k=1}^K \bd_k^{1/\gamma}
    \p{1 \pm O\p{\sqrt{\frac{\log n}{n}}}} \\
    {\Cn}^{1/\gamma}
    &= D_\gamma \pm \sum_{k=1}^K \bd_k^{1/\gamma}O\p{\sqrt{\frac{\log n}{n}}} \\
    |{\Cn}^{1/\gamma} - D_\gamma|
    &= O\p{\sqrt{\frac{\log n}{n}}} \\
    |{\Cn} - D_\gamma^{\gamma}|
    &= O\p{\sqrt{\frac{\log n}{n}}} \tag{By \Cref{lem:gamma-power-convergence}}
  \end{align*}
  For any $k$,
  \begin{align*}
    |\ckn - D_\gamma^\gamma|
    &\le |\ckn - \Cn| + |\Cn - D_\gamma^\gamma| \\
    &\le O\p{\sqrt{\frac{\log n}{n}}}.
  \end{align*}
  With this, we can bound
  \begin{align*}
    |\barpn{n}_k - \bplim_k|
    &= \left| \p{\frac{\bd_k}{\ckn}}^{1/\gamma} -
    \p{\frac{\bd_k}{D_\gamma^\gamma}}^{1/\gamma} \right| \\
    &= \bd_k^{1/\gamma} \left| {\ckn}^{-1/\gamma} - D_\gamma^{-1} \right|.
  \end{align*}
  Again using \Cref{lem:gamma-power-convergence}, because $|\ckn -
  D_\gamma^\gamma| \le O\p{\sqrt{\frac{\log n}{n}}}$,
  \begin{align*}
    \left| {\ckn}^{-1/\gamma} - D_\gamma^{-1} \right|
    &\le O\p{\sqrt{\frac{\log n}{n}}}.
  \end{align*}
  Therefore,
  \begin{align*}
    |\barpn{n}_k - \bplim_k|
    &\le O\p{\sqrt{\frac{\log n}{n}}}.
  \end{align*}
\end{proof}

\begin{lemma}
  \label{lem:lower-bound-on-p}
  Choose $\{\bPn{n}\}_{n=2}^\infty$ such that $\bPn{n} \in
  \eqsetn{n}$. Let $\barpn{n} = \frac{1}{n} \sum_{i=1}^n \bPn{n}(i)$.
  Then, $\barpn{n}_k = \Omega(1)$ for all $k \in [K]$.
\end{lemma}
\begin{proof}
  We proceed by induction on $k$. The constraint $\bP(i) \in \Delta(\pi)$
  implies $\bP(i)_1 \ge \frac{1}{K}$. Therefore, for any $n$, $\barpn{n}_1 \ge
  \frac{1}{K}$.

  Assume that $\barpn{n}_k = \Omega(1)$.
  This means that there exist $a_k, n_k$
  such that $\barpn{n}_k \ge a_k$ for all $n \ge n_k$. Assume towards
  contradiction that $\barpn{n}_{k+1} = o(1)$. Then, there must be an infinite
  sequence $\{n_j^*\}_{j=1}^\infty$ such that $\barpn{n_j^*} < a_k$. In the
  remainder of this analysis, we restrict ourselves to $n^*$'s coming from this
  sequence.

  Because $\barpn{n^*}_k > \barpn{n^*}_{k+1}$, there exists some player $i$ such
  that $\bPn{n^*}_k(i) > \bPn{n^*}_{k+1}(i)$. Because $\bPn{n^*} \in \eqsetn{n^*}$,
  \begin{align*}
    \bPn{n^*}_k(i) > \bPn{n^*}_{k+1}(i) \Longrightarrow
    V_k(\bPn{n^*}(-i)) \ge V_{k+1}(\bPn{n^*}(-i)),
  \end{align*}
  because otherwise player $i$ could deviate by shifting mass from $k$ to $k+1$.
  Thus, for player $i$,
  \begin{equation}
    \label{eq:player-i-deviate}
    V_k(\bPn{n^*}(-i)) \ge V_{k+1}(\bPn{n^*}(-i)).
  \end{equation}
  We will show that this entails a contradiction.

  Let $V_k^{(i)} = V_k(\bPn{n^*}(-i))$, and define $V_{k+1}^{(i)}$ analogously.
  By definition,
  \begin{align*}
    V_k^{(i)}
    &= \bd_k \E{\frac{1}{(1+X_k(\bP^{(n^*)}(-i)))^{\gamma}}}.
  \end{align*}
  Because $\bPn{n^*} \ge a_k$ by assumption $X_k(\bPn{n^*}(-i))$ stochastically
  dominates $X(n^*-1, a_k)$, so
  \begin{align*}
    V_k^{(i)}
    &\ge \bd_k \E{\frac{1}{(1 + X(n^*-1, a_k))^\gamma}}.
  \end{align*}
  Let $\mu_k^{(n^*)} = (n^*-1)a_k$. Using \Cref{lem:pb-expectation},
  \ifthenelse{\boolean{smallEqs}}{
  \begin{align*}
    \bd_k &\E{\frac{1}{(1 + X(n^*-1, a_k))^{-\gamma}}} \\
    &\le \frac{\bd_k}{(1 + \mu_k^{(n^*)})^\gamma}\p{1 + O\p{\sqrt{\frac{\log
    \mu_k^{(n^*)}}{\mu_k^{(n^*)}}}}} \\
    &= \frac{\bd_k}{(1 + (n^*-1)a_k)^\gamma}\p{1 + O\p{\sqrt{\frac{\log
    n^*}{n^*}}}} \\
  \end{align*}
}{\begin{align*}
    \bd_k \E{\frac{1}{(1 + X(n^*-1, a_k))^{-\gamma}}}
    &\le \frac{\bd_k}{(1 + \mu_k^{(n^*)})^\gamma}\p{1 + O\p{\sqrt{\frac{\log
    \mu_k^{(n^*)}}{\mu_k^{(n^*)}}}}} \\
    &= \frac{\bd_k}{(1 + (n^*-1)a_k)^\gamma}\p{1 + O\p{\sqrt{\frac{\log
    n^*}{n^*}}}} \\
  \end{align*}
}
  Therefore,
  \begin{equation}
    \label{eq:V_k-upper-bound}
    V_k^{(i)}
    \le \bd_k (1 + (n^*-1))^{-\gamma} + O\p{\sqrt{\frac{\log
    n^*}{n^*}}}.
  \end{equation}
  Let $\mu_{k+1}^{(n^*)} = (n^*-1)\barpn{n^*}_{k+1}$. Also
  using~\Cref{lem:pb-expectation},
  \begin{align*}
    V_{k+1}^{(i)}
    &= \bd_{k+1} \E{\frac{1}{(1+X_{k+1}(\bP^{(n^*)}(-i)))^{\gamma}}} \\
    &\ge \bd_{k+1} (1 + \mu_{k+1}^{(n^*)})^{-\gamma}
  \end{align*}
Combining with~\eqref{eq:player-i-deviate} and~\eqref{eq:V_k-upper-bound},
\ifthenelse{\boolean{smallEqs}}{
  \begin{align*}
    &\bd_k (1 + (n^*-1)a_k)^{-\gamma}\p{1 + O\p{\sqrt{\frac{\log n^*}{n^*}}}} \\
    &\ge \bd_{k+1} (1 + \mu_{k+1}^{(n^*)})^{-\gamma} \\
    &\frac{\bd_k}{\bd_{k+1}} \p{1 + O\p{\sqrt{\frac{\log
    n^*}{n^*}}}}
    \ge \p{\frac{1 + (n^*-1)a_k}{1 + \mu_{k+1}^{(n^*)}}}^\gamma
  \end{align*}
}{\begin{align*}
    \bd_k (1 + (n^*-1)a_k)^{-\gamma}\p{1 + O\p{\sqrt{\frac{\log
    n^*}{n^*}}}}
    &\ge \bd_{k+1} (1 + \mu_{k+1}^{(n^*)})^{-\gamma} \\
    \frac{\bd_k}{\bd_{k+1}} \p{1 + O\p{\sqrt{\frac{\log
    n^*}{n^*}}}}
    &\ge \p{\frac{1 + (n^*-1)a_k}{1 + \mu_{k+1}^{(n^*)}}}^\gamma
  \end{align*}
}
  Observe that as $n^* \to \infty$, the left hand side is upper-bounded by a
  constant. Under the assumption that $\mu_{k+1}^{(n^*)} = o(n^*)$, the right
  hand side goes to infinity, which is a contradiction. Therefore, the inductive
  step holds, and $\barpn{n}_{k+1} = \Omega(1)$.
\end{proof}

\begin{lemma}
  \label{lem:V-k-asymptotic}
  Choose $\{\bPn{n}\}_{n=2}^\infty$ such that $\bPn{n} \in
  \eqsetn{n}$. Then,
  \ifthenelse{\boolean{smallEqs}}{
  \begin{align*}
    \bd_k
    &(1 + \mu_k(\bPn{n}))^{-\gamma} \\
    &\le V_k(\bPn{n}) \\
    &\le \bd_k (1 + \mu_k(\bPn{n}))^{-\gamma} \p{1 + O\p{\sqrt{\frac{\log
    n}{n}}}}.
  \end{align*}
}{\begin{align*}
    \bd_k (1 + \mu_k(\bPn{n}))^{-\gamma}
    \le V_k(\bPn{n})
    \le \bd_k (1 + \mu_k(\bPn{n}))^{-\gamma} \p{1 + O\p{\sqrt{\frac{\log
    n}{n}}}}.
  \end{align*}
}
  Moreover, for any player $i$,
  \ifthenelse{\boolean{smallEqs}}{
  \begin{align*}
    \bd_k
    &(1 + \mu_k(\bPn{n}(-i)))^{-\gamma} \\
    &\le V_k(\bPn{n}(-i)) \\
    &\le \bd_k (1 + \mu_k(\bPn{n}(-i)))^{-\gamma} \p{1 + O\p{\sqrt{\frac{\log
    n}{n}}}}.
  \end{align*}
}{\begin{align*}
    \bd_k (1 + \mu_k(\bPn{n}(-i)))^{-\gamma}
    \le V_k(\bPn{n}(-i))
    \le \bd_k (1 + \mu_k(\bPn{n}(-i)))^{-\gamma} \p{1 + O\p{\sqrt{\frac{\log
    n}{n}}}}.
  \end{align*}
}
\end{lemma}
\begin{proof}
  This follows directly from \Cref{lem:pb-expectation}, using the result from
  \Cref{lem:lower-bound-on-p} that $\mu_k(\bPn{n}) = \Omega(n)$.
\end{proof}

\begin{lemma}
  \label{lem:V-increase-bound}
  Assume without loss of generality that $\bd_k$ are decreasing.
  Choose $\{\bPn{n}\}_{n=2}^\infty$ such that $\bPn{n} \in
  \eqsetn{n}$. Then, for all $k < k'$,
  \begin{align*}
    V_{k'}(\bPn{n})\le V_k(\bPn{n}) + O(n^{-1-\gamma}).
  \end{align*}
\end{lemma}
\begin{proof}
  First, observe that
  \ifthenelse{\boolean{smallEqs}}{
  \begin{align*}
    &V_k(\bPn{n}(-i)) - V_k(\bPn{n}) \\
    &= \bd_k \E{(1 + X_k(\bPn{n}(-i)))^{-\gamma} - (1 + X_k(\bPn{n}))^{-\gamma}}
    \\
    &= \bd_k \bPn{n}(i)_k \cdot \\
    &\E{(1 + X_k(\bPn{n}(-i)))^{-\gamma} - (2 +
    X_k(\bPn{n}(-i)))^{-\gamma}}
    \numberthis \label{eq:V-i-diff}
  \end{align*}
}{\begin{align*}
    V_k(\bPn{n}(-i)) - V_k(\bPn{n})
    &= \bd_k \E{(1 + X_k(\bPn{n}(-i)))^{-\gamma} - (1 + X_k(\bPn{n}))^{-\gamma}}
    \\
    &= \bd_k \bPn{n}(i)_k \E{(1 + X_k(\bPn{n}(-i)))^{-\gamma} - (2 +
    X_k(\bPn{n}(-i)))^{-\gamma}}
    \numberthis \label{eq:V-i-diff}
  \end{align*}
}
  Let $b(x) = \E{(1 + x)^{-\gamma} - (2 + x)^{-\gamma}}$. Note that $b(x) \ge 0$
  and $b'(x) \le 0$.
  Let $\delta_k(i; n) = \E{(1 + X_k(\bPn{n}(-i)))^{-\gamma} - (2 +
  X_k(\bPn{n}(-i)))^{-\gamma}} = \E{b(X_k(\bPn{n}(-i)))}$. 
  Because $X_k(\bPn{n}(-i))$ stochastically dominates
  $X_{k+1}(\bPn{n}(-i))$,
  \begin{equation}
    \label{eq:delta-bound}
    \delta_k(i; n) \le \delta_{k+1}(i; n).
  \end{equation}
  We will show that for any $k$,
  \begin{equation}
    \label{eq:V-diff-bound}
    V_{k+1}(\bPn{n}) - V_k(\bPn{n}) \le \frac{\bd_k}{k} \delta_k(i; n).
  \end{equation}
  First, we handle the case where $\barpn{n}_k = \barpn{n}_{k+1}$. In this case,
  we must have $\bPn{n}_k = \bPn{n}_{k+1}$ because of the constraint $\bPn{n}(i)
  \in \Delta(\pi)$ for all $i$. As a result, $X_k(\bPn{n})$ and
  $X_{k+1}(\bPn{n})$ are identically distributed, so
  \ifthenelse{\boolean{smallEqs}}{
  \begin{align*}
    V_{k+1}
    &(\bPn{n}) - V_k(\bPn{n}) \\
    &= \bd_{k+1} \E{(1 + X_k(\bPn{n}))^{-\gamma}} - \bd_k \E{(1 +
    X_k(\bPn{n}))^{-\gamma}} \\
    &\le 0.
  \end{align*}
}{\begin{align*}
    V_{k+1}(\bPn{n}) - V_k(\bPn{n})
    &= \bd_{k+1} \E{(1 + X_k(\bPn{n}))^{-\gamma}} - \bd_k \E{(1 +
    X_k(\bPn{n}))^{-\gamma}} \\
    &\le 0.
  \end{align*}
}

  Next, we handle the case where $\barpn{n}_k > \barpn{n}_{k+1}$.
  Assume towards contradiction that
  \begin{align*}
    V_{k+1}(\bPn{n}) - V_k(\bPn{n})
    &> \frac{\bd_k}{k} \delta_k(i; n).
  \end{align*}
  This implies that for all $i$,
  \ifthenelse{\boolean{smallEqs}}{
  \begin{align*}
    V_{k+1}(\bPn{n}(-i)) - \bd_{k+1} \bPn{n}(i)_{k+1} \delta_{k+1}(i; n) \\
    - (V_k(\bPn{n}(-i)) - \bd_k \bPn{n}(i)_k \delta_k(i; n)) \\
    > \frac{\bd_k}{k} \delta_k(i; n).
  \end{align*}
}{\begin{align*}
    V_{k+1}(\bPn{n}(-i)) - \bd_{k+1} \bPn{n}(i)_{k+1} \delta_{k+1}(i; n)
    - (V_k(\bPn{n}(-i)) - \bd_k \bPn{n}(i)_k \delta_k(i; n))
    &> \frac{\bd_k}{k} \delta_k(i; n).
  \end{align*}
}
  Rearranging,
  \ifthenelse{\boolean{smallEqs}}{
  \begin{align*}
    &V_{k+1}(\bPn{n}(-i)) \\
    &> V_k(\bPn{n}(-i)) + \bd_{k+1} \bPn{n}(i)_{k+1} \delta_{k+1}(i; n) \\
    &- \bd_k \bPn{n}(i)_k \delta_k(i; n) + \frac{\bd_k}{k} \delta_k(i; n) \\
    &\ge V_k(\bPn{n}(-i)) + \bd_{k+1} \bPn{n}(i)_{k+1} \delta_k(i; n) \\
    &- \bd_k \bPn{n}(i)_k \delta_k(i; n) + \frac{\bd_k}{k} \delta_k(i; n) \\
    &= V_k(\bPn{n}(-i)) + \delta_k(i; n) \cdot \\
    &\p{\bd_{k+1} \bPn{n}(i)_{k+1} - \bd_k \bPn{n}(i)_k + \frac{\bd_k}{k}} \\
    &\ge V_k(\bPn{n}(-i)) + \delta_k(i; n) \p{\frac{\bd_k}{k} - \bd_k
    \bPn{n}(i)_k} \\
    &\ge V_k(\bPn{n}(-i)) + \delta_k(i; n) \bd_k\p{\frac{1}{k} - \bPn{n}(i)_k} \\
    &\ge V_k(\bPn{n}(-i)) \tag{$\bPn{n}(i)_k \le \frac{1}{k}$}
  \end{align*}
}{\begin{align*}
    V_{k+1}(\bPn{n}(-i))
    &> V_k(\bPn{n}(-i)) + \bd_{k+1} \bPn{n}(i)_{k+1} \delta_{k+1}(i; n) - \bd_k
    \bPn{n}(i)_k \delta_k(i; n) + \frac{\bd_k}{k} \delta_k(i; n) \\
    &\ge V_k(\bPn{n}(-i)) + \bd_{k+1} \bPn{n}(i)_{k+1} \delta_k(i; n) - \bd_k
    \bPn{n}(i)_k \delta_k(i; n) + \frac{\bd_k}{k} \delta_k(i; n) \\
    &= V_k(\bPn{n}(-i)) + \delta_k(i; n) \p{\bd_{k+1} \bPn{n}(i)_{k+1} - \bd_k
    \bPn{n}(i)_k + \frac{\bd_k}{k}} \\
    &\ge V_k(\bPn{n}(-i)) + \delta_k(i; n) \p{\frac{\bd_k}{k} - \bd_k
    \bPn{n}(i)_k} \\
    &\ge V_k(\bPn{n}(-i)) + \delta_k(i; n) \bd_k\p{\frac{1}{k} - \bPn{n}(i)_k} \\
    &\ge V_k(\bPn{n}(-i)) \tag{$\bPn{n}(i)_k \le \frac{1}{k}$}
  \end{align*}
}
  Under the assumption that $\barpn{n}_k > \barpn{n}_{k+1}$, there must be some
  $i$ such that $\bPn{n}(i)_k > \bPn{n}(i)_{k+1}$. This player $i$ could
  transfer mass from $\bPn{n}(i)_{k+1}$ to $\bPn{n}(i)_k$ and strictly increase
  their utility without violating the $\bPn{n}(i) \in \Delta(\pi)$ constraint.
  This contradicts the fact that $\bPn{n}$ is an equilibrium and
  proves~\eqref{eq:V-diff-bound}.

  We can use this to derive
  \begin{align*}
    V_{k'}(\bpn{n}) - V_k(\bpn{n})
    &= \sum_{\ell=k}^{k'-1} V_{\ell+1} - V_\ell \\
    &\le \sum_{\ell=k}^{k'-1} \frac{\bd_k}{k} \delta_k(i; n) \\
    &\le \bd_1 \sum_{\ell=1}^{k-1} \frac{\delta_k(i; n)}{k}
    \numberthis \label{eq:V-bound-sum-delta}
  \end{align*}
  We bound $\delta_k(i; n)$ as follows:
  \ifthenelse{\boolean{smallEqs}}{
  \begin{align*}
    &\delta_k(i; n) \\
    &= \E{(1 + X_k(\bPn{n}(-i)))^{-\gamma} - (2 + X_k(\bPn{n}(-i)))^{-\gamma}}
    \\
    &= \E{\gamma(1 + \xi(X_k(\bPn{n}(-i))) + X_k(\bPn{n}(-i)))^{-1-\gamma}}
  \end{align*}
}{\begin{align*}
    \delta_k(i; n)
    &= \E{(1 + X_k(\bPn{n}(-i)))^{-\gamma} - (2 + X_k(\bPn{n}(-i)))^{-\gamma}}
    \\
    &= \E{\gamma(1 + \xi(X_k(\bPn{n}(-i))) + X_k(\bPn{n}(-i)))^{-1-\gamma}}
  \end{align*}
}
  for some $\xi(X_k(\bPn{n}(-i))) \in [0, 1]$ by Taylor's theorem. Because the
  function in our expectation is monotone decreasing,
  \begin{align*}
    \delta_k(i; n)
    &\le \E{\gamma(1 + X_k(\bPn{n}(-i)))^{-1-\gamma}} \\
    &\le \gamma (1 + \mu_k(\bPn{n}(-i))^{-1-\gamma} \p{1 + O\p{\sqrt{\frac{\log
    n}{n}}}} \tag{\Cref{lem:pb-expectation,lem:lower-bound-on-p}} \\
    &\le O\p{n^{-1-\gamma}}
    \numberthis \label{eq:delta-k-bound}
  \end{align*}
  where the last line again uses \Cref{lem:lower-bound-on-p}. In other words,
  because $\barpn{n}_k = \Omega(1)$, the terms in question concentrate
  asymptotically.

  Putting this into~\eqref{eq:V-bound-sum-delta},
  \begin{align*}
    V_{k'}(\bpn{n}) - V_k(\bpn{n})
    &\le \bd_1 \sum_{\ell=1}^{k-1} \frac{\delta_k(i; n)}{k} \\
    &\le \bd_1 O(n^{-1-\gamma}) \sum_{\ell=1}^{k-1} \frac{1}{k} \\
    &\le \bd_1 (\log K + 1) O(n^{-1-\gamma}) \\
    &\le O(n^{-1-\gamma}).
  \end{align*}
\end{proof}

\begin{lemma}
  \label{lem:V-decrease-bound}
  Assume without loss of generality that $\bd_k$ are decreasing.
  Choose $\{\bPn{n}\}_{n=2}^\infty$ such that $\bPn{n} \in
  \eqsetn{n}$. Then,
  \begin{align*}
    V_{1}(\bPn{n}) \le V_K(\bPn{n}) + O(n^{-1-\gamma}).
  \end{align*}
\end{lemma}
\begin{proof}
  Here, we show that $V_K(\bPn{n})$ can't be ``too much'' smaller than
  $V_1(\bPn{n})$. This holds because any player can always transfer
  mass from $\bPn{n}(i)_K$ to $\bPn{n}(i)_1$, so for all $i$ such that
  $\bPn{n}(i)_K > 0$,
  \begin{align*}
    V_K(\bPn{n}(-i)) \ge V_1(\bPn{n}(-i)).
  \end{align*}
  Because $\barpn{n}_K = \Omega(1)$, asymptotically, there must exist some
  player $i$ such that $\bPn{n}(i)_K > 0$. Choose such a player $i$.
  Using~\eqref{eq:V-i-diff} twice,
  \ifthenelse{\boolean{smallEqs}}{
  \begin{align*}
    V_K(\bPn{n}) - \bd_K \bPn{n}(i)_K \delta_K(i; n)
    &\ge V_1(\bPn{n}) \\
    &- \bd_1 \bPn{n}(i)_1 \delta_1(i; n) \\
    V_K(\bPn{n})
    &\ge V_1(\bPn{n}) - \bd_1 \delta_1(i; n) \\
    V_K(\bPn{n}) + O(n^{-1-\gamma})
    &\ge V_1(\bPn{n}) \tag{By~\eqref{eq:delta-k-bound}}
  \end{align*}
}{\begin{align*}
    V_K(\bPn{n}) - \bd_K \bPn{n}(i)_K \delta_K(i; n)
    &\ge V_1(\bPn{n}) - \bd_1 \bPn{n}(i)_1 \delta_1(i; n) \\
    V_K(\bPn{n})
    &\ge V_1(\bPn{n}) - \bd_1 \delta_1(i; n) \\
    V_K(\bPn{n}) + O(n^{-1-\gamma})
    &\ge V_1(\bPn{n}) \tag{By~\eqref{eq:delta-k-bound}}
  \end{align*}
}
\end{proof}

\begin{lemma}
  \label{lem:V-overall-bound}
  Choose $\{\bPn{n}\}_{n=2}^\infty$ such that $\bPn{n} \in \eqsetn{n}$. Then,
  for all $k < k'$,
  \begin{align*}
    |V_{k}(\bPn{n}) - V_{k'}(\bPn{n})| \le O(n^{-1-\gamma}).
  \end{align*}
\end{lemma}
\begin{proof}
  We must show
  \begin{align*}
    V_{k'}(\bPn{n}) - V_k(\bPn{n})
    &\le O(n^{-1-\gamma}) \\
    V_k(\bPn{n}) - V_{k'}(\bPn{n})
    &\le O(n^{-1-\gamma})
  \end{align*}
  The first inequality follows directly from \Cref{lem:V-increase-bound}. For
  the second, we proceed as follows:
  Therefore,
  \begin{align*}
    V_k(\bPn{n})
    &\le V_1(\bPn{n}) + O(n^{-1-\gamma})
    \tag{By \Cref{lem:V-increase-bound}}
    \\
    &\le V_K(\bPn{n}) + O(n^{-1-\gamma})
    \tag{By \Cref{lem:V-decrease-bound}} \\
    &\le V_{k'}(\bPn{n}) + O(n^{-1-\gamma}).
    \tag{By \Cref{lem:V-increase-bound}}
  \end{align*}
\end{proof}

\begin{lemma}
  \label{lem:gamma-power-convergence}
  Suppose we have:
  \begin{itemize}
    \item a sequence of nonnegative values $\{a_n\}_{n=1}^\infty$,
    \item positive constant $b > 0$,
    \item real number $\gamma$,
    \item $f(n)$ such that $f(n)$ is monotonically decreasing and $\lim_{n \to
      \infty} f(n) = 0$.
  \end{itemize}
  If $|a_n - b| \le O(f(n))$, then $|a_n^\gamma - b^\gamma| \le O(f(n))$.
\end{lemma}
\begin{proof}
  If $\gamma = 0$, the claim is trivially true. Assume $\gamma \ne 0$.
  Let $g(x) = x^\gamma$. Then,
  \begin{align*}
    a_n^\gamma - b^\gamma
    &= g(a_n) - g(b) \\
    &= g'(z_n)(a_n - b) \tag{Mean value theorem, $z_n$ between $a_n$ and $b$} \\
    &= \gamma z_n^{\gamma-1} (a_n - b)
  \end{align*}
  For sufficiently large $n$, $b/2 \le a_n \le 2b$, so
  \begin{align*}
    |\gamma z_n^{\gamma - 1}| \le |\gamma| (2b)^{|\gamma - 1|} \triangleq C,
  \end{align*}
  which is a constant. Therefore,
  \begin{align*}
    |a_n^\gamma - b^\gamma|
    &\le C |a_n - b| \\
    &\le O(f(n))
  \end{align*}
\end{proof}

\begin{lemma}
  \label{lem:pb-expectation}
  Let $Y$ be a Poisson binomial random variable: $Y = \sum_{i=1}^n Y_i$, where
  each $Y_i$ is a Bernoulli with parameter $p_i$. Let $\mu = \sum_{i=1}^n p_i$.
  Then, for any fixed $\gamma > 0$,
  \begin{align*}
    (1+\mu)^{-\gamma}
    \le \E{(1+Y)^{-\gamma}}
    \le (1 + \mu)^{-\gamma} \p{1 + O\p{\sqrt{\frac{\log \mu}{\mu}}}}.
  \end{align*}
\end{lemma}
\begin{proof}
  The lower bound follows from Jensen's inequality. For the upper bound, we
  use concentration. For any $\delta \in (0, 1)$,
  \ifthenelse{\boolean{smallEqs}}{
  \begin{align*}
    \E{(1+Y)^{-\gamma}}
    &= \E{(1+Y)^{-\gamma} \cdot \ind{Y > (1-\delta) \mu}} \\
    &+ \E{(1+Y)^{-\gamma} \cdot \ind{Y \le (1-\delta) \mu}}.
    \numberthis \label{eq:Y-breakdown}
  \end{align*}
}{\begin{equation}
    \label{eq:Y-breakdown}
    \E{(1+Y)^{-\gamma}}
    = \E{(1+Y)^{-\gamma} \cdot \ind{Y > (1-\delta) \mu}}
    + \E{(1+Y)^{-\gamma} \cdot \ind{Y \le (1-\delta) \mu}}.
  \end{equation}
}
  A Chernoff bound yields
  \begin{align*}
    \Pr[Y \le (1-\delta) \mu]
    \le e^{-\frac{\mu \delta^2}{2}}.
  \end{align*}
  Because $(1+Y)^{-\gamma} \in [0, 1]$,
  \ifthenelse{\boolean{smallEqs}}{
  \begin{align*}
    &\E{(1+Y)^{-\gamma} \cdot \ind{Y \le (1-\delta) \mu}} \\
    &= \E{(1+Y)^{-\gamma} \given Y \le (1-\delta) \mu} \Pr[Y \le (1-\delta) \mu]
    \\
    &\le e^{-\frac{\mu \delta^2}{2}}.
    \numberthis \label{eq:Y-prob-bound}
  \end{align*}
}{\begin{align*}
    \E{(1+Y)^{-\gamma} \cdot \ind{Y \le (1-\delta) \mu}}
    &= \E{(1+Y)^{-\gamma} \given Y \le (1-\delta) \mu} \Pr[Y \le (1-\delta) \mu]
    \\
    &\le e^{-\frac{\mu \delta^2}{2}}.
    \numberthis \label{eq:Y-prob-bound}
  \end{align*}
}
  Putting this into~\eqref{eq:Y-breakdown},
  \begin{equation}
    \label{eq:Y-first-bound}
    \E{(1+Y)^{-\gamma}}
    \le \E{(1+Y)^{-\gamma} \cdot \ind{Y > (1-\delta) \mu}}
    + e^{-\frac{\mu \delta^2}{2}}.
  \end{equation}
  Next, we use the fact that $(1+y)^{-\gamma}$ is decreasing in $y$ to get
  \ifthenelse{\boolean{smallEqs}}{
  \begin{align*}
    &\E{(1+Y)^{-\gamma} \cdot \ind{Y > (1-\delta) \mu}} \\
    &\le \E{(1+Y)^{-\gamma} \given Y > (1-\delta) \mu} \\
    &\le (1+(1-\delta)\mu)^{-\gamma}.
  \end{align*}
}{\begin{align*}
    \E{(1+Y)^{-\gamma} \cdot \ind{Y > (1-\delta) \mu}}
    &\le \E{(1+Y)^{-\gamma} \given Y > (1-\delta) \mu} \\
    &\le (1+(1-\delta)\mu)^{-\gamma}.
  \end{align*}
}
  Let $f(y) = (1+y)^{-\gamma}$. By Taylor's theorem,
  \begin{align*}
    (1+(1-\delta)\mu)^{-\gamma}
    &= (1 - \delta \mu + \mu)^{-\gamma} \\
    &= (1 + \mu)^{-\gamma} + f'(\xi)(-\delta \mu)
  \end{align*}
  for some $\xi \in [\mu - \delta \mu, \mu]$. Using the fact that $f'(y)
  = -\gamma (1+y)^{-1-\gamma}$ (which is an increasing function),
  \begin{align*}
    (1+(1-\delta)\mu)^{-\gamma}
    &= (1 + \mu)^{-\gamma} +f'(\xi) (-\delta \mu) \\
    &= (1 + \mu)^{-\gamma} + \gamma (1+\xi)^{-1-\gamma} (\delta \mu) \\
    &\le (1 + \mu)^{-\gamma} + \gamma (1 - \delta \mu + \mu)^{-1-\gamma} (\delta
    \mu) \\
    &= (1 + \mu)^{-\gamma} + \gamma \delta \mu (1 - \delta \mu +
    \mu)^{-1-\gamma}.
  \end{align*}
  Putting this into~\eqref{eq:Y-first-bound},
  \begin{align*}
    \E{(1+Y)^{-\gamma}}
    &\le (1 + \mu)^{-\gamma} + \gamma \delta \mu (1 - \delta \mu +
    \mu)^{-1-\gamma}
    + e^{-\frac{\mu \delta^2}{2}}.
  \end{align*}
  Choose $\delta = \sqrt{\frac{2(1+\gamma) \log \mu}{\mu}}$. Then,
  \ifthenelse{\boolean{smallEqs}}{
  \begin{align*}
    &\E{(1+Y)^{-\gamma}} \\
    &\le (1 + \mu)^{-\gamma} \\
    &+ \gamma \sqrt{\frac{2(1+\gamma) \log \mu}{\mu}}
    \mu \p{1 - \sqrt{\frac{2(1+\gamma) \log \mu}{\mu}} \mu + \mu}^{-1-\gamma} \\
    &+ e^{-\frac{\mu \frac{2(1+\gamma)\log \mu}{\mu}}{2}} \\
    &= (1 + \mu)^{-\gamma} \\
    &+ \gamma \sqrt{2(1+\gamma) \mu \log \mu} \p{1 - \sqrt{2(1+\gamma) \mu \log
    \mu} + \mu}^{-1-\gamma} \\
    &+ e^{- \log \mu} \\
    &= (1 + \mu)^{-\gamma} \\
    &+ \gamma \sqrt{2(1+\gamma) \mu \log \mu} \p{1 - \sqrt{2(1+\gamma) \mu \log
    \mu} + \mu}^{-1-\gamma} \\
    &+ \mu^{-1-\gamma} \\
    &= (1 + \mu)^{-\gamma} \cdot \\
    &\left(1 + \gamma \sqrt{2(1+\gamma) \mu \log \mu}
      \p{\frac{1+\mu}{1 - \sqrt{2(1+\gamma) \mu \log \mu} + \mu}}^{\gamma}
      \cdot \right. \\
    &\left. \p{1 - \sqrt{2(1+\gamma) \mu \log \mu} + \mu}^{-1}
    + \p{1 + \frac{1}{\mu}}^\gamma \mu^{-1}\right)
  \end{align*}
}{\begin{align*}
    &\E{(1+Y)^{-\gamma}} \\
    &\le (1 + \mu)^{-\gamma} + \gamma \sqrt{\frac{2(1+\gamma) \log \mu}{\mu}}
    \mu \p{1 - \sqrt{\frac{2(1+\gamma) \log \mu}{\mu}} \mu + \mu}^{-1-\gamma}
    + e^{-\frac{\mu \frac{2(1+\gamma)\log \mu}{\mu}}{2}} \\
    &= (1 + \mu)^{-\gamma} + \gamma \sqrt{2(1+\gamma) \mu \log \mu} \p{1 -
    \sqrt{2(1+\gamma) \mu \log \mu} + \mu}^{-1-\gamma}
    + e^{- \log \mu} \\
    &= (1 + \mu)^{-\gamma} + \gamma \sqrt{2(1+\gamma) \mu \log \mu} \p{1 -
    \sqrt{2(1+\gamma) \mu \log \mu} + \mu}^{-1-\gamma}
    + \mu^{-1-\gamma} \\
    &= (1 + \mu)^{-\gamma} \left(1 + \gamma \sqrt{2(1+\gamma) \mu \log \mu}
      \p{\frac{1+\mu}{1 - \sqrt{2(1+\gamma) \mu \log \mu} + \mu}}^{\gamma} \p{1
    - \sqrt{2(1+\gamma) \mu \log \mu} + \mu}^{-1} \right. \\
    &\left.+
    \p{1 + \frac{1}{\mu}}^\gamma \mu^{-1}\right)
  \end{align*}
}
  For fixed $\gamma$, this is
  \begin{align*}
    \E{(1 + Y)^{-\gamma}}
    \le (1+\mu)^{-\gamma} \p{1 + O\p{\sqrt{\frac{\log \mu}{\mu}}}}.
  \end{align*}
\end{proof}

\subsection{Omitted proofs for \Cref{sec:poa}}
\label{app:poa}

\citet[Thm. 5]{vetta2002nash} shows that the price of anarchy is at most 2 for any
\textit{valid utility game}. We will show that $\game(n, \bd, s)$ is a valid utility
game.
As before, we denote the $i$th player's strategy by $\bP(i)$, the strategy
matrix for the remaining $n-1$ players by $\bP(-i)$, and for any $S \subseteq
[n]$, $\bP(S)$ denotes the strategy matrix given by restricting to only players
in $S$. We use the same notation as defined in \Cref{sec:support-asymp}.

A valid utility game \citep{vetta2002nash} satisfies the following three
conditions:
\begin{enumerate}
  \item \textbf{Submodularity:} The social welfare function is submodular in the
    set of players participating. For $A \subseteq B \subseteq [n]$,
    \begin{align*}
      W(\bP(A \cup \{i\})) - W(\bP(A)) \ge W(\bP(B \cup \{i\})) - W(\bP(B)).
    \end{align*}
  \item \textbf{Players' marginal contributions are rewarded:} For each player,
    their utility is at least their marginal contribution to social welfare
    given all other players' behaviors:
    \begin{align*}
      U_i(\bP) \ge W(\bP) - W(\bP(-i)).
    \end{align*}
  \item \textbf{Social welfare is at least players' utilities:} The sum of all
    players' utilities is at most the social welfare:
    \begin{align*}
      \sum_{i \in [n]} U_i(\bP) \le W(\bP).
    \end{align*}
\end{enumerate}
With this, we can prove \Cref{thm:poa}.

\poa*
\begin{proof}

  Assume without loss of generality that $\bd$ is decreasing.
  We will show that $\game(n, \bd, s)$ is a valid utility game.

  \paragraph*{Submodularity.}

  It suffices to consider the case where $s \in \sup$, since when $s \in
  \sdown$, social welfare is equivalent to the case $s = s_1 \in \sup$.
  We want to show that for $A \subseteq B \subseteq [n]$,
  \begin{equation}
    \label{eq:W-submodularity}
    W(\bP(A \cup \{i\})) - W(\bP(A))
    \ge W(\bP(B \cup \{i\})) - W(\bP(B)).
  \end{equation}
  To do so, it suffices to show that for each $k \in [K]$,
  \begin{align*}
    \bd_k
     \E{\frac{X_k(\bP(A \cup \{i\}))}{s(X_k(\bP(A \cup
    \{i\})))}}
    - \E{\frac{X_k(\bP(A))}{s(X_k(\bP(A)))}} \\
    \ge
    \bd_k \E{\frac{X_k(\bP(B \cup \{i\}))}{s(X_k(\bP(B \cup
    \{i\})))}}
    - \E{\frac{X_k(\bP(B))}{s(X_k(\bP(B)))}}.
  \end{align*}
  Observe that because $B \supseteq A$,
  \begin{align*}
    X_k(\bP(B)) = X_k(\bP(B \backslash A)) + X_k(\bP(A)).
  \end{align*}
  Because $X_k(\cdot) \ge 0$, this means
  \begin{align*}
    X_k(\bP(B)) \ge X_k(\bP(A)).
  \end{align*}
  Because $s \in \sup$, by assumption, $x/s(x)$ is increasing and concave. For
  any $a \le b$ and $c \ge 0$,
  \begin{align*}
    \frac{a+c}{s(a+c)} - \frac{a}{s(a)} \ge \frac{b+c}{s(b+c)} - \frac{b}{s(b)}.
  \end{align*}
  Thus,
  \ifthenelse{\boolean{smallEqs}}{
  \begin{align*}
    &\frac{X_k(\bP(A) + X_k(\bP(i))}{s(X_k(\bP(A) + X_k(\bP(i)))}
    - \frac{X_k(\bP(A))}{s(X_k(\bP(A)))} \\
    &\ge
    \frac{X_k(\bP(B)) X_k(\bP(i)}{s(X_k(\bP(B)) X_k(\bP(i))}
    - \frac{X_k(\bP(B))}{s(X_k(\bP(B)))}.
  \end{align*}
}{\begin{align*}
    \frac{X_k(\bP(A) + X_k(\bP(i))}{s(X_k(\bP(A) + X_k(\bP(i)))}
    - \frac{X_k(\bP(A))}{s(X_k(\bP(A)))}
    \ge
    \frac{X_k(\bP(B)) X_k(\bP(i)}{s(X_k(\bP(B)) X_k(\bP(i))}
    - \frac{X_k(\bP(B))}{s(X_k(\bP(B)))}.
  \end{align*}
}
  Taking expectations of both sides, multiplying both sides by $\bd_k$, and
  summing over $k \in [K]$ yields~\eqref{eq:W-submodularity}.

  \paragraph*{Players' marginal contributions are rewarded.}

  Next, we must show 
  \begin{equation}
    \label{eq:marginal-cont}
      U_i(\bP) \ge W(\bP) - W(\bP(-i)).
  \end{equation}
 For $s \in \sup$,
 \ifthenelse{\boolean{smallEqs}}{
  \begin{align*}
    &W(\bP) \\
    &= \sum_{k \in [K]} \bd_k \E{\frac{X_k(\bP)}{s(X_k(\bP))}} \\
    &= \sum_{k \in [K]} \bd_k \E{\frac{X_k(\bP(-i)) +
    X_k(\bP(i))}{s(X_k(\bP(-i)) + X_k(\bP(i)))}} \\
    &= \sum_{k \in [K]} \bd_k \p{\E{\E{\frac{X_k(\bP(-i)) +
    X_k(\bP(i))}{s(X_k(\bP(-i)) + X_k(\bP(i)))}} \; \Big| \; X_k(\bP(i))}} \\
    &= \sum_{k \in [K]} \bd_k
    \left(\bP(i)_k \E{\frac{1}{s(1 + X_k(\bP(-i)))}} \right. \\
    &\left.+ (1
    - \bP(i)_k) \E{\frac{X_k(\bP(-i))}{s(X_k(\bP(-i)))}}\right) \\
    &= U_i(\bP) + \sum_{k \in [K]} \bd_k (1 - \bP(i)_k)
    \E{\frac{X_k(\bP(-i))}{s(X_k(\bP(-i)))}} \\
    &\le U_i(\bP) + \sum_{k \in [K]} \bd_k
    \E{\frac{X_k(\bP(-i))}{s(X_k(\bP(-i)))}} \\
    &= U_i(\bP)  + W(\bP(-i)).
  \end{align*}
}{\begin{align*}
    W(\bP)
    &= \sum_{k \in [K]} \bd_k \E{\frac{X_k(\bP)}{s(X_k(\bP))}} \\
    &= \sum_{k \in [K]} \bd_k \E{\frac{X_k(\bP(-i)) +
    X_k(\bP(i))}{s(X_k(\bP(-i)) + X_k(\bP(i)))}} \\
    &= \sum_{k \in [K]} \bd_k \p{\E{\E{\frac{X_k(\bP(-i)) +
    X_k(\bP(i))}{s(X_k(\bP(-i)) + X_k(\bP(i)))}} \; \Big| \; X_k(\bP(i))}} \\
    &= \sum_{k \in [K]} \bd_k \p{\bP(i)_k \E{\frac{1}{s(1 + X_k(\bP(-i)))}} + (1
    - \bP(i)_k) \E{\frac{X_k(\bP(-i))}{s(X_k(\bP(-i)))}}} \\
    &= U_i(\bP) + \sum_{k \in [K]} \bd_k (1 - \bP(i)_k)
    \E{\frac{X_k(\bP(-i))}{s(X_k(\bP(-i)))}} \\
    &\le U_i(\bP) + \sum_{k \in [K]} \bd_k
    \E{\frac{X_k(\bP(-i))}{s(X_k(\bP(-i)))}} \\
    &= U_i(\bP)  + W(\bP(-i)).
  \end{align*}
}
  For $s \in \sdown$,
  \ifthenelse{\boolean{smallEqs}}{
  \begin{align*}
    &W(\bP) \\
    &= \sum_{k \in [K]} \bd_k \Pr[X_k(\bP) > 0] \\
    &= \sum_{k \in [K]} \bd_k \Pr[X_k(\bP(-i)) > 0 \vee X_k(\bP(i)) > 0] \\
    &= \sum_{k \in [K]} \bd_k (\Pr[X_k(\bP(i)) > 0 \wedge X_k(\bP(-i)) = 0] \\
    &+ \Pr[X_k(\bP(-i)) > 0]) \\
    &= W(\bP(-i)) + \sum_{k \in [K]} \bd_k \bP(i)_k \Pr[X_k(\bP(-i)) = 0] \\
    &= W(\bP(-i)) \\
    &+ \sum_{k \in [K]} \bd_k \bP(i)_k \E{\frac{1}{s(1 +
    X_k(\bP(-i)))} \cdot \ind{X_k(\bP(-i)) = 0}} \\
    &\le W(\bP(-i)) +  \sum_{k \in [K]} \bd_k \bP(i)_k \E{\frac{1}{s(1 +
    X_k(\bP(-i)))}} \\
    &= U_i(\bP) + W(\bP(-i)).
  \end{align*}
}{\begin{align*}
    W(\bP)
    &= \sum_{k \in [K]} \bd_k \Pr[X_k(\bP) > 0] \\
    &= \sum_{k \in [K]} \bd_k \Pr[X_k(\bP(-i)) > 0 \vee X_k(\bP(i)) > 0] \\
    &= \sum_{k \in [K]} \bd_k \p{\Pr[X_k(\bP(i)) > 0 \wedge X_k(\bP(-i)) = 0] +
    \Pr[X_k(\bP(-i)) > 0]} \\
    &= W(\bP(-i)) +  \sum_{k \in [K]} \bd_k \bP(i)_k \Pr[X_k(\bP(-i)) = 0] \\
    &= W(\bP(-i)) +  \sum_{k \in [K]} \bd_k \bP(i)_k \E{\frac{1}{s(1 +
    X_k(\bP(-i)))} \cdot \ind{X_k(\bP(-i)) = 0}} \\
    &\le W(\bP(-i)) +  \sum_{k \in [K]} \bd_k \bP(i)_k \E{\frac{1}{s(1 +
    X_k(\bP(-i)))}} \\
    &= U_i(\bP) + W(\bP(-i)).
  \end{align*}
}
  Thus,~\eqref{eq:marginal-cont} holds.

  \paragraph*{Social welfare is at least players' utilities.}

  Finally, we must show
  \begin{equation}
    \label{eq:sum-utilities-welfare}
    \sum_{i \in [n]} U_i(\bP) \le W(\bP).
  \end{equation}
  We proceed as follows.
  \begin{align*}
    \sum_{i \in [n]} U_i(\bP)
    &= \sum_{k \in [K]} \bd_k \sum_{i \in [n]} \bP(i)_k \E{\frac{1}{s(1 +
    X_k(\bP(-i)))}} \\
    &= \sum_{k \in [K]} \bd_k \sum_{i \in [n]}
    \E{\frac{X_k(\bP(i))}{s(X_k(\bP(i)) + X_k(\bP(-i)))}} \\
    &= \sum_{k \in [K]} \bd_k \sum_{i \in [n]}
    \E{\frac{X_k(\bP(i))}{s(X_k(\bP))}} \\
    &= \sum_{k \in [K]} \bd_k
    \E{\frac{X_k(\bP)}{s(X_k(\bP))}}.
  \end{align*}
  For $s \in \sup$, this last line is equal to $W(\bP)$,
  and~\eqref{eq:sum-utilities-welfare} holds. For $s \in \sdown$,
  \begin{align*}
    \sum_{k \in [K]} \bd_k
    \E{\frac{X_k(\bP)}{s(X_k(\bP))}}
    &\le \sum_{k \in [K]} \bd_k \E{\ind{X_k(\bP) > 0}} \\
    &= \sum_{k \in [K]} \bd_k \Pr[X_k(\bP) > 0] \\
    &= W(\bP).
  \end{align*}
  Thus, $\game(n, \bd, s)$ is a valid utility game.

  To show that this is tight, for some $n$, let
  \begin{align*}
    \bd = [1; ~ \nicefrac{1}{n}; ~ \dots ~ \nicefrac{1}{n}].
  \end{align*}
  Choose $K > n^2$.
  Let the score function be the identity $s_1$. We will show that there
  exists a strategy profile with social welfare approaching $2$ and a
  (symmetric) equilibrium with social welfare $1$. Define $\bPp$ as follows:
  \begin{align*}
    \bPp(1) &= [1; ~ 0; ~ \dots ~ 0] \\
    \bPp(i) &= [\nicefrac{1}{K}; ~ \dots ~ \nicefrac{1}{K}]. \tag{$i > 1$}
  \end{align*}
  Because $s_1 \in \sdown$, social welfare under $\bPp$ is
  \begin{align*}
    W(\bPp)
    &= \sum_{k \in [K]} \bd_k (1 - \Pr[X_k(\bPp) = 0]) \\
    &= \bd_1 + \sum_{k=2}^K \bd_k \p{1 - \Pr\b{X\p{n-1, \frac{1}{K}} = 0}} \\
    &= 1 + \sum_{k=2}^K \frac{1}{n} \p{1 - \p{1 - \frac{1}{K}}^{n-1}} \\
    &= 1 + \frac{K-1}{n} \p{1 - \p{1 - \frac{1}{K}}^{n-1}} \\
    &\ge 1 + \frac{K-1}{n}
    \p{1 - \p{1 - \frac{n-1}{K} + \binom{n-1}{2} \frac{1}{K^2}}}
    \tag{by \Cref{lem:1-x-bound} since $K > n$} \\
    &= 1 + \frac{K-1}{n}
    \p{\frac{n-1}{K} - \binom{n-1}{2} \frac{1}{K^2}} \\
    &= 1 + \frac{(K-1)(n-1)}{Kn} \p{1  - \frac{n-2}{2K}} \\
    &\ge 1 + \frac{(n^2-1)(n-1)}{n^3} \p{1 - \frac{1}{2n}} \tag{$K > n^2$} \\
    &\ge 2 - \frac{n^2 + n}{n^3} \p{1 - \frac{1}{2n}} - \frac{1}{2n} \\
    &\ge 2 - \frac{3}{2n} - \frac{1}{n^2}.
  \end{align*}
  For sufficiently large $n$, we can make this arbitrarily close to 2. In
  contrast, we will show that $\bP$ is an equilibrium, where
  \begin{align*}
    \bP(i) &= [1; ~ 0; ~ \dots ~ 0]
  \end{align*}
  for all $i$. To show that it is an equilibrium, observe that each player gets
  utility $U_i(\bP) = \frac{1}{n}$, which weakly greater than $\bd_k$ for any $k
  > 1$. Therefore, $\bP$ is (weakly) an equilibrium. However, $W(\bP) = 1$. The
  price of anarchy is therefore 2.
\end{proof}

\subsection{Auxiliary results}
\label{app:aux}

\subsubsection{Results from prior work}

\begin{lemma}[{\citet[Thm 3.41]{rudin1964principles}}]
  \label{lem:rudin-partial-sum}
  For any sequences $\{x_i\}, \{y_i\}$, let $X_i = \sum_{j=0}^i x_i$. Then,
  \begin{equation}
    \label{eq:rudin-partial-sum}
    \sum_{i=0}^n x_i y_i = -\sum_{i=0}^{n-1} X_i (y_{i+1} - y_i) + X_n y_n.
  \end{equation}
\end{lemma}

\begin{lemma}[{\citet{sah1991effects}}]
  \label{lem:sah-diff}
  Define
  \begin{align*}
    f(n, p)
    &\triangleq \sum_{\ell=0}^n b(\ell, n, p) g(\ell) \\
    \df(n, p)
    &\triangleq f(n+1, p) - f(n, p) \\
    \ddf(n, p)
    &\triangleq \df(n+1, p) - \df(n, p) \\
    \dg(\ell)
    &\triangleq g(1+\ell) - g(\ell) \\
    \ddg(\ell)
    &\triangleq \dg(1+\ell) - \dg(\ell)
  \end{align*}
  Then,
  \begin{align}
    \label{eq:sah-diff-1}
    \df(n, p)
    &= p \sum_{\ell=0}^n b(\ell, n, p) \dg(\ell) \\
    \label{eq:sah-diff-2}
    \ddf(n, p)
    &= p^2 \sum_{\ell=0}^n b(\ell, n, p) \ddg(\ell) \\
    \label{eq:sah-partial-1}
    \frac{\partial}{\partial p} f(n, p)
    &= \frac{n}{p} \df(n-1, p) \\
    \label{eq:sah-partial-2}
    \frac{\partial}{\partial p} \df(n, p)
    &= \frac{1}{p} \p{\df(n, p) + n \ddf(n-1, p)}.
  \end{align}
\end{lemma}
\begin{proof}

  \eqref{eq:sah-diff-1},~\eqref{eq:sah-diff-2}, and~\eqref{eq:sah-partial-2} are
  directly from \citet[eq.~(4--6)]{sah1991effects}.
  A similar argument, which we reproduce below for completeness,
  gives us~\eqref{eq:sah-partial-1}. See \citet{sah1989comparative} for
  additional detail.

  From \citet[eq. (A6)]{sah1991effects} (see also \citet[eq.
  (10.9) p. 173]{feller1971introduction}),
  \begin{equation}
    \label{eq:feller-partial}
    \frac{\partial}{\partial p} B(\ell, n, p) = -n b(\ell, n-1, p).
  \end{equation}
  Therefore,
  \begin{align*}
    \frac{\partial}{\partial p} f(n, p)
    &= \frac{\partial}{\partial p} \sum_{\ell=0}^n b(\ell, n, p) g(\ell) \\
    &= \frac{\partial}{\partial p} \b{g(n) - \sum_{\ell=0}^{n-1} B(\ell, n, p)
    \dg(\ell)} \tag{by \Cref{lem:rudin-partial-sum}} \\
    &= - \sum_{\ell=0}^{n-1} (-n b(\ell, n-1, p)) \dg(\ell)
    \tag{by~\eqref{eq:feller-partial}} \\
    &= n \sum_{\ell=0}^{n-1} b(\ell, n-1, p) \dg(\ell) \\
    &= \frac{n}{p} \df(n-1, p).
    \tag{by \Cref{eq:sah-diff-1}}
  \end{align*}

\end{proof}

\subsubsection{Additional results.}

\begin{lemma}
  \label{lem:potential-game}
  $\game$ is a potential game.
\end{lemma}
\begin{proof}
  Define potentially asymmetric strategy profiles $\bP$ as in \Cref{app:poa}.
  $\game$ is a potential game if there exists some function $\Phi$ such that,
  for any $\bP$ and $\bP'$ differing only in player $i$'s strategy,
  \begin{equation}
    \label{eq:potential-game}
    \Phi(\bP') - \Phi(\bP)
    = U_i(\bP') - U_i(\bP).
  \end{equation}
  We will show that the this holds for the following definition of $\Phi$, which
  is fairly standard in the literature on congestion
  games~\citep{rosenthal1973class}.
  \begin{align*}
    \Phi(\bP)
    \triangleq \E{\sum_{k \in [K]} \bd_k
      \sum_{\ell=1}^{X_k(\bP)} \frac{1}{s(\ell)}}.
  \end{align*}
  Observe that
  \ifthenelse{\boolean{smallEqs}}{
  \begin{align*}
    &\E{\sum_{k \in [K]} \bd_k \sum_{\ell=1}^{X_k(\bP)} \frac{1}{s(\ell)}} \\
    &= \sum_{k \in [K]} \bd_k
    \E{\frac{X_k(\bP(i))}{s(1+X_k(\bP(-i)))} + \sum_{\ell=1}^{X_k(\bP(-i))}
    \frac{1}{s(\ell)}} \\
    &= \sum_{k \in [K]} \bd_k \bP(i)_k \frac{1}{1 + X_k(\bP(-i))} \\
    &+ \sum_{k \in [K]} \bd_k
    \E{\sum_{\ell=1}^{X_k(\bP(-i))} \frac{1}{s(\ell)}} \\
    &= U_i(\bP) + \Phi(\bP(-i)).
  \end{align*}
}{\begin{align*}
    \E{\sum_{k \in [K]} \bd_k \sum_{\ell=1}^{X_k(\bP)} \frac{1}{s(\ell)}}
    &= \sum_{k \in [K]} \bd_k
    \E{\frac{X_k(\bP(i))}{s(1+X_k(\bP(-i)))} + \sum_{\ell=1}^{X_k(\bP(-i))}
    \frac{1}{s(\ell)}} \\
    &= \sum_{k \in [K]} \bd_k \bP(i)_k \frac{1}{1 + X_k(\bP(-i))}
    + \sum_{k \in [K]} \bd_k
    \E{\sum_{\ell=1}^{X_k(\bP(-i))}
    \frac{1}{s(\ell)}} \\
    &= U_i(\bP) + \Phi(\bP(-i)).
  \end{align*}
}
  Therefore, if $\bP'$ and $\bP$ only differ in column $i$,
  \begin{align*}
    \Phi(\bP') - \Phi(\bP)
    = U_i(\bP') - U_i(\bP).
  \end{align*}
\end{proof}

\begin{lemma}
  \label{lem:dp-dn-ratio}
  \begin{align*}
    \frac{\partial}{\partial p} \frac{u(n, p)}{u(n-1, p)} \le 0.
  \end{align*}
\end{lemma}
\begin{proof}
  We proceed as follows.
  \ifthenelse{\boolean{smallEqs}}{
  \begin{align*}
    &
    \frac{\partial}{\partial p} \frac{u(n, p)}{u(n-1, p)} \le 0 \\
    \Longleftrightarrow &
    \frac{\partial}{\partial p} \log \p{\frac{u(n, p)}{u(n-1, p)}} \le 0 \\
    \Longleftrightarrow &
    \frac{\frac{\partial}{\partial p} u(n, p)}{u(n, p)} - 
    \frac{\frac{\partial}{\partial p} u(n-1, p)}{u(n-1, p)} \le 0 \\
    \Longleftrightarrow &
    \frac{\frac{n}{p} \du(n-1, p)}{u(n, p)} \le
    \frac{\frac{n-1}{p} \du(n-2, p)}{u(n-1, p)} \\
    \Longleftrightarrow &
    \frac{n(u(n, p) - u(n-1, p))}{u(n, p)} \\
    &
    \le \frac{(n-1)(u(n-1, p) - u(n-2,
    p))}{u(n-1, p)} \\
    \Longleftrightarrow &
    n - \frac{n u(n-1, p)}{u(n, p)} \le (n-1) - \frac{(n-1)u(n-2, p)}{u(n-1, p)}
    \\
    \Longleftrightarrow &
    1 - \frac{n u(n-1, p)}{u(n, p)} \le - \frac{(n-1)u(n-2, p)}{u(n-1, p)}
    \\
    \Longleftrightarrow &
    u(n-1, p) u(n, p) + (n-1) u(n-2, p) u(n, p) \\
                        &\le n u(n-1, p)^2.
  \end{align*}
}{\begin{align*}
    \frac{\partial}{\partial p} \frac{u(n, p)}{u(n-1, p)} \le 0
    \Longleftrightarrow &
    \frac{\partial}{\partial p} \log \p{\frac{u(n, p)}{u(n-1, p)}} \le 0 \\
    \Longleftrightarrow &
    \frac{\frac{\partial}{\partial p} u(n, p)}{u(n, p)} - 
    \frac{\frac{\partial}{\partial p} u(n-1, p)}{u(n-1, p)} \le 0 \\
    \Longleftrightarrow &
    \frac{\frac{n}{p} \du(n-1, p)}{u(n, p)} \le
    \frac{\frac{n-1}{p} \du(n-2, p)}{u(n-1, p)} \\
    \Longleftrightarrow &
    \frac{n(u(n, p) - u(n-1, p))}{u(n, p)} \le \frac{(n-1)(u(n-1, p) - u(n-2,
    p))}{u(n-1, p)} \\
    \Longleftrightarrow &
    n - \frac{n u(n-1, p)}{u(n, p)} \le (n-1) - \frac{(n-1)u(n-2, p)}{u(n-1, p)}
    \\
    \Longleftrightarrow &
    1 - \frac{n u(n-1, p)}{u(n, p)} \le - \frac{(n-1)u(n-2, p)}{u(n-1, p)}
    \\
    \Longleftrightarrow &
    u(n-1, p) u(n, p) + (n-1) u(n-2, p) u(n, p) \le n u(n-1, p)^2.
  \end{align*}
}
  Using the fact that $u(n, p) \le u(n-1, p)$ (because $X(n, p)$
  stochastically dominates $X(n-1, p)$), $u(n, p) u(n-1, p) \le u(n-1, p)^2$.
  Therefore, it suffices to show that
  \begin{align*}
    (n-1) u(n-2, p) u(n, p)
    &\le (n-1) u(n-1, p)^2 \\
    u(n-2, p) u(n, p)
    &\le u(n-1, p)^2.
    \numberthis \label{eq:need-n-convex}
  \end{align*}
  From \Cref{lem:sah-diff}, we have
  \ifthenelse{\boolean{smallEqs}}{
  \begin{align*}
    &u(n, p) - u(n-1, p) - (u(n-1, p) - u(n-2, p))\\
    =\, &\ddu(n-2, p) \\
    =\, &p^2 \sum_{\ell=0}^{n-2} b(\ell, n-2, p) \cdot \\
        &\p{\frac{1}{s(3+\ell)} -
    \frac{1}{s(2+\ell)} - \p{\frac{1}{s(2+\ell)} - \frac{1}{s(1+\ell)}}}.
  \end{align*}
}{\begin{align*}
    &u(n, p) - u(n-1, p) - (u(n-1, p) - u(n-2, p))\\
    =\, &\ddu(n-2, p) \\
    =\, &p^2 \sum_{\ell=0}^{n-2} b(\ell, n-2, p) \p{\frac{1}{s(3+\ell)} -
    \frac{1}{s(2+\ell)} - \p{\frac{1}{s(2+\ell)} - \frac{1}{s(1+\ell)}}}.
  \end{align*}
}
  Under the assumption that $1/s(x)$ is convex (\labelcref{def:convex}), each
  term in the sum is weakly negative. As a result, $u$ is discrete convex in
  $n$:
  \ifthenelse{\boolean{smallEqs}}{
  \begin{align*}
    u(n, p) - u(n-1, p) - (u(n-1, p) - u(n-2, p))
    \le 0 \\
    u(n-1, p)
    \ge \frac{u(n, p) + u(n-2, p)}{2}.
  \end{align*}
}{\begin{align*}
    u(n, p) - u(n-1, p) - (u(n-1, p) - u(n-2, p))
    & \le 0 \\
    u(n-1, p)
    &\ge \frac{u(n, p) + u(n-2, p)}{2}.
  \end{align*}
}
  Putting this into~\eqref{eq:need-n-convex}, we have
  \ifthenelse{\boolean{smallEqs}}{
  \begin{align*}
    &u(n-1, p)^2 - u(n, p) u(n-2, p) \\
    &\ge \p{\frac{u(n, p) + u(n-2, p)}{2}}^2 - u(n, p) u(n-2, p) \\
    &= \frac{u(n, p)^2 + 2u(n, p) u(n-2, p) + u(n-2, p)^2}{4} \\
    &- u(n, p) u(n-2, p) \\
    &= \frac{u(n, p)^2 + u(n-2, p)^2}{4} - \frac{u(n, p) u(n-2, p)}{2} \\
    &= \frac{1}{2} \p{u(n, p)^2 - 2u(n, p) u(n-2, p) + u(n-2, p)^2} \\
    &= \frac{1}{2} \p{u(n, p) - u(n-2, p)}^2 \\
    &\ge 0.
  \end{align*}
}{\begin{align*}
    u(n-1, p)^2 - u(n, p) u(n-2, p)
    &\ge \p{\frac{u(n, p) + u(n-2, p)}{2}}^2 - u(n, p) u(n-2, p) \\
    &= \frac{u(n, p)^2 + 2u(n, p) u(n-2, p) + u(n-2, p)^2}{4} \\
    &- u(n, p) u(n-2, p) \\
    &= \frac{u(n, p)^2 + u(n-2, p)^2}{4} - \frac{u(n, p) u(n-2, p)}{2} \\
    &= \frac{1}{2} \p{u(n, p)^2 - 2u(n, p) u(n-2, p) + u(n-2, p)^2} \\
    &= \frac{1}{2} \p{u(n, p) - u(n-2, p)}^2 \\
    &\ge 0.
  \end{align*}
}
\end{proof}

\begin{lemma}
  \label{lem:u-dec-in-p}
  Suppose $g$ is weakly decreasing and $g(1) < g(0)$. Let $f(n, p) = \E{g(X(n,
  p))}$. Then,
  \begin{align*}
    \frac{\partial}{\partial p} f(n, p) < 0.
  \end{align*}
\end{lemma}
\begin{proof}
  If $n = 1$, the claim is true because $\frac{\partial}{\partial p} f(n, p) =
  g(0) - g(1)$. Otherwise,
  \begin{align*}
    \frac{\partial}{\partial p} f(n, p)
    &= \frac{n}{p} \df(n-1, p) \tag{\Cref{lem:sah-diff}} \\
    &= n \sum_{\ell=0}^{n-1} b(\ell, n-1, p) \dg(\ell) \tag{\Cref{lem:sah-diff}}
    \\
    &= n \sum_{\ell=0}^{n-1} b(\ell, n-1, p) (g(1+\ell) - g(\ell)) \\
    &\le n b(0, n-1, p) (g(1) - g(0)) \\
    &= n (1-p)^{n-1} \p{g(1) - g(0)} \\
    &< 0.
  \end{align*}
\end{proof}

\begin{lemma}
  \label{lem:binom-ex-convex}
  If $g(x)$ is convex, then $\E{g(X(n, p))}$ is convex in $p$.
  If $g(x)$ is concave, then $\E{g(X(n, p))}$ is concave in $p$.
  Additionally, if $n > 1$ and
  \begin{align*}
    \frac{g(2) + g(0)}{2} \ne g(1),
  \end{align*}
  then these are strict.
\end{lemma}
\begin{proof}
  Let $f(n, p) = \E{g(X(n, p))}$. If $n = 1$, $f(1, p) = p g(1) + (1-p) g(0)$,
  which is both convex and concave. Otherwise,
  \ifthenelse{\boolean{smallEqs}}{
  \begin{align*}
    &\frac{\partial^2}{\partial p^2} f(n, p) \\
    &= \frac{\partial}{\partial p} \frac{n}{p} \df(n-1, p)
    \tag{\Cref{lem:sah-diff}} \\
    &= n \left(\frac{1}{p} \cdot \frac{1}{p} \p{\df(n-1, p) + (n-1) \ddf(n-2,
      p)} \right. \\
    &+ \left. \df(n-1, p) (-p^{-2})\right) \tag{\Cref{lem:sah-diff}} \\
    &= \frac{n(n-1)}{p^2} \ddf(n-2, p) \\
    &= n(n-1) \sum_{\ell=0}^{n-2} b(\ell, n, p) \ddg(\ell)
    \tag{\Cref{lem:sah-diff}} \\
    &= n(n-1) \sum_{\ell=0}^{n-2} b(\ell, n, p) (\dg(1+\ell) - \dg(\ell)) \\
    &= n(n-1) \cdot \\
    &\sum_{\ell=0}^{n-2} b(\ell, n, p) (g(2+\ell) - g(1+\ell) -
    (g(1+\ell) - g(\ell))) \\
    &= n(n-1) \sum_{\ell=0}^{n-2} b(\ell, n, p) (g(2+\ell) + g(\ell) -
    2g(1+\ell)).
  \end{align*}
}{\begin{align*}
    \frac{\partial^2}{\partial p^2} f(n, p)
    &= \frac{\partial}{\partial p} \frac{n}{p} \df(n-1, p)
    \tag{\Cref{lem:sah-diff}} \\
    &= n \p{\frac{1}{p} \cdot \frac{1}{p} \p{\df(n-1, p) + (n-1) \ddf(n-2, p)} +
    \df(n-1, p) (-p^{-2})} \tag{\Cref{lem:sah-diff}} \\
    &= \frac{n(n-1)}{p^2} \ddf(n-2, p) \\
    &= n(n-1) \sum_{\ell=0}^{n-2} b(\ell, n, p) \ddg(\ell)
    \tag{\Cref{lem:sah-diff}} \\
    &= n(n-1) \sum_{\ell=0}^{n-2} b(\ell, n, p) (\dg(1+\ell) - \dg(\ell)) \\
    &= n(n-1) \sum_{\ell=0}^{n-2} b(\ell, n, p) (g(2+\ell) - g(1+\ell) -
    (g(1+\ell) - g(\ell))) \\
    &= n(n-1) \sum_{\ell=0}^{n-2} b(\ell, n, p) (g(2+\ell) + g(\ell) -
    2g(1+\ell)).
  \end{align*}
}
  If $g$ is convex, then
  \begin{align*}
    \frac{g(2+\ell) + g(\ell)}{2} \ge g(1+\ell).
  \end{align*}
  Similarly, if $g$ is concave, then
  \begin{align*}
    \frac{g(2+\ell) + g(\ell)}{2} \le g(1+\ell).
  \end{align*}
  This proves convexity or concavity. To show strictness, observe that under the
  given assumption, the term
  \begin{align*}
    b(0, n, p) ((g(2) - g(1)) - (g(1) - g(0)) \ne 0,
  \end{align*}
  meaning $f$ is strictly convex/concave in $p$.
\end{proof}

\begin{lemma}
  \label{lem:inc-product}
  Let $\{a_\ell\}_{\ell=1}^{k}$ be a positive, non-decreasing sequence. Let $f$
  be a concave, nonnegative, increasing function. Then, for
  any positive $c$,
  \begin{align*}
    \max_b
    &\sum_{\ell=1}^k a_\ell f(b_\ell) \\
    \text{s.t.}
    &\sum_{\ell} b_\ell \le c \\
     & \{b_\ell\}_{\ell=1}^{k} \text{ is non-increasing} 
  \end{align*}
  is attained by $b_\ell^* = c/k$ for all $\ell$. If $a_k > a_1$ or $f$ is
  strictly concave, then this is the unique maximizer.
\end{lemma}
\begin{proof}
  In the case where $a_k = a_1$, then all the $a_\ell$'s are equal and the claim
  is trivially true. From here, we consider the case where $a_k > a_1$.
  The Lagrangian is
  \ifthenelse{\boolean{smallEqs}}{
  \begin{align*}
    \mc{L}(b; \mu, \lambda)
    &= -\sum_{\ell=1}^k a_\ell f(b_\ell) + \mu \p{\sum_{\ell=1}^k b_\ell - c} \\
    &+ \sum_{\ell=1}^{k-1} \lambda_\ell (b_{\ell+1} - b_\ell).
  \end{align*}
}{\begin{align*}
    \mc{L}(b; \mu, \lambda)
    &= -\sum_{\ell=1}^k a_\ell f(b_\ell) + \mu \p{\sum_{\ell=1}^k b_\ell - c} +
    \sum_{\ell=1}^{k-1} \lambda_\ell (b_{\ell+1} - b_\ell).
  \end{align*}
}
  An optimal solution must satisfy stationarity: for all $\ell$,
  \begin{align*}
    \frac{\partial}{\partial b_\ell} \mc{L}(b; \mu, \lambda)
    &= 0.
  \end{align*}
  This means
  \begin{align*}
    -a_1 f'(b_1) + \mu - \lambda_1
    &= 0 \\
    -a_\ell f'(b_\ell) + \mu + \lambda_{\ell-1} - \lambda_\ell
    &= 0 \tag{$2 \le \ell \le k-1$} \\
    -a_k f'(b_k) + \mu + \lambda_{k-1}
    &= 0.
  \end{align*}
  Summing all these equations, we have
  \begin{align*}
    -\sum_{\ell=1}^k a_\ell f'(b_\ell) + k \mu
    &= 0 \\
    \mu = \frac{1}{k} \sum_{\ell=1}^k a_\ell f'(b_\ell).
  \end{align*}
  We will show that $\lambda_j > 0$ for all $j \le k-1$.
  Sum the first $j$ equations above to get
  \ifthenelse{\boolean{smallEqs}}{
  \begin{align*}
    &-\sum_{\ell=1}^j a_\ell f'(b_\ell) + j \mu - \lambda_j
    = 0 \\
    \lambda_j
    &= j \mu - \sum_{\ell=1}^j a_\ell f'(b_\ell) \\
    &= \frac{j}{k} \sum_{\ell=1}^k a_\ell f'(b_\ell) - \sum_{\ell=1}^j a_\ell
    f'(b_\ell) \\
    &= \frac{j}{k} \p{\sum_{\ell=1}^j a_\ell f'(b_\ell) + \sum_{j+1}^k a_\ell
    f'(b_\ell)} -
    \sum_{\ell=1}^j a_\ell f'(b_\ell) \\
    &= \frac{j}{k} \sum_{j+1}^k a_\ell f'(b_\ell) -
    \frac{k - j}{k} \sum_{\ell=1}^j a_\ell f'(b_\ell) \\
    &= \frac{j(k-j)}{k} \frac{1}{k-j} \sum_{j+1}^k a_\ell f'(b_\ell) -
    \frac{j(k - j)}{k} \frac{1}{j} \sum_{\ell=1}^j a_\ell f'(b_\ell) \\
    &> 0.
  \end{align*}
}{\begin{align*}
    -\sum_{\ell=1}^j a_\ell f'(b_\ell) + j \mu - \lambda_j
    &= 0 \\
    \lambda_j
    &= j \mu - \sum_{\ell=1}^j a_\ell f'(b_\ell) \\
    &= \frac{j}{k} \sum_{\ell=1}^k a_\ell f'(b_\ell) - \sum_{\ell=1}^j a_\ell
    f'(b_\ell) \\
    &= \frac{j}{k} \p{\sum_{\ell=1}^j a_\ell f'(b_\ell) + \sum_{j+1}^k a_\ell
    f'(b_\ell)} -
    \sum_{\ell=1}^j a_\ell f'(b_\ell) \\
    &= \frac{j}{k} \sum_{j+1}^k a_\ell f'(b_\ell) -
    \frac{k - j}{k} \sum_{\ell=1}^j a_\ell f'(b_\ell) \\
    &= \frac{j(k-j)}{k} \frac{1}{k-j} \sum_{j+1}^k a_\ell f'(b_\ell) -
    \frac{j(k - j)}{k} \frac{1}{j} \sum_{\ell=1}^j a_\ell f'(b_\ell) \\
    &> 0.
  \end{align*}
}
  The last step follows because both $a_\ell$ and $f'(b_\ell)$ are weakly
  increasing, and as long as either $a_k > a_1$ or $f'(b_k) > f(b_1)$, the
  overall sequence must strictly increase somewhere. The average of the last
  $k-j$ terms of a non-decreasing sequence that strictly increases somewhere
  must be strictly larger than the average of the first $j$ terms. Complementary
  slackness tells us that $\lambda_\ell > 0$ implies $b_\ell = b_{\ell+1}$ for
  all $\ell \le k-1$, meaning that the unique optimizer is $b_\ell^* = c/k$ for
  all $\ell$.
\end{proof}

\begin{lemma}
  \label{lem:1-x-bound}
  For $n > 0$ and $0 < x < 1/n$,
  \begin{align*}
    (1-x)^n \le 1 - nx + \binom{n}{2} x^2.
  \end{align*}
\end{lemma}
\begin{proof}
  \begin{align*}
    (1-x)^n
    &= \sum_{i=0}^\infty
    \binom{n}{i} (-x)^i
    = 1 - nx + \binom{n}{2} x^2 + \sum_{i=3}^\infty \binom{n}{i} (-x)^i.
  \end{align*}
  We need to show
  \begin{align*}
    \sum_{i=3}^\infty \binom{n}{i} (-x)^i \le 0.
  \end{align*}
  Observe that the first term is negative, and for $x < 1/n$, the absolute value
  of each term is decreasing. Therefore, the sum is bounded from above by 0.
\end{proof}

\section{Omitted proofs for Section~\ref{sec:inter-tool}}
\label{app:pairwise-proofs}

\dominantstrength*
\begin{proof}
  We will provide an example where $|\Pi| = 2$. Let $\bd = [1; ~ 1 -
  \varepsilon; ~ 1 - \varepsilon; ~ 0; ~\dots ~ 0]$. Let $\bd(\pi)$ denote $\bd$
  ordered according to $\pi$. We consider two GAITs $\pin{1}$ and $\pin{2}$ given by
  \ifthenelse{\boolean{smallEqs}}{
  \begin{align*}
    \pin{1}
    &= [1; ~ 4; ~ 5; ~ \dots ~ K; ~ 2; ~ 3] \\
    &\Longrightarrow
    \bd(\pin{1}) = [1; ~ 0; ~ \dots ~ 0; ~ 1 - \varepsilon; ~ 1 - \varepsilon] \\
    \pin{2}
    &= [2; ~ 3; ~ 4; ~ \dots ~ K; ~ 1] \\
    &\Longrightarrow
    \bd(\pin{2}) = [1 - \varepsilon; ~ 1 - \varepsilon; ~ 0; ~ \dots ~ 0; ~ 1]
  \end{align*}
}{\begin{align*}
    \pin{1}
    &= [1; ~ 4; ~ 5; ~ \dots ~ K; ~ 2; ~ 3]
    &\Longrightarrow&
    & \bd(\pin{1}) &= [1; ~ 0; ~ \dots ~ 0; ~ 1 - \varepsilon; ~ 1 - \varepsilon] \\
    \pin{2}
    &= [2; ~ 3; ~ 4; ~ \dots ~ K; ~ 1]
    &\Longrightarrow&
    & \bd(\pin{2}) &= [1 - \varepsilon; ~ 1 - \varepsilon; ~ 0; ~ \dots ~ 0; ~ 1]
  \end{align*}
}
  First, observe that $\minm_1(1, \pi, \bd) = 1$. This is because the
  distribution $\bp_1 = [1; ~ 0; ~ \dots ~ 0] \in \Delta(\pin{1})$ yields utility
  1. In contrast, a single player's optimal strategy in $\Delta(\pin{2})$ is $[0;
  ~ 1; ~ 0; ~ \dots ~ 0]$, which yields utility $1 - \varepsilon$.

  Next, we will show that there exists $n$ such that $\maxm_1(n, \Pi, \bd) <
  \minm_2(n, \Pi, \bd)$. In particular, choose $n=3$. Assume towards
  contradiction that there is an equilibrium in which $m_1 \ge 2$. For
  sufficiently large $K$, these $m_1$ players must choose the equilibrium
  strategy $\bp = [1; ~ 0; ~ \dots ~ 0]$, yielding utility at most $1/m_1$. This
  is because the constraint $\bp \in \Delta(\pin{1})$ means that $\bp_1 +
  \bp_2(K-2) \le 1$ and $\bp_1 + \bp_3(K-1) \le 1$, so for large $K$, it is
  optimal to choose $\bp_2 = \bp_3 = 0$.

  Since $m_2 \le 1$, if a player chooses $\pin{2}$, their equilibrium strategy
  $\bq$ satisfies $\bq_2 + \bq_3 = 1$ for a similar reason: for $\bq \in
  \Delta(\pin{2})$, $\bq_2 + \bq_3 + \bq_1(K-2) \le 1$. Again, for large $K$, it
  is optimal to choose $\bq_1 = 0$. This means that all players who choose
  $\pin{1}$ get utility $1/m_1$. They could deviate to choosing $\pin{2}$ with $\bq'
  = [0; ~ \nicefrac{1}{2}; ~ \nicefrac{1}{2}; ~ 0; ~ \dots ~ 0]$, which would
  yield utility at least
  \ifthenelse{\boolean{smallEqs}}{
  \begin{align*}
    &(1-\varepsilon) \frac{1}{2} \p{\E{\frac{1}{1 + X(1, \bq_2)}} + \E{\frac{1}{1
    + X(1, \bq_3)}}} \\
    &= \frac{1-\varepsilon}{2} \p{1 - \bq_2 + \frac{\bq_2}{2} + 1 - \bq_3 +
    \frac{\bq_3}{2}} \\
    &= \frac{1-\varepsilon}{2} \p{2 - \frac{\bq_2}{2} - \frac{\bq_3}{2}} \\
    &= \frac{1-\varepsilon}{2} \p{2 - \frac{1}{2}} \tag{$\bq_2 + \bq_3 = 1$} \\
    &= \frac{3(1-\varepsilon)}{4} \\
    &> \frac{1}{m_1} 
  \end{align*}
}{\begin{align*}
    (1-\varepsilon) \frac{1}{2} \p{\E{\frac{1}{1 + X(1, \bq_2)}} + \E{\frac{1}{1
    + X(1, \bq_3)}}}
    &= \frac{1-\varepsilon}{2} \p{1 - \bq_2 + \frac{\bq_2}{2} + 1 - \bq_3 +
    \frac{\bq_3}{2}} \\
    &= \frac{1-\varepsilon}{2} \p{2 - \frac{\bq_2}{2} - \frac{\bq_3}{2}} \\
    &= \frac{1-\varepsilon}{2} \p{2 - \frac{1}{2}} \tag{$\bq_2 + \bq_3 = 1$} \\
    &= \frac{3(1-\varepsilon)}{4} \\
    &> \frac{1}{m_1} 
  \end{align*}
}
  for sufficiently small $\varepsilon$. By contradiction, there is no pure
  partially symmetric equilibrium in which $m_1 \ge 2$, meaning $\maxm_1(3, \Pi,
  \bd) \le 1$. Finally, we will show that there is an pure partially symmetric
  equilibrium in which $m_1 = 1$ and $m_2 = 2$. Let
  \begin{align*}
    \bp^+ = [1; ~ 0; ~ \dots ~ 0]
    &&
    \bq^+ = [0; ~ \nicefrac{1}{2}; ~ \nicefrac{1}{2}; ~ 0; ~ \dots ~ 0].
  \end{align*}
  With $m_1 = 1$ and $m_2 = 2$, the player who chooses $\pin{1}$ gets utility 1,
  and players who choose $\pin{2}$ each get utility $3(1-\varepsilon)/4$. No
  player can strictly improve their utility by either switching to the other
  tool or changing their distribution. Therefore, $\maxm_1(3, \Pi, \bd) <
  \minm_2(3, \Pi, \bd)$.
\end{proof}

\strict*
\begin{proof}

  Without loss of generality, let $\pin{1} = [1; ~ 2; ~ \dots ~ K]$ and let
  $\bd_{\hat k} = 0$. Let $\pin{2} = [\hat k; ~ 1; ~ \dots ~ \hat k-1; ~ \hat k+1;
  ~ K]$. Our goal is to show that $\maxm_2(n, \Pi, \bd) = 0$ for all $n$, $\Pi
  \subseteq \{\pin{1}, \pin{2}\}$. To do so, assume towards contradiction that at
  equilibrium, some player $i$ chooses $\pin{2}$ with distribution $\bp$. Then,
  player $i$ could strictly improve their utility by:
  \begin{itemize}
    \item switching to $\pin{1}$, and
    \item transferring all probability mass from $\bp_{\hat k}$ to $\bp_1$.
  \end{itemize}
  This is because for all $k \ne \hat k$, player $i$ still puts at least as much
  probability mass on type $k$. Unless $\bp$ is the uniform distribution, at
  least one of these increases must be strict, and because $s(x) < \infty$,
  this strictly increases expected utility. If $\bp$ is the uniform
  distribution, then transferring probability mass from $\hat k$ to 1 strictly
  increases utility. Therefore, there can be no equilibrium in which a player
  chooses $\pin{2}$.
\end{proof}

\perfect*
\begin{proof}
  
  Without loss of generality, let $\pin{1} = [1; ~ 2; ~ \dots ~ K]$. First, we
  consider the case where $\pin{1}$ is a perfect ordering. For a given $n$, let
  $\bp$ be the equilibrium distribution for $\game(n, \bd, s)$ when players
  can only use $\pin{1}$. We will show that this remains an equilibrium even when
  they are given the option to switch to other GAITs.

  To do so, recall that by \Cref{lem:unique-eq-decreasing}, $\bp$ satisfies
  \begin{align*}
    U(\bp; n, \bd, s) \ge \bd_k u(n-1, \bp_k)
  \end{align*}
  for all $k$. Suppose a player deviates to some distribution $\bp'$ via some
  other GAIT $\pi'$. Given that the remaining $n-1$ players are all playing
  $\bp$, the deviating player's utility is
  \ifthenelse{\boolean{smallEqs}}{
  \begin{align*}
    \sum_{k \in [K]} \bp_k' \bd_k u(n-1, \bp_k)
    &\le \sum_{k \in [K]} \bp_k' U(\bp; n, \bd, s) \\
    &= U(\bp; n, \bd, s).
  \end{align*}
}{\begin{align*}
    \sum_{k \in [K]} \bp_k' \bd_k u(n-1, \bp_k)
    \le \sum_{k \in [K]} \bp_k' U(\bp; n, \bd, s)
    = U(\bp; n, \bd, s).
  \end{align*}
}
  Thus, no player can strictly improve their utility by deviating, meaning that
  $\bp$ is an equilibrium.

  Next, if $\pin{1}$ is not a perfect ordering, then let $\hat k$ be such that
  $\bd_{\hat k-1} < \bd_{\hat k}$. Assume towards contradiction that there
  exists an equilibrium in which all players use $\pin{1}$. This means they must
  play the strategy $\bp = \peq(n, \bd, s) = \peq(n, \tbd, s)$ where $\tbd$ is
  the result of the application of the Pool Adjacent Violators algorithm to
  $\bd$ (\Cref{thm:unique-eq}). Let $n_0$ be large enough that $\peq(n_0, \tbd,
  s)_{\hat k} > 0$. (Such an $n_0$ must exist by a similar argument to the one
  used in the proof of \Cref{lem:n-to-infty}.) Using
  \Cref{lem:unique-eq-decreasing}, as long as $n \ge n_0$,
  \ifthenelse{\boolean{smallEqs}}{
  \begin{align*}
    \bd_{\hat k} u(n-1, \bp_{\hat k})
    > \tbd_{\hat k} u(n-1, \bp_{\hat k})
    &= U(\bp; n, \tbd, s) \\
    &= U(\bp; n, \bd, s).
  \end{align*}
}{\begin{align*}
    \bd_{\hat k} u(n-1, \bp_{\hat k})
    > \tbd_{\hat k} u(n-1, \bp_{\hat k})
    = U(\bp; n, \tbd, s)
    = U(\bp; n, \bd, s).
  \end{align*}
}
  Let $\pin{2}$ rank $\hat k$ first. Then, a player could deviate to choosing
  $\pin{2}$ and place all of their probability mass on $\hat k$ to get utility
  $\bd_{\hat k} u(n-1, \bp_{\hat k}) > U(\bp; n, \bd, s)$. Therefore, there is
  no equilibrium in which all players use $\pin{1}$.
\end{proof}

\section{Additional Details and Figures for \Cref{sec:empirical}}
\label{app:empirical}

\subsection{Additional figures}
\label{app:additional-figures}

\ifdefined\smallfigs
\begin{figure}[ht]
  \centering
  \begin{subfigure}[ht]{0.49\textwidth}
    \centering
    \includegraphics[width=\textwidth]{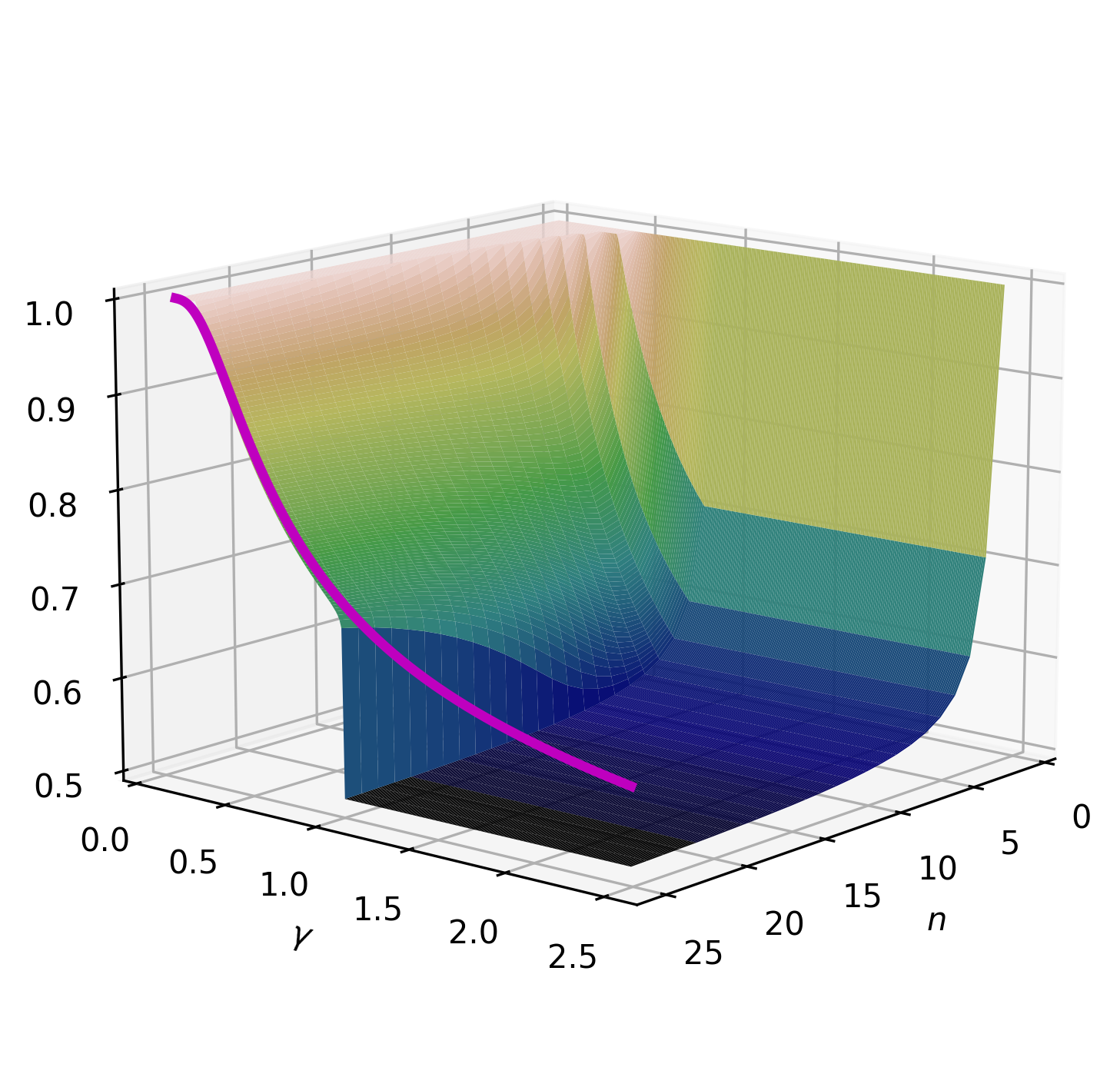}
    \caption{$\popt(n, \bd, s_\gamma)_1$ for $\bd = [5; ~ 2]$.}
    \label{fig:ps-to-inf-opt}
  \end{subfigure}
  \hfill
  \begin{subfigure}[ht]{0.49\textwidth}
    \centering
    \includegraphics[width=\textwidth]{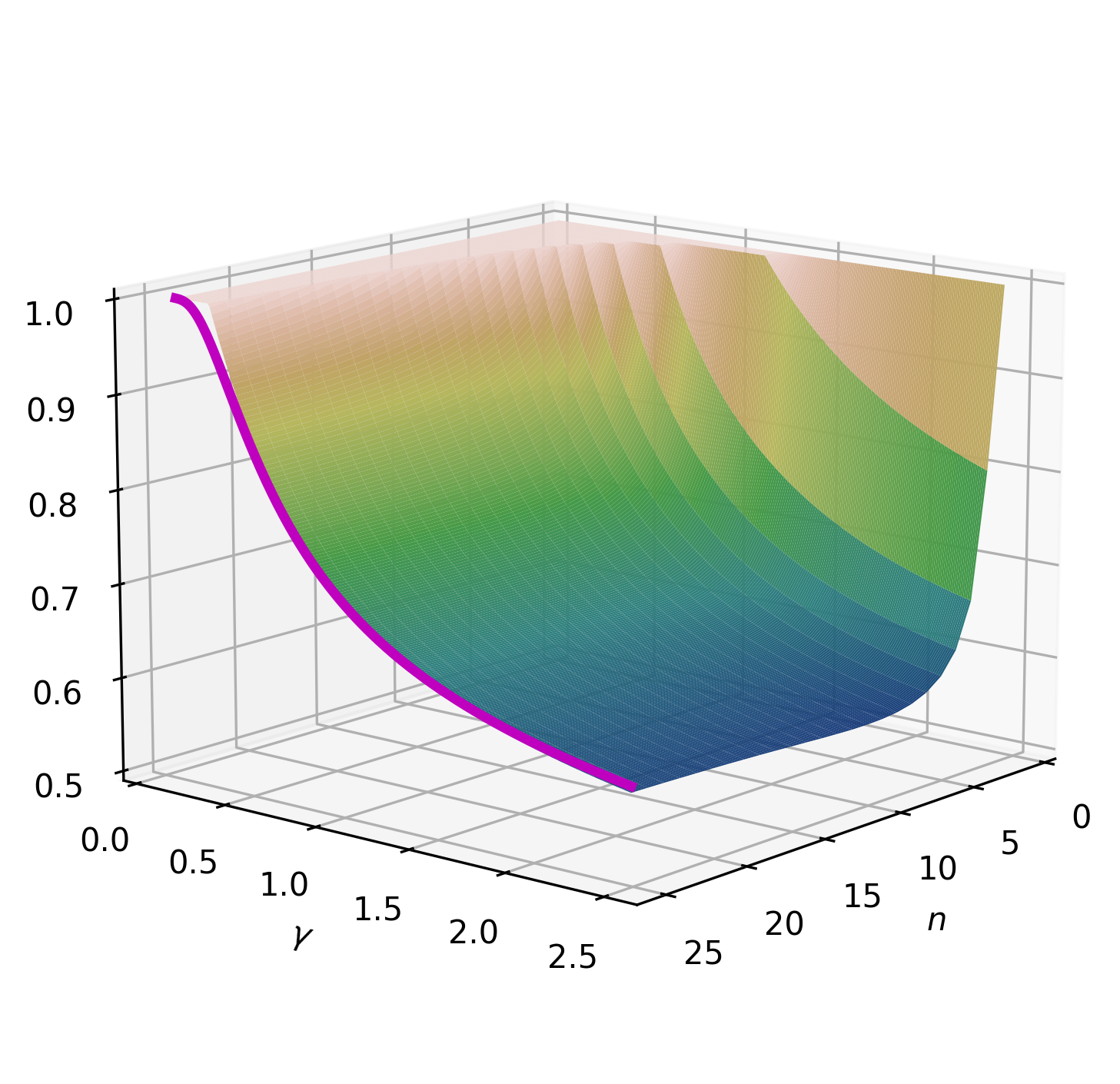}
    \caption{$\peq(n, \bd, s_\gamma)_1$ for $\bd = [5; ~ 2]$.}
    \label{fig:ps-to-inf-eq}
  \end{subfigure}
  \caption{For further intuition on \Cref{lem:n-to-infty},
    we show how $\popt(n, \bd, s_\gamma)_1$ and $\peq(n, \bd, s_\gamma)_1$
    respectively vary with both $n$ and $\gamma$ for the instance given by $\bd
    = [5; ~ 2]$. For $\gamma < 1$, both $\popt$ and $\peq$ converge to $\bplim$,
    shown in magenta. For $\gamma \ge 1$, $\popt$ converges to the uniform
    distribution, with a discontinuity at $\gamma = 1$. Further, while $\peq$ is
    monotone in both $n$ and $\gamma$ (as predicted by
    \Cref{thm:diversity-n,thm:diversity-gamma}), $\popt$ is not. Increasing $n$
  can lead to a \textit{less} diverse symmetric socially optimal strategy.}
  \label{fig:ps-to-inf}
\end{figure}
\fi

\def\panelwidth{0.45}
\begin{figure}[ht]
    \centering
    \begin{subfigure}[t]{\panelwidth\textwidth}
        \centering
        \includegraphics[width=\textwidth]{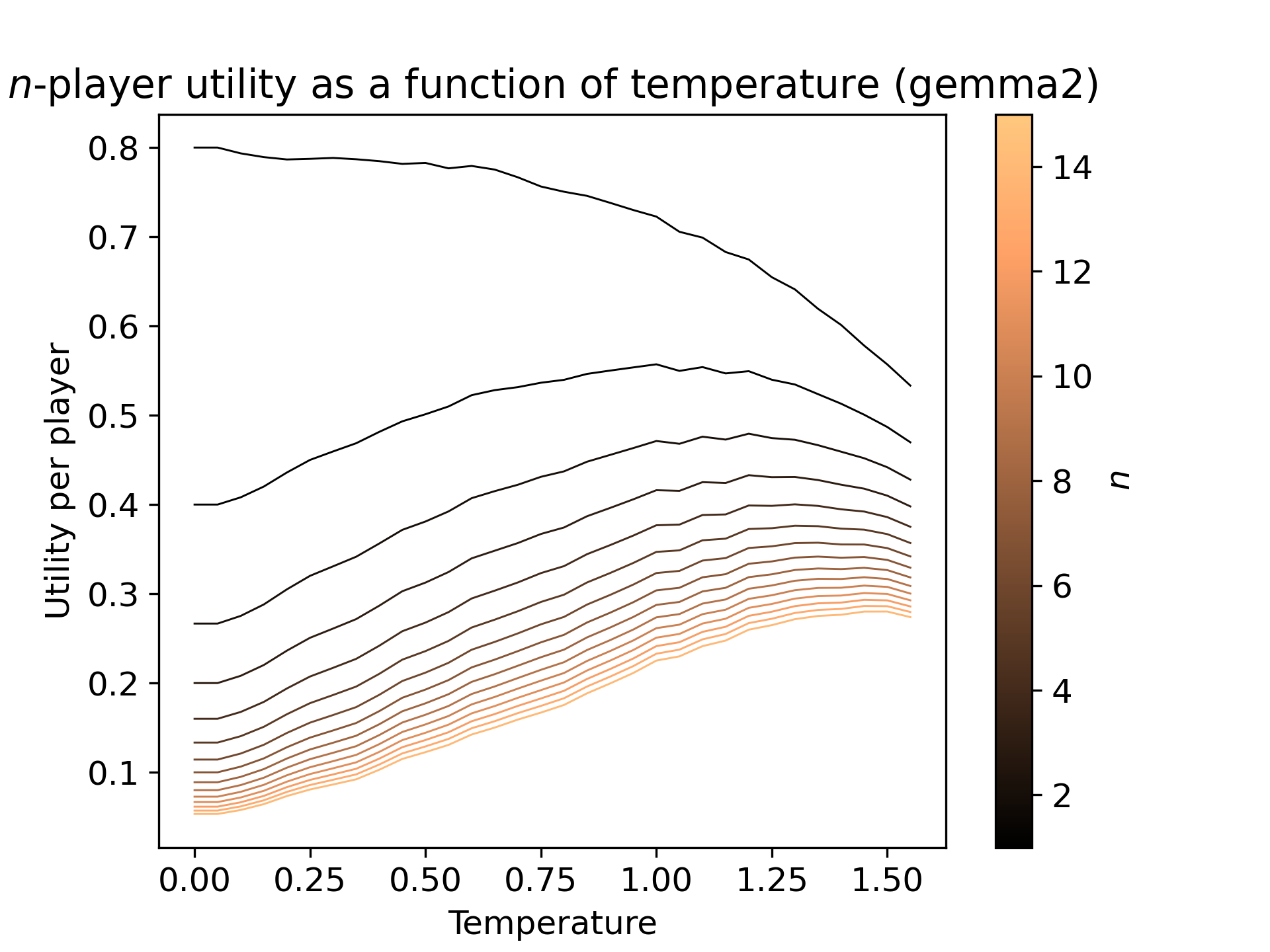}
    \end{subfigure}
    \hfill
    \begin{subfigure}[t]{\panelwidth\textwidth}
        \centering
        \includegraphics[width=\textwidth]{./img/llama3.1_opt_over_temp.png}
    \end{subfigure}

    \begin{subfigure}[t]{\panelwidth\textwidth}
        \centering
        \includegraphics[width=\textwidth]{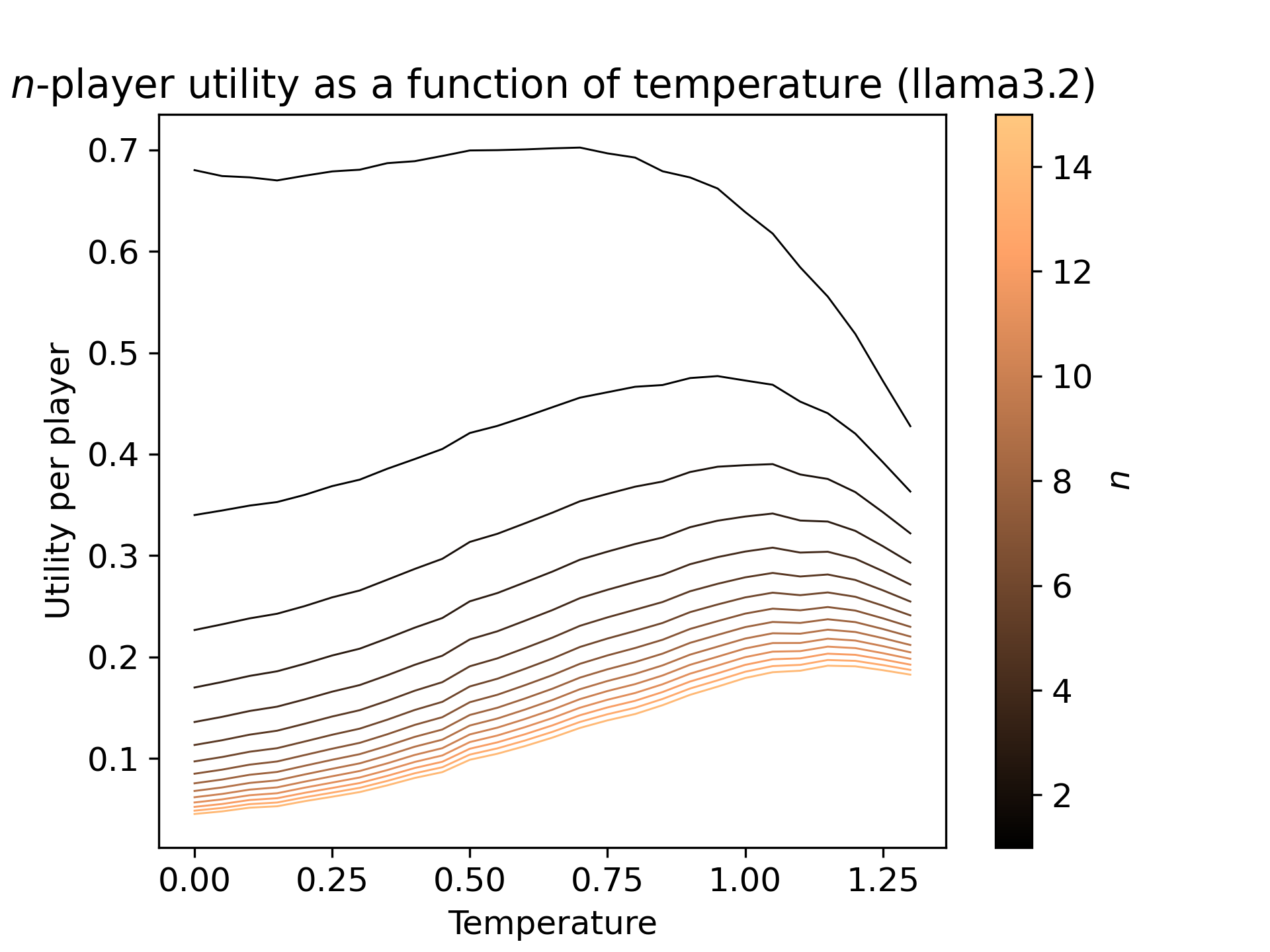}
    \end{subfigure}
    \hfill
    \begin{subfigure}[t]{\panelwidth\textwidth}
        \centering
        \includegraphics[width=\textwidth]{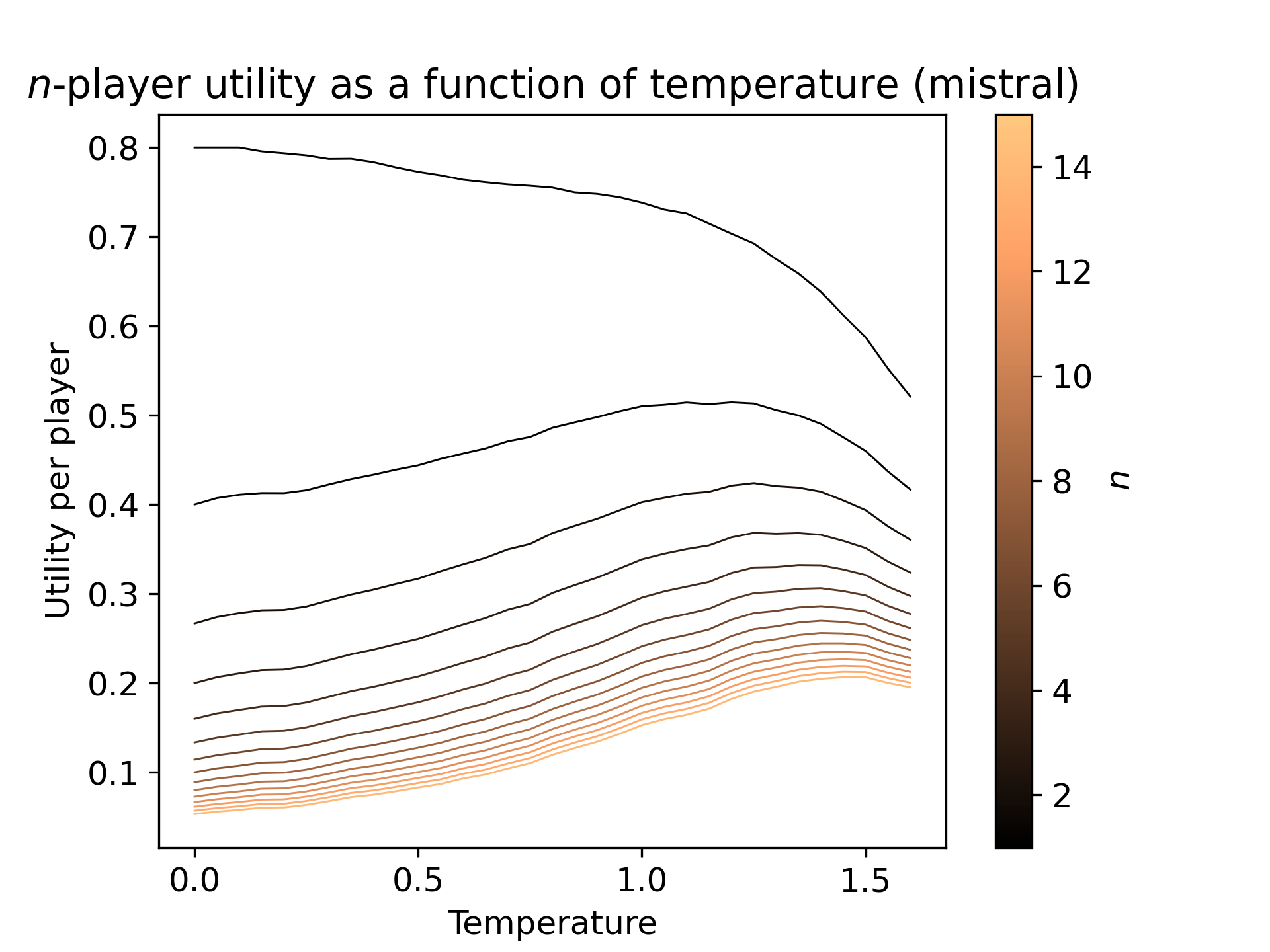}
    \end{subfigure}

    \begin{subfigure}[t]{\panelwidth\textwidth}
        \centering
        \includegraphics[width=\textwidth]{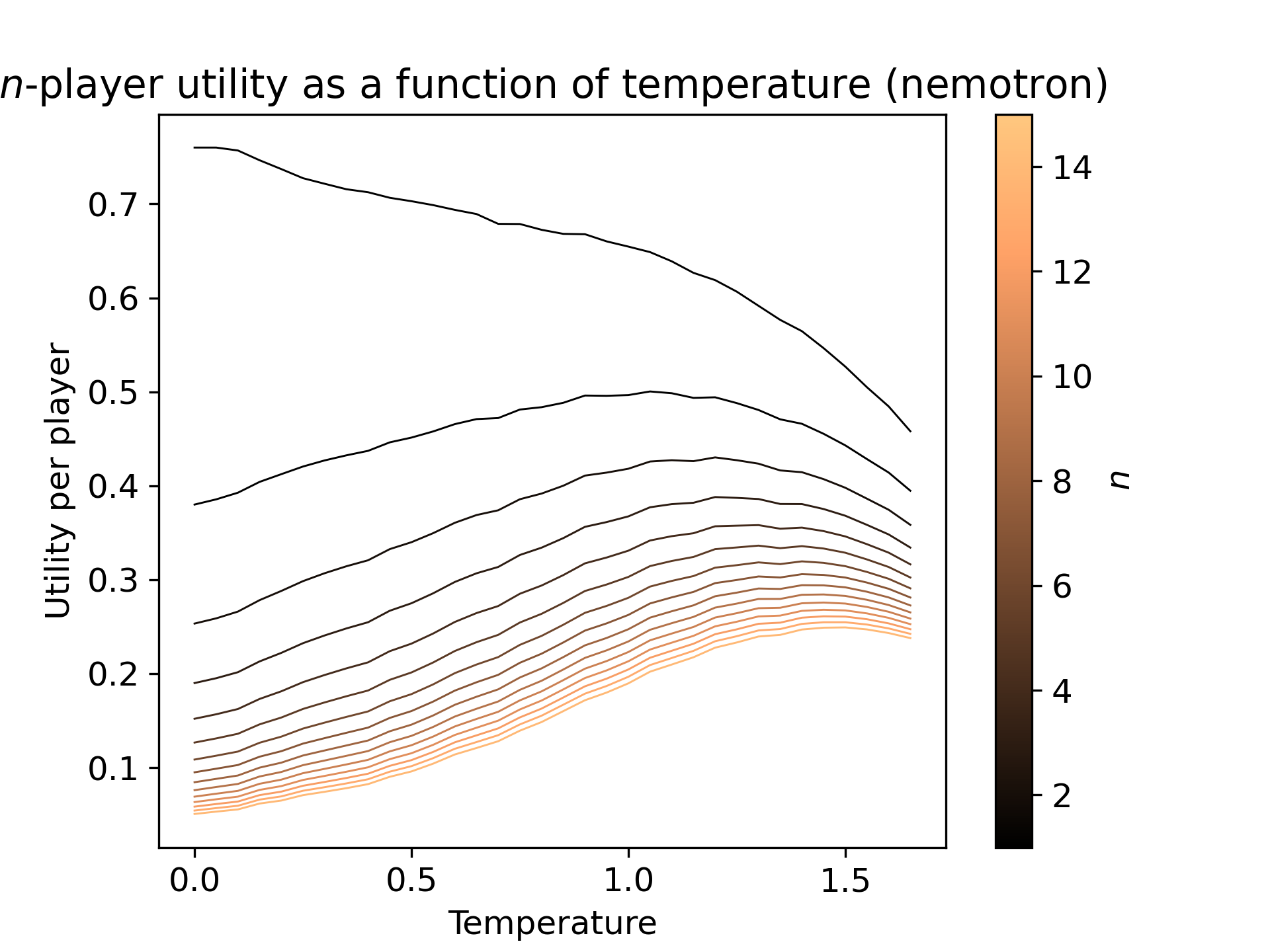}
    \end{subfigure}
    \hfill
    \begin{subfigure}[t]{\panelwidth\textwidth}
        \centering
        \includegraphics[width=\textwidth]{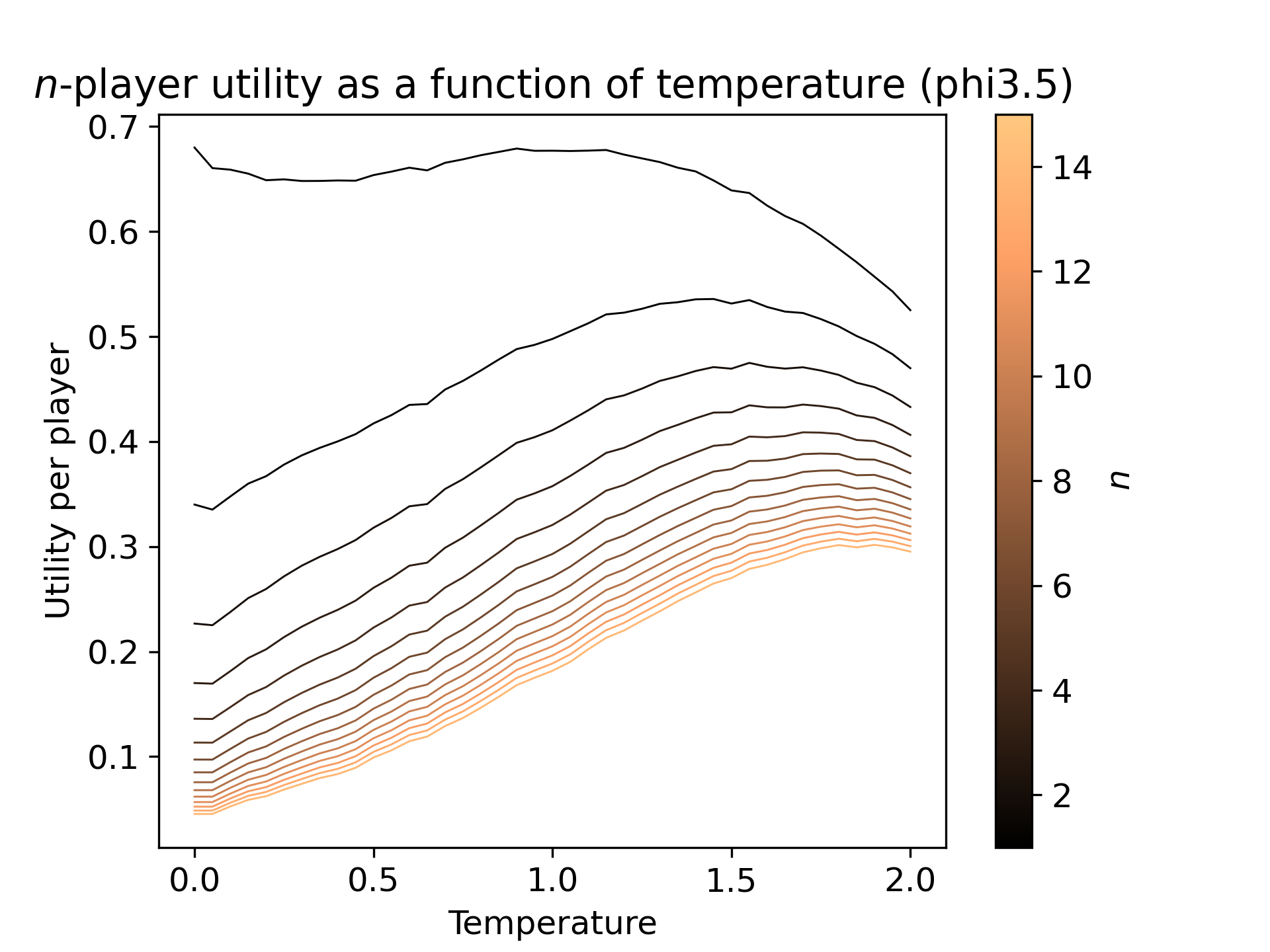}
    \end{subfigure}

    \caption{Social welfare as a function of $\tau$ for different values of $n$.
    $\gamma = 1.0$.}
    \label{fig:utility_over_temp}
\end{figure}

\ifdefined\smallfigs
\Cref{fig:ps-to-inf} visualizes \Cref{lem:n-to-infty} and shows that $\popt$ can
be non-monotone.
\fi
\Cref{fig:utility_over_temp} replicates \Cref{fig:llama3.1-opt-over-temp} for
all the language models we use.
\Cref{fig:opt_and_eq_temp_over_gamma,fig:opt_and_eq_welfare_over_gamma,fig:opt_and_eq_util_over_gamma},
provide analogous results to
\Cref{fig:opt_and_eq_temp_over_n,fig:opt_and_eq_util_over_n}
respectively using $\gamma$ instead of $n$. Because social welfare is different
from the sum of players' utilities when $\gamma > 1.0$, we plot them separately.
\Cref{fig:all-pairwise} shows pairwise equilibria for all pairs of language
models in a style similar to \Cref{fig:market-shares}.
We provide the prompts we use in
\Cref{tab:default-prompt,tab:verification-prompt}. For \gemma, we omit the
``System'' role because \gemma\ does not support it.

\begin{figure}[ht]
  \centering
  \includegraphics[width=0.8\textwidth]{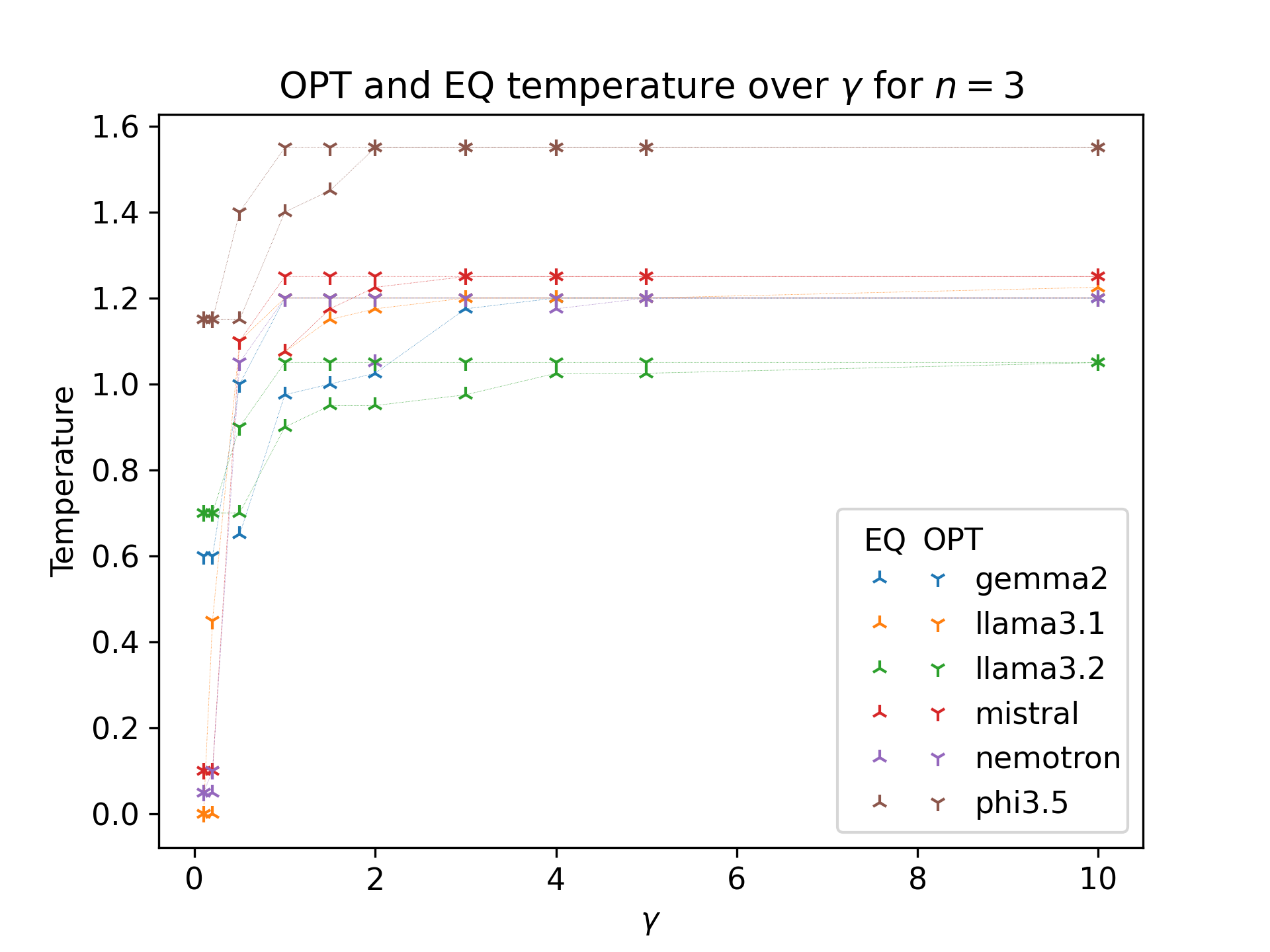}
  \caption{Socially optimal and equilibrium temperatures for each language model
    as a function of $\gamma$. Observe that (1) all curves are increasing in
    $\gamma$, and (2) equilibrium temperatures are lower than their socially
    optimal counterparts. Optimal and equilibrium temperatures converge for
    large $\gamma$, as predicted by \Cref{lem:opt-inf-equivalence}. Note that
    \llamaone\ appears to have higher equilibrium temperature than optimal
    temperature for $\gamma \ge 5$. This is an artifact of discretization of the
  temperature space. $n=3$.}
  \label{fig:opt_and_eq_temp_over_gamma}
\end{figure}

\begin{figure}[ht]
  \centering
  \includegraphics[width=0.8\textwidth]{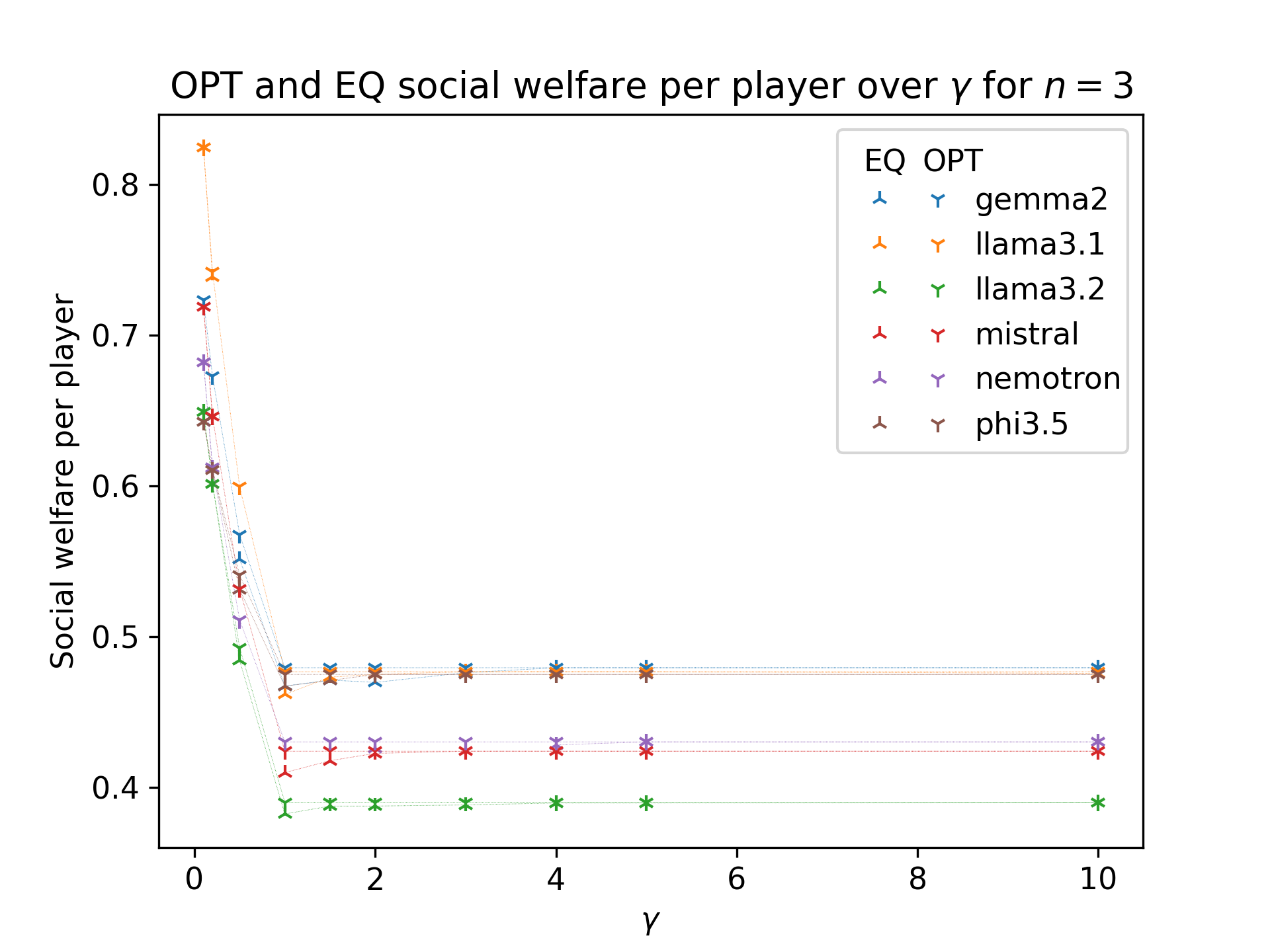}
  \caption{Social welfare for each
    language model at optimal and equilibrium temperatures
    as a function of $\gamma$. Note that for $\gamma \ge 1$, optimal social
    welfare is constant.
    Observe that (1) social welfare is decreasing in
    $\gamma$ for $\gamma \le 1$, and (2) equilibrium welfare is not too far from
    optimal. The nonmonotonicity at $\gamma = 1$ is due to the fact that our
    social welfare changes at this point. Interestingly, social welfare for
    \gemma, \phithree, and \llamaone\ are nearly identical for large values of
    $\gamma$ despite the fact that they yield very different per-player
    utilities (\Cref{fig:opt_and_eq_util_over_gamma}). $n=3$.}
  \label{fig:opt_and_eq_welfare_over_gamma}
\end{figure}

\begin{figure}[ht]
  \centering
  \includegraphics[width=0.8\textwidth]{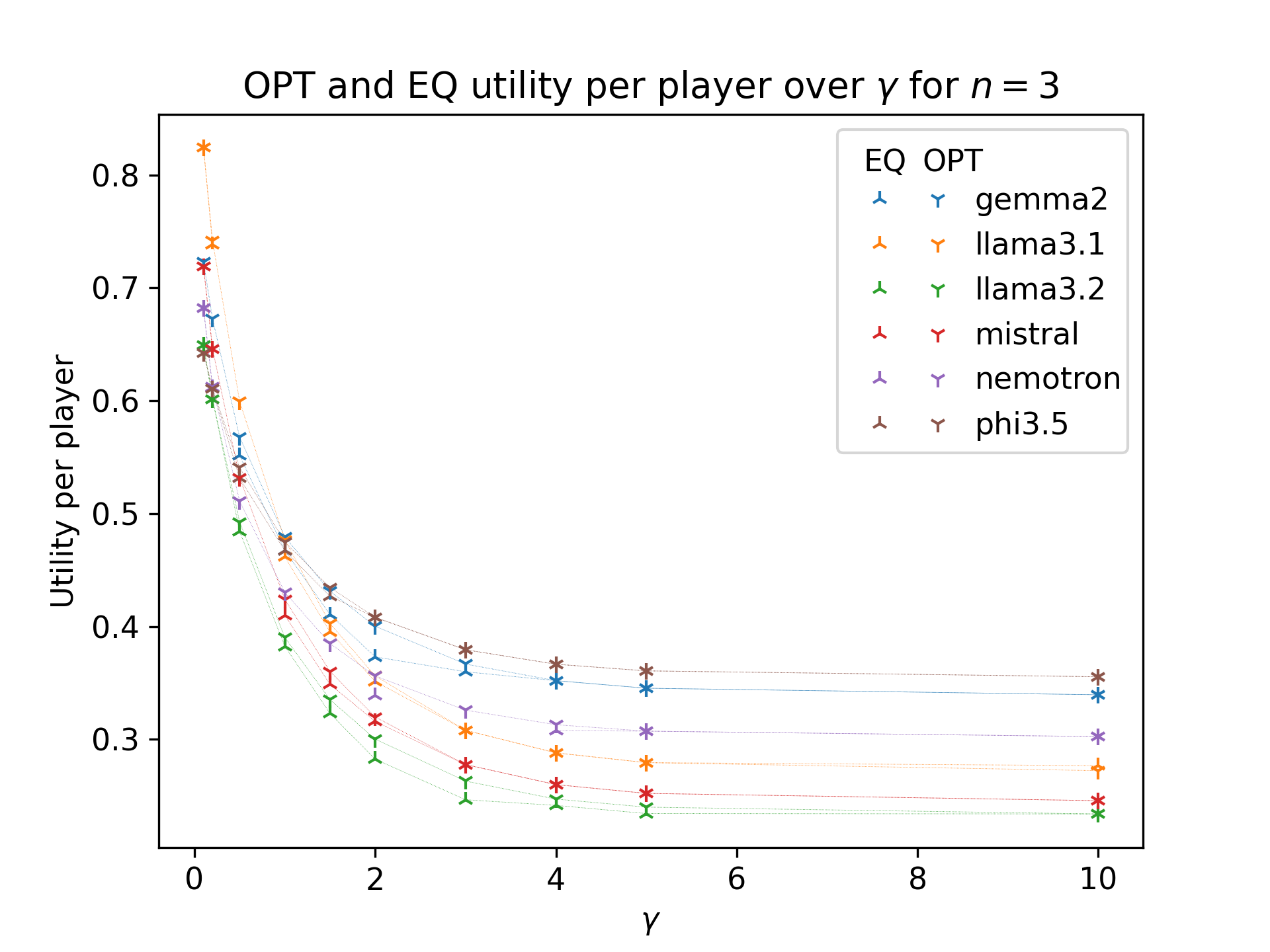}
  \caption{Socially optimal and equilibrium utility (per player) for each
    language model as a function of $\gamma$. Note that for $\gamma > 1$, social
    welfare (\Cref{fig:opt_and_eq_welfare_over_gamma}) is strictly larger than
    the sum of player utilities. Observe that (1) utility is decreasing in
    $\gamma$, and (2) equilibrium welfare is not too far from optimal. Note that
    \phithree\ and \gemma\ seem to perform relatively better in the presence of
    stronger competition. $n=3$.}
  \label{fig:opt_and_eq_util_over_gamma}
\end{figure}

\begin{figure}[ht]
  \centering
  \includegraphics[width=0.91\textwidth]{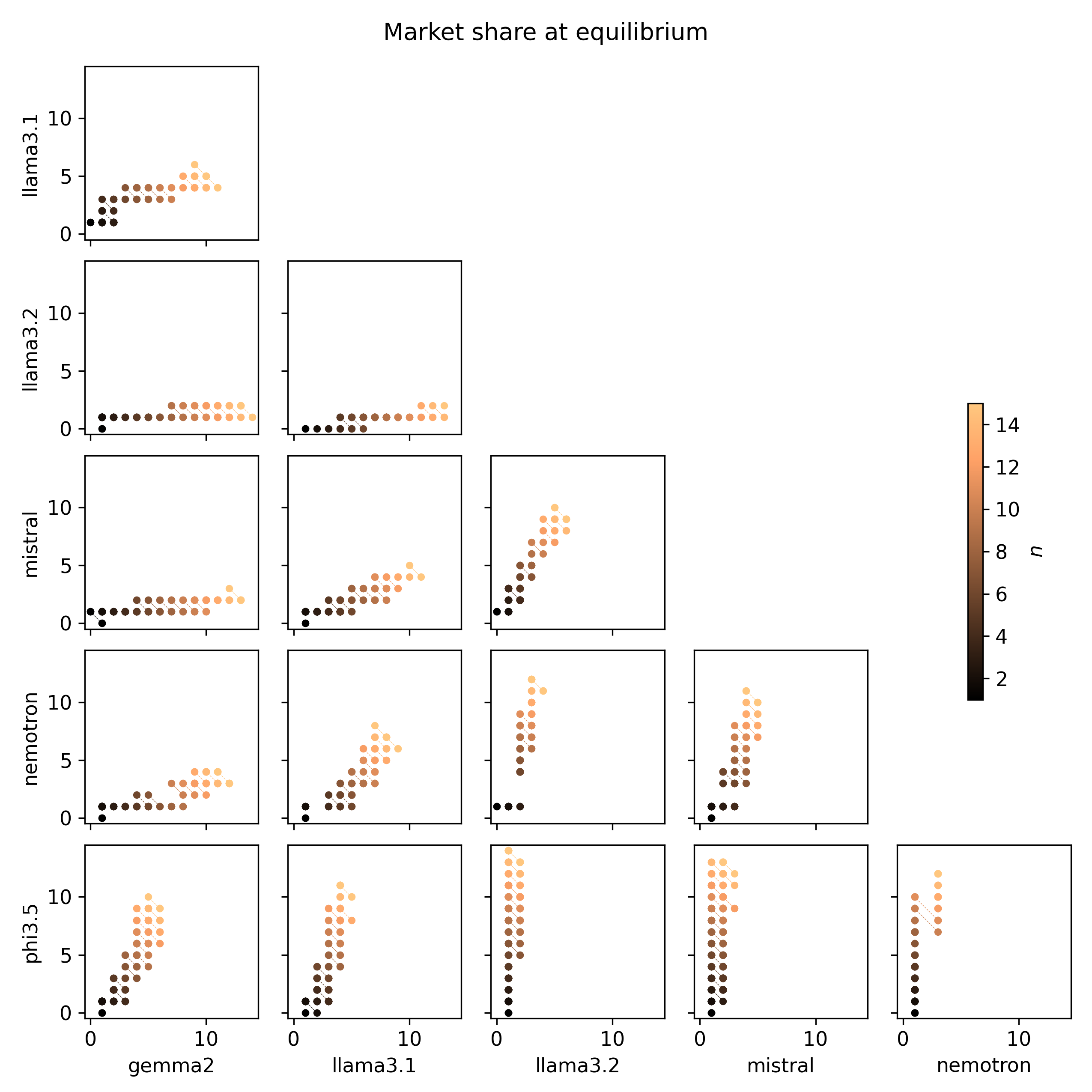}
  \caption{Pairwise equilibria for all pairs of GAITs. $\gamma = 1.0$.}
  \label{fig:all-pairwise}
\end{figure}

\begin{figure}[ht]
  \centering
  \includegraphics[width=0.8\textwidth]{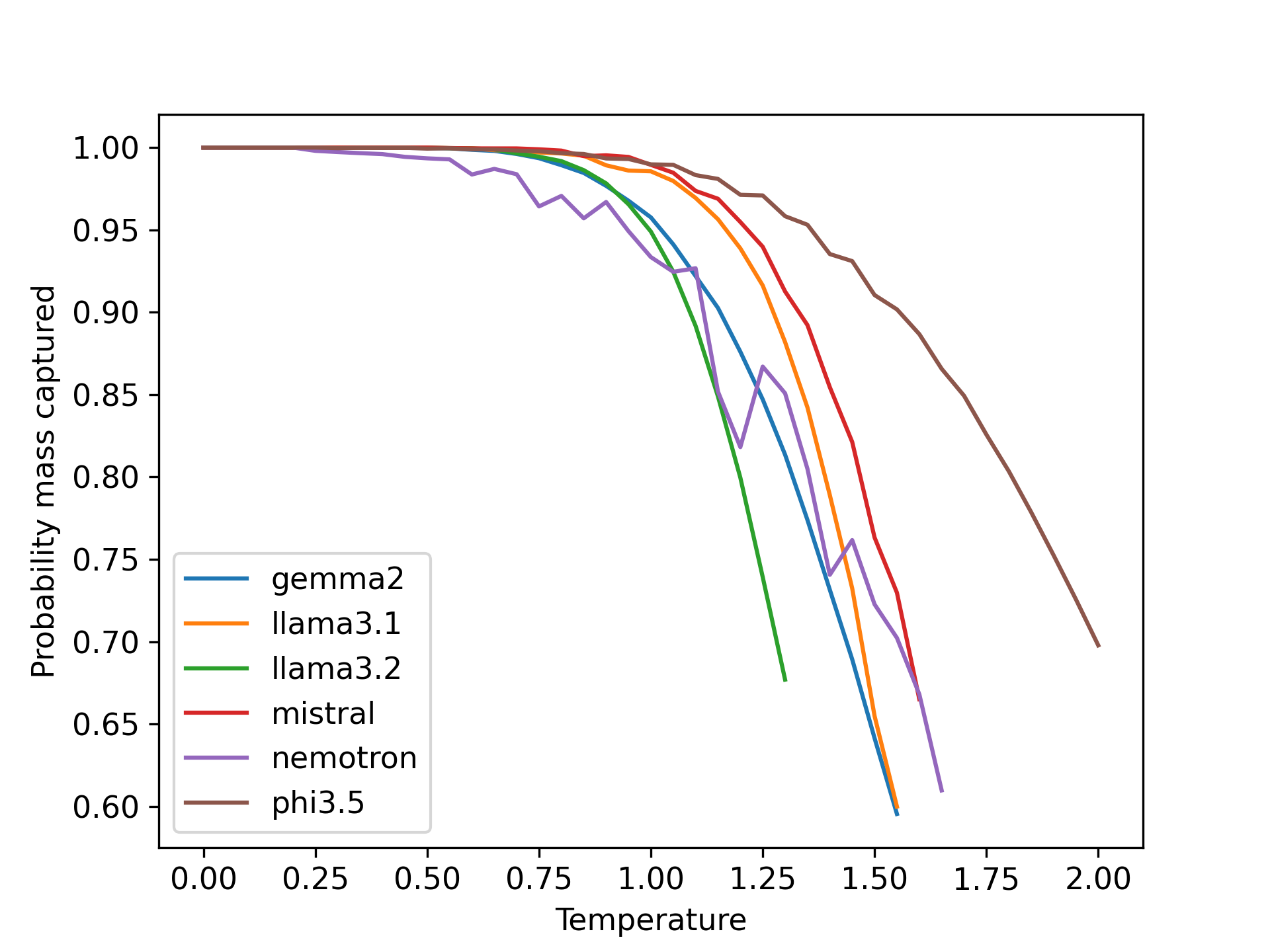}
  \caption{Probability mass captured by our sample at different temperatures.}
  \label{fig:mass-captured}
\end{figure}

\subsection{Experimental details for \Cref{sec:empirical}}
\label{app:exp-details}

\begin{sidewaystable}
  \centering
  \begin{tabular}{llclcc}
    \toprule
    Model name
    & Hugging Face model ID
    & \# params
    & Ref.
    & $\tau_{\max}$
    & $|\mc V|$
    \\
    \midrule 
    \gemma
    & \href{https://huggingface.co/google/gemma-2-2b-it}{google/gemma-2-2b-it}
    & 2.6B
    & \citep{gemma_2024,team2024gemma}
    & $1.55$
    & 169,069
    \\
    \llamaone
    & \href{https://huggingface.co/meta-llama/Llama-3.1-8B-Instruct}{meta-llama/Llama-3.1-8B-Instruct}
    & 8.0B
    & \citep{dubey2024llama}
    & $1.55$
    & 72,399
    \\
    \llamatwo
    &
    \href{https://huggingface.co/meta-llama/Llama-3.2-1B-Instruct}{meta-llama/Llama-3.2-1B-Instruct}
    & 1.2B
    & \citep{dubey2024llama}
    & $1.3$
    & 72,399
    \\
    \mistral
    &
    \href{https://huggingface.co/mistralai/Mistral-7B-Instruct-v0.3}{mistralai/Mistral-7B-Instruct-v0.3}
    & 7.2B
    & \citep{jiang2023mistral}
    & $1.6$
    & 25,093
    \\
    \nemotron
    &
    \href{https://huggingface.co/nvidia/Nemotron-Mini-4B-Instruct}{nvidia/Nemotron-Mini-4B-Instruct}
    & 4.2B
    & \citep{parmar2024nemotron,minitron2024}
    & $1.65$
    & 169,159
    \\
    \phithree
    &
    \href{https://huggingface.co/microsoft/Phi-3.5-mini-instruct}{microsoft/Phi-3.5-mini-instruct}
    & 3.8B
    & \citep{abdin2024phi}
    & $2.0$
    & 24,158
    \\
    \qwen
    &
    \href{https://huggingface.co/Qwen/Qwen2.5-7B-Instruct}{Qwen/Qwen2.5-7B-Instruct}
    & 7.6B
    & \citep{qwen2.5}
    & --
    & --
    \\
    \bottomrule
  \end{tabular}
  \caption{Language model details. $\tau_{\max}$ is the maximum temperature at
    which we sample using that model. $\mc V$ is the set of tokens we allow
    (i.e., those that only contain alphanumeric characters plus the built
    in end-of-message character). We use \qwen\ only for answer verification,
  not for sample generation.}
  \label{tab:models}
\end{sidewaystable}

\paragraph*{Using samples optimally.} Given a finite number of samples for each
(language model, temperature, instance) combination, we could naively estimate
utilities in various competitive settings by resampling from that sample and
simulating games. Critically, we must resample \textit{without replacement} when
simulating a game; if a particular answer appears once in our sample, resampling
with replacement would allow two competing players to both sample it, which
would bias our estimates of utility downwards. To simulate $n$ players using a
finite sample $\mc K$ of size $|\mc K| = S$ drawn i.i.d. from $\bp$, we would
randomly partition the sample into $n$ subsets, and simulate $\lfloor S/n
\rfloor$ games where we use each sample at most once. Of course, we could repeat
this random partitioning step to reduce variance.

In fact, a variant of this scheme turns out to be optimal. Let $K_1, \dots, K_n$
be random variables representing realized types of the $n$ players in our
simulation, where each $K_i$ is drawn without replacement from $\mc K$. Let
$\E{g(K_1, \dots, K_n)}$ be the average utility of the $n$ players with these
sampled responses and parameters $\bd, s$. Note that $g$ is symmetric in its
arguments. We are interested in the quantity
\begin{align*}
  U(\bp; n, \bd, s)
  &= \EE{K_1, \dots, K_n \sim \bp}{g(K_1, \dots, K_n)},
\end{align*}
but we cannot compute it without access to $\bp$. However, given a finite
sample $\mc K$ drawn i.i.d. from $\bp$, the
minimum-variance~\citep{rao1945information,blackwell1947conditional} estimator
for $U(\bp)$ is the U-statistic~\citep{hoeffding1948class}
\begin{align*}
  \widehat U(\mc K) \triangleq \EE{K_1, \dots, K_n \stackrel{\text{w/o rep.}}{\sim}
  \mc K}{g(K_1, \dots, K_n)}.
\end{align*}
Note that the $K_i$'s are drawn without replacement from $\mc K$. We can
directly compute $\widehat U(\mc K)$ as follows. For each type $k$, let $a_k$ be the
number of times $k$ appears in $\mc K$. Suppose that player 1 samples type $\hat
k$, i.e., $K_1 = \hat k$. The probability that exactly $\ell$ of the remaining
$n-1$ players also sample type $\hat k$ is
\begin{align*}
  \frac{\binom{a_{\hat k} - 1}{\ell} \binom{S - a_{\hat k}}{n - 1 -
  \ell}}{\binom{S-1}{n-1}}.
\end{align*}
By symmetry,
\begin{align*}
  \widehat U(\mc K)
  &= \sum_{k \in K} \bd_k \EE{K_1, \dots, K_n \stackrel{\text{w/o rep.}}{\sim}
  \mc K}{\frac{\ind{K_1=k}}{s\p{\sum_{i \in [n]} \ind{K_i = k}} }} \\
  &= \sum_{k \in \mc K} \frac{a_k \bd_k}{S}
  \sum_{\ell=0}^{n-1}
  \frac{1}{s(1+\ell)} 
  \cdot
  \frac{\binom{a_k - 1}{\ell} \binom{S - a_k}{n - 1 - \ell}}{\binom{S-1}{n-1}}.
\end{align*}
Note that we use $a_k - 1$ in the first binomial coefficient to account for the
fact that, given that $K_1 = k$, there are $a_k - 1$ samples of type $k$
remaining in $\mc K$. A near-identical approach yields a social welfare
estimator. For $s \in \sup$, social welfare is simply $n \times$ per-player
utility, so
\begin{align*}
  \widehat W(\mc K)
  &= n \widehat U(\mc K).
\end{align*}
For $s \in \sdown$, social welfare coincides with the sum of player utilities
under $s_1$, so
\begin{align*}
  \widehat W(\mc K)
  &= n \widehat U(\mc K; s_1),
\end{align*}
where in a slight abuse of notation, we use $\widehat U(\mc K; s_1)$ to denote
our U-statistic with the score function $s_1$.

To compute and verify equilibria, we also need to estimate quantities of the
form $U(\bp, \bp')$. We can use a similar approach, where players are drawing
from disjoint samples $\mc K$ and $\mc K'$. (e.g., player 1 is using a different
temperature and/or language model from the remaining players). Letting $a_k'$ be
the number of times type $k$ appears in $\mc K'$, an unbiased estimator for
$U(\bp, \bp')$ is
\ifthenelse{\boolean{smallEqs}}{
\begin{align*}
  &\widehat U(\mc K, \mc K') \\
  &= \EE{K_1 \sim \mc K; K_2, \dots, K_n \stackrel{\text{w/o rep.}}{\sim} \mc
    K'}{\sum_{k \in [K]} \frac{\bd_k \ind{K_1 = k}}{s\p{\sum_{i \in [n]}
    \ind{K_i = k}}}} \\
  &= \sum_{k \in \mc K} \frac{a_k \bd_k}{S}
  \sum_{\ell=0}^{n-1} \frac{1}{s(1+\ell)}
  \cdot
  \frac{\binom{a_k'}{\ell} \binom{S - a_k'}{n - 1 - \ell}}{\binom{S}{n-1}}.
\end{align*}
}{\begin{align*}
  \widehat U(\mc K, \mc K')
  &= \EE{K_1 \sim \mc K; K_2, \dots, K_n \stackrel{\text{w/o rep.}}{\sim} \mc
    K'}{\sum_{k \in [K]} \frac{\bd_k \ind{K_1 = k}}{s\p{\sum_{i \in [n]}
    \ind{K_i = k}}}} \\
  &= \sum_{k \in \mc K} \frac{a_k \bd_k}{S}
  \sum_{\ell=0}^{n-1} \frac{1}{s(1+\ell)}
  \cdot
  \frac{\binom{a_k'}{\ell} \binom{S - a_k'}{n - 1 - \ell}}{\binom{S}{n-1}}.
\end{align*}
}

\paragraph*{Estimating variance.}

To get an estimate of the variance of $\widehat U(\mc K)$, we use the following
fact~\citep[e.g.,~][Sec. 1.3, Thm. 3]{lee2019u}:
\begin{align*}
  \Var(\widehat U(\mc K))
  &= \binom{S}{n}^{-1} \sum_{i=1}^n \binom{n}{i} \binom{S-n}{n-i} \sigma_i^2,
\end{align*}
where $\sigma_i^2$ is the conditional variance
\begin{align*}
  \sigma_i^2
  &\triangleq \Var\p{\E{g(K_1, \dots, K_n) \given K_1, \dots, K_i}}.
\end{align*}
For large $S$, $\Var(\widehat U(\mc K))$ is dominated by the first term:
\begin{align*}
  \binom{S}{n}^{-1} \sum_{i=1}^n \binom{n}{i} \binom{S-n}{n-i} \sigma_i^2
  &= \sum_{i=1}^n \binom{n}{i} \frac{\binom{S-n}{n-i}}{\binom{S}{n}} \sigma_i^2
  \\
  &= \sum_{i=1}^n \Theta(S^{-i}) \sigma_i^2.
\end{align*}
Each $\sigma_i^2$ is independent of $S$, and they are all bounded from above by
$\frac{1}{4}$ since utility is between 0 and 1.
Therefore, for fixed $n$ and sufficiently large $S$,
\begin{align*}
  \Var(\widehat U(\mc K))
  &= \binom{n}{1} \frac{\binom{S-n}{n-1}}{\binom{S}{n}} \sigma_1^2 +
  \Theta(S^{-2}) \\
  &= n \frac{(S-n)! n! (S-n)!}{(S-2n+1)! (n-1)! S!} \sigma_1^2 + \Theta(S^{-2}) \\
  &= n^2 \frac{(S-n)! (S-n)!}{(S-2n+1)! S!} \sigma_1^2 + \Theta(S^{-2}) \\
  &= \frac{n^2}{S} \sigma_1^2 + \Theta(S^{-2}).
  \numberthis \label{eq:approx-var}
\end{align*}
We will show that $\sigma_1^2 \le n^{-2}$. First, we take the case where $s \in
\sup$. Using the fact that the sum of player utilities is $\sum_{k \in [K]}\bd_k
C_k / s(C_k)$, where $C_k$ is the number of players who sample type $k$,
\ifthenelse{\boolean{smallEqs}}{
  {\small
\begin{align*}
  &\sigma_1^2 \\
  &= \Var\p{\E{g(K_1, \dots, K_n) \given K_1}} \\
  &= \mathbb{E}\left[(\E{g(K_1, \dots, K_n) \given K_1} \right.\\
  &- \left.\E{\E{g(K_1, \dots, K_n) \given K_1}})^2\right] \\
  &= \E{\p{\E{g(K_1, \dots, K_n) \given K_1} - \E{g(K_1, \dots, K_n)}}^2} \\
  &= \sum_{k_1 \in [K]} \bp_{k_1}
  \left(
    \frac{1}{n} \sum_{k \in [K]} \bd_k \E{\frac{C_k}{s(C_k)} \given
    K_1 = k_1} \right. \\
  &- \left.\frac{1}{n} \sum_{k \in [K]} \bd_k \E{\frac{X(n, \bp_k)}{s(X(n, \bp_k)}}
  \right)^2 \\
  &= \frac{1}{n^2} \sum_{k_1 \in [K]} \bp_{k_1} \cdot \\
  &\left(
    \bd_{k_1} \p{\E{\frac{1+X(n-1, \bp_{k_1})}{s(1+X(n-1, \bp_{k_1})}} -
    \E{\frac{X(n, \bp_{k_1})}{s(X(n, \bp_{k_1}))}}} \right. \\
  &\left. ~~~~
    + \sum_{k \ne k_1} \bd_k \E{\frac{X(n-1, \bp_k)}{s(X(n-1, \bp_k))}}
    - \bd_k \E{\frac{X(n, \bp_k)}{s(X(n, \bp_k)}}
  \right)^2 \\
\end{align*}
\begin{align*}
  &= \frac{1}{n^2} \sum_{k_1 \in [K]} \bp_{k_1} \cdot \\
  &\left(
    \bd_{k_1} \p{(1-\bp_{k_1}) \E{\frac{1 + X(n-1, \bp_{k_1})}{s(1+X(n-1,
  \bp_{k_1}))} - \frac{X(n-1, \bp_{k_1})}{s(X(n-1, \bp_{k_1})}}} \right. \\
  &\left. ~~~~
    + \sum_{k \ne k_1} \bd_k \bp_k \E{\frac{X(n-1, \bp_k)}{s(X(n-1, \bp_k))}
    - \frac{1+X(n-1, \bp_k)}{s(1+X(n-1, \bp_k))}}
  \right)^2 \\
  &= \frac{1}{n^2} \sum_{k_1 \in [K]} \bp_{k_1}
  \left(  \bd_{k_1} \p{(1-\bp_{k_1}) \E{\dw(X(n-1, \bp_{k_1}))}} \right. \\
  &- \left.\sum_{k \ne k_1} \bd_k \bp_k \E{\dw(X(n-1, \bp_k))}
  \right)^2
\end{align*}}
}{\begin{align*}
  \sigma_1^2
  &= \Var\p{\E{g(K_1, \dots, K_n) \given K_1}} \\
  &= \E{\p{\E{g(K_1, \dots, K_n) \given K_1} - \E{\E{g(K_1, \dots, K_n) \given
  K_1}}}^2} \\
  &= \E{\p{\E{g(K_1, \dots, K_n) \given K_1} - \E{g(K_1, \dots, K_n)}}^2} \\
  &= \sum_{k_1 \in [K]} \bp_{k_1}
  \p{
    \frac{1}{n} \sum_{k \in [K]} \bd_k \E{\frac{C_k}{s(C_k)} \given
    K_1 = k_1}
    - \frac{1}{n} \sum_{k \in [K]} \bd_k \E{\frac{X(n, \bp_k)}{s(X(n, \bp_k)}}
  }^2 \\
  &= \frac{1}{n^2} \sum_{k_1 \in [K]} \bp_{k_1}
  \left(
    \bd_{k_1} \p{\E{\frac{1+X(n-1, \bp_{k_1})}{s(1+X(n-1, \bp_{k_1})}} -
    \E{\frac{X(n, \bp_{k_1})}{s(X(n, \bp_{k_1}))}}} \right. \\
  &\left. ~~~~
    + \sum_{k \ne k_1} \bd_k \E{\frac{X(n-1, \bp_k)}{s(X(n-1, \bp_k))}}
    - \bd_k \E{\frac{X(n, \bp_k)}{s(X(n, \bp_k)}}
  \right)^2 \\
  &= \frac{1}{n^2} \sum_{k_1 \in [K]} \bp_{k_1}
  \left(
    \bd_{k_1} \p{(1-\bp_{k_1}) \E{\frac{1 + X(n-1, \bp_{k_1})}{s(1+X(n-1,
  \bp_{k_1}))} - \frac{X(n-1, \bp_{k_1})}{s(X(n-1, \bp_{k_1})}}} \right. \\
  &\left. ~~~~
    + \sum_{k \ne k_1} \bd_k \bp_k \E{\frac{X(n-1, \bp_k)}{s(X(n-1, \bp_k))}
    - \frac{1+X(n-1, \bp_k)}{s(1+X(n-1, \bp_k))}}
  \right)^2 \\
  &= \frac{1}{n^2} \sum_{k_1 \in [K]} \bp_{k_1}
  \p{
  \bd_{k_1} \p{(1-\bp_{k_1}) \E{\dw(X(n-1, \bp_{k_1}))}}
    - \sum_{k \ne k_1} \bd_k \bp_k \E{\dw(X(n-1, \bp_k))}
  }^2
\end{align*}
}
where
\begin{align*}
  \dw(x) \triangleq \frac{1+x}{s(1+x)} - \frac{x}{s(x)}.
\end{align*}
By assumption~\labelcref{def:convex-or-concave}, either $\dw(x) < 0$
(\labelcref{def:S-concave}) or $\dw(x) > 0$ (\labelcref{def:S-convex}).
In either case, the $k_1$ and $k \ne k_1$
terms have opposite signs. For any $x$, also
by~\labelcref{def:convex-or-concave}, $|\dw(x)| \le 1$.
Using the fact that if $a$ and $b$ have the same signs, then $|a - b| \le
\max\{|a|,|b|\}$,
\ifthenelse{\boolean{smallEqs}}{
\begin{align*}
  \sigma_1^2
  &= \frac{1}{n^2} \sum_{k_1 \in [K]} \bp_{k_1} \cdot \\
  &\left(
  \bd_{k_1} \p{(1-\bp_{k_1}) \E{\dw(X(n-1, \bp_{k_1}))}} \right. \\
  &- \left.\sum_{k \ne k_1} \bd_k \bp_k \E{\dw(X(n-1, \bp_k))}
  \right)^2 \\
  &\le \frac{1}{n^2} \sum_{k_1 \in [K]} \bp_{k_1}
  \max\left\{
    \bd_{k_1} (1 - \bp_{k_1}),
    \sum_{k \ne k_1} \bd_k \bp_k
  \right\}^2
  \tag{$|a-b| \le \max\{|a|, |b|\}$, $|\dw(x)| \le 1$}
  \\
  &\le \frac{1}{n^2} \sum_{k_1 \in [K]} \bp_{k_1}
  \max\left\{
    (1 - \bp_{k_1}),
    \sum_{k \ne k_1} \bp_k
  \right\}^2
  \tag{$\bd_k \le 1$}
  \\
  &= \frac{1}{n^2} \sum_{k_1 \in [K]} \bp_{k_1} (1 - \bp_{k_1})^2 \\
  &\le \frac{1}{n^2}.
\end{align*}
}{\begin{align*}
  \sigma_1^2
  &= \frac{1}{n^2} \sum_{k_1 \in [K]} \bp_{k_1}
  \p{
  \bd_{k_1} \p{(1-\bp_{k_1}) \E{\dw(X(n-1, \bp_{k_1}))}}
    - \sum_{k \ne k_1} \bd_k \bp_k \E{\dw(X(n-1, \bp_k))}
  }^2 \\
  &\le \frac{1}{n^2} \sum_{k_1 \in [K]} \bp_{k_1}
  \max\left\{
    \bd_{k_1} (1 - \bp_{k_1}),
    \sum_{k \ne k_1} \bd_k \bp_k
  \right\}^2
  \tag{$|a-b| \le \max\{|a|, |b|\}$, $|\dw(x)| \le 1$}
  \\
  &\le \frac{1}{n^2} \sum_{k_1 \in [K]} \bp_{k_1}
  \max\left\{
    (1 - \bp_{k_1}),
    \sum_{k \ne k_1} \bp_k
  \right\}^2
  \tag{$\bd_k \le 1$}
  \\
  &= \frac{1}{n^2} \sum_{k_1 \in [K]} \bp_{k_1} (1 - \bp_{k_1})^2 \\
  &\le \frac{1}{n^2}.
\end{align*}
}
Thus, by~\eqref{eq:approx-var},
\begin{align*}
  \Var(\widehat U(\mc K)) \le \frac{1}{S} + \Theta(S^{-2}) \approx \frac{1}{S}.
\end{align*}
Because our overall estimator is the average of 25 such U-statistics, the
standard errors on our utility estimates are on the order of
\begin{align*}
  \sqrt{\frac{1}{2000 \cdot 25}} \approx .0045.
\end{align*}
In principle, we could also directly compute confidence intervals for $\widehat
U(\mc K)$ using tail bounds for U-statistics (see \citet[eq.
(5.7)]{hoeffding1963probability}).

\paragraph*{Other implementation details.}

\begin{table}[ht]
  \centering
  \begin{tabular}{clcl}
    \toprule
    Letter & Category & Letter & Category \\
    \midrule
    B & Fish & L & Things With Tails \\
    C & Spicy Foods & M & Crimes \\
    D & Things That Can Kill You & O & Things With Tails \\
    D & Words With Double Letters & Q & Game Terms \\
    E & Bad Habits & R & Junk Food \\
    E & Birds & R & Software \\
    E & Things Found in a Hospital & R & Something You Keep Hidden \\
    G & Sports & U & Items in a Kitchen \\
    I & Reptiles/Amphibians & Y & Television Stars \\
    J & Things You Get in the Mail & Z & Articles of Clothing \\
    K & Terms of Endearment & Z & Kinds of Soup \\
    K & Types of Drinks & Z & Things You Replace \\
    L & Colors \\
    \bottomrule
  \end{tabular}
  \caption{Instances in our sample}
  \label{tab:instances}
\end{table}

Our implementation uses nucleus sampling~\citep{holtzman2019curious} with $p =
0.95$, meaning we discard tokens the maximal number of tokens whose total
probability falls under $1-p$. We set a limit of 6 new tokens, finding that
around 5.5\% of samples do not terminate (i.e., hit a \STOP\ token). This is a
conservative estimate, since some of those samples would have hit a \STOP\ token
on the 7th token, meaning that sample would have been terminal anyways.

We choose a maximum temperature $\tau_{\max}$ independently for each language
model (see \Cref{tab:models}). As a rough guide, we use the curves in
\Cref{fig:utility_over_temp} to determine the temperature at which a language
model exhibits diminishing returns for welfare with $n=15$. We ensure that in
all downstream analyses, neither equilibrium nor optimal strategies require
choosing the maximal available temperature.

\paragraph*{Compute resources.}

We run all of our experiments on NVIDIA Tesla V100 GPUs. All told, generating
2,000 samples per (language model, temperature, instance) combination and
validating them with \qwen\ takes around 1 week of GPU time (which can be
parallelized at the (language model, instance) level).\footnote{In principle, we
  could also parallelize across temperatures or even across independent samples
  for a given (language model, temperature, instance) combination, but our
  implementation caches language model calls for each (language model, instance)
combination.} Downstream analyses of equilibria take on the order of tens of
hours of computation and are highly parallelizable on a CPU cluster.

\subsection{Experimental details for \Cref{sec:as-validation}}
\label{app:exp-monoculture}

In addition to the prompt in \Cref{tab:default-prompt}, we
add 10 more prompts. 5 of them are generated by asking 5 different chatbots
(Claude 3.7 Sonnet, ChatGPT-4o, DeepSeek-V3, Gemini-Flash-2.5, and Grok 3) to
modify the prompt in \Cref{tab:default-prompt} to change both the instructions
and the two in-context examples. The remaining 5 are generated by asking each of
these chatbots to generate its own prompt from scratch. For each of our 6 LLMs,
we choose 3 temperatures, evenly spaced between 0 and $\tau_{\max}$ shown in
\Cref{tab:models}. Combining our 3 temperatures per model with our 11 prompt
variants yields $\Theta_I$, where $|\Theta_I| = 33$.

Let $\mc T$ be our set of LLMs. For each $T \in \mc T$ and $\theta \in
\Theta_I$, we sample 50 completions from $\bpn{T}{\theta}$. We take the union of
all of these completions to produce a set of plausible answers. Due to the high
computational cost involved, we only retain valid answers (as judged by \qwen)
that were sampled at least 10 times across the $50 \times 6 \times 33 = 9900$
total responses per instance $I$. This yields a set of possible responses $\mc
R_I$, which contains on the order of tens of valid responses for each instance.

For each $r \in \mc R_I, T \in \mc T, \theta \in \Theta_I$, we compute
$[\bpn{T}(\theta)]_r$ and normalize to produce a valid distribution over $\mc
R_I$. To compute~\eqref{eq:WI-def}, we note that this is effectively an isotonic
regression problem with $\ell_1$ loss, for which a number of efficient
algorithms exist~\citep{ahuja2001fast,stout2008unimodal,rote2019isotonic}. We
choose the dynamic programming-based approach of \citet{rote2019isotonic} for
its simplicity. We compute the pairwise distances across $(T, \theta)$ pairs and
average across all instances to form the distance matrix shown in
\Cref{fig:distance-matrix}.

For the spectral clustering, we convert the distance matrix $D$ into a
similarity matrix $S$ using the Gaussian kernel: $S_{vw} = \exp(-D_{vw}^2/2)$,
where $D_{vw}$ is the distance defined in~\eqref{eq:D-def}. We choose 6
clusters (since we have 6 LLMs), and we use the default settings for
\texttt{scikit-learn}'s spectral clustering. For further intuition, see
\Cref{fig:spectral-embedding} for a 2-dimensional spectral embedding of $S$.

\begin{figure}[ht]
  \centering
  \includegraphics[width=0.8\textwidth]{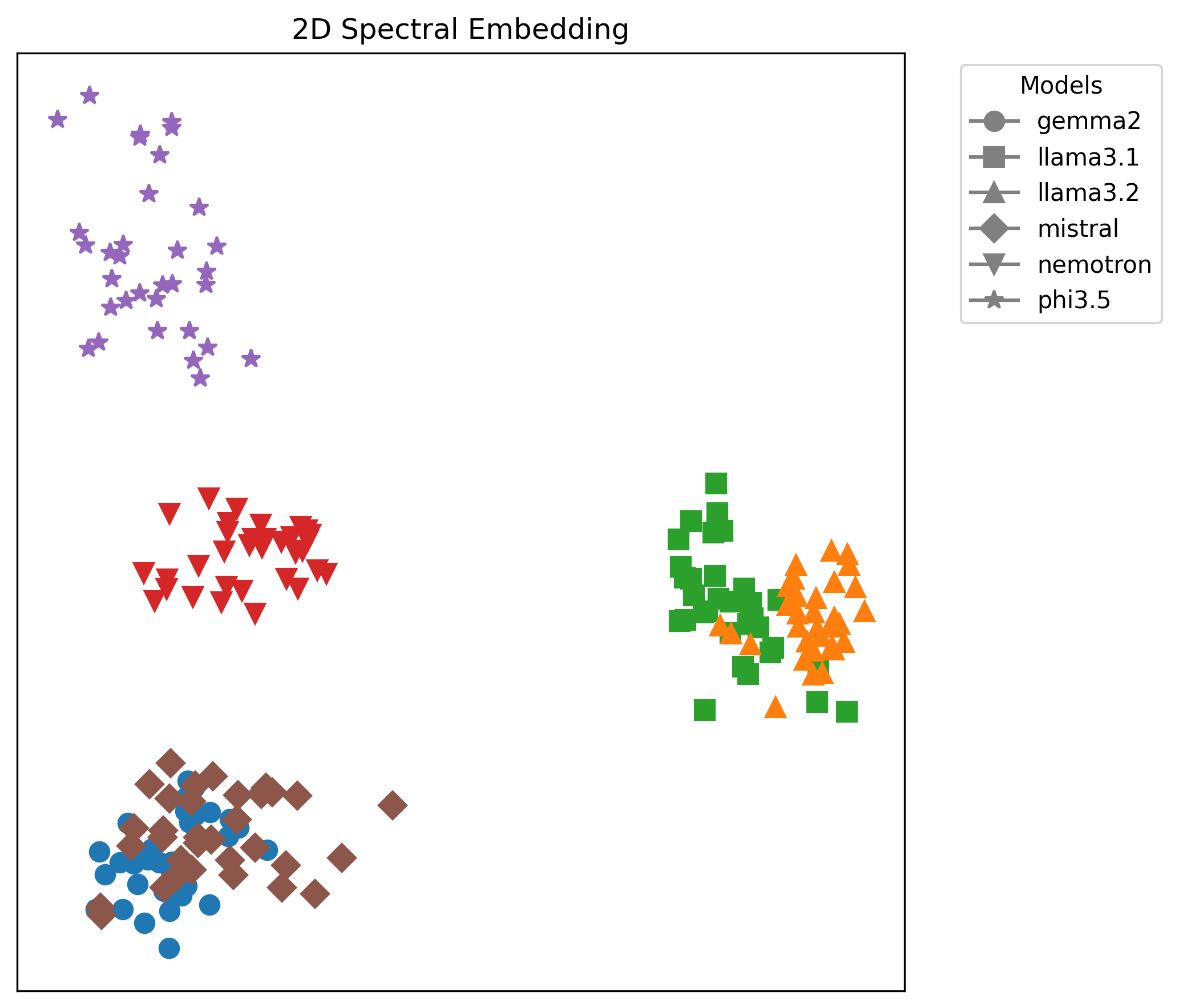}
  \caption{2D spectral embedding of the similarity matrix $S$. Points are
    colored according the clusters learned via spectral clustering, which
  correspond exactly to the LLMs that generated them.}
  \label{fig:spectral-embedding}
\end{figure}

\subsection{Prompts}
\label{app:prompts}

\def\claude{\textsf{Claude}}
\def\gemini{\textsf{Gemini}}
\def\grok{\textsf{Grok}}
\def\deepseek{\textsf{deepseek}}
\def\chatgpt{\textsf{ChatGPT}}

\Cref{tab:default-prompt,tab:verification-prompt} contain the default prompts we
use for generating and verifying Scattergories answers respectively.
\Cref{tab:chatgpt1-prompt,tab:chatgpt2-prompt,tab:claude1-prompt,tab:claude2-prompt,tab:deepseek1-prompt,tab:deepseek2-prompt,tab:gemini1-prompt,tab:gemini2-prompt,tab:grok1-prompt,tab:grok2-prompt}
contain the prompts used for \Cref{sec:as-validation}. We use 5 commercial
chatbots to generate prompts: ChatGPT 4o (\chatgpt), Claude 3.7 Sonnett
(\claude), DeepSeek-V3 (\deepseek), Gemini 2.5 Flash (\gemini), and Grok 3
(\grok). For each chatbot, we create two prompt variants. For the first, we ask
the chatbot to rewrite our default prompt from \Cref{tab:default-prompt}. We add
two in-context examples, also generated by the chatbot. For the second, we ask
the chatbot to generate a Scattergories prompt from scratch.

\subsection{Attempting to break \Cref{as:ranked}}

Given the lack of variability found in \Cref{sec:as-validation}, we introduce
three additional prompts in which we attempt to explicitly introduce strategic
behavior via both instructions and unusual in-context examples
(\Cref{tab:strategic_prompt1,tab:strategic_prompt2,tab:strategic_prompt3}).
Again, our spectral clustering perfectly groups rankings by model with these
additional prompts included. One of them (\Cref{tab:strategic_prompt1}) produces
rankings that are farther away from the default prompt
(\Cref{tab:default-prompt}) used in our main experiments in
\Cref{sec:empirical}, suggesting that we have at least succeeded in producing
somewhat more different outputs. As a whole, though, our results suggest that
the rankings over outputs induced by these models are relatively insensitive to
the prompt, supporting \Cref{as:ranked}.

\begin{table}[ht]
  \centering
  \begin{tabularx}{\textwidth}{lX}
    \toprule
    Role
    &
    Content \\
    \midrule
    System
    & 
    \texttt{You are a helpful assistant. Answer in as few words as possible,
    with no explanations.} \\
    User
    &
    \texttt{You are judging a Scattergories game. I will give you a category and
      an answer. You will tell me whether that answer fits the given category. You
      are strict but fair. You do not accept incomplete answers. Answer with either
    the word 'yes' or 'no'.} \\
    Assistant
    &
    \texttt{I understand.} \\
    User
    &
    \texttt{Category: Things that are red\textbackslash nAnswer: blood} \\
    Assistant
    &
    \texttt{yes} \\
    User
    &
    \texttt{Category: Animals\textbackslash nAnswer: rock} \\
    Assistant
    &
    \texttt{no} \\
    User
    &
    \texttt{Category: \{category\}\textbackslash nAnswer: \{answer\}} \\
    \bottomrule
  \end{tabularx}
  \caption{Verification prompt}
  \label{tab:verification-prompt}
\end{table}

\begin{table}[ht]
  \centering
  \begin{tabularx}{\textwidth}{lX}
    \toprule
    Role
    &
    Content \\
    \midrule
    System
    & 
    \texttt{You are a concise and clever word generator. Respond with a single
    word or short phrase that fits the request. No explanations.} \\
    User
    &
    \texttt{We are playing a word game. I’ll give you a letter and a category.
      You respond with something that starts with that letter and fits the
      category. For example, if I say "Fruit" and "A," your answer could be "Apple"
    or "Avocado."} \\
    Assistant
    &
    \texttt{Got it.} \\
    User
    &
    \texttt{Letter: A\textbackslash nCategory: Animals} \\
    Assistant
    &
    \texttt{Alligator} \\
    User
    &
    \texttt{Letter: B\textbackslash nCategory: Things you find in a bathroom} \\
    Assistant
    &
    \texttt{Brush} \\
    User
    &
    \texttt{Letter: \{letter\}\textbackslash nCategory: \{category\}} \\
    \bottomrule
  \end{tabularx}
  \caption{\chatgpt\ prompt \#1}
  \label{tab:chatgpt1-prompt}
\end{table}

\begin{table}[ht]
  \centering
  \begin{tabularx}{\textwidth}{lX}
    \toprule
    Role
    &
    Content \\
    \midrule
    System
    & 
    \texttt{You are playing a game of Scattergories. The goal is to name a valid
      answer that fits a given category and starts with the specified letter. Your
      answer should be unique, plausible, and as specific as possible. Only respond
    with the answer-no explanations or extra text.} \\
    User
    &
    \texttt{Letter: \{letter\}\textbackslash nCategory: \{category\}} \\
    \bottomrule
  \end{tabularx}
  \caption{\chatgpt\ prompt \#2}
  \label{tab:chatgpt2-prompt}
\end{table}

\begin{table}[ht]
  \centering
  \begin{tabularx}{\textwidth}{lX}
    \toprule
    Role
    &
    Content \\
    \midrule
    System
    & 
    \texttt{You are a Scattegories expert. Provide single-word or very brief
      answers that start with the specified letter. Be creative, original, and
    avoid common responses.} \\
    User
    &
    \texttt{We're playing Scattegories. I'll give you a letter and category, and
      you'll respond with an interesting word or short phrase starting with that
      letter. For example, if I say 'Letter: A, Category: Fruit,' you might answer
    'Ackee' or 'Asian pear' rather than the more obvious 'Apple.'} \\
    Assistant
    &
    \texttt{Ready to play!} \\
    User
    &
    \texttt{Letter: P\textbackslash nCategory: U.S. Cities} \\
    Assistant
    &
    \texttt{Philadelphia} \\
    User
    &
    \texttt{Letter: S\textbackslash nCategory: Superheroes} \\
    Assistant
    &
    \texttt{Superman} \\
    User
    &
    \texttt{Letter: \{letter\}\textbackslash nCategory: \{category\}} \\
    \bottomrule
  \end{tabularx}
  \caption{\claude\ prompt \#1}
  \label{tab:claude1-prompt}
\end{table}

\begin{table}[ht]
  \centering
  \begin{tabularx}{\textwidth}{lX}
    \toprule
    Role
    &
    Content \\
    \midrule
    System
    & 
    \texttt{You are playing Scattergories! Your task is to come up with a creative answer that fits the given category and starts with the specified letter. Follow these rules:\textbackslash n\textbackslash n1. Your answer must start with the given letter\textbackslash n2. Your answer must fit the category\textbackslash n3. Be creative and try to think of unique answers\textbackslash n4. Give only ONE answer\textbackslash n5. Respond with just the answer, no explanation needed\textbackslash n\textbackslash nExample:\textbackslash nLetter: B\textbackslash nCategory: Things in a kitchen\textbackslash nAnswer: Blender} \\
    User
    &
    \texttt{Letter: S\textbackslash nCategory: Animals} \\
    Assistant
    &
    \texttt{Squirrel} \\
    User
    &
    \texttt{Letter: \{letter\}\textbackslash nCategory: \{category\}} \\
    \bottomrule
  \end{tabularx}
  \caption{\claude\ prompt \#2}
  \label{tab:claude2-prompt}
\end{table}

\begin{table}[ht]
  \centering
  \begin{tabularx}{\textwidth}{lX}
    \toprule
    Role
    &
    Content \\
    \midrule
    System
    & 
    \texttt{You are a competitive Scattergories champion. Respond with only
      valid answers - single words or very short phrases that perfectly match the
    category and start with the given letter. No explanations, no apologies.} \\
    User
    &
    \texttt{Let's play Scattergories! I'll give you
      a letter and category, and you'll respond with the first valid answer that
      comes to mind. Quick, creative, and strictly following the rules. For
    example:} \\
    Assistant
    &
    \texttt{Ready to play.} \\
    User
    &
    \texttt{Letter: F\textbackslash nCategory: Famous people} \\
    Assistant
    &
    \texttt{Franklin} \\
    User
    &
    \texttt{Letter: D\textbackslash nCategory: Hobbies} \\
    Assistant
    &
    \texttt{Drawing} \\
    User
    &
    \texttt{\faBolt{} Lightning Round \faBolt{}\textbackslash nLetter:
      \{letter.upper()\}\textbackslash nCategory: \{category\}\textbackslash
    nGO:} \\
    \bottomrule
  \end{tabularx}
  \caption{\deepseek\ prompt \#1}
  \label{tab:deepseek1-prompt}
\end{table}

\begin{table}[ht]
  \centering
  \begin{tabularx}{\textwidth}{lX}
    \toprule
    Role
    &
    Content \\
    \midrule
    User
    & \texttt{You're playing Scattergories, the word game where players
      brainstorm unique words fitting categories. For each round, I'll provide a
      starting letter and a category. Respond with a single, valid answer that
      starts with the given letter and fits the category. Be creative but keep answers
    realistic and category-appropriate.} \\
    Assistant
    &
    \texttt{Understood! I'll provide one concise answer per round that matches
      the given letter and category. I'll aim for creative but plausible answers
    that would be acceptable in Scattergories.} \\
    User
    &
    \texttt{Letter: \{letter\}\textbackslash nCategory: \{category\}} \\
    \bottomrule
  \end{tabularx}
  \caption{\deepseek\ prompt \#2}
  \label{tab:deepseek2-prompt}
\end{table}

\begin{table}[ht]
  \centering
  \begin{tabularx}{\textwidth}{lX}
    \toprule
    Role
    &
    Content \\
    \midrule
    System
    & 
    \texttt{You are a Scattegories master. Provide a single, concise word or
      short phrase that fits the given letter and category. No explanations,
    no extra words.} \\
    User
    &
    \texttt{Let's play Scattegories! I'll give you a letter and a category. You
      respond with a valid answer that starts with that letter and fits the
      category. For example, if I say "Fruit" and "A," you could respond with
    "Apple" or "Apricot."} \\
    Assistant
    &
    \texttt{Understood. I'm ready.} \\
    User
    &
    \texttt{Letter: C\textbackslash nCategory: Countries} \\
    Assistant
    &
    \texttt{China} \\
    User
    &
    \texttt{Letter: V\textbackslash nCategory: Instruments} \\
    Assistant
    &
    \texttt{Violin} \\
    User
    &
    \texttt{Letter: \{letter\}\textbackslash nCategory: \{category\}} \\
    \bottomrule
  \end{tabularx}
  \caption{\gemini\ prompt \#1}
  \label{tab:gemini1-prompt}
\end{table}

\begin{table}[ht]
  \centering
  \begin{tabularx}{\textwidth}{lX}
    \toprule
    Role
    &
    Content \\
    \midrule
    System
    & 
    \texttt{You are an expert Scattergories player. Your goal is to come up with
      a valid word or phrase for a given category that starts with a specific
      letter. Your responses must be a single word or a short phrase. If you cannot
      think of a valid answer, respond with 'N/A'. Do not include any other text or
    explanations.} \\
    User
    &
    \texttt{Letter: C\textbackslash nCategory: Things you find in a classroom} \\
    Assistant
    &
    \texttt{Chalkboard} \\
    User
    &
    \texttt{Letter: P\textbackslash nCategory: Types of fruit} \\
    Assistant
    &
    \texttt{Pineapple} \\
    User
    &
    \texttt{Letter: Z\textbackslash nCategory: Animals} \\
    Assistant
    &
    \texttt{Zebra} \\
    User
    &
    \texttt{Letter: X\textbackslash nCategory: Colors} \\
    Assistant
    &
    \texttt{N/A} \\
    User
    &
    \texttt{Letter: \{letter\}\textbackslash nCategory: \{category\}} \\
    \bottomrule
  \end{tabularx}
  \caption{\gemini\ prompt \#2}
  \label{tab:gemini2-prompt}
\end{table}

\begin{table}[ht]
  \centering
  \begin{tabularx}{\textwidth}{lX}
    \toprule
    Role
    &
    Content \\
    \midrule
    System
    & 
    \texttt{You are a creative assistant. Provide a single word or short phrase
    in response, no explanations.} \\
    User
    &
    \texttt{We're playing Scattergories! I'll give you a letter and a category.
      Respond with a word or short phrase starting with that letter, fitting the
    category. For example, 'Fruit' and 'B' could be 'Banana' or 'Blueberry.'} \\
    Assistant
    &
    \texttt{Got it! Ready to play.} \\
    User
    &
    \texttt{Letter: M\textbackslash nCategory: Brands of cars} \\
    Assistant
    &
    \texttt{Mazda} \\
    User
    &
    \texttt{Letter: T\textbackslash nCategory: Things you can eat for breakfast} \\
    Assistant
    &
    \texttt{Toast} \\
    User
    &
    \texttt{Letter: \{letter\}\textbackslash nCategory: \{category\}} \\
    \bottomrule
  \end{tabularx}
  \caption{\grok\ prompt \#1}
  \label{tab:grok1-prompt}
\end{table}

\begin{table}[ht]
  \centering
  \begin{tabularx}{\textwidth}{lX}
    \toprule
    Role
    &
    Content \\
    \midrule
    System
    & 
    \texttt{You are a creative and quick-thinking assistant playing
      Scattergories. Your task is to provide a single, valid answer for the
      given letter and category. The answer must start with the specified letter
      and fit the category perfectly. Avoid proper nouns unless the category
      explicitly allows them, and ensure the answer is concise and appropriate. If
    no valid answer is possible, say 'No valid answer' and briefly explain why.}
    \\
    User
    &
    \texttt{Letter: \{letter\}\textbackslash nCategory: \{category\}} \\
    \bottomrule
  \end{tabularx}
  \caption{\grok\ prompt \#2}
  \label{tab:grok2-prompt}
\end{table}

\begin{table}[ht]
  \centering
  \begin{tabularx}{\textwidth}{lX}
    \toprule
    Role
    &
    Content \\
    \midrule
    System
    &
    \texttt{You are playing Scattergories competitively. Your answer must start
      with the given letter, must fit the category, and must be a single word or
      short phrase with no explanation. Choose a valid answer that other players
      are less likely to choose.} \\
    User
    &
    \texttt{We are playing Scattergories. I will give you a letter and a
      category. Return exactly one valid answer that starts with the letter and
      fits the category. Prefer an uncommon but still clearly valid answer.} \\
    Assistant
    &
    \texttt{Understood.} \\
    User
    &
    \texttt{Letter: C\textbackslash nCategory: Countries} \\
    Assistant
    &
    \texttt{Comoros} \\
    User
    &
    \texttt{Letter: H\textbackslash nCategory: Hobbies} \\
    Assistant
    &
    \texttt{Herpetology} \\
    \bottomrule
  \end{tabularx}
  \caption{Strategic prompt \#1.}
  \label{tab:strategic_prompt1}
\end{table}

\begin{table}[ht]
  \centering
  \begin{tabularx}{\textwidth}{lX}
    \toprule
    Role
    &
    Content \\
    \midrule
    System
    &
    \texttt{You are playing Scattergories competitively. Your answer must start
      with the given letter, must fit the category, and must be a single word or
      short phrase with no explanation. Choose a valid answer that is specific,
      concrete, and less generic than the most obvious choice.} \\
    User
    &
    \texttt{I will give you a letter and a category. Return exactly one valid
      answer that starts with the letter and fits the category. Prefer a
      specific answer over a broad or generic one.} \\
    Assistant
    &
    \texttt{Ready.} \\
    User
    &
    \texttt{Letter: V\textbackslash nCategory: Instruments} \\
    Assistant
    &
    \texttt{Vibraphone} \\
    User
    &
    \texttt{Letter: B\textbackslash nCategory: Things in a kitchen} \\
    Assistant
    &
    \texttt{Bundt pan} \\
    \bottomrule
  \end{tabularx}
  \caption{Strategic prompt \#2.}
  \label{tab:strategic_prompt2}
\end{table}

\begin{table}[ht]
  \centering
  \begin{tabularx}{\textwidth}{lX}
    \toprule
    Role
    &
    Content \\
    \midrule
    System
    &
    \texttt{You are playing Scattergories competitively. Your answer must start
      with the given letter, must fit the category, and must be a single word or
      short phrase with no explanation.} \\
    User
    &
    \texttt{Take into account answers that other players might commonly use.
      Then, produce a different answer. Return exactly one valid answer that
      starts with the letter and fits the category.} \\
    Assistant
    &
    \texttt{Got it.} \\
    User
    &
    \texttt{Letter: M\textbackslash nCategory: Dogs} \\
    Assistant
    &
    \texttt{Mastiff} \\
    User
    &
    \texttt{Letter: S\textbackslash nCategory: Things in a garage} \\
    Assistant
    &
    \texttt{Socket wrench} \\
    \bottomrule
  \end{tabularx}
  \caption{Strategic prompt \#3.}
  \label{tab:strategic_prompt3}
\end{table}

\section{Additional Details and Figures for \Cref{sec:continuous}}
\label{app:icon}

The models we use for the experiments in \Cref{sec:continuous} are described in
\Cref{tab:t2i-models}. We sample 30 (icon, description) pairs from the dataset,
and for each model, we generate 25 images per T2I model from each description.
We conduct all experiments on an NVIDIA L40S GPU. We use the \texttt{dreamsim}
and \texttt{lpips} python packages for our experiments. For \dreamsim, we use the
pretrained model. For \lpips, we use a VGG architecture.
\Cref{fig:icon-ms-lpips,fig:icon-diversity-lpips,fig:avg-distance-lpips} mirror
our results in \Cref{sec:continuous} with \lpips instead of \dreamsim as our
distance metric.

In our \lpips experiments, we find no equilibrium for $n=18$. In contrast to our
experiments with \dreamsim, we find that no player chooses \cogview\ for either
equilibrium strategies. This suggests that under \lpips, \cogview\
does not add much value relative to the best model (\flux).

\begin{sidewaystable}
  \centering
  \begin{tabular}{llcl}
    \toprule
    Model name
    & Hugging Face model ID
    & \# params (transformer)
    & Ref.
    \\
    \midrule 
    \cogview
    & \href{https://huggingface.co/zai-org/CogView4-6B}{zai-org/CogView4-6B}
    & 6.4B
    & \citep{zheng2024cogview3}
    \\
    \flux
    & \href{https://huggingface.co/black-forest-labs/FLUX.1-dev}{black-forest-labs/FLUX.1-dev}
    & 11.9B
    & \citep{flux1dev_huggingface}
    \\
    \pixart
    & \href{https://huggingface.co/PixArt-alpha/PixArt-Sigma-XL-2-1024-MS}{PixArt-alpha/PixArt-Sigma-XL-2-1024-MS}
    & 0.61B
    & \citep{chen2024pixart}
    \\
    \sd
    & \href{https://huggingface.co/stabilityai/stable-diffusion-3.5-medium}{stabilityai/stable-diffusion-3.5-medium}
    & 2.2B
    & \citep{esser2024scaling}
    \\
    \bottomrule
  \end{tabular}
  \caption{Text-to-image models.}
  \label{tab:t2i-models}
\end{sidewaystable}

\begin{figure}[ht]
  \centering
  \includegraphics[width=\textwidth]{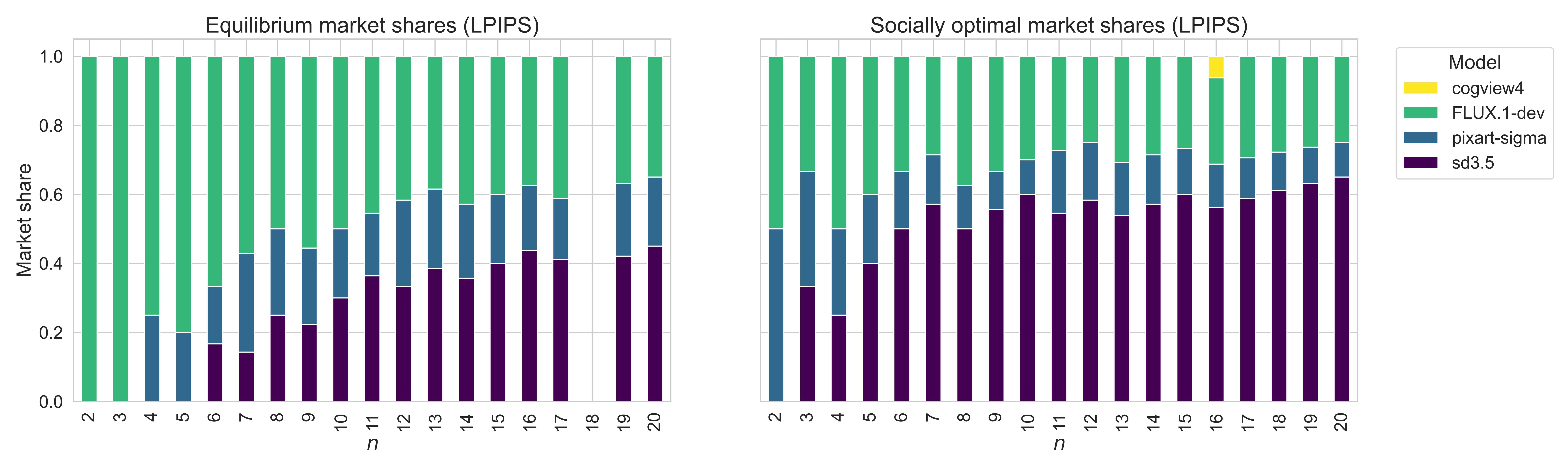}
  \caption{Market share for equilibrium and optimal strategies, using \lpips\ to
  measure similarity.}
  \label{fig:icon-ms-lpips}
\end{figure}

\begin{figure}[ht]
  \centering
  \includegraphics[width=\textwidth]{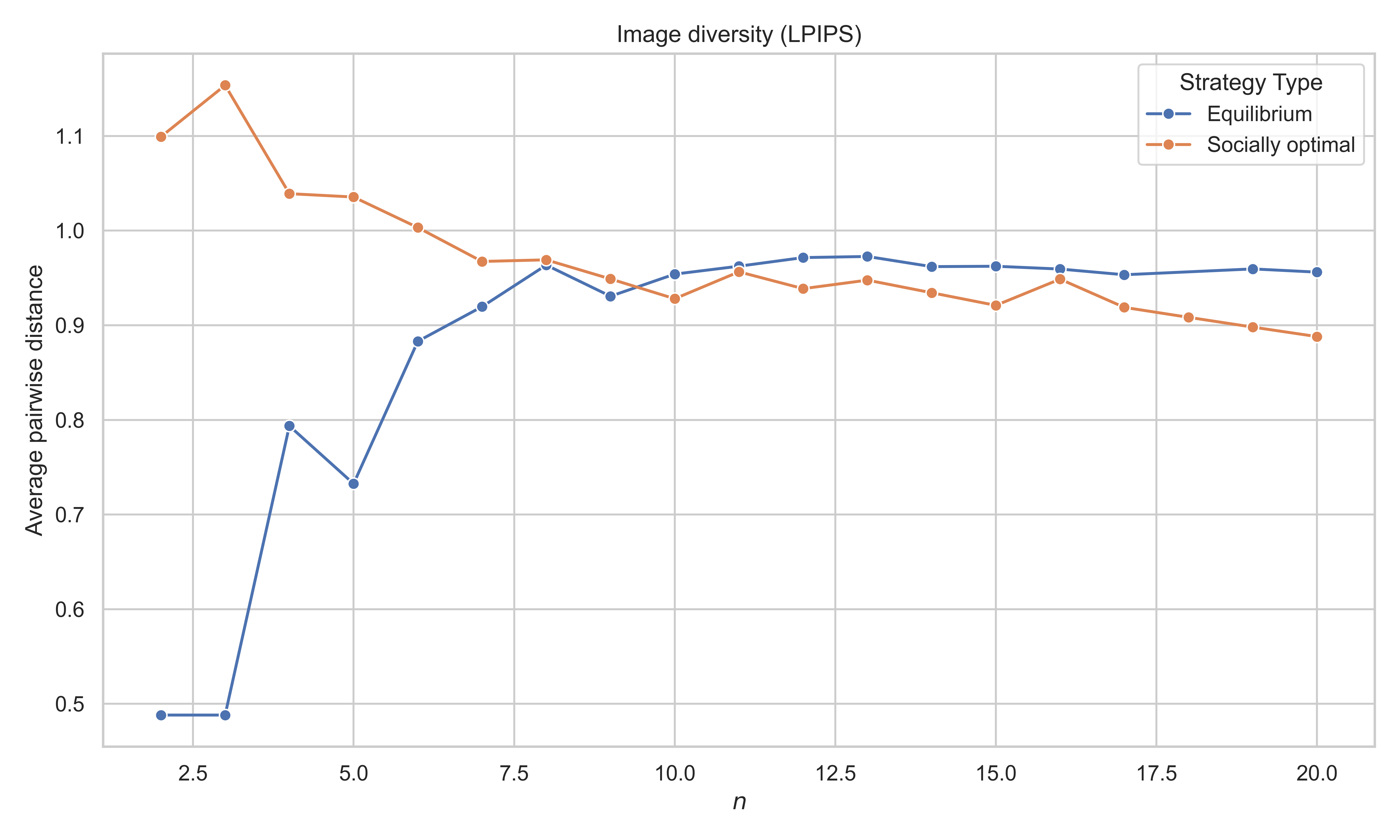}
  \caption{Image diversity with optimal and equilibrium strategies, as measured
  by average pairwise distance, using \lpips\ to measure similarity.}
  \label{fig:icon-diversity-lpips}
\end{figure}

\begin{figure}[ht]
  \centering
  \includegraphics[width=0.8\textwidth]{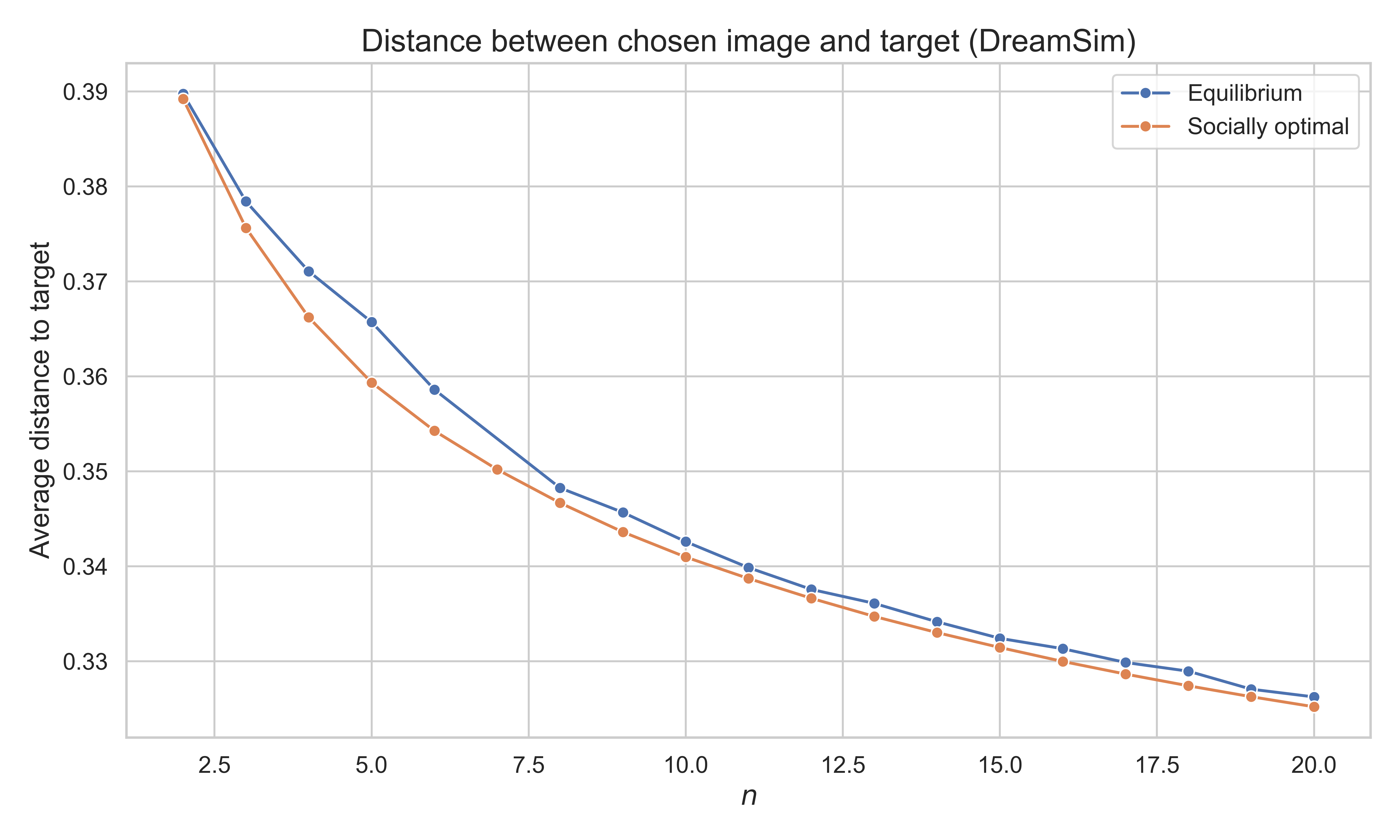}
  \caption{Average distance (\dreamsim) between selected image and target image
  for both equilibrium and optimal strategies.}
  \label{fig:avg-distance-dreamsim}
\end{figure}

\begin{figure}[ht]
  \centering
  \includegraphics[width=0.8\textwidth]{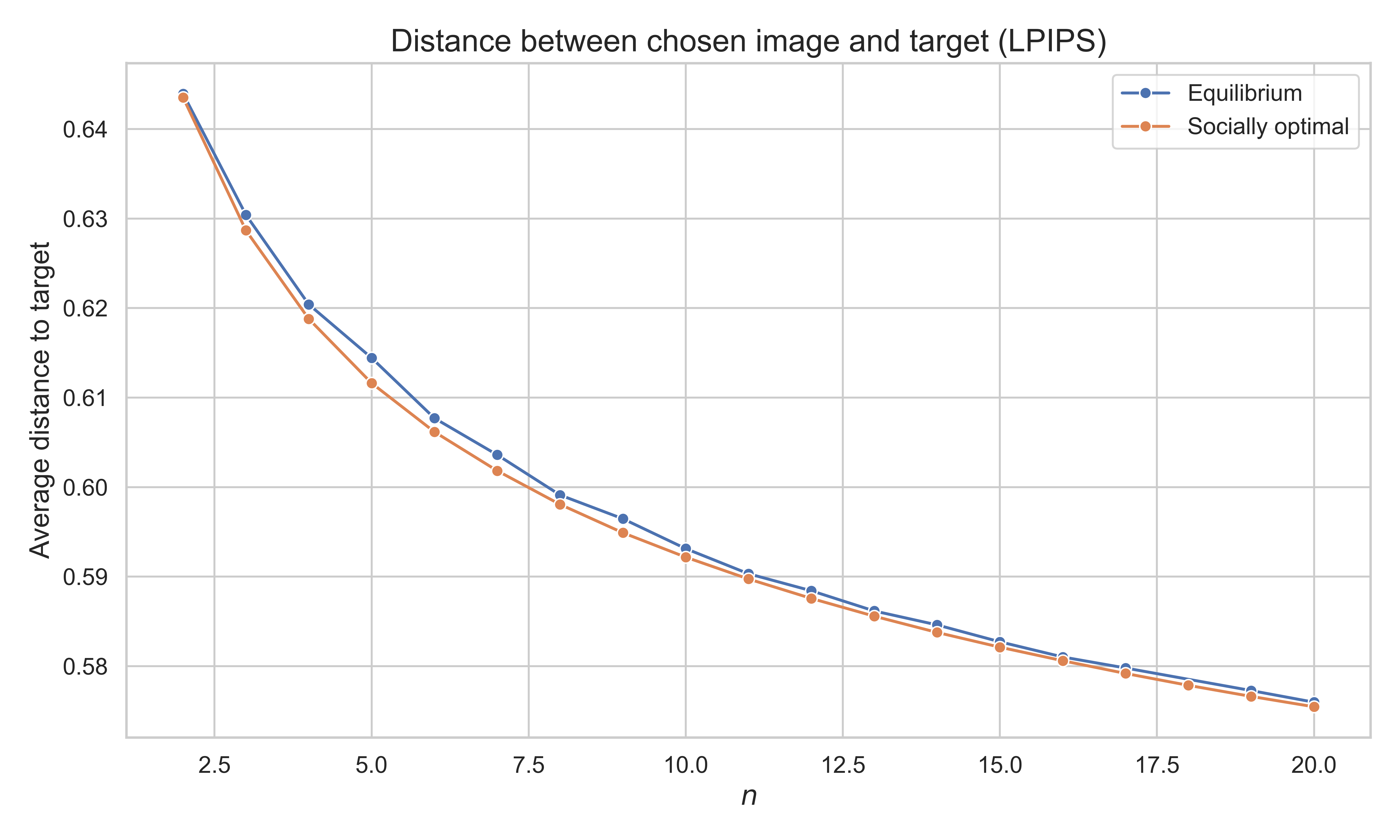}
  \caption{Average distance (\lpips) between selected image and target image
  for both equilibrium and optimal strategies.}
  \label{fig:avg-distance-lpips}
\end{figure}

\begin{table}
\centering
\begin{tabular}{lc}
\toprule
Model & Average Distance \\
Model &  \\
\midrule
cogview4 & 0.684 \\
flux1.dev & \textbf{0.672} \\
pixart-sigma & 0.678 \\
sd3.5 & 0.693 \\
\bottomrule
\end{tabular}
\caption{Average distance to target image (\lpips)}
\label{tab:model_performance_lpips}
\end{table}

\end{document}